\pdfoutput=1 % This tells arXiv to use pdflatex.

\documentclass[a4paper,USenglish]{lipics-v2021}

\nolinenumbers
\hideLIPIcs
\EventShortTitle{Preprint}

\title{A LOCAL View of the Polynomial Hierarchy}

\author{%
  Fabian Reiter%
}{%
  LIGM, Université Gustave Eiffel, Marne-la-Vallée, France%
}{%
  fabian.reiter@gmail.com%
}{%
  https://orcid.org/0000-0003-1268-4107%
}{}

\authorrunning{F. Reiter}

\Copyright{Fabian Reiter}



\relatedversion{%
  A preliminary version of this paper appeared in the
  \emph{Proceedings of the 43rd ACM Symposium on Principles of Distributed Computing}
  (\href{https://www.podc.org/podc2024}{PODC 2024})
  \cite{DBLP:conf/podc/Reiter24}.
}
\relatedversiondetails{\faFile*{} Conference paper}{https://doi.org/10.1145/3662158.3662774}

\supplement{%
  Video presentations based on earlier versions of this paper
  are available online.
  A short talk was presented at \href{https://www.podc.org/podc2024}{PODC 2024},
  and a longer one at the \href{https://adga-workshop.org/2024}{ADGA 2024} workshop.
}
\supplementdetails[subcategory={short talk}]{\faVideo{} Video}{https://doi.org/10.5281/zenodo.16316792}
\supplementdetails[subcategory={long talk}]{\faVideo{} Video}{https://doi.org/10.5281/zenodo.16311191}

\funding{%
  Partially supported by the ERC projects
  \href{http://www.lsv.fr/~bouyer/equalis}{\emph{EQualIS} (308087)} and
  \href{https://paves.model.in.tum.de/}{\emph{PaVeS} (787367)}.
}

\acknowledgements{%
  The author gratefully acknowledges the anonymous reviewers
  for their helpful comments and suggestions.
}

\usepackage{sty/main}        % Package loading & main configuration.
\usepackage{sty/colors}      % Colors used throughout the document.
\usepackage{sty/utilities}   % Some useful helper commands.
\usepackage{sty/notation}    % Commands for mathematical notation.
\usepackage{sty/terminology} % Database used by the 'knowledge' package.

\begin{document}

\pagenumbering{roman}
\maketitle
\begin{abstract}
  We extend classical methods of computational complexity
  to the realm of distributed computing,
  where they sometimes prove more effective than in their original context.
  Our focus is on decision problems in the \textsc{local} model,
  a setting in which networked computers use synchronous message passing
  to collectively answer questions about their network topology.
  We impose two time constraints on this model:
  the number of \kl{communication rounds} is bounded by a constant,
  and the number of \kl{computation steps} of each computer
  is polynomially bounded in
  the size of its local input and received messages.

  By letting two players alternately
  assign \kl{certificates} to all computers,
  we obtain a distributed generalization of the \kl{polynomial hierarchy}
  (and thus of the complexity classes $\PTIME$~and~$\NP$).
  We then extend key results of complexity theory to this setting,
  including the Cook\nobreakdash--Levin theorem
  (which identifies \kl{Boolean satisfiability}
  as a \kl{complete} problem for~$\NP$)
  and Fagin's~theorem
  (which characterizes~$\NP$ as the class of problems expressible in
  \kl[existential fragment]{existential second-order logic}).
  The original results can be recovered as the special case
  where the network consists of a single computer.

  But perhaps more surprisingly,
  separating complexity classes
  becomes easier in the distributed setting:
  we can show that
  \kl[locally polynomial hierarchy]{our hierarchy} is \emph{infinite},
  while it remains notoriously open
  whether the same holds when restricted to a single computer.
  (By contrast,
  a collapse of \kl[locally polynomial hierarchy]{our hierarchy}
  would have implied
  a collapse of the classical \kl{polynomial hierarchy}.)

  As an application,
  we propose quantifier alternation as a new tool
  for measuring the locality of problems in distributed computing.
\end{abstract}

%%% Local Variables:
%%% mode: latex
%%% TeX-master: "../lph-paper"
%%% End:

\KnowledgeFootnote{}
\newpage
\tableofcontents

\newpage
\pagenumbering{arabic}
\section{Introduction}

In this paper,
we revisit classical computational complexity theory
from the perspective of distributed network computing.
As we will see,
certain standard notions and techniques
not only extend well to the distributed setting,
but in some cases
allow us to achieve more there than in the centralized setting.
We begin by setting the context,
and then present our approach and~results.

% –––––––––––––––––––––––––––––––––––––––––––––––––––––––––––––––––––––––––––––––
\subsection{Background}
\label{ssec:background}

When solving a problem in a computer network using a distributed algorithm,
a major concern is the issue of \emph{locality}.
At its core lies the question of how much information
each computer needs to obtain about the rest of the network
in order to solve the given problem.
The less information needed,
the more local the problem.

\subparagraph*{The LOCAL model.}

In the late 1980s,
Linial~\cite{DBLP:journals/siamcomp/Linial92}
introduced an influential model of distributed computing
that focuses entirely on locality,
while abstracting away many other issues such as
failures, asynchrony, and bandwidth limitations.
In this model,
which Peleg~\cite{Peleg00} later called the \textsc{local} model,
a network consists of several computers
that communicate with their \kl{neighbors}
by exchanging messages through a sequence of fault-free synchronous \kl{rounds}.
The computers,
referred to as \kl{nodes},
are all identical except for possessing \kl{globally unique} \kl{identifiers}.
They have unlimited computational power
to process their local input and the messages they receive in each \kl{round},
and there is no limitation on the message sizes.
The goal in this setting is for the \kl{nodes} to collectively solve
some \kl{graph} problem related to the topology of their network.
That is,
the network serves both
as the communication infrastructure and as the input \kl{graph}.
Typically,
the problem is a construction task such as finding
a (vertex or edge) \kl{coloring},
a maximal matching,
a maximal independent set,
or a spanning tree.
After a finite number of \kl{rounds},
each \kl{node} should produce a local output such as
“my color is blue” or “I~belong to the independent set”,
and the combined output of all \kl{nodes} should yield a valid solution
to the considered problem.

Since the \textsc{local} model imposes no constraints
on computational power and message size,
once the \kl{nodes} have communicated for a number of \kl{rounds}
greater than the \kl{diameter} of the network \kl{graph},
they can in principle know the entire \kl{graph}
and thus solve any problem
that can be solved by a single computer in the centralized setting.
Therefore,
if we equate the complexity of a problem
with the number of \kl{rounds} required to solve it,
call this number the \emph{\kl{round-time} complexity},
and measure it as a function of the number of \kl{nodes},
then all problems have a complexity that lies between
constant \kl{round time}
(the purely local problems)
and linear \kl{round time}
(the inherently global problems).
From this perspective,
investigating the locality of a problem amounts to
determining its \kl{round-time} complexity.

\subparagraph*{Constant round time.}

The role of constant \kl{round time} in the \textsc{local} model
is vaguely analogous to
the role of polynomial time in centralized computing
in that it provides a first approximation
of what constitutes an efficiently solvable problem.
The rigorous study of
the class of problems solvable in constant \kl{round time}
was initiated in the early 1990s
by Naor and Stockmeyer~\cite{DBLP:journals/siamcomp/NaorS95}.
To narrow down the area of investigation,
they focused on construction problems for which the validity of a proposed solution
can at least be \emph{verified} in constant \kl{round time}.
For instance,
a proposed vertex \kl{coloring}
can be easily verified in a single \kl{round} of communication
(each \kl{node} compares its own color with those of its \kl{neighbors}),
whereas a proposed spanning tree cannot be verified locally
(a~sufficiently long cycle is indistinguishable from a line).
Given the analogy with centralized computing,
construction problems verifiable in constant \kl{round time}
are sometimes referred to as the distributed analog of
the complexity class~$\FNP$,
i.e., the function problem variant of~$\NP$~\cite{DBLP:conf/swat/Suomela20}.
For technical reasons,
Naor and Stockmeyer considered only
the subclass of locally verifiable problems
for which there are constant bounds on
the maximum \kl{degree} of the \kl{graphs}
and on the size of the local inputs and outputs.
This subclass,
which they called~$\LCL$ (for \emph{locally checkable labelings}),
became the foundation of
a fruitful research program on locality in distributed computing
(see~“\nameref{par:construction-problems}”
at the end of Section~\ref{ssec:related-work}).

\subparagraph*{Decision problems.}

Although research in distributed computing
has traditionally focused on construction problems,
a newer branch called \emph{distributed decision}~%
\cite{DBLP:journals/eatcs/FeuilloleyF16},
takes more inspiration from classical complexity theory.
In his PODC~2010 keynote talk~\cite{DBLP:conf/podc/Fraigniaud10},
Fraigniaud suggested that
decision problems,
on which standard complexity theory is built,
could also serve as the basis for a complexity theory of distributed computing.
The rationale is that decision problems
are easier to reduce to one another than construction problems,
while still being general enough to express challenges
that arise in a wide variety of models of distributed computing.
Reductions between such problems could therefore reveal
connections between different areas of distributed computing,
or even connections to other~fields.

To make joint decisions in a distributed setting,
the simplest and most widely studied mechanism is acceptance by unanimity.
This requires all computers to \kl{accept} on yes-instances,
and at least one computer to \kl{reject} on no-instances.
When viewed in this context,
$\LCL$ can be reinterpreted as
a class of decision problems on \kl{labeled graphs},
which we will refer to as \emph{\kl{graph properties}}.
The idea is that
a \kl{graph} is a yes-instance of a given \kl{graph property}
if its \kl{labeling} represents
a valid solution to the corresponding $\LCL$ problem
(e.g., a valid \kl{coloring}, or a maximal independent~set).
By generalizing this to arbitrary \kl{graphs}
with \kl{labels} of arbitrary size,
we arrive at the class of \kl{graph properties}
that are decidable in constant \kl{round time} in the \textsc{local} model.
This class was given the name~$\LD$ (for \emph{local decision})
in~\cite{DBLP:journals/jacm/FraigniaudKP13}.
Following the above analogy with centralized computing,
one can think of~$\LD$ as
a distributed analog of the complexity class~$\PTIME$.

\subparagraph*{Nondeterminism.}
\label{par:nondeterminism}

Problems in~$\LD$ are by definition purely local,
and it is easy to come up with \kl{graph properties}
that lie outside this class.
For instance,
the \kl{nodes} of a \kl{graph} cannot locally \kl{decide}
whether the \kl{graph} is a tree
(again, because a sufficiently long cycle is indistinguishable from a line).
However,
a much larger class of \kl(graph){properties} can be \kl{verified}
if we take inspiration from the complexity class~$\NP$
and allow some external entity to nondeterministically assign
each \kl{node} an additional \kl{label}
that acts as a \kl{certificate}.
In fact,
if arbitrary \kl{certificates} are allowed,
then the \kl{nodes} can verify any \kl(graph){property}
decidable by a single computer in the centralized setting,
because each \kl{certificate} can in principle \kl{encode}
the entire \kl{graph} along with the \kl{node}'s \kl{identifiers}
(see, e.g., \cite[\S\,4.1]{DBLP:journals/dmtcs/Feuilloley21}).

To obtain more interesting complexity classes,
two types of restrictions on the \kl{certificates}
have been considered in the literature.
The first is to require the \kl{certificates}
to be independent of the \kl{nodes}' \kl{identifiers}.
Fraigniaud, Korman, and Peleg~\cite{DBLP:journals/jacm/FraigniaudKP13}
explored this restriction
in a model called \emph{nondeterministic local decision}
and showed that it strictly weakens expressiveness,
as some \kl(graph){properties} dependent on
the number of \kl{nodes} in the \kl{graph}
can no longer be locally verified.
(It~was subsequently shown
that the non-verifiable \kl(graph){properties}
are precisely those
that are not closed under lift~\cite{DBLP:conf/opodis/FraigniaudHK12}.)
The full power of unrestricted nondeterminism can only be recovered
in combination with additional resources
such as randomization or an oracle providing the number of \kl{nodes}.

The second type of restriction limits the size of \kl{certificates}
with respect to the total number of \kl{nodes}.
This idea was introduced by
Korman, Kutten, and Peleg~\cite{DBLP:journals/dc/KormanKP10}
in a model called \emph{proof labeling schemes},
and further developed by
Göös and Suomela~\cite{DBLP:journals/toc/GoosS16}
in a more general model called \emph{locally checkable proofs}.
As later argued by Feuilloley~\cite{DBLP:journals/dmtcs/Feuilloley21},
\kl{certificate} size provides, in a sense, an alternative measure of locality:
purely local \kl(graph){properties} do not require any \kl{certificates},
while inherently global \kl(graph){properties}
require quadratic-size \kl{certificates}
(to \kl{encode} an adjacency matrix of the entire \kl{graph}).
The \kl(graph){property} of $3$-\kl{colorability} is almost local,
requiring only constant-size \kl{certificates},
and interestingly,
many natural \kl(graph){properties}
such as non-$2$-\kl{colorability} and \kl{Hamiltonicity}
require logarithmic-size \kl{certificates}.

\subparagraph*{Alternation.}

Since nondeterminism with restrictions on the \kl{certificates}
provides additional power,
but not enough to express all \kl{graph properties},
a natural follow-up is to explore more computational resources
from standard complexity theory
and assess their impact on expressiveness.
One such resource is quantifier alternation,
the key concept underlying the \kl{polynomial hierarchy}
(a hierarchy of complexity classes
that contains $\PTIME$, $\NP$, and~$\coNP$ at its lowest levels).
Adapted to the \textsc{local} model,
alternation can be thought of as a game between two players,
\kl(certificate){Eve} and \kl(certificate){Adam},
who take turns assigning \kl{certificates} to all \kl{nodes}.
Intuitively,
\kl{Eve} (the~existential player) tries to prove that
the input \kl{graph} satisfies a given \kl(graph){property},
while \kl{Adam} (the~universal player) tries to disprove it.
A constant-\kl[round time]{round-time} distributed algorithm
then serves as an \kl{arbiter}
to determine the winner based on the \kl{certificates} provided.

This framework was investigated by
Balliu, D'Angelo, Fraigniaud, and Olivetti~\cite{DBLP:journals/jcss/BalliuDFO18}
for \kl{identifier}-independent \kl{certificates},
and by
Feuilloley, Fraigniaud, and Hirvonen~\cite{DBLP:journals/tcs/FeuilloleyFH21}
for logarithmic-size \kl{certificates}.
The two resulting alternation hierarchies turned out to be radically different.
In the case of \kl{identifier}-independent \kl{certificates},
a single alternation between \kl(certificate){Adam} and \kl(certificate){Eve}
already suffices to \kl{arbitrate} any \kl(graph){property}
decidable in the centralized setting.
This means that the entire hierarchy collapses to its second level.
On the other hand,
for logarithmic-size \kl{certificates},
there are \kl{graph properties} that lie outside the corresponding hierarchy,
and it remains open whether the hierarchy is infinite.
The latter question is believed to be difficult,
as it has been shown to be closely related to
a long-standing open problem in communication complexity~%
\cite{DBLP:conf/wdag/FeuilloleyH18}.

%–––––––––––––––––––––––––––––––––––––––––––––––––––––––––––––––––––––––––––––––
\subsection{Contribution}

This paper is motivated by the following question raised by
Fraigniaud, Korman, and Peleg~\cite[\S\,5.1]{DBLP:journals/jacm/FraigniaudKP13}:
\emph{What are the connections between
  classical computational complexity theory and local complexity theory?}
Rather than viewing the classical theory as merely a source of inspiration,
we aim to extend it directly to the setting of distributed decision.
We approach this from the perspective
that centralized computing corresponds to a special case of the \textsc{local} model,
where the network consists of a single computer.
More specifically,
we introduce the class~$\LP$
(for \emph{locally polynomial time}),
which consists of the \kl{graph properties}
that can be decided
in a constant number of \kl{rounds} in the \textsc{local} model
under the following constraints:
the number of \kl{computation steps} of each computer in each \kl{round}
must be polynomially bounded
in the size of its local input and received messages,
and the algorithm must work correctly
even under \kl{identifier assignments}
that are only \kl{locally unique} within a fixed radius.
This class generalizes both
the complexity class~$\PTIME$
(its~restriction to \kl{graphs}~consisting of a single \kl{labeled} \kl{node}),
and the class of $\LCL$ decision problems
(its~restriction to
\kl{graphs} of bounded maximum \kl{degree} and constant \kl{label} size).

Building on~$\LP$,
we then define the \emph{\kl{locally polynomial hierarchy}}
$\Set{\SigmaLP{\Level}\!, \PiLP{\Level}}_{\Level \in \Naturals}$
analogously to the alternation hierarchies mentioned above,
i.e., as a game between \kl(certificate){Eve} and \kl(certificate){Adam}
who alternately assign \kl{certificates} to all \kl{nodes}.
The size of these \kl{certificates} must be polynomially bounded
with respect to a constant-radius \kl{neighborhood} of the \kl{nodes}.
Hence,
when restricted to \kl{single-node graphs},
our hierarchy coincides exactly with the classical \kl{polynomial hierarchy}.
For example,
the restriction of $\SigmaLP{1} = \NLP$ to \kl{single-node graphs}
coincides with
the complexity class $\SigmaP{1} = \NP$.

Aiming more at a conceptual than a technical contribution,
we make four main~points:
\begin{enumerate}
\item
  \emph{Several key concepts and results from standard complexity theory
    generalize well to the distributed setting.}
  To illustrate this,
  we extend the notion of polynomial-time reductions
  to our model of computation,
  and then use it to establish
  a number of \kl{hardness} and \kl{completeness} results
  for the two lowest levels of the \kl{locally polynomial hierarchy}.
  While some of these are meaningful only in the distributed setting,
  our results also include distributed generalizations of well-known classics,
  such as
  the Cook\nobreakdash--Levin theorem
  (which identifies \kl{Boolean satisfiability}
  as a \kl{complete} problem for~$\NP$),
  and the fact that
  $3$\nobreakdash-\kl{colorability} is $\NP$\nobreakdash-\kl{complete}.
  Similarly,
  we prove a distributed generalization of Fagin's theorem
  (which characterizes~$\NP$ as the class of problems expressible in
  \kl[existential fragment]{existential second-order logic}).
  This gives us a logical, and thus machine-independent,
  characterization of the entire \kl{locally polynomial hierarchy},
  demonstrating the robustness of our definition.
  Moreover,
  whenever we generalize a classical result,
  the original version can be recovered
  by restricting networks to single computers.
\item
  \emph{Sometimes standard techniques
    get us further in the distributed setting
    than they do in the centralized setting.}
  Specifically,
  we are able to show that
  the \kl{locally polynomial hierarchy} is infinite,
  while it remains notoriously open
  whether this is also true when restricted to a single computer.
  (A~collapse of our hierarchy would have implied
  a collapse of the classical \kl{polynomial hierarchy},
  but the converse does not hold.)
  As a consequence,
  our \kl{hardness} and \kl{completeness} results
  for the \kl{locally polynomial hierarchy}
  immediately yield \emph{unconditional} lower bounds
  on the complexity of the \kl{graph properties} in question,
  i.e., lower bounds that do not rely on any complexity-theoretic assumptions.
  In addition,
  the constraints imposed by the distributed setting
  allow us to identify
  natural \kl{graph properties} that lie outside our hierarchy.
\item
  \emph{Descriptive complexity theory,
    the discipline of characterizing complexity classes
    in terms of equivalent logical formalisms,
    is particularly helpful in the distributed setting.}
  \begin{itemize}
  \item
    \emph{On the one hand,
      this approach gives us access
      to a large body of existing results in logic and automata theory.}
    In particular,
    our infiniteness result for the \kl{locally polynomial hierarchy}
    leverages a corresponding result on \kl{monadic second-order logic}
    established by Matz, Schweikardt, and Thomas~%
    \cite{DBLP:journals/iandc/MatzST02},
    as well as a logical characterization of finite automata on \kl{pictures}
    (so-called \kl{tiling systems})
    established by Giammarresi, Restivo, Seibert and Thomas~%
    \cite{DBLP:journals/iandc/GiammarresiRST96}.
    Despite not being explicitly concerned with distributed computing,
    these results rely significantly on a form of locality.
    Moreover,
    to show that some \kl{graph properties} lie outside our hierarchy,
    we make direct use of classical results from automata theory,
    namely the pumping lemma for regular languages
    and the Büchi-Elgot-Trakhtenbrot theorem
    (which provides a logical characterization of finite automata on words).
  \item
    \emph{On the other hand,
      descriptive complexity can offer a fresh perspective on distributed computing
      by introducing unconventional constraints on familiar concepts.}
    For instance,
    the restriction to algorithms
    that work correctly under \kl{locally unique} \kl{identifiers}
    is necessary to prove our generalization of Fagin’s theorem.
    But this restriction is also meaningful
    from a pure distributed computing point of view:
    in a sense,
    $\LP$~fully preserves the locality of~$\LCL$
    (which can be defined without \kl{identifiers}),
    whereas $\LD$ is somewhat less local
    due to its reliance on \kl{globally unique} \kl{identifiers}.
    Similarly,
    proving our generalization of Fagin’s theorem
    requires a polynomial bound on \kl{certificate} sizes
    with respect to the \kl{nodes}' constant-radius \kl{neighborhoods}.
    This again places a strong emphasis on locality
    and contrasts sharply with previous approaches
    to distributed nondeterminism and alternation,
    which allow \kl{certificate} sizes to depend on the whole~\kl{graph}.
  \end{itemize}
\item
  \emph{Quantifier alternation provides a new tool
    for measuring the locality of problems in distributed computing.}
  This perspective emerges naturally from
  the strict local constraints required to generalize Fagin’s theorem,
  yet it extends beyond the specifics of the \kl{locally polynomial hierarchy}.
  The resulting measure of locality
  aligns surprisingly well with
  the previously mentioned \kl{certificate}-size metric,
  despite being based on an entirely different principle.
  We introduce this topic
  in Section~\ref{ssec:measure-locality}
  and further elaborate on it in Section~\ref{sec:discussion}.
\end{enumerate}

%–––––––––––––––––––––––––––––––––––––––––––––––––––––––––––––––––––––––––––––––
\subsection{Related work}
\label{ssec:related-work}

Our base class~$\LP$ generalizes the $\LCL$~problems
of Naor and Stockmeyer~\cite{DBLP:journals/siamcomp/NaorS95}
to arbitrary \kl{labeled graphs}
(when interpreting these problems as decision problems).
However,
it is less general than the class~$\LD$
of Fraigniaud, Korman, and Peleg~\cite{DBLP:journals/jacm/FraigniaudKP13},
since it imposes restrictions
on the individual processing power of the \kl{nodes}
and requires correctness under \kl{locally unique} \kl{identifiers}.
Hence,
we have
$\LCL \subseteq \LP \subseteq \LD$,
and it is easy to check that these inclusions are strict.

\subparagraph*{Alternation hierarchies.}

The work most closely related to this paper
includes the different alternation hierarchies based on~$\LD$.
The relationship is particularly clear
for the previously mentioned \kl{identifier}-independent hierarchy
$\Set{\indSigmaLD{\Level}\!, \indPiLD{\Level}}_{\Level \in \Naturals}$
of Balliu, D'Angelo, Fraigniaud, and Olivetti~\cite{DBLP:journals/jcss/BalliuDFO18}.
Since that hierarchy collapses to its second level~$\indPiLD{2}$
and contains all decidable \kl(graph){properties},
it obviously subsumes our hierarchy
$\Set{\SigmaLP{\Level}\!, \PiLP{\Level}}_{\Level \in \Naturals}$,
which is infinite and excludes some decidable \kl(graph){properties}.
But even on the lower levels,
it is easy to see that
$\SigmaLP{\Level} \subseteq \indSigmaLD{\Level}$ and
$\PiLP{\Level} \subseteq \indPiLD{\Level}$,
essentially because
the \kl{identifier}-independent \kl{certificates} chosen by the first player
can be used to provide each \kl{node} with
a new, \kl{locally unique} \kl{identifier}
whose validity can be verified in a constant number of \kl{communication rounds}.
The inclusion on the nondeterministic level is strict
because $\NOTALLSELECTED$,
the \kl(graph){property} of
a \kl{labeled  graph} having at least one unselected \kl{node},
lies in $\indSigmaLD{1}$ but not~in~$\SigmaLP{1}$.

Recently,
a polynomial-time version
$\Set{\indSigmaLDP{\Level}\!, \indPiLDP{\Level}}_{\Level \in \Naturals}$
of the \kl{identifier}-independent hierarchy
was investigated by Aldema Tshuva and Oshman~\cite{DBLP:conf/podc/TshuvaO22}.
Although at first glance their definition may seem similar to ours,
it differs in a crucial point:
the polynomial bound they impose on the processing time of the \kl{nodes}
is relative to the size of the entire input \kl{graph}
(including \kl{labels}),
rather than
relative to the amount of information that the \kl{nodes} receive locally.
As a result,
from its second level $\indPiLDP{2}$\! onward,
their hierarchy is essentially equivalent to
the centralized \kl{polynomial hierarchy}
(restricted to \kl{encodings} of \kl{graphs}),
and thus it is unknown whether it collapses or not.
Nevertheless,
its relationship to our hierarchy mirrors that of
the original \kl{identifier}-independent hierarchy, i.e.,
$\SigmaLP{\Level} \subseteq \indSigmaLDP{\Level}$ and
$\PiLP{\Level} \subseteq \indPiLDP{\Level}$
for all $\Level \in \Naturals$,
and the \kl(graph){property} $\NOTALLSELECTED$
separates $\indSigmaLDP{1}$ from~$\SigmaLP{1}$.

It is less obvious
how exactly our hierarchy relates to
the previously mentioned logarithmic-size hierarchy
$\Set{\logSigmaLD{\Level}\!, \logPiLD{\Level}}_{\Level \in \Naturals}$
of Feuilloley, Fraigniaud, and Hirvonen~\cite{DBLP:journals/tcs/FeuilloleyFH21}.
But at least when restricted to
\kl{graphs} of bounded maximum \kl{degree} and constant \kl{label} size,
each level of the logarithmic-size hierarchy
contains the corresponding level of our hierarchy,
since our bound on the \kl{certificates} sizes
reduces to a constant bound for such \kl{graphs}.
That is,
$\SigmaLP{\Level}\On{\bGRAPH{\MaxDegree}} \subseteq
 \logSigmaLD{\Level}\On{\bGRAPH{\MaxDegree}}$
and
$\PiLP{\Level}\On{\bGRAPH{\MaxDegree}} \subseteq
 \logPiLD{\Level}\On{\bGRAPH{\MaxDegree}}$,
where
$\Class\On{\bGRAPH{\MaxDegree}}$
denotes the aforementioned restriction of a class~$\Class$.
Moreover,
the nondeterministic classes are again separated
by the \kl(graph){property} $\NOTALLSELECTED$,
which lies in~$\logSigmaLD{1}$ but not in~$\SigmaLP{1}$.

What most fundamentally distinguishes this work
from all three $\LD$-based hierarchies
is that our hierarchy preserves some degree of locality.
This is because
we bound the size of a \kl{node}'s \kl{certificates}
with respect to its constant-radius \kl{neighborhood},
so that each \kl{certificate} can \kl{encode}
only a very limited amount of global information about the input \kl{graph}.
By contrast,
in both \kl{identifier}-independent hierarchies,
“only the first few levels of alternation are needed
to overcome the locality of a distributed algorithm”,
as noted by Aldema Tshuva and Oshman~\cite[\S\,1]{DBLP:conf/podc/TshuvaO22}.
This is particularly evident in the original version
of Balliu, D'Angelo, Fraigniaud, and Olivetti,
where the second level~$\indPiLD{2}$
already contains all global \kl(graph){properties}.
Similarly,
in the logarithmic-size hierarchy
of Feuilloley, Fraigniaud, and Hirvonen,
the third level~$\logSigmaLD{3}$ is already powerful enough
to express the existence of a \kl{nontrivial automorphism},
a \kl(graph){property} that is inherently global
when using \kl{certificate} size as the measure of locality.
(Göös and Suomela~\cite{DBLP:journals/toc/GoosS16}
have shown that it requires quadratic-size \kl{certificates},
which is the highest possible complexity.)

\subparagraph*{Descriptive complexity.}

Another line of research closely related to this paper is the development of
descriptive complexity in the setting of distributed computing.
This was initiated by Hella~et~al.~\cite{DBLP:journals/dc/HellaJKLLLSV15},
who used several variants of modal logic
to characterize synchronous constant-\kl[round time]{round-time} algorithms
for various models of distributed computing in anonymous networks.
Their idea was later extended to
arbitrary-\kl[round time]{round-time} algorithms~\cite{DBLP:conf/csl/Kuusisto13},
asynchronous algorithms~\cite{DBLP:conf/icalp/Reiter17},
and a stronger model with \kl[globally unique]{unique} \kl{identifiers}~%
\cite{DBLP:conf/fossacs/BolligBR19}.
Our generalization of Fagin's theorem to the \textsc{local} model
remains close in spirit to the work of Hella~et~al.
The main difference is that we consider
polynomial-time Turing machines instead of finite-state automata,
and \kl(quantifier){bounded} \kl{first-order quantifiers}
instead of modal operators.
Also,
the result of Hella~et~al. already holds for deterministic models,
whereas our result requires the presence of nondeterminism,
or more generally, alternation.
This parallels the situation in the centralized setting,
where Fagin's theorem characterizes the class~$\NP$,
but it remains a major open question
whether the class~$\PTIME$ admits a similar characterization.

\subparagraph*{Construction problems.}
\label{par:construction-problems}

More distantly related to this paper,
recent years have also seen
significant progress in the study of $\LCL$~problems
as originally defined by Naor and Stockmeyer
(i.e., as construction problems
whose solutions can be verified locally).
Much research has focused on classifying $\LCL$~problems
according to the \kl{round time} required
to construct solutions for them in the \textsc{local} model.
For over two decades,
progress was slow as efforts were driven by individual problems
rather than entire classes of problems.
While there were known examples of $\LCL$~problems with
constant, iterated logarithmic, and linear \kl{round-time} complexities,
it remained unclear
whether problems with other complexities existed
in the spectrum between constant and linear \kl{round time}.
However,
this picture changed drastically in the mid 2010's,
when a large number of positive and negative results
were published within a few years,
proving the existence of problems
in some intermediate regions of the spectrum,
and ruling out the existence of problems in other regions.
Taken together,
those results now provide a nearly complete classification of $\LCL$~problems,
revealing essentially four complexity classes.
In his SWAT~2020 keynote talk~\cite{DBLP:conf/swat/Suomela20},
Suomela interpreted those classes as follows:
purely local problems,
symmetry-breaking problems,
inherently global problems,
and an intriguing class of problems
for which randomness provides a significant speedup.
As the topic is well beyond the scope of this paper,
and the publications are numerous,
the reader is referred to (the transcript of) Suomela's talk,
which summarizes recent progress and provides many references.

%–––––––––––––––––––––––––––––––––––––––––––––––––––––––––––––––––––––––––––––––
\subsection{Organization}

We begin with an informal overview of the paper
in Section~\ref{sec:overview}.
The material covered there
will be repeated later in much greater detail and formality.
This is necessary because descriptive complexity involves
mechanical translations between algorithms and \kl{logical formulas},
forcing us to deal with the low-level aspects of both frameworks.
In Section~\ref{sec:preliminaries},
we give some preliminaries on \kl{graphs} and relational \kl{structures}.
Then, in Section~\ref{sec:turing-machines},
we introduce our model of computation,
which extends standard Turing machines to the distributed setting,
and define the \kl{locally polynomial hierarchy} based on this model.
Section~\ref{sec:logic} introduces the corresponding logical formalism,
along with some examples of
\kl{graph properties} expressed as \kl{logical formulas}.
In Section~\ref{sec:restrictive-arbiters},
we provide a more flexible characterization of our complexity classes
in order to simplify subsequent proofs.
The actual results begin in Section~\ref{sec:fagin},
where we present Fagin’s theorem
and generalize it to the \kl{locally polynomial hierarchy}.
In Section~\ref{sec:reductions},
we introduce the notion of \kl{locally polynomial reductions},
based on which
we establish a number of \kl{hardness} and \kl{completeness} results,
including a generalization of the Cook--Levin theorem and
the $\NLP$\nobreakdash-\kl{completeness} of $3$\nobreakdash-\kl{colorability}.
Section~\ref{sec:hierarchy} constitutes the longest part of the paper,
where we prove that the \kl{locally polynomial hierarchy} is infinite.
This involves a detour through
\kl{tiling systems} and \kl{monadic second-order logic} on \kl{pictures}.
Finally,
in Section~\ref{sec:discussion},
we discuss how the preceding results
may be relevant to the study of locality in distributed~computing.

%%% Local Variables:
%%% mode: latex
%%% TeX-master: "../lph-paper"
%%% End:

\section{Informal overview}
\label{sec:overview}

In this paper,
we investigate the computational complexity of \kl{graph properties}
within a distributed model of computation.
As is common in this type of setting,
we assume that all \kl{graphs} are
finite, simple, undirected, and \kl{connected}.
Additionally,
our \kl{graphs} are equipped with a \kl{labeling} function
that assigns a bit string to each \kl{node}.
The focus is exclusively on
\kl{graph properties} that are invariant under isomorphism.
These \kl(graph){properties} typically depend on the \kl{graph}'s topology
(e.g.,
$3$-\kl{colorability},
\kl{Eulerianness}, or
\kl{Hamiltonicity}),
but may also depend on its \kl{node labels}
(e.g.,
having all \kl{nodes} \kl{labeled} the same, or
having the \kl{labeling} form a valid $3$-\kl{coloring}).

% –––––––––––––––––––––––––––––––––––––––––––––––––––––––––––––––––––––––––––––––
\subsection{Our complexity classes}
\label{ssec:overview-complexity-classes}

To classify \kl{graph properties},
we extend standard complexity classes from strings to \kl{graphs},
treating strings as \kl{graphs} consisting of a single \kl{labeled} \kl{node}.

\subparagraph*{Model of computation.}

We use distributed algorithms as decision procedures for \kl{graph properties}.
Given an input \kl{graph}~$\Graph$
and an \kl[identifier assignment]{assignment}~$\IdMap$
of \kl{identifiers} to the \kl{nodes} of~$\Graph$,
the goal is for the \kl{nodes} to collectively decide whether
$\Graph$~has a certain \kl(graph){property}~$\Property$.
To do so,
they proceed in a sequence of synchronous \kl{communication rounds}.
In each \kl{round},
each \kl{node}
first receives the messages sent by its \kl{neighbors} in the previous \kl{round},
then performs some local computations, and
finally sends new messages to its \kl{neighbors}.
After a finite number of \kl{rounds},
each \kl{node} must have reached an individual \kl{verdict},
and $\Graph$~is \kl{accepted}
if and only if the \kl{nodes} unanimously agree on it.
The collective decision must be independent of
the particular \kl{identifier assignment}~$\IdMap$,
as long as the latter satisfies
a basic requirement of \emph{\kl{local uniqueness}}:
$\IdMap$~must assign different \kl{identifiers} to any two \kl{nodes}
that lie within some fixed \kl{distance} of each other.
This can be seen as a precondition for the algorithm to work~correctly.

In the following,
we restrict our attention to distributed algorithms
that are guaranteed to \kl{terminate}
in a constant number of \kl{communication rounds},
and where
the number of \kl{computation steps} of each \kl{node} in each \kl{round}
is polynomially bounded in
the size of its local input and the messages it receives.
We call such algorithms \emph{\kl{locally polynomial machines}},
as we will formalize them using
a model based on \kl[distributed Turing machines]{Turing machines}
(see Section~\ref{sec:turing-machines}).

\subparagraph*{The locally polynomial hierarchy.}

Our complexity classes are based on a game between two players
who intuitively argue whether
a given \kl{graph}~$\Graph$ has some \kl(graph){property}~$\Property$:
\kl(certificate){Eve} (the~existential player)
tries to prove that $\Graph$~has \kl(graph){property}~$\Property$,
while \kl(certificate){Adam} (the~universal player)
tries to prove the opposite.
Given some \kl{locally unique} \kl{identifier assignment} of~$\Graph$,
the two players take turns assigning additional \kl{labels},
called \kl{certificates},
to the \kl{nodes} of~$\Graph$.
These \kl{certificates} can be thought of as
proofs (in \kl(certificate){Eve}'s case)
and
counterproofs (in~\kl(certificate){Adam}'s case).
They may depend on the provided \kl{identifiers},
but their size must be polynomially bounded
with respect to the amount of information
contained in a \kl{node}'s constant-radius \kl{neighborhood}
(including all \kl{labels} and \kl{identifiers} therein).
After a fixed number~$\Level$ of moves,
the winner is determined by a \kl{locally polynomial machine}
that acts as an \kl{arbiter}.
Ultimately,
the \kl{graph}~$\Graph$ has \kl(graph){property}~$\Property$
if and only if \kl(certificate){Eve} has a winning strategy in this game,
i.e., if she always wins when playing optimally.
Depending on who makes the first move,
$\Property$~is classified
as a~$\SigmaLP{\Level}$-\kl(graph){property}
(if~\kl(certificate){Eve} starts)
or a~$\PiLP{\Level}$-\kl(graph){property}
(if~\kl(certificate){Adam}~starts).

To take a specific case,
$\Property$~belongs to~$\SigmaLP{3}$
if it satisfies the following equivalence
for every \kl{graph}~$\Graph$ and
every admissible \kl{identifier assignment}~$\IdMap$ of~$\Graph$:
\begin{equation*}
  \Graph \in \Property
  \; \iff \;
  \exists \CertifMap_1 \,
  \forall \CertifMap_2 \,
  \exists \CertifMap_3:
  \Result{\Machine}{
    \Graph, \,
    \IdMap, \,
    \CertifMap_1 \CertifConcat
    \CertifMap_2 \CertifConcat
    \CertifMap_3
  } \equiv
  \Accept,
\end{equation*}
where
$\Machine$ is an appropriately chosen \kl{locally polynomial machine},
and all quantifiers range over \kl{certificate assignments}
that satisfy the aforementioned polynomial bound.
(The notation used here
will be formally introduced in Section~\ref{sec:turing-machines}.)

\begin{example}
  \label{ex:3-round-3-colorable}
  One \kl(graph){property} that clearly belongs to~$\SigmaLP{3}$
  is \emph{$3$-round $3$-colorability},
  a variant of classical $3$\nobreakdash-\kl{colorability}
  introduced by
  Ajtai, Fagin, and Stockmeyer~\cite[\S\,11]{DBLP:journals/jcss/AjtaiFS00}.
  A \kl{graph} has this \kl(graph){property}
  if \kl(certificate){Eve} can always force a valid $3$-\kl{coloring}
  in the following game:
  First,
  \kl(certificate){Eve} chooses the colors of \kl{nodes} of \kl{degree}~$1$;
  then \kl(certificate){Adam} chooses the colors of \kl{nodes} of \kl{degree}~$2$;
  and finally,
  \kl(certificate){Eve} chooses the colors of all remaining \kl{nodes}.
  Consider, for instance,
  the \kl{graph} in Figure~\ref{fig:graph-not-3-round-3-colorable}.
  Although it is $3$-\kl{colorable} in the classical sense,
  it is not $3$-round $3$-colorable
  because \kl(certificate){Adam} has a winning strategy:
  After \kl(certificate){Eve} assigns a color
  $i \in \Set{0, 1, 2}$
  to \kl{node}~$\Node[1]$,
  \kl(certificate){Adam} assigns
  the same color~$i$ to \kl{node}~$\Node[2]_1$
  and a different color
  $j \in \Set{0, 1, 2} \setminus \Set{i}$
  to \kl{node}~$\Node[2]_2$.
  No matter which colors in $\Set{0, 1, 2}$ \kl(certificate){Eve} then assigns
  to the remaining \kl{nodes} $\Node[3]_1$, $\Node[3]_2$, $\Node[3]_3$,
  there will always be
  two \kl{adjacent} \kl{nodes} with the same color.
  However,
  if we remove the \kl{edge} between \kl{nodes} $\Node[3]_1$ and~$\Node[3]_3$,
  as shown in Figure~\ref{fig:graph-3-round-3-colorable},
  we obtain a \kl{graph}
  that is both $3$-\kl{colorable} and $3$-round $3$-colorable:
  regardless of the colors assigned to \kl{nodes} $\Node[1]$, $\Node[2]_1$, $\Node[2]_2$,
  \kl(certificate){Eve} can always find
  non-conflicting colors for the remaining~\kl{nodes}.
  \lipicsEnd
\end{example}

\begin{figure}[htb]
  \centering
  \begin{subfigure}[t]{0.4\textwidth}
    \begin{tikzpicture}[%
    semithick,
    every node/.style={draw,circle,minimum size=4ex,inner sep=0},
  ]
  \def\nodeDist{10ex}
  \node (u) {$\Node[1]$};
  \node (w1) at ($(u)+(0:\nodeDist)$) {$\Node[3]_1$};
  \node (w2) at ($(w1)+(0:\nodeDist)$) {$\Node[3]_2$};
  \node (v2) at ($(w1)+(-90:\nodeDist)$) {$\Node[2]_2$};
  \node (w3) at ($(w2)+(-90:\nodeDist)$) {$\Node[3]_3$};
  \node (v1) at ($(w2)+(-30:\nodeDist)$) {$\Node[2]_1$};
  \draw
    (u) edge (w1)
    (w1) edge (w2)
    (w1) edge (v2)
    (w1) edge (w3)
    (w2) edge (w3)
    (w2) edge (v1)
    (v2) edge (w3)
    (w3) edge (v1)
    ;
\end{tikzpicture}

%%% Local Variables:
%%% mode: latex
%%% TeX-master: "../lph-paper"
%%% End:
    \caption{no-instance}
    \label{fig:graph-not-3-round-3-colorable}
  \end{subfigure}
  \qquad
  \begin{subfigure}[t]{0.4\textwidth}
    \begin{tikzpicture}[%
    semithick,
    every node/.style={draw,circle,minimum size=4ex,inner sep=0},
  ]
  \def\nodeDist{10ex}
  \node (u) {$\Node[1]$};
  \node (w1) at ($(u)+(0:\nodeDist)$) {$\Node[3]_1$};
  \node (w2) at ($(w1)+(0:\nodeDist)$) {$\Node[3]_2$};
  \node (v2) at ($(w1)+(-90:\nodeDist)$) {$\Node[2]_2$};
  \node (w3) at ($(w2)+(-90:\nodeDist)$) {$\Node[3]_3$};
  \node (v1) at ($(w2)+(-30:\nodeDist)$) {$\Node[2]_1$};
  \draw
    (u) edge (w1)
    (w1) edge (w2)
    (w1) edge (v2)
    (w1.south east) edge[bend right=30] ($(w1)+(0.45*\nodeDist,-0.3*\nodeDist)$)
    (w2) edge (w3)
    (w2) edge (v1)
    (v2) edge (w3)
    (w3) edge (v1)
    (w3.north west) edge[bend right=30] ($(w3)+(-0.45*\nodeDist,0.3*\nodeDist)$)
    ;
\end{tikzpicture}

%%% Local Variables:
%%% mode: latex
%%% TeX-master: "../lph-paper"
%%% End:
    \caption{yes-instance}
    \label{fig:graph-3-round-3-colorable}
  \end{subfigure}
  \caption{
    A negative and a positive instance of $3$-round $3$-colorability.
  }
\end{figure}

We refer to the family of classes
$\Set{\SigmaLP{\Level}\!, \PiLP{\Level}}_{\Level \in \Naturals}$
as the \emph{\kl{locally polynomial hierarchy}}.
Two classes at the lowest levels are of particular interest:
$\LP = \SigmaLP{0}$
(for \emph{locally polynomial time})
and
$\NLP = \SigmaLP{1}$
(for \emph{nondeterministic locally polynomial time}).
Due to the asymmetric nature of acceptance by unanimity,
classes on the same level of the \kl{locally polynomial hierarchy}
are neither \kl{complement classes} of each other,
nor are they closed under \kl{complementation}
(see Corollary~\ref{cor:complement-classes}
on page~\pageref{cor:complement-classes}).
Therefore,
it makes sense to also consider
the \kl[complement hierarchy]{hierarchy} of \kl{complement classes}
$\Set{\coSigmaLP{\Level}\!, \coPiLP{\Level}}_{\Level \in \Naturals}$.
The two hierarchies are illustrated in Figure~\ref{fig:hierarchy-full-overview},
along with the inclusion and separation results shown in this~paper.

\begin{figure}[tb]
  \centering
  \begin{tikzpicture}[
    semithick,on grid,node distance=7.8ex,
    every node/.style={draw,rounded rectangle,
                       minimum height=4ex,minimum width=10ex},
    extra/.style={draw=none,rectangle,minimum size=0,inner sep=0},
    strict/.style={},
    non strict/.style={thin,dashed},
    heavy/.style={ultra thick},
    pseudo/.style={thin,dotted},
  ]
  \def\nodeDist{7.8ex}
  % Sigma and Pi nodes
  \node[heavy] (s0) {$\LP$};
  \node[heavy] (s1) [above left=\nodeDist and 2/3*\nodeDist of s0] {$\SigmaLP{1}$};
  \node (p1) [above right=\nodeDist and 2/3*\nodeDist of s0] {$\PiLP{1}$};
  \node (s2) [above of=s1]        {$\SigmaLP{2}$};
  \node[heavy] (p2) [above of=p1] {$\PiLP{2}$};
  \node[heavy] (s3) [above of=s2] {$\SigmaLP{3}$};
  \node (p3) [above of=p2]        {$\PiLP{3}$};
  \node (s4) [above of=s3]        {$\SigmaLP{4}$};
  \node[heavy] (p4) [above of=p3] {$\PiLP{4}$};
  \node[extra] at ([yshift=2/3*\nodeDist]$(s4)!0.5!(p4)$) {$\vdots$};
  % co-Sigma and co-Pi nodes
  \node[heavy] (cs0) [right=4.3*\nodeDist of s0] {$\coLP$};
  \node[heavy] (cs1) [above left=\nodeDist and 2/3*\nodeDist of cs0] {$\coSigmaLP{1}$};
  \node (cp1) [above right=\nodeDist and 2/3*\nodeDist of cs0] {$\coPiLP{1}$};
  \node (cs2) [above of=cs1]        {$\coSigmaLP{2}$};
  \node[heavy] (cp2) [above of=cp1] {$\coPiLP{2}$};
  \node[heavy] (cs3) [above of=cs2] {$\coSigmaLP{3}$};
  \node (cp3) [above of=cp2]        {$\coPiLP{3}$};
  \node (cs4) [above of=cs3]        {$\coSigmaLP{4}$};
  \node[heavy] (cp4) [above of=cp3] {$\coPiLP{4}$};
  \node[extra] at ([yshift=2/3*\nodeDist]$(cs4)!0.5!(cp4)$) {$\vdots$};
  % Pseudo Sigma nodes
  \node (s0x) [draw=none,right=4.3*\nodeDist of cs0] {};
  \node (s1x) [pseudo,above left=\nodeDist and 2/3*\nodeDist of s0x] {$\SigmaLP{1}$};
  \node (s2x) [pseudo,above of=s1x] {$\SigmaLP{2}$};
  \node (s3x) [pseudo,above of=s2x] {$\SigmaLP{3}$};
  % Additional labels
  \node [extra,anchor=east,left=0ex of s0.west] {$\SigmaLP{0} \,{=}\,$};
  \node [extra,anchor=west,right=0ex of s0.east] {$\,{=}\, \PiLP{0}$};
  \node [extra,anchor=east,left=0ex of s1.west] {$\NLP \,{=}\,$};
  \node [extra,anchor=east,left=0ex of cs0.west] {$\coSigmaLP{0} \,{=}\,$};
  \node [extra,anchor=west,right=0ex of cs0.east] {$\,{=}\, \coPiLP{0}$};
  \node [extra,anchor=east,left=0ex of cs1.west] {$\coNLP \,{=}\,$};
  % Sigma and Pi edges
  \path[non strict]
    (s3) edge (s4)
    (p2) edge (p3)
    (s1) edge (s2)
    (s0) edge (p1);
  \path[strict]
    (s3) edge (p4)
    (p3) edge (s4)
    (p3) edge (p4)
    (s2) edge (s3)
    (s2) edge (p3)
    (p2) edge (s3)
    (s1) edge (p2)
    (p1) edge (s2)
    (p1) edge (p2)
    (s0) edge (s1);
  % co-Sigma and co-Pi edges
  \path[non strict]
    (cs3) edge (cs4)
    (cp2) edge (cp3)
    (cs1) edge (cs2)
    (cs0) edge (cp1);
  \path[strict]
    (cs3) edge (cp4)
    (cp3) edge (cs4)
    (cp3) edge (cp4)
    (cs2) edge (cs3)
    (cs2) edge (cp3)
    (cp2) edge (cs3)
    (cs1) edge (cp2)
    (cp1) edge (cs2)
    (cp1) edge (cp2)
    (cs0) edge (cs1);
  % Edges connecting non-co to co nodes
  \path[strict]
    (p1)  edge (cs3)
    (cs2) edge (p4);
  \path[strict]
    (cp1) edge (s3x)
    (s2x) edge (cp4);
\end{tikzpicture}

%%% Local Variables:
%%% mode: latex
%%% TeX-master: "../lph-paper"
%%% End:
  \caption{
    The \kl{locally polynomial hierarchy} (left)
    and its \kl{complement hierarchy} (right).
    The dotted classes on the far right
    are the same as those on the far left,
    repeated for the sake of readability.
    Only the lowest five levels are shown,
    but the pattern extends infinitely.
    (See Figure~\ref{fig:hierarchy-full} on page~\pageref{fig:hierarchy-full}
    for an extended version of this figure
    with references to the corresponding proofs.)
    Each line (whether solid or dashed) indicates
    an inclusion of the lower class in the higher class.
    All inclusions represented by solid lines are proved to be strict,
    and classes on the same level (regardless of which hierarchy)
    are proved to be pairwise distinct,
    even when restricted to
    \kl{graphs} of bounded maximum \kl{degree} and constant \kl{label} size.
    The inclusions represented by dashed lines
    are in fact equalities when restricted to the latter class of \kl{graphs},
    but this statement is unlikely to generalize to arbitrary \kl{graphs},
    where it holds if and only if $\PTIME = \coNP$.
    This means that
    from a distributed computing perspective,
    the classes shown with thick borders are the most meaningful.
  }
  \label{fig:hierarchy-full-overview}
\end{figure}

\subparagraph*{Connection to standard complexity classes.}

By restricting the classes
$\Set{\SigmaLP{\Level}\!, \PiLP{\Level}}_{\Level \in \Naturals}$
to strings
(i.e.,~\kl{labeled graphs} consisting of a single \kl{node}),
we obtain the classes
$\Set{\SigmaP{\Level}, \PiP{\Level}}_{\Level \in \Naturals}$
of the original \kl{polynomial hierarchy}
introduced by Meyer and Stockmeyer~\cite{DBLP:conf/focs/MeyerS72}.
Since the same observation holds for the \kl{complement classes},
this means that the \kl{locally polynomial hierarchy}
is identical to its \kl{complement hierarchy} on strings.
In particular,
$\PTIME = \LP\On\NODE = \coLP\On\NODE$
and
$\NP = \NLP\On\NODE = \coPiLP{1}\On\NODE$,
where
$\Class\On\NODE$
denotes the restriction of a complexity class~$\Class$
to \kl{single-node graphs}.
This also means that
any inclusion result for the \kl{locally polynomial hierarchy}
would imply
the corresponding inclusion result for the \kl{polynomial hierarchy}
(e.g.,~$\LP = \NLP$ would imply $\PTIME = \NP$),
but not vice versa.
Conversely,
any separation result for the \kl{polynomial hierarchy}
would imply
the corresponding separation result for the \kl{locally polynomial hierarchy}
(e.g.,~$\PTIME \neq \NP$ would imply $\LP \neq \NLP$),
but not vice versa.
Thus, unfortunately,
our infiniteness result for the \kl{locally polynomial hierarchy}
does not imply
a corresponding result for the original \kl{polynomial hierarchy}.

% –––––––––––––––––––––––––––––––––––––––––––––––––––––––––––––––––––––––––––––––
\subsection{Extending classical reductions}

Aiming to apply standard techniques of complexity theory
to the distributed setting,
we extend
Karp's~\cite{DBLP:conf/coco/Karp72} notion of polynomial-time reduction
to computer networks.

\subparagraph*{Locally polynomial reductions.}

In a nutshell,
if there is a \emph{\kl{locally polynomial reduction}}
from a \kl(graph){property}~$\Property$
to a \kl(graph){property}~$\Property'$,
then this means that there exists
a \kl{locally polynomial machine}~$\Machine_{\Tag{red}}$
that transforms an input \kl{graph}~$\Graph$ into a new \kl{graph}~$\Graph'$
such that
$\Graph$~has \kl(graph){property}~$\Property$
if and only if
$\Graph'$~has \kl(graph){property}~$\Property'$.
Hence,
the existence of such a \kl{reduction} implies that
$\Property'$~is at least as hard as~$\Property$,
since
an efficient \kl{decider}~$\Machine'$ for~$\Property'$
could be converted into an efficient \kl{decider}~$\Machine$ for~$\Property$,
which would first run~$\Machine_{\Tag{red}}$
and then simulate~$\Machine'$ on the resulting~\kl{graph}.

To transform a \kl{graph}~$\Graph$ into a \kl{graph}~$\Graph'$
with a distributed algorithm,
each \kl{node}~$\Node[1]$ of the input \kl{graph}~$\Graph$
computes a string
that \kl{encodes} a \kl{subgraph} of the output \kl{graph}~$\Graph'$,
including the \kl{labels} of all \kl{nodes} therein.
We call this \kl{subgraph}
the \emph{\kl{cluster}} representing~$\Node[1]$ in~$\Graph'$.
\kl{Clusters} of different \kl{nodes} may not overlap,
and \kl{edges} between different \kl{clusters} are only permitted
if the \kl{clusters} represent
\kl{adjacent} \kl{nodes} in the original \kl{graph}~$\Graph$.
This setup allows the \kl{nodes} of~$\Graph$ to simulate
a distributed algorithm running on~$\Graph'$
by simulating the algorithm within their respective \kl{clusters}
and exchanging messages with their \kl{neighbors}
to simulate inter-\kl{cluster} communication.
(For a more formal presentation, see Section~\ref{sec:reductions}.)

\subparagraph*{Hardness and completeness results.}

Given the above notion of \kl{reduction},
our definitions of \kl{hardness} and \kl{completeness}
for different levels of the \kl{locally polynomial hierarchy}
should come as no surprise:
a \kl{graph property}~$\Property$
is \kl{hard} for a complexity class~$\Class$
if there is a \kl{locally polynomial reduction} to~$\Property$
from every \kl{graph property} in~$\Class$,
and $\Property$~is \kl{complete} for~$\Class$
if $\Property$~itself additionally lies in that class.

The basic approach to establishing \kl{reductions}
between specific problems in our setting
is quite similar to that in the centralized setting,
so conventional techniques continue to work well.
Using fairly simple constructions,
we can show that
\kl{Eulerianness} is $\LP$\nobreakdash-\kl{complete},
while \kl{Hamiltonicity} is both $\LP$-\kl{hard} and $\coLP$-\kl{hard}.
Because of the incomparability of~$\LP$ and~$\coLP$
(see~Figure~\ref{fig:hierarchy-full-overview} and its caption),
this immediately tells us that
\kl{Hamiltonicity} is a strictly harder problem than \kl{Eulerianness}
in our model of computation.
What's more,
we can sometimes even build directly on classical \kl{reductions}
by extending them to the distributed setting.
In particular,
we can generalize the Cook--Levin theorem from~$\NP$ to $\NLP$ and,
based on that,
establish the $\NLP$-\kl{completeness} of $3$-\kl{colorability}.
Again,
this has a direct implication:
$3$-\kl{colorability} is neither in~$\LP$ nor in $\coNLP$,
since both classes are separate from~$\NLP$.
We now sketch two of the above \kl{reductions} as examples.
(More details are given in Section~\ref{sec:reductions}.)

\subparagraph*{An LP-hardness proof.}

To show that \kl{Hamiltonicity} is $\LP$-\kl{hard},
we provide a \kl{reduction} to it from $\ALLSELECTED$,
a trivially $\LP$-\kl{complete} \kl{graph property}
that requires all \kl{nodes} to be \kl{labeled} with the bit string~$1$.
This \kl{reduction} is illustrated
in Figure~\ref{fig:allselected-to-hamiltonian-overview}.

Given an arbitrary \kl{graph}~$\Graph$,
we construct a \kl{graph}~$\Graph'$
that has a \kl{Hamiltonian cycle}
if and only if
all \kl{nodes} of~$\Graph$ have \kl{label}~$1$.
The main idea is that
a \kl{Hamiltonian cycle} in~$\Graph'$
represents a depth-first traversal of a \kl{spanning tree} of~$\Graph$,
using a method known as the Euler tour technique.
For this purpose,
each \kl{edge} of~$\Graph$ is represented by two \kl{edges} in~$\Graph'$,
so that it can be traversed twice by a \kl{Hamiltonian cycle} in~$\Graph'$.
If all \kl{nodes} of~$\Graph$ are \kl{labeled} with~$1$,
then any \kl{spanning tree} of~$\Graph$
yields a \kl{Hamiltonian cycle} of~$\Graph'$.
However,
if at least one \kl{node} of~$\Graph$
has a \kl{label} different from~$1$
(such as \kl{node}~$\Node[1]_2$
in Figure~\ref{fig:allselected-to-hamiltonian-overview}),
then our construction includes an additional \kl{node} of \kl{degree}~$1$
to ensure that $\Graph'$ is not \kl{Hamiltonian}.
Note that
the \kl{nodes} of~$\Graph$ can compute~$\Graph'$
in a constant number of \kl{communication rounds}
and a number of \kl{computation steps} polynomial
in the size of their local input and the messages they receive.
(For more details,
see Proposition~\ref{prp:hamiltonian-lp-hard}
on page~\pageref{prp:hamiltonian-lp-hard}.)

\begin{figure}[htb]
  \centering
  \begin{tikzpicture}[
    semithick,>=stealth',on grid,
    vertex/.style={draw,fill=white,circle,minimum size=1.8ex,inner sep=0},
    named vertex/.style={draw,circle,inner sep=0,minimum size=4ex},
    box/.style={draw,dotted,rounded corners=2ex},
    hamilton/.style={ultra thick},
    label/.style={inner sep=0,align=left},
  ]
  \def\unitDist{3ex}
  \def\horizDist{22ex}
  \def\vertDist{18ex}
  \def\boxMargin{7ex}
  \def\cycleRad{5.5ex}
  % u1's gadget
  \coordinate (u1_mid);
  \draw[hamilton] ($(u1_mid)+(60:\cycleRad)$) arc (60:210:\cycleRad);
  \draw[hamilton] ($(u1_mid)+(30:\cycleRad)$) arc (30:-30:\cycleRad);
  \draw[hamilton] ($(u1_mid)+(-120:\cycleRad)$) arc (-120:-60:\cycleRad);
  \draw ($(u1_mid)+(30:\cycleRad)$) arc (30:60:\cycleRad);
  \draw ($(u1_mid)+(-30:\cycleRad)$) arc (-30:-60:\cycleRad);
  \draw ($(u1_mid)+(-150:\cycleRad)$) arc (-150:-120:\cycleRad);
  \node[vertex] (u1_east1) at ($(u1_mid)+(60:\cycleRad)$) {};
  \node[vertex] (u1_east2) at ($(u1_mid)+(30:\cycleRad)$) {};
  \node[vertex] (u1_southeast1) at ($(u1_mid)+(-60:\cycleRad)$) {};
  \node[vertex] (u1_southeast2) at ($(u1_mid)+(-30:\cycleRad)$) {};
  \node[vertex] (u1_south1) at ($(u1_mid)+(-150:\cycleRad)$) {};
  \node[vertex] (u1_south2) at ($(u1_mid)+(-120:\cycleRad)$) {};
  % u2's gadget
  \node[vertex] (u2_mid) at ($(u1_mid)+(0:\horizDist)$) {};
  \draw[hamilton] ($(u2_mid)+(120:\cycleRad)$) arc (120:-210:\cycleRad);
  \draw ($(u2_mid)+(120:\cycleRad)$) arc (120:150:\cycleRad);
  \node[vertex] (u2_west1) at ($(u2_mid)+(120:\cycleRad)$) {};
  \node[vertex] (u2_west2) at ($(u2_mid)+(150:\cycleRad)$) {};
  \node[vertex] (u2_south1) at ($(u2_mid)+(-60:\cycleRad)$) {};
  \node[vertex] (u2_south2) at ($(u2_mid)+(-30:\cycleRad)$) {};
  \path (u2_west1) edge (u2_mid);
  % u3's gadget
  \coordinate (u3_mid) at ($(u1_mid)+(-90:\vertDist)$);
  \draw[hamilton] ($(u3_mid)+(120:\cycleRad)$) arc (120:-210:\cycleRad);
  \draw ($(u3_mid)+(120:\cycleRad)$) arc (120:150:\cycleRad);
  \node[vertex] (u3_east1) at ($(u3_mid)+(-30:\cycleRad)$) {};
  \node[vertex] (u3_east2) at ($(u3_mid)+(-60:\cycleRad)$) {};
  \node[vertex] (u3_north1) at ($(u3_mid)+(150:\cycleRad)$) {};
  \node[vertex] (u3_north2) at ($(u3_mid)+(120:\cycleRad)$) {};
  % u4's gadget
  \coordinate (u4_mid) at ($(u2_mid)+(-90:\vertDist)$);
  \draw[hamilton] ($(u4_mid)+(120:\cycleRad)$) arc (120:-210:\cycleRad);
  \draw ($(u4_mid)+(120:\cycleRad)$) arc (120:150:\cycleRad);
  \node[vertex] (u4_west1) at ($(u4_mid)+(-150:\cycleRad)$) {};
  \node[vertex] (u4_west2) at ($(u4_mid)+(-120:\cycleRad)$) {};
  \node[vertex] (u4_northwest1) at ($(u4_mid)+(150:\cycleRad)$) {};
  \node[vertex] (u4_northwest2) at ($(u4_mid)+(120:\cycleRad)$) {};
  \node[vertex] (u4_north1) at ($(u4_mid)+(60:\cycleRad)$) {};
  \node[vertex] (u4_north2) at ($(u4_mid)+(30:\cycleRad)$) {};
  % Connections between the gadgets
  \path[hamilton]
    (u1_east1) edge (u2_west1)
    (u1_east2) edge (u2_west2)
    (u1_southeast1) edge (u4_northwest1)
    (u1_southeast2) edge (u4_northwest2)
    (u1_south1) edge (u3_north1)
    (u1_south2) edge (u3_north2);
  \path
    (u2_south1) edge (u4_north1)
    (u2_south2) edge (u4_north2)
    (u3_east1) edge (u4_west1)
    (u3_east2) edge (u4_west2);
  % Cluster boxes
  \node[box] (cluster_u1) [fit={($(u1_mid)+(-\boxMargin,\boxMargin)$)
                                ($(u1_mid)+(\boxMargin,-\boxMargin)$)},rounded corners=6ex] {};
  \node[box] (cluster_u2) [fit={($(u2_mid)+(-\boxMargin,\boxMargin)$)
                                ($(u2_mid)+(\boxMargin,-\boxMargin)$)},rounded corners=6ex] {};
  \node[box] (cluster_u3) [fit={($(u3_mid)+(-\boxMargin,\boxMargin)$)
                                ($(u3_mid)+(\boxMargin,-\boxMargin)$)},rounded corners=6ex] {};
  \node[box] (cluster_u4) [fit={($(u4_mid)+(-\boxMargin,\boxMargin)$)
                                ($(u4_mid)+(\boxMargin,-\boxMargin)$)},rounded corners=6ex] {};
  \node[label,anchor=south] at ($(cluster_u1.north)+(0,0.3*\unitDist)$)
       {\kl{cluster} of~$\Node[1]_1$};
  \node[label,anchor=south] at ($(cluster_u2.north)+(0,0.3*\unitDist)$)
       {\kl{cluster} of~$\Node[1]_2$};
  \node[label,anchor=north] at ($(cluster_u3.south)+(0,-0.3*\unitDist)$)
       {\kl{cluster} of~$\Node[1]_3$};
  \node[label,anchor=north] at ($(cluster_u4.south)+(0,-0.3*\unitDist)$)
       {\kl{cluster} of~$\Node[1]_4$};
  \node (graph') [label,above left=0*\unitDist and 0*\unitDist of cluster_u1.north west,anchor=north east]
                 {$\Graph'$:};
  % Original graph
  \node[named vertex] (u1) [left=14*\unitDist of cluster_u1] {$\Node[1]_1$};
  \node[named vertex] (u2) at ($(u1)+(0:\horizDist)$)  {$\Node[1]_2$};
  \node[named vertex] (u3) at ($(u1)+(-90:\vertDist)$) {$\Node[1]_3$};
  \node[named vertex] (u4) at ($(u2)+(-90:\vertDist)$) {$\Node[1]_4$};
  \node[label,above=0.3*\unitDist of u1.north] {$1$};
  \node[label,above=0.3*\unitDist of u2.north] {$0$};
  \node[label,below=0.3*\unitDist of u3.south] {$1$};
  \node[label,below=0.3*\unitDist of u4.south] {$1$};
  \node[label] at ($(u1 |- graph')+(-1.4*\unitDist,0)$) {$\Graph$:};
  \path[hamilton]
    (u1) edge (u2)
         edge (u3)
         edge (u4);
  \path
    (u2) edge (u4)
    (u3) edge (u4);
\end{tikzpicture}

%%% Local Variables:
%%% mode: latex
%%% TeX-master: "../lph-paper"
%%% End:
  \caption{
    Example illustrating
    the \kl{reduction} from $\ALLSELECTED$ to~$\HAMILTONIAN$,
    used to show that the latter is $\LP$-\kl{hard}.
    The \kl{graph}~$\Graph$ has all \kl{node labels} equal to~$1$
    if and only if
    the \kl{graph}~$\Graph'$ has a \kl{Hamiltonian cycle}.
    The thick \kl{edges} in~$\Graph$ form a \kl{spanning tree},
    which is replicated by the thick \kl{edges} in~$\Graph'$.
    In this particular case,
    if \kl{node}~$\Node[1]_2$ of~$\Graph$ had \kl{label}~$1$,
    then its \kl{cluster} in~$\Graph'$ would lack the “central” \kl{node},
    and thus the thick \kl{edges} in~$\Graph'$
    would form a \kl{Hamiltonian cycle}.
  }
  \label{fig:allselected-to-hamiltonian-overview}
\end{figure}

\subparagraph*{An NLP-completeness proof.}

$3$-\kl{colorability} clearly lies in~$\NLP$.
To show that it is also $\NLP$\nobreakdash-\kl{hard},
we build on the classical \kl{reduction}
from $\tSAT$
to the \kl{string-encoded} version~of $3$\nobreakdash-\kl{colorability},
which gives us the desired result almost for free.
Our extension of this construction to the distributed setting
is illustrated in Figure~\ref{fig:3satgraph-to-3colorable-overview}.
Here we generalize $\tSAT$ to \kl{graphs} as follows
to obtain an $\NLP$-\kl{complete} \kl(graph){property}:
each \kl{node} of the input \kl{graph} is \kl{labeled} with a \kl{Boolean formula},
and the \kl{graph} is said to be \kl(graph){satisfiable}
if there exists
an assignment of \kl(Boolean){variable} \kl{valuations} to its \kl{nodes}
such that each \kl{valuation}
satisfies the \kl(Boolean){formula} of the corresponding \kl{node}
while being consistent with the \kl{valuations} of all \kl{adjacent} \kl{nodes}.
(Two \kl{adjacent} \kl{nodes} can have different \kl(Boolean){variables},
but any \kl(Boolean){variables} shared by both must be assigned the same~values.)

Given an input \kl{graph}~$\Graph$,
we construct a \kl{graph}~$\Graph'$
that is $3$-\kl{colorable} if and only if $\Graph$~is \kl(graph){satisfiable}.
For this purpose,
each \kl{node} of~$\Graph$ is represented by a \kl{cluster}
that encodes its \kl(Boolean){formula}
in such a way that
a valid $3$\nobreakdash-\kl{coloring} of the \kl{cluster}
represents a satisfying \kl{valuation} of the \kl(Boolean){formula}.
This is done by directly using the classical construction as it stands.
In addition,
to ensure that
the \kl(Boolean){variable} \kl{valuations} of \kl{adjacent} \kl{nodes}
are consistent for all shared \kl(Boolean){variables},
the corresponding \kl{clusters} are connected
with auxiliary gadgets
that force certain \kl{nodes} to have the same color.
Again,
the \kl{nodes} of~$\Graph$ can compute~$\Graph'$
in a constant number of \kl{communication rounds}
and a polynomial number of \kl{computation steps}.
(For more details,
see Theorem~\ref{thm:three-colorable}
on page~\pageref{thm:three-colorable};
for the classical \kl{reduction},
see, e.g., \cite[Prp.~2.27]{DBLP:books/daglib/0019967}.)

\begin{figure}[htb]
  \centering
  \begin{tikzpicture}[
    semithick,>=stealth',on grid,
    vertex/.style={draw,circle,minimum size=1.8ex,inner sep=0},
    named vertex/.style={draw,circle,inner sep=0,minimum size=4ex},
    box/.style={draw,dotted,rounded corners=2ex},
    label/.style={inner sep=0,align=left},
  ]
  \def\nodeDist{3ex}
  % Nodes of u's gadget
  \node[vertex] (false_u) {};
  \node[vertex] (ground_u) [below=2.0*\nodeDist of false_u] {};
  \node[vertex] (P2_u)    [below right=2.0*\nodeDist and 0.5*\nodeDist of ground_u] {};
  \node[vertex] (notP2_u) [below left =2.0*\nodeDist and 0.5*\nodeDist of ground_u] {};
  \node[vertex] (P1_u)    [left=1.5*\nodeDist of notP2_u] {};
  \node[vertex] (notP1_u) [left=\nodeDist of P1_u] {};
  \node[vertex] (notP3_u) [right=1.5*\nodeDist of P2_u] {};
  \node[vertex] (P3_u)    [right=\nodeDist of notP3_u] {};
  \node[vertex] (clause1_u) [below left=2.4*\nodeDist and 0.5*\nodeDist of notP2_u] {};
  \node[vertex] (clause2_u) [right=\nodeDist of clause1_u] {};
  \node[vertex] (clause3_u) at ($(clause1_u)+(-60:\nodeDist)$) {};
  \node[vertex] (clause4_u) [below=\nodeDist of clause3_u] {};
  \node[vertex] (clause5_u) [right=\nodeDist of clause4_u] {};
  \node[vertex] (clause6_u) at ($(clause4_u)+(-60:\nodeDist)$) {};
  \node[box] (clause_u) [fit=(clause1_u)(clause2_u)(clause3_u)(clause4_u)(clause5_u)(clause6_u)] {};
  \node[box] (formula_u) [fit={($(false_u)+(0,1.4*\nodeDist)$)
                               ($(notP1_u)-(2.4*\nodeDist,0)$)
                               ($(P3_u)+(0.8*\nodeDist,0)$)
                               ($(clause6_u)-(0,1.0*\nodeDist)$)},rounded corners=6ex] {};
  \node[label,above=0.1*\nodeDist of false_u.north] {$\mathit{false}$};
  \node[label,anchor=east,above left=0.15*\nodeDist and 0*\nodeDist of ground_u.west] {$\mathit{ground}$};
  \node[label,anchor=center,below=0.25*\nodeDist of P1_u.south] {$\BoolVar_1$};
  \node[label,anchor=center,below=0.1*\nodeDist of notP1_u.south] {$\Negated{\BoolVar}_1$};
  \node[label,anchor=center,below=0.25*\nodeDist of P2_u.south] {$\BoolVar_2$};
  \node[label,anchor=center,below=0.1*\nodeDist of notP2_u.south] {$\Negated{\BoolVar}_2$};
  \node[label,anchor=center,below=0.25*\nodeDist of P3_u.south] {$\BoolVar_3$};
  \node[label,anchor=center,below=0.1*\nodeDist of notP3_u.south] {$\Negated{\BoolVar}_3$};
  \node[label,rotate=90,anchor=north] at ($(clause_u.east)+(0.15*\nodeDist,0)$) {clause \\[-0.8ex] gadget};
  \node[label,rotate=90,anchor=north] at ($(formula_u.east |- clause_u)+(0.15*\nodeDist,0)$)
       {\kl(Boolean){formula} \\[-0.8ex] gadget};
  % Nodes of v's gadget
  \node[vertex] (false_v) [right=17*\nodeDist of false_u] {};
  \node[vertex] (ground_v) [below=2.0*\nodeDist of false_v] {};
  \node[vertex] (P4_v)    [below left =2.0*\nodeDist and 0.5*\nodeDist of ground_v] {};
  \node[vertex] (notP4_v) [below right=2.0*\nodeDist and 0.5*\nodeDist of ground_v] {};
  \node[vertex] (notP3_v) [left=1.5*\nodeDist of P4_v] {};
  \node[vertex] (P3_v)    [left=\nodeDist of notP3_v] {};
  \node[vertex] (P5_v)    [right=1.5*\nodeDist of notP4_v] {};
  \node[vertex] (notP5_v) [right=\nodeDist of P5_v] {};
  \node[vertex] (clause1_v) [below left=2.4*\nodeDist and 0.5*\nodeDist of P4_v] {};
  \node[vertex] (clause2_v) [right=\nodeDist of clause1_v] {};
  \node[vertex] (clause3_v) at ($(clause1_v)+(-60:\nodeDist)$) {};
  \node[vertex] (clause4_v) [below=\nodeDist of clause3_v] {};
  \node[vertex] (clause5_v) [right=\nodeDist of clause4_v] {};
  \node[vertex] (clause6_v) at ($(clause4_v)+(-60:\nodeDist)$) {};
  \node[box] (clause_v) [fit=(clause1_v)(clause2_v)(clause3_v)(clause4_v)(clause5_v)(clause6_v)] {};
  \node[box] (formula_v) [fit={($(false_v)+(0,1.4*\nodeDist)$)
                               ($(P3_v)-(0.8*\nodeDist,0)$)
                               ($(notP5_v)+(2.4*\nodeDist,0)$)
                               ($(clause6_v)-(0,1.0*\nodeDist)$)},rounded corners=6ex] {};
  \node[label,above=0.1*\nodeDist of false_v.north] {$\mathit{false}$};
  \node[label,anchor=west,above right=0.15*\nodeDist and 0.1*\nodeDist of ground_v.east] {$\mathit{ground}$};
  \node[label,anchor=center,below=0.25*\nodeDist of P3_v.south] {$\BoolVar_3$};
  \node[label,anchor=center,below=0.1*\nodeDist of notP3_v.south] {$\Negated{\BoolVar}_3$};
  \node[label,anchor=center,below=0.25*\nodeDist of P4_v.south] {$\BoolVar_4$};
  \node[label,anchor=center,below=0.1*\nodeDist of notP4_v.south] {$\Negated{\BoolVar}_4$};
  \node[label,anchor=center,below=0.25*\nodeDist of P5_v.south] {$\BoolVar_5$};
  \node[label,anchor=center,below=0.1*\nodeDist of notP5_v.south] {$\Negated{\BoolVar}_5$};
  \node[label,rotate=90,anchor=south] at ($(clause_v.west)-(0.15*\nodeDist,0)$) {clause \\[-0.8ex] gadget};
  \node[label,rotate=90,anchor=south] at ($(formula_v.west |- clause_v)-(0.15*\nodeDist,0)$)
       {\kl(Boolean){formula} \\[-0.8ex] gadget};
  % Nodes of connector gadgets
  \node[vertex] (false_uv) at ($(false_u)!0.4!(false_v)$) {};
  \node[vertex] (false_vu) at ($(false_u)!0.6!(false_v)$) {};
  \node[vertex] (ground_uv) at ($(ground_u)!0.4!(ground_v)$) {};
  \node[vertex] (ground_vu) at ($(ground_u)!0.6!(ground_v)$) {};
  \node[vertex] (P3_uv) [below=2.0*\nodeDist of ground_uv] {};
  \node[vertex] (P3_vu) [below=2.0*\nodeDist of ground_vu] {};
  % Cluster boxes
  \node[box] (cluster_u) [fit={($(false_uv)+(0.8*\nodeDist,2.0*\nodeDist)$)
                               ($(notP1_u)-(3.0*\nodeDist,0)$)
                               ($(clause6_u)-(0,1.6*\nodeDist)$)},rounded corners=6ex] {};
  \node[box] (cluster_v) [fit={($(false_vu)+(-0.8*\nodeDist,2.0*\nodeDist)$)
                               ($(notP5_v)+(3.0*\nodeDist,0)$)
                               ($(clause6_v)-(0,1.6*\nodeDist)$)},rounded corners=6ex] {};
  \node[label,anchor=north] at ($(cluster_u.south)+(0,-0.3*\nodeDist)$)
       {\kl{cluster} representing \kl{node}~$\Node[1]$};
  \node[label,anchor=north] at ($(cluster_v.south)+(0,-0.3*\nodeDist)$)
       {\kl{cluster} representing \kl{node}~$\Node[2]$};
  \node (graph') [label,above left=0.2*\nodeDist and -0.2*\nodeDist of cluster_u.north west,anchor=north east]
                 {$\Graph'$:};
  % Edges of u's gadget
  \path (false_u) edge (ground_u)
        (ground_u) edge (P1_u)
        (ground_u) edge (notP1_u)
        (ground_u) edge (P2_u)
        (ground_u) edge (notP2_u)
        (ground_u) edge (P3_u)
        (ground_u) edge (notP3_u)
        (P1_u) edge (notP1_u)
        (P2_u) edge (notP2_u)
        (P3_u) edge (notP3_u)
        (clause1_u) edge (clause2_u)
        (clause1_u) edge (clause3_u)
        (clause2_u) edge (clause3_u)
        (clause3_u) edge (clause4_u)
        (clause4_u) edge (clause5_u)
        (clause4_u) edge (clause6_u)
        (clause5_u) edge (clause6_u)
        (P1_u)    edge[out=-55,in=90]  (clause1_u)
        (notP2_u) edge[out=-55,in=90]  (clause2_u)
        (notP3_u) edge[out=-130,in=90] (clause5_u)
        (clause6_u) edge[out=170,in=180,looseness=1.8] (false_u)
        (clause6_u) edge[out=160,in=180,looseness=1.9] (ground_u)
        ;
  % Edges of v's gadget
  \path (false_v) edge (ground_v)
        (ground_v) edge (P3_v)
        (ground_v) edge (notP3_v)
        (ground_v) edge (P4_v)
        (ground_v) edge (notP4_v)
        (ground_v) edge (P5_v)
        (ground_v) edge (notP5_v)
        (P3_v) edge (notP3_v)
        (P4_v) edge (notP4_v)
        (P5_v) edge (notP5_v)
        (clause1_v) edge (clause2_v)
        (clause1_v) edge (clause3_v)
        (clause2_v) edge (clause3_v)
        (clause3_v) edge (clause4_v)
        (clause4_v) edge (clause5_v)
        (clause4_v) edge (clause6_v)
        (clause5_v) edge (clause6_v)
        (P3_v.-48)     edge[out=-75,in=175] (clause1_v)
        (P4_v)         edge[out=-55,in=90]  (clause2_v)
        (notP5_v.-130) edge[out=-112,in=16,looseness=1.1] (clause5_v)
        (clause6_v) edge[out=10,in=0,looseness=1.8] (false_v)
        (clause6_v) edge[out=20,in=0,looseness=1.9] (ground_v)
        ;
  % Edges of connector gadgets
  \path (false_u)  edge (false_uv)
        (false_u)  edge[bend left=13] (false_vu)
        (false_v)  edge[bend left=13] (false_uv)
        (false_v)  edge (false_vu)
        (false_uv) edge (false_vu)
        (ground_u)  edge (ground_uv)
        (ground_u)  edge[bend left=13] (ground_vu)
        (ground_v)  edge[bend left=13] (ground_uv)
        (ground_v)  edge (ground_vu)
        (ground_uv) edge (ground_vu)
        (P3_u)  edge (P3_uv)
        (P3_u)  edge[bend left=15] (P3_vu)
        (P3_v)  edge[bend left=15] (P3_uv)
        (P3_v)  edge (P3_vu)
        (P3_uv) edge (P3_vu)
        ;
  % Original graph
  \node[named vertex] (u) [above=1.0*\nodeDist of cluster_u.north] {$\Node[1]$};
  \node[named vertex] (v) [above=1.0*\nodeDist of cluster_v.north] {$\Node[2]$};
  \node[label,left=0.5*\nodeDist of u.west] {
    $\BoolVar_1 \OR \Negated{\BoolVar}_2 \OR \Negated{\BoolVar}_3$};
  \node[label,right=0.5*\nodeDist of v.east] {
    $\BoolVar_3 \OR \BoolVar_4 \OR \Negated{\BoolVar}_5$};
  \node[label] at (graph' |- u) {$\Graph$:};
  \path (u) edge (v);
\end{tikzpicture}

%%% Local Variables:
%%% mode: latex
%%% TeX-master: "../lph-paper"
%%% End:
  \caption{
    Example illustrating
    the \kl{reduction} from $\tSATGRAPH$ to~$\COLORABLE{3}$,
    used to show that the latter is $\NLP$-\kl{complete}.
    The \kl{Boolean graph}~$\Graph$ is \kl(graph){satisfiable}
    if and only if
    the \kl{graph}~$\Graph'$ is $3$-\kl{colorable}.
    The labels in~$\Graph'$ serve explanatory purposes only
    and are not part of the \kl{graph}.
  }
  \label{fig:3satgraph-to-3colorable-overview}
\end{figure}

% –––––––––––––––––––––––––––––––––––––––––––––––––––––––––––––––––––––––––––––––
\subsection{A logical characterization}

A central tool and recurring theme of this paper
is Fagin's theorem~\cite{Fagin74}.
In its original form
(see Theorem~\ref{thm:fagin} on page~\pageref{thm:fagin}),
it states that
a formal language lies in $\NP$
if and only if
it can be \kl{defined} by a \kl{formula} of
\kl[existential fragment]{existential second-order logic}.
Such \kl{formulas} are of the form
$\ExistsRel{\SOVar_1} \dots \ExistsRel{\SOVar_n} \, \Formula[1]$,\,
where
$\SOVar_1, \dots, \SOVar_n$
are \kl{second-order variables}
and $\Formula[1]$~is a \kl{first-order} \kl{formula}.
In the context of this paper,
these \kl{formulas} are evaluated
on bit strings represented as relational \kl{structures}.
More precisely,
the bits of a string are represented by
a sequence of \kl{elements} connected by a binary successor relation,
and their values are determined by a unary relation.
For instance,
the string $010011$ is represented by the \kl{structure}
\begin{center}
  \begin{tikzpicture}[
    baseline,
    semithick,>=stealth',
    bit/.style={draw,circle,minimum size=2ex,inner sep=0},
    set/.style={fill=black},
  ]
  \def\bitDist{4.5ex}
  \node[bit,anchor=base] (u1) {\phantom{x}};
  \node[bit,set] (u2) at ([shift={(0:\bitDist)}]u1) {};
  \node[bit]     (u3) at ([shift={(0:\bitDist)}]u2) {};
  \node[bit]     (u4) at ([shift={(0:\bitDist)}]u3) {};
  \node[bit,set] (u5) at ([shift={(0:\bitDist)}]u4) {};
  \node[bit,set] (u6) at ([shift={(0:\bitDist)}]u5) {};
  \path[->] (u1) edge (u2)
  (u2) edge (u3)
  (u3) edge (u4)
  (u4) edge (u5)
  (u5) edge (u6);
\end{tikzpicture}%
%
%%% Local Variables:
%%% mode: latex
%%% TeX-master: "../lph-paper"
%%% End:
\:\!,
\end{center}
where the \kl{elements} belonging to the unary relation are marked in~black.

\subparagraph*{Extension to the distributed setting.}

We show that Fagin's result can be generalized
to obtain a similar logical characterization of the class~$\NLP$
(see Theorem~\ref{thm:local-fagin} on page~\pageref{thm:local-fagin}).
To evaluate \kl{logical formulas} on \kl{labeled graphs},
we use the \kl{structural representation}
illustrated in Figure~\ref{fig:graph}.
This \kl[structural representation]{representation}
contains an \kl{element}
for every \kl{node} and every \kl{labeling bit} of the \kl{graph}.
The \kl{nodes} are connected symmetrically to their \kl{neighbors}
and asymmetrically to their \kl{labeling bits}
by two binary relations.
In turn,
the \kl{labeling bits} of each \kl{node}
are interconnected by a successor relation
and assigned a value by a unary relation,
just as in the string representation described~above.

\begin{figure}[htb]
  \centering
  \begin{tikzpicture}[
    semithick,>=stealth',
    vertex/.style={draw,circle,minimum size=2ex},
    label/.style={},
    bit/.style={vertex},
    set/.style={fill=black},
  ]
  \def\nodeDist{15ex}
  \def\labelDist{3ex}
  \def\bitDist{4.5ex}
  \def\bitNodeDist{6ex}
  \begin{scope}
    \node[vertex] (u) {};
    \node[vertex] (v) at ([shift={(-120:\nodeDist)}]u) {};
    \node[vertex] (w) at ([shift={(-60:\nodeDist)}]u) {};
    \node[vertex] (x) at ([shift={(0:\nodeDist)}]u) {};
    \node[label] at ([shift={(+90:\labelDist)}]u) {010};
    \node[label] at ([shift={(-90:\labelDist)}]v) {10};
    \node[label] at ([shift={(-90:\labelDist)}]w) {1101};
    \node[label] at ([shift={(+90:\labelDist)}]x) {001};
    \node[label] at ([shift={(180:6ex)}]u) {$\Graph$:};
    \path (u) edge (v)
          (u) edge (w)
          (u) edge (x)
          (v) edge (w);
  \end{scope}
  \begin{scope}[xshift=40ex]
    \node[vertex] (u) {};
    \node[vertex] (v) at ([shift={(-120:\nodeDist)}]u) {};
    \node[vertex] (w) at ([shift={(-60:\nodeDist)}]u) {};
    \node[vertex] (x) at ([shift={(0:\nodeDist)}]u) {};
    \node[bit,set] (u2) at ([shift={(90:\bitNodeDist)}]u) {};
    \node[bit]     (u1) at ([shift={(180:\bitDist)}]u2) {};
    \node[bit]     (u3) at ([shift={(0:\bitDist)}]u2) {};
    \node[bit,set] (v1) at ([xshift=-0.5*\bitDist,yshift=-\bitNodeDist]v) {};
    \node[bit]     (v2) at ([xshift=0.5*\bitDist,yshift=-\bitNodeDist]v) {};
    \node[bit,set] (w2) at ([xshift=-0.5*\bitDist,yshift=-\bitNodeDist]w) {};
    \node[bit]     (w3) at ([xshift=0.5*\bitDist,yshift=-\bitNodeDist]w) {};
    \node[bit,set] (w1) at ([shift={(180:\bitDist)}]w2) {};
    \node[bit,set] (w4) at ([shift={(0:\bitDist)}]w3) {};
    \node[bit]     (x2) at ([shift={(90:\bitNodeDist)}]x) {};
    \node[bit]     (x1) at ([shift={(180:\bitDist)}]x2) {};
    \node[bit,set] (x3) at ([shift={(0:\bitDist)}]x2) {};
    \node[label] at ([shift={(180:6ex)}]u) {$\StructRepr{\Graph}$:};
    \path[<->] (u) edge (v)
               (u) edge (w)
               (u) edge (x)
               (v) edge (w);
    \path[->] (u1) edge (u2)
              (u2) edge (u3)
              (v1) edge (v2)
              (w1) edge (w2)
              (w2) edge (w3)
              (w3) edge (w4)
              (x1) edge (x2)
              (x2) edge (x3);
    \path[->,densely dotted]
              (u) edge (u1)
              (u) edge (u2)
              (u) edge (u3)
              (v) edge (v1)
              (v) edge (v2)
              (w) edge (w1)
              (w) edge (w2)
              (w) edge (w3)
              (w) edge (w4)
              (x) edge (x1)
              (x) edge (x2)
              (x) edge (x3);
  \end{scope}
\end{tikzpicture}

%%% Local Variables:
%%% mode: latex
%%% TeX-master: "../lph-paper"
%%% End:
  \caption{
    A \kl{labeled graph}~$\Graph$
    and its \kl{structural representation}~$\StructRepr{\Graph}$\!.
    For later reference:
    the unary relation~$\BitSet{1}{\StructRepr{\Graph}}$\!
    is represented by black \kl{elements},
    and the binary relations
    $\LinkRel{1}{\StructRepr{\Graph}}$\! and~$\LinkRel{2}{\StructRepr{\Graph}}$\!
    by solid and dotted arrows,
    respectively.
  }
  \label{fig:graph}
\end{figure}

Our generalization of Fagin's theorem states that
a \kl{graph property} lies in $\NLP$
if and only if
it can be \kl{defined}
(on \kl{structural representations} as above)
by a \kl{formula} of the following
fragment of \kl[existential fragment]{existential second-order logic}:
\kl{formulas} are of the form
$\ExistsRel{\SOVar_1} \dots \ExistsRel{\SOVar_n} \,
 \ForAll{\FOVar[1]} \, \Formula[1]$,\,
where
$\SOVar_1, \dots, \SOVar_n$
are \kl{second-order variables},
$\FOVar[1]$ is a \kl{first-order variable},
and $\Formula[1]$~is a \kl{first-order} \kl{formula}
in which all \kl{quantifiers} are \emph{\kl(quantifier){bounded}}
to range only over locally accessible \kl{elements}.
For instance,
\kl{existential quantification} must be of the form
$\ExistsNb{\FOVar[3]}{\FOVar[2]} \; \Formula[2]$,
which can be read as
“there exists an \kl{element}~$\FOVar[3]$
connected to a known \kl{element}~$\FOVar[2]$
such that \kl{formula}~$\Formula[2]$~is \kl{satisfied}”.
This means that
\kl{first-order quantification} in~$\Formula[1]$
is always relative to some \kl{element} already fixed at an outer scope,
and thus in effect that
$\Formula[1]$~is~\emph{\kl[bounded fragment]{bounded}}
around the \kl{variable}~$\FOVar[1]$.
(For formal definitions and a series of detailed examples,
see Section~\ref{sec:logic}.)

\begin{example}
  \label{ex:3-colorable-overview}
  We can \kl{define} $3$-\kl{colorability}
  with a \kl{second-order} \kl{formula} of the form
  \begin{equation*}
    \ExistsRel{C_0, C_1, C_2} \;
    \ForAll{\FOVar[1]} \, (\IsNode\Of{\FOVar[1]} \IMP \WellColored\Of{\FOVar[1]}),
  \end{equation*}
  where
  $C_0$, $C_1$ and~$C_2$ are \kl{unary} \kl{second-order variables}
  intended to represent
  the sets of \kl{nodes} with colors $0$, $1$ and~$2$,
  respectively,
  $\FOVar[1]$ is a \kl{first-order variable},
  and
  $\IsNode$ and $\WellColored$
  are \kl{first-order} \kl{formulas}
  \kl[bounded fragment]{bounded} around~$\FOVar[1]$.
  To ensure that $\FOVar[1]$ represents a \kl{node}
  (rather than a \kl{labeling bit}),
  $\IsNode\Of{\FOVar[1]}$ states that
  “$\FOVar[1]$ does not have any dotted arrow pointing to it”
  (see Figure~\ref{fig:graph}).
  To express that
  $\FOVar[1]$~is correctly colored,
  $\WellColored\Of{\FOVar[1]}$ states that
  “$\FOVar[1]$~is assigned exactly one color”
  and that
  “$\FOVar[1]$'s color differs from its \kl{neighbors}'~colors”.
  (We~will describe this more formally in Section~\ref{sec:logic};
  see Example~\ref{ex:3-colorable} on page~\pageref{ex:3-colorable}.)
  \lipicsEnd
\end{example}

\subparagraph*{Extension to higher levels of alternation.}

Stockmeyer~\cite{DBLP:journals/tcs/Stockmeyer76}
showed that Fagin's theorem extends
to the higher levels of the \kl{polynomial hierarchy}.
For example,
the complexity class~$\PiP{2}$ is characterized
by \kl{formulas} of the form
$\ForAllRel{\SOVar[1]_1} \dots \ForAllRel{\SOVar[1]_m} \,
 \ExistsRel{\SOVar[2]_1} \dots \ExistsRel{\SOVar[2]_n} \, \Formula[1]$,
consisting of a block of \kl{universal} \kl{second-order quantifiers},
followed by a block of \kl{existential} \kl{second-order quantifiers},
and then a \kl{first-order} \kl{formula}~$\Formula[1]$.
We similarly extend our generalization of Fagin's theorem
to the higher levels of the \kl{locally polynomial hierarchy}
(see Theorem~\ref{thm:local-hierarchy-equivalence}
on page~\pageref{thm:local-hierarchy-equivalence}).
For instance,
the complexity class~$\PiLP{2}$ is characterized
by \kl{formulas} of the form
$\ForAllRel{\SOVar[1]_1} \dots \ForAllRel{\SOVar[1]_m} \,
 \ExistsRel{\SOVar[2]_1} \dots \ExistsRel{\SOVar[2]_n} \,
 \ForAll{\FOVar[1]} \, \Formula[1]$,
where
the prefix of \kl{second-order quantifiers} is as above,
$\FOVar[1]$ is a \kl{first-order variable},
and $\Formula[1]$~is a \kl{first-order} \kl{formula}
\kl[bounded fragment]{bounded} around~$\FOVar[1]$.
We refer to
this logical characterization of the \kl{locally polynomial hierarchy}
as the \mbox{\emph{\kl{local second-order hierarchy}}}.
All \kl{graph properties} in this hierarchy
can be \kl{defined} by \kl{formulas} consisting of
alternating blocks of
\kl{existential} and \kl{universal} \kl{second-order quantifiers},
followed by a single \kl{universal} \kl{first-order quantifier},
and then a \kl[bounded fragment]{bounded} \kl{first-order} \kl{formula}.

For each alternation level,
we can recover Stockmeyer's result
by restricting our corresponding statement to \kl{single-node graphs}.
Indeed,
when the input \kl{graph} consists of a single \kl{node},
all \kl{elements} of its \kl{structural representation}
lie within distance~$2$ of each other,
so the distinction between
\kl(quantifier){bounded} and \kl(quantifier){unbounded}
\kl{first-order quantification}
becomes~irrelevant.

\subparagraph*{Proof outline.}

Unfortunately,
there does not seem to be a straightforward way
to use Fagin's original result as a black box
to prove our generalization.
So instead,
we give a full proof from scratch,
adapting the ideas of the original proof to the distributed setting.
Here we sketch only the nondeterministic case,
which is the easiest to present,
but the proof extends to arbitrarily high levels of quantifier alternation.

The easy part is
to translate a \kl{formula} of the form
$\ExistsRel{\SOVar_1} \dots \ExistsRel{\SOVar_n} \,
 \ForAll{\FOVar[1]} \, \Formula[1]$
into a \kl{distributed Turing machine}
that \kl{verifies} the same \kl(graph){property}.
In essence,
the \kl{certificates} chosen by the prover~(\kl(certificate){Eve})
are used to \kl{encode} the \kl{existentially quantified} relations
$\SOVar_1, \dots, \SOVar_n$,
so that the \kl{nodes} \kl{executing} the \kl{machine}
just have to run a local algorithm
to evaluate~$\Formula[1]$ in their constant-radius \kl{neighborhood}.
They can do this in a polynomial number of \kl{computation steps}
by simply iterating over all possible \kl{interpretations}
of the \kl{first-order variables} in~$\Formula$.

The reverse translation, however, is more complicated.
It involves encoding the space-time diagram
of every \kl[distributed Turing machine]{Turing machine} in the network
as a collection of relations
over the corresponding \kl{structural representation}.
The key insight that makes this possible is
the same as in Fagin's original proof:
since the number of \kl{computation steps} of each \kl{machine}
is polynomially bounded in the size of its input,
each cell of the corresponding space-time diagram
can be represented by a tuple of nearby \kl{elements}
whose length depends on the degree of the bounding polynomial.
What makes our generalized proof somewhat more cumbersome
are the additional technicalities imposed by the distributed setting,
in particular
the assignment of \kl{locally unique} \kl{identifiers}
(chosen small enough to be representable),
and the exchange of messages between \kl{adjacent}~\kl{nodes}.
The latter requires that
our \kl{formula} keep track of
the tape positions of sent and received messages
for each pair of \kl{adjacent} \kl{machines},
so that
the appropriate section of one \kl{machine}'s \kl{sending tape}
is copied to
the appropriate section of the other \kl{machine}'s \kl{receiving tape}.

\subparagraph*{Implications.}

Our generalization of Fagin's theorem
serves several purposes in this paper:
\begin{enumerate}
\item It provides evidence that
  our definition of the \kl{locally polynomial hierarchy} is robust,
  in that
  the complexity classes defined do not inherently depend on technical details
  such as the chosen model of computation.
\item It gives us a convenient way to prove
  our generalization of the Cook--Levin theorem mentioned above.
  This is analogous to the centralized setting,
  where the Cook--Levin theorem can be obtained as
  a corollary of Fagin's theorem.
\item As we will see next,
  we make extensive use of the provided connection to logic
  to prove that the \kl{locally polynomial hierarchy} is infinite.
\end{enumerate}

% –––––––––––––––––––––––––––––––––––––––––––––––––––––––––––––––––––––––––––––––
\subsection{Infiniteness of our hierarchy}

While in centralized computing
the question of whether $\PTIME$~equals~$\NP$
remains a major open problem,
the corresponding question in distributed computing%
---whether $\LP$~equals~$\NLP$---%
is easily settled with an elementary argument:
nondeterminism provides a means to break symmetry,
which is impossible in a purely deterministic setting
(see Proposition~\ref{prp:lp-vs-nlp} on page~\pageref{prp:lp-vs-nlp}).
What seems less obvious, however,
is how to separate complexity classes
higher up in the \kl{locally polynomial hierarchy},
especially when seeking separations that remain valid for
\kl{graphs} of bounded maximum \kl{degree} and constant \kl{label} size.
This is where our generalization of Fagin's theorem
proves particularly helpful,
as it allows us to reformulate the problem
in the well-studied framework of logic and automata theory.
Our separation proof builds on two results from that area,
both concerning \kl{monadic second-order logic} on \kl{pictures}.
The main ideas are outlined below.

\subparagraph*{Logic on pictures.}

\emph{\kl{Monadic second-order logic}}
is the fragment of \kl{second-order logic}
that can only \kl{quantify} over sets instead of arbitrary relations.
This means, for instance, that
the \kl{formulas} of
\kl[existential fragment]{existential monadic second-order logic}
are of the form
$\ExistsRel{\MSOVar_1} \dots \ExistsRel{\MSOVar_n} \, \Formula[1]$,
where
the \kl{variables}
$\MSOVar_1, \dots, \MSOVar_n$
represent sets of \kl{elements}
and $\Formula[1]$~is a \kl{first-order} \kl{formula}.

Meanwhile,
\emph{\kl{pictures}} are matrices of fixed-length binary strings.
To describe the \kl(picture){properties} of \kl{pictures}
using \kl{logical formulas},
every \kl{picture} is given a \kl(picture){structural representation}
as shown in Figure~\ref{fig:picture-overview}.
Specifically,
the entries of the \kl{picture} are represented by \kl{elements}
connected by a “vertical” and a “horizontal” successor relation,
and the value of each bit is represented by a unary relation.
(Formal definitions are given in Section~\ref{ssec:climb-up}.)

\begin{figure}[htb]
  \centering
  % https://tex.stackexchange.com/a/55604
\makeatletter
\tikzset{circle split part fill/.style  args={#1,#2}{%
    alias=tmp@name, % Jake's idea !!
    postaction={%
      insert path={%
        \pgfextra{%
          \pgfpointdiff{\pgfpointanchor{\pgf@node@name}{center}}%
          {\pgfpointanchor{\pgf@node@name}{east}}%
          \pgfmathsetmacro\insiderad{\pgf@x}
          \fill[#1] (\pgf@node@name.base) ([xshift=-\pgflinewidth]\pgf@node@name.east) arc
          (0:180:\insiderad-\pgflinewidth)--cycle;
          \fill[#2] (\pgf@node@name.base) ([xshift=\pgflinewidth]\pgf@node@name.west)  arc
          (180:360:\insiderad-\pgflinewidth)--cycle;
        }}}}}
\makeatother

\begin{tikzpicture}[
    semithick,>=stealth',on grid,
    inner sep=0,
    pixel/.style={minimum size=7ex},
    vertex/.style={circle split,rotate=90,minimum size=2ex},
    00/.style={draw=white,circle split part fill={white,white}},
    01/.style={draw=black,circle split part fill={white,black}},
    10/.style={draw=black,circle split part fill={black,white}},
    11/.style={draw=black,circle split part fill={black,black}},
  ]
  \def\pixelDist{7ex}
  \begin{scope}
    \matrix (p) [matrix of math nodes,left delimiter={[},right delimiter={]},
                 inner xsep=0ex,inner ysep=0.5ex,
                 row sep={\pixelDist,between origins},
                 column sep={\pixelDist,between origins},
                 matrix anchor=p-1-1.center] {
      00 & 01 & 00 & 01 \\
      10 & 11 & 10 & 11 \\
      00 & 01 & 00 & 01 \\
    };
    \node [left=6ex of p-1-1] {$\Picture$:};
  \end{scope}
  \begin{scope}[xshift=40ex]
    \node[vertex,00] (p11) {};
    \node[vertex,10] (p21) [below=\pixelDist of p11] {};
    \node[vertex,00] (p31) [below=\pixelDist of p21] {};
    \node[vertex,01] (p12) [right=\pixelDist of p11] {};
    \node[vertex,11] (p22) [below=\pixelDist of p12] {};
    \node[vertex,01] (p32) [below=\pixelDist of p22] {};
    \node[vertex,00] (p13) [right=\pixelDist of p12] {};
    \node[vertex,10] (p23) [below=\pixelDist of p13] {};
    \node[vertex,00] (p33) [below=\pixelDist of p23] {};
    \node[vertex,01] (p14) [right=\pixelDist of p13] {};
    \node[vertex,11] (p24) [below=\pixelDist of p14] {};
    \node[vertex,01] (p34) [below=\pixelDist of p24] {};
    \node [left=5ex of p11] {$\StructReprP{\Picture}$:};
    \foreach \n in {p11,p31,p13,p33}
      \node[draw,circle,minimum size=2ex] at (\n) {};
    \path[->] (p11) edge (p21)
              (p21) edge (p31)
              (p12) edge (p22)
              (p22) edge (p32)
              (p13) edge (p23)
              (p23) edge (p33)
              (p14) edge (p24)
              (p24) edge (p34);
    \path[->,densely dotted]
              (p11) edge (p12)
              (p12) edge (p13)
              (p13) edge (p14)
              (p21) edge (p22)
              (p22) edge (p23)
              (p23) edge (p24)
              (p31) edge (p32)
              (p32) edge (p33)
              (p33) edge (p34);
  \end{scope}
\end{tikzpicture}

%%% Local Variables:
%%% mode: latex
%%% TeX-master: "../lph-paper"
%%% End:
  \caption{
    A $2$-bit \kl{picture}~$\Picture$
    and its \kl(picture){structural representation}~$\StructReprP{\Picture}$.
  }
  \label{fig:picture-overview}
\end{figure}

\subparagraph*{Proof outline.}

Our proof of the infiniteness of the \kl{locally polynomial hierarchy}
consists of two main parts,
which remain mostly in the realm of logic:
\begin{enumerate}
\item First,
  we show that the \kl{local second-order hierarchy} is infinite
  when restricted to \kl{pictures},
  and more precisely that
  all levels ending with a block of existential quantifiers are distinct
  (see Section~\ref{sssec:digression-on-pictures}).
  This is obtained by combining the following two results:
  \begin{enumerate}
  \item We show that
    the \kl[local second-order hierarchy]{local} and
    the \kl{monadic second-order hierarchies}
    on \kl{pictures} are levelwise equivalent
    for all levels ending with a block of existential quantifiers.
    The main ingredient to prove this is
    an automata-theoretic characterization of
    \kl[existential fragment]{existential monadic second-order logic}
    on \kl{pictures}
    due to Giammarresi, Restivo, Seibert, and Thomas~%
    \cite{DBLP:journals/iandc/GiammarresiRST96}
    (see Theorem~\ref{thm:equivalence-ts-emso}
    on page~\pageref{thm:equivalence-ts-emso}).
    This characterization,
    which is itself based on a locality property of \kl{first-order logic},
    gives us a convenient way to establish the equivalence
    of the \kl{existential fragments}
    of \kl[local second-order logic]{local} and \kl{monadic second-order logic}
    on \kl{pictures}.
    For the higher levels of the hierarchies,
    the equivalence is then obtained
    by induction on the number of quantifier alternations.
  \item Matz, Schweikardt, and Thomas~\cite{DBLP:journals/iandc/MatzST02}
    have shown that
    the \kl{monadic second-order hierarchy} on \kl{pictures} is infinite
    (see Theorem~\ref{thm:mso-hierarchy} on page~\pageref{thm:mso-hierarchy}).
    Interestingly,
    one way to prove their result
    is based on the automata-theoretic characterization mentioned above,
    and thus ultimately on the same locality property of \kl{first-order logic}.
  \end{enumerate}
\item Second,
  we transfer the previous infiniteness result
  for the \kl{local second-order hierarchy}
  from \kl{pictures} to \kl{graphs}
  (see Section~\ref{sssec:from-pictures-to-graphs}).
  We do this by \kl(graph){encoding} \kl{pictures} as \kl{graphs}
  in such a way that
  \kl{formulas} can be translated from one type of \kl{structure} to the other
  without changing the alternation level of \kl{second-order quantifiers}.
  By our generalization of Fagin's theorem,
  this implies that
  all levels of the \kl{locally polynomial hierarchy}
  ending with a block of existential quantifiers are distinct,
  even when restricted to
  \kl{graphs} of bounded maximum \kl{degree} and constant \kl{label} size
  (see Theorem~\ref{thm:locally-polynomial-hierarchy}
  on page~\pageref{thm:locally-polynomial-hierarchy}).
  We then complete this partial separation result
  with some additional arguments
  to arrive at the fuller separation result
  shown in Figure~\ref{fig:hierarchy-full-overview}
  (see Sections~\ref{ssec:warm-up} and~\ref{ssec:complete-picture}).
\end{enumerate}
While the basic infiniteness result
could also be proved without an explicit detour through \kl{pictures},
a more direct approach would not allow us to easily separate
individual levels of the \kl[locally polynomial hierarchy]{hierarchy}
on \kl{graphs} of bounded maximum \kl{degree} and constant \kl{label} size
(see~Section~\ref{sssec:direct-infiniteness-proof}).

\subparagraph*{Implications.}

Besides the result itself,
we derive two main benefits from
the infiniteness of the \kl{locally polynomial hierarchy}:
\begin{enumerate}
\item \kl{Hardness} results give us unconditional lower bounds.
  That is,
  if we can show that a \kl{graph property} is \kl{hard}
  for some class of the \kl[locally polynomial hierarchy]{hierarchy},
  then we immediately know that it does not lie
  in any class below or incomparable to that class
  (see Corollaries~%
  \ref{cor:three-colorable-not-in-lp},
  \ref{cor:non-three-colorable-not-in-nlp}, and
  \ref{cor:noneulerian-hamiltonian-nonhamiltonian-not-in-nlp}
  starting on page~\pageref{cor:three-colorable-not-in-lp}).
\item Since alternation is the only way
  for \kl{nodes} to obtain global information about their network,
  the infiniteness result also suggests that
  the higher a \kl{graph property} lies
  in the \kl{locally polynomial hierarchy},
  the more global it is.
  We elaborate on this next.
\end{enumerate}

% –––––––––––––––––––––––––––––––––––––––––––––––––––––––––––––––––––––––––––––––
\subsection{Alternation as a measure of locality}
\label{ssec:measure-locality}

Inspired by Feuilloley's observation that
\kl{certificate} size can serve as a measure of locality
in the context of nondeterminism with unbounded \kl{certificates}~%
\cite[\S\,4.4]{DBLP:journals/dmtcs/Feuilloley21},
we propose that
the number of alternations plays a similar role
when considering alternation with \kl(certificate){locally bounded} \kl{certificates}.
To explore this idea in the most general setting possible,
we introduce the \emph{\kl{locally bounded hierarchy}}
$\Set{\SigmaLB{\Level}\!, \PiLB{\Level}}_{\Level \in \Naturals}$,
which is defined analogously to the \kl{locally polynomial hierarchy},
except that polynomial bounds
on local processing time and \kl{certificate} size
are replaced with arbitrary computable bounds.
Most of our previous results carry over directly to this generalized setting,
including all separation results.
In particular,
the hierarchy remains infinite.
This is because our separations already hold
on \kl{graphs} of bounded maximum \kl{degree} and constant \kl{label} size,
where
the \kl[locally bounded hierarchy]{locally bounded} and
\kl{locally polynomial hierarchies}
are equivalent.
On arbitrary \kl{graphs},
the \kl{locally bounded hierarchy} even exhibits a “cleaner” structure:
the inclusions corresponding to the dashed lines
in Figure~\ref{fig:hierarchy-full-overview}
become provable equalities,
yielding a strict linear order of complexity classes.

\begin{figure}[htb]
  \centering
  \begin{tikzpicture}[%
    semithick,on grid,>=stealth',node distance=6.5ex,
    property/.style={rectangle,draw=none,minimum size=0},
    strict/.style={},
    class/.style={draw,rounded rectangle,minimum height=4ex},
    lb/.style={class,minimum width=10ex},
    lcp/.style={class,minimum width=20ex},
    ruler/.style={minimum size=0,inner sep=0,font=\footnotesize,label distance=1.5ex},
    measure/.style={},
    ]
  \def\nodeDist{6.5ex}
  \def\propertySep{2.5ex}
  \def\hierarchyDist{53ex}
  % LB hierarchy
  \node[lb] (s0) {$\LB\vphantom{\SigmaLB{0}}$};
  \node[lb] (s1) [above of=s0] {$\SigmaLB{1}$};
  \node[lb] (p2) [above of=s1] {$\PiLB{2}$};
  \node[lb] (s3) [above of=p2] {$\SigmaLB{3}$};
  \node[lb] (p4) [above of=s3] {$\PiLB{4}$};
  \node (dots) at ([yshift=3/4*\nodeDist]p4) {$\vdots$};
  \coordinate (outside) at ([yshift=2*\nodeDist]p4);
  \path[strict]
    (s3) edge (p4)
    (p2) edge (s3)
    (s1) edge (p2)
    (s0) edge (s1)
    ;
  % LCP hierarchy
  \node[lcp] (lcp0)     [right=\hierarchyDist of s0] {$\LCP{0}$};
  \node[lcp] (lcpConst) [right=\hierarchyDist of s1] {$\LCP{O(1)}$};
  \node[lcp] (lcpLog)   [right=\hierarchyDist of s3] {$\LCP{O(\log n)}$};
  \node[lcp] (lcpPoly)  [right=\hierarchyDist of outside,anchor=center] {$\LCP{\poly(n)}$};
  \path[strict]
    (lcpPoly)  edge (lcpLog)
    (lcpLog)   edge (lcpConst)
    (lcpConst) edge (lcp0)
    ;
  % Graph properties
  \node (eulerian) [property]
    at ($(s0.east)!0.5!(lcp0.west)$) {$\EULERIAN$};
  \node (3-colorable) [property]
    at ($(s1.east)!0.5!(lcpConst.west)$) {$\COLORABLE{3}$};
  \node (odd) [property]
    at ($(s3.east)!0.5!(lcpLog.west)$) {$\ODD$};
  \node (acyclic) [property]
    at ([yshift=-\propertySep]odd) {$\ACYCLIC$};
  \node (hamiltonian) [property]
    at ([yshift=-\propertySep]acyclic) {$\HAMILTONIAN$};
  \node (non-2-colorable) [property]
    at ([yshift=-\propertySep]hamiltonian) {$\NONCOLORABLE{2}$};
  \node (non-3-colorable) [property]
    at ($(p4.east)!0.5!(lcpLog.west|-p4)$) {$\NONCOLORABLE{3}$};
  \node (automorphic) [property]
    at ($(p4.east|-lcpPoly)!0.5!(lcpPoly.west)$) {$\AUTOMORPHIC$};
  \node (prime) [property]
    at ([yshift=-\propertySep]automorphic) {$\PRIME$};
  \path[->,shorten >=1ex]
    (eulerian) edge (s0)
               edge (lcp0)
    (3-colorable) edge (s1)
                  edge (lcpConst)
    (odd.east) edge (lcpLog.west)
    (acyclic.east) edge (lcpLog.184)
    (hamiltonian.east) edge (lcpLog.187)
    (non-2-colorable.east) edge (lcpLog.190)
    (non-3-colorable.east) edge[out=0,in=180] (lcpPoly.186)
    (automorphic) edge (lcpPoly)
    (prime.east) edge[out=0,in=180] (lcpLog.175)
    ;
  \path[->,shorten >=1ex,densely dashed]
    (odd.west) edge (s3.east)
    (acyclic.west) edge (s3.-9)
    (hamiltonian.west) edge (s3.-17)
    (non-2-colorable.west) edge (s3.-24)
    (non-3-colorable) edge (p4)
    ;
  \path[black!30,line width=0.6ex]
    ([xshift=-6ex]$(prime-|dots)!0.55!(dots)$) edge ([xshift=-5ex]$(prime)!0.55!(dots-|prime)$);
  % LB measure of locality
  \node (r4) [ruler,left=3ex of p4.west,label={[ruler]left:4}] {\large$-$};
  \node (r3) [ruler,label={[ruler]left:3}] at (s3-|r4) {\large$-$};
  \node (r2) [ruler,label={[ruler]left:2}] at (p2-|r4) {\large$-$};
  \node (r1) [ruler,label={[ruler]left:1}] at (s1-|r4) {\large$-$};
  \node (r0) [ruler,label={[ruler]left:0}] at (s0-|r4) {\large$-$};
  \path[ruler] ([yshift=-1.5ex]r0.center) edge ([yshift=3ex]r4.center);
  % LCP measure of locality
  \node (global) [measure,anchor=west] at ([xshift=2ex]lcpPoly.east) {global};
  \node (local)  [measure,anchor=west] at ([xshift=2ex]lcp0.east) {local};
  \path[measure,->] (local) edge (global.south-|local);
\end{tikzpicture}

%%% Local Variables:
%%% mode: latex
%%% TeX-master: "../lph-paper"
%%% End:
  \caption{
    An example-based comparison between
    the \kl{locally bounded hierarchy}
    $\Set{\SigmaLB{\Level}\!, \PiLB{\Level}}_{\Level \in \Naturals}$
    introduced in this paper
    and the locally-checkable-proofs hierarchy
    $\Set{\LCP{f}}_{f:\, \Naturals \to \Naturals}$
    introduced by Göös and Suomela~\cite{DBLP:journals/toc/GoosS16}.
    (The interpretation of the latter as a measure of locality
    is due to Feuilloley~\cite[\S\,4.4]{DBLP:journals/dmtcs/Feuilloley21}.)
    A vertical line between two complexity classes
    indicates that
    the lower class is strictly included in the higher one.
    An arrow from a \kl{graph property} to a complexity class
    indicates that
    the \kl(graph){property} is contained in the class.
    If the arrow is solid,
    there is also a matching lower bound,
    i.e., the \kl(graph){property} is not contained in any lower class.
    \kl(graph){Properties} above the thick horizontal line
    lie outside the \kl{locally bounded hierarchy}.
  }
  \label{fig:lb-vs-lcp-overview}
\end{figure}

To test our hypothesis
that alternation provides a meaningful measure of locality,
we compare
the \kl{locally bounded hierarchy}
to a fundamentally different hierarchy based on \kl{certificate} size.
Specifically,
we focus on the locally-checkable-proofs hierarchy
$\Set{\LCP{f}}_{f:\, \Naturals \to \Naturals}$
introduced by Göös and Suomela~\cite{DBLP:journals/toc/GoosS16},
which is the most general nondeterministic model
in the literature on distributed decision.
This hierarchy is defined through a single-round game
in which \kl(certificate){Eve} plays alone
but can choose arbitrarily large \kl{certificates}
that may reference \kl{globally unique} \kl{identifiers}.
Instead of constraining \kl(certificate){Eve}'s \kl{certificates},
one evaluates
how large they must be
to convince the \kl{nodes} of a given \kl{graph property}.

Despite their different foundations,
the \kl[locally bounded hierarchy]{locally bounded}
and locally-checkable-proofs hierarchies
exhibit notable parallels.
Figure~\ref{fig:lb-vs-lcp-overview} informally compares the two
by mapping various \kl{graph properties} onto them.
The measures roughly align
on four levels of locality,
allowing us to classify \kl{graph properties} as
\emph{purely local}
(in $\LB$ and $\LCP{0}$),
\emph{almost local}
(in $\SigmaLB{1}$ and $\LCP{O(1)}$),
\emph{intermediate}
(in $\SigmaLB{3}$ and $\LCP{O(\log n)}$),
or
\emph{inherently global}
(outside $\bigcup_{\Level \in \Naturals} \SigmaLB{\Level}$
but within $\LCP{\poly(n)}$).
This close correspondence suggests
that both hierarchies capture essential aspects
of what makes \kl(graph){properties} local or global.
While mismatches do occur,
particularly for \kl(graph){properties} involving counting or \kl{complementation},
these differences may reflect quirks of the respective formalisms
rather than fundamental disagreements about the nature of locality.
(We will return to this discussion in Section~\ref{sec:discussion}.)

%%% Local Variables:
%%% mode: latex
%%% TeX-master: "../lph-paper"
%%% End:

\section{Preliminaries}
\label{sec:preliminaries}

We now begin the formal development
of the material outlined in the previous section,
starting with basic definitions and notation.

\AP
We denote
the empty set by~$\intro*\EmptySet$,
the set of nonnegative integers by~$\intro*\Naturals$,
the set of positive integers by~$\intro*\Positives$,
and the set of integers by~$\intro*\Integers$.
The cardinality of any set~$A$ is written as
$\intro*\Card{A}$,
its power set as
$\intro*\PowerSet{A}$,
and the set of finite strings over~$A$
as~$\intro*A^\KleeneStar$.
The length of a string~$\String$ is denoted by
$\intro*\Length{\String}$,
and its $i$-th symbol by~%
$\intro*\String\AtPos{i}$.
By a slight abuse of notation,
we sometimes lift functions from elements to sets, i.e.,
given $f \colon X \to Y$ and $A \subseteq X$,
we write~$f(A)$ for $\SetBuilder{f(a)}{a \in A}$.
To denote integer intervals,
we define
$\intro*\Range[m]{n} = \SetBuilder{i \in \Integers}{m \leq i \leq n}$
and
$\reintro*\Range{n} = \Range[0]{n}$,
for any $m, n \in \Integers$.
Angle brackets indicate excluded endpoints, e.g.,
$\reintro*\Range*[m]{n} = \Range[m + 1]{n}$
and
$\reintro*\Range{n}* = \Range[0]{n - 1}$.

\AP
Throughout this paper,
we assume some fixed but unspecified
\intro{encoding} of finite objects
(e.g., integers, \kl{graphs}, or tuples of finite objects)
as binary strings.
Sometimes we also implicitly identify such objects
with their string representations.

\subparagraph*{Graphs.}

\AP
All \kl{graphs} we consider are
finite, simple, undirected, and \kl{connected}.
Formally,
a \intro{labeled graph}, or~simply \reintro{graph},
is represented by a triple
$\Graph = \Tuple{\NodeSet{\Graph}, \EdgeSet{\Graph}, \Labeling{\Graph}}$,
where~$\intro*\NodeSet{\Graph}$ is a finite nonempty set of \intro{nodes};
$\intro*\EdgeSet{\Graph}$ is a set of undirected \intro{edges}
consisting of 2-element subsets of~$\NodeSet{\Graph}$
and containing,
for every partition $\Set{\NodeSet[0]{}, \NodeSet[1]{}}$ of~$\NodeSet{\Graph}$,
at least one \kl{edge}
$\Set{\Node[1], \Node[2]}$
with
$\Node[1] \in \NodeSet[0]{}$
and
$\Node[2] \in \NodeSet[1]{}$;
and
$\intro*\Labeling{\Graph} \colon
 \NodeSet{\Graph} \to \Set{0, 1}^\KleeneStar$
is a \intro{labeling} function
that assigns a bit string to each \kl{node}.
We refer to the string~$\Labeling{\Graph}(\Node[1])$
as the \intro{label} of \kl{node}~$\Node[1]$
and to the symbol~$\Labeling{\Graph}(\Node[1])\AtPos{i}$
as the $i$-th \intro{labeling bit} of~$\Node[1]$,
for $i \in \Range[1]{\Length{\Labeling{\Graph}(\Node[1])}}$.
To simplify notation,
we often write
$\Node[1] \intro*\inG \Graph$
instead of
$\Node[1] \in \NodeSet{\Graph}$,
and we define $\intro*\CardG{\Graph}$,
the \intro(graph){cardinality} of~$\Graph$,
as~$\Card{\NodeSet{\Graph}}$.

\AP
We denote by $\intro*\GRAPH$
the set of all \kl{labeled graphs}
and by $\intro*\NODE$
the set of \intro{single-node graphs},
i.e., \kl{labeled graphs} consisting of a single \kl{node}.
A \intro{graph property}
(sometimes called a “graph language”)
is a set $\Property \subseteq \GRAPH$
that is closed under isomorphism.
If a \kl{graph}~$\Graph$ belongs to~$\Property$,
then we also say that $\Graph$~has
the \kl(graph){property}~$\Property$.

\AP
\phantomintro{adjacent}%
\phantomintro{incident}%
\phantomintro{leaf}%
\phantomintro{cycle}%
\phantomintro{subgraph}%
\phantomintro{spanning subgraph}%
\phantomintro{regular}%
\phantomintro{path}%
\phantomintro{connected}%
\phantomintro{component}%
We follow Diestel~\cite{DBLP:books/daglib/0030488}
for standard graph-theoretic terms such as
\intro{neighbor}, \intro{degree}, \intro{distance},
\intro{diameter}, \intro{induced subgraph}, and so on.
The \kl{diameter} of a \kl{graph}~$\Graph$
is denoted by~$\intro*\Diameter(\Graph)$.
For $\Radius \in \Naturals$ and $\Node \inG \Graph$,
the $\Radius$-\intro{neighborhood}
$\intro*\Neighborhood{\Graph}{\Radius}{\Node}$
of~$\Node$ in~$\Graph$
is the \kl{subgraph} of~$\Graph$ \kl{induced}
by all \kl{nodes} at \kl{distance} at most~$\Radius$ from~$\Node$.
That is,
$\Neighborhood{\Graph}{\Radius}{\Node}$ is
the \kl{graph}~$\Graph'$ that consists of
the \kl{nodes} at \kl{distance} at most~$\Radius$ from~$\Node$
and all \kl{edges} connecting them,
and whose \kl{labeling}~$\Labeling{\Graph'}$
is the restriction of~$\Labeling{\Graph}$ to~$\NodeSet{\Graph'}$.

\subparagraph*{Identifier assignments.}

\AP
An \intro{identifier assignment} of a \kl{graph}~$\Graph$ is a function
$\IdMap \colon
\NodeSet{\Graph} \to \Set{0, 1}^\KleeneStar$
whose purpose is to distinguish between different \kl{nodes} of~$\Graph$.
We refer to~$\IdMap(\Node[1])$
as the \intro{identifier} of \kl{node}~$\Node[1]$ under~$\IdMap$.
\kl{Identifiers} are ordered lexicographically, i.e.,
the \intro{identifier order} is such that
$\IdMap(\Node[1]) \intro*\SmallerThan \IdMap(\Node[2])$
if either $\IdMap(\Node[1])$~is a proper prefix of~$\IdMap(\Node[2])$,
or $\IdMap(\Node[1])\AtPos{i} < \IdMap(\Node[2])\AtPos{i}$
at the first position~$i$ where the two strings differ.

\AP
We say that
$\IdMap$~is $\IdentRadius$-\intro{locally unique}
for some $\IdentRadius \in \Naturals$
if it satisfies
$\IdMap(\Node[1]) \neq \IdMap(\Node[2])$
for all distinct \kl{nodes}~$\Node[1]$ and~$\Node[2]$
that lie in the $\IdentRadius$-\kl{neighborhood}
of a common \kl{node}~$\Node[3]$
(or equivalently, in the $2\IdentRadius$-\kl{neighborhood} of each other).
If $\IdentRadius \geq \Diameter(\Graph) / 2$,
we say that $\IdMap$~is \intro{globally unique}.

\AP
An $\IdentRadius$-\kl{locally unique}
\kl{identifier assignment}~$\IdMap$ of a \kl{graph}~$\Graph$
is called \intro{small}
(with respect to~$\IdentRadius$)
if for every \kl{node} $\Node \inG \Graph$,
the length of~$\IdMap(\Node)$
is at most
$\Ceiling{\log_2\CardG{\Neighborhood{\Graph}{2\IdentRadius}{\Node}}}$,
i.e., logarithmically bounded in the \kl(graph){cardinality}
of $\Node$'s $2\IdentRadius$-\kl{neighborhood} in~$\Graph$.
When we want to emphasize that
an $\IdentRadius$-\kl{locally unique} \kl{identifier assignment}
is not necessarily \kl{small},
we call it \intro{arbitrary-sized}.

\begin{remark}
  \label{rem:small-sized-identifiers}
  For every \kl{graph}~$\Graph$
  and integer $\IdentRadius \in \Naturals$,
  there exists an $\IdentRadius$-\kl{locally unique}
  \kl{identifier assignment}~$\IdMap$ of~$\Graph$
  that is \kl{small}.
\end{remark}

\begin{claimproof}
  By definition,
  an \kl{identifier assignment}~$\IdMap$ of~$\Graph$
  is $\IdentRadius$-\kl{locally unique}
  if the \kl{identifier} $\IdMap(\Node)$
  of every \kl{node} $\Node \inG \Graph$
  is distinct from the \kl{identifiers}
  of all other \kl{nodes} in $\Neighborhood{\Graph}{2\IdentRadius}{\Node}$.
  Such an \kl{identifier assignment} can be easily constructed
  if we may choose among at least
  $\CardG{\Neighborhood{\Graph}{2\IdentRadius}{\Node}}$
  possible values of $\IdMap(\Node)$.
  Hence,
  a bit string of length at most
  $\Ceiling{\log_2\CardG{\Neighborhood{\Graph}{2\IdentRadius}{\Node}}}$
  is sufficient.
  \claimqedhere
\end{claimproof}

\subparagraph*{Certificate assignments.}

\AP
For any \kl{graph}~$\Graph$
and any \kl{identifier assignment}~$\IdMap$ of~$\Graph$,
a \intro{certificate assignment} of
$\Tuple{\Graph, \IdMap}$
is a function
$\CertifMap \colon
\NodeSet{\Graph} \to \Set{0, 1}^\KleeneStar$.
We refer to~$\CertifMap(\Node[1])$
as the \intro{certificate} of \kl{node}~$\Node[1]$ under~$\CertifMap$.
Given $\CertifRadius \in \Naturals$
and $\CertifPolynomial \colon \Naturals \to \Naturals$,
we say that $\CertifMap$ is
$\Tuple{\CertifRadius, \CertifPolynomial}$-\intro(certificate){bounded}
if for every \kl{node} $\Node[1] \inG \Graph$,
the length of $\Node[1]$'s \kl{certificate}
is bounded by~$\CertifPolynomial$ with respect to
the \kl(graph){cardinality} of $\Node[1]$'s $\Radius$-\kl{neighborhood}
and the lengths of all \kl{labels} and \kl{identifiers} therein, i.e.,
\begin{equation*}
  \Length{\CertifMap(\Node[1])}
  \, \leq \,
  \CertifPolynomial
  \Bigl( \,
    \sum_{
      \lalign{
        \Node[2] \inG \Neighborhood{\Graph}{\CertifRadius}{\Node[1]}
      }
    }
    1 +
    \Length{\Labeling{\Graph}(\Node[2])} +
    \Length{\IdMap(\Node[2])} \,
  \Bigr).
\end{equation*}

\AP
We often represent several \kl{certificate assignments}
as a single function
$\CertifListMap \colon
\NodeSet{\Graph} \to \Set{0, 1, \Separator}^\KleeneStar$,
called a \intro{certificate-list assignment},
where the symbol~$\Separator$ is used
to separate the individual \kl{certificates} of each \kl{node}.
Given \kl{certificate assignments}
$\CertifMap_1, \CertifMap_2, \dots, \CertifMap_{\Level}$,
we write
${\CertifMap_1 \intro*\CertifConcat
  \CertifMap_2 \reintro*\CertifConcat \dots \reintro*\CertifConcat
  \CertifMap_{\Level}}$
for the \kl{certificate-list assignment}~$\CertifListMap$
such that
$\CertifListMap(\Node[1]) =
 \CertifMap_1(\Node[1]) \, \Separator \,
 \CertifMap_2(\Node[1]) \, \Separator \dots \Separator \,
 \CertifMap_{\Level}(\Node[1])$
for all $\Node[1] \inG \Graph$.
We say that $\CertifListMap$~is
$\Tuple{\CertifRadius, \CertifPolynomial}$-\intro(certificate-list){bounded}
if $\CertifMap_i$ is
$\Tuple{\CertifRadius, \CertifPolynomial}$-\kl(certificate){bounded}
for every $i \in \Range[1]{\Level}$.

\subparagraph*{Structural representations.}

\AP
We will evaluate \kl{logical formulas}
on relational \intro{structures} of the form
$\Struct =
 \Tuple{
   \Domain{\Struct} \!,
   \BitSet{1}{\Struct}, \dots, \BitSet{m}{\Struct},
   {\LinkRel{1}{\Struct}}, \dots, {\LinkRel{n}{\Struct}}
 }$,
where
$\intro*\Domain{\Struct}$ is a finite nonempty set of \intro{elements},
called the \intro{domain} of~$\Struct$,
each $\intro*\BitSet{i}{\Struct}$
is a subset of~$\Domain{\Struct}$,
for $i \in \Range[1]{m}$,
and
each $\intro*\LinkRel{i}{\Struct}$
is a binary relation over~$\Domain{\Struct}$,
for $i \in \Range[1]{n}$.
We refer to~$\Tuple{m, n}$ as the \intro{signature} of~$\Struct$.
To simplify notation,
we often write
$\Element[1] \intro*\inS \Struct$
instead of
$\Element[1] \in \Domain{\Struct}$,
and we define $\intro*\CardS{\Struct}$,
the \intro(structure){cardinality} of~$\Struct$,
as~$\Card{\Domain{\Struct}}$.
We also write
$\Element[1] \intro*\NeighborRel{\Struct} \Element[2]$
to indicate that
$\Element[1] \LinkRel{i}{\Struct} \Element[2]$
or
$\Element[2] \LinkRel{i}{\Struct} \Element[1]$
for some $i \in \Range[1]{n}$.

\AP
To evaluate \kl{logical formulas} on \kl{graphs},
we identify each \kl{graph}~$\Graph$
with a \kl{structure}
$\intro*\StructRepr{\Graph} \! =
 \Tuple{
   \Domain{\StructRepr{\Graph}} \!,
   \BitSet{1}{\StructRepr{\Graph}} \!,
   {\LinkRel{1}{\StructRepr{\Graph}}} \!, {\LinkRel{2}{\StructRepr{\Graph}}}
 }$
of \kl{signature}~$\Tuple{1, 2}$,
called the \intro{structural representation} of~$\Graph$.
This \kl{structure} contains one \kl{element}~$\Node[1]$
for each \kl{node}~$\Node[1] \in \Graph$
and one \kl{element}~$\Tuple{\Node[1], i}$
for each of $\Node[1]$'s \kl{labeling bits}, i.e.,
\begin{equation*}
  \Domain{\StructRepr{\Graph}} \! = \,
  \NodeSet{\Graph}
  \, \cup \,
  \SetBuilder{
    \Tuple{\Node[1], i}
  }{
    \Node[1] \in \NodeSet{\Graph} \!, \:
    i \in \Range[1]{\Length{\Labeling{\Graph}(\Node[1])}}
  }.
\end{equation*}
The set~$\BitSet{1}{\StructRepr{\Graph}}$\!
corresponds to the \kl{labeling bits} whose value is~$1$, i.e.,
$\Tuple{\Node[1], i} \in \BitSet{1}{\StructRepr{\Graph}}$\!
if and only if
$\Labeling{\Graph}(\Node[1])\AtPos{i} = 1$.
The relation~$\LinkRel{1}{\StructRepr{\Graph}}$\!
represents the edges in~$\EdgeSet{\Graph}$
and the successor relation of the \kl{labeling bits}, i.e.,
$\Node[1] \LinkRel{1}{\StructRepr{\Graph}} \! \Node[2]$
if and only if
$\Set{\Node[1], \Node[2]} \in \EdgeSet{\Graph}$,
and
$\Tuple{\Node[1], i} \LinkRel{1}{\StructRepr{\Graph}} \! \Tuple{\Node[2], j}$
if and only if
$\Node[1] = \Node[2]$
and
$j = i + 1$.
Finally,
the relation~$\LinkRel{2}{\StructRepr{\Graph}}$\!
determines which \kl{node} owns which \kl{labeling bits}, i.e.,
$\Node[1] \LinkRel{2}{\StructRepr{\Graph}} \! \Tuple{\Node[2], i}$
if and only if
$\Node[1] = \Node[2]$.
An example is provided in Figure~\ref{fig:graph} on page~\pageref{fig:graph}.

\AP
For $\Radius \in \Naturals$ and $\Node \inG \Graph$,
the \kl{structural representation}
of $\Node$'s $\Radius$-\kl{neighborhood}
$\Neighborhood{\Graph}{\Radius}{\Node}$
is denoted by
$\intro*\StructNeighborhood{\Graph}{\Radius}{\Node}$.
For instance,
if $\Node$ is the upper right \kl{node}
of the \kl{graph}~$\Graph$ depicted in Figure~\ref{fig:graph},
then
$\CardS{\StructNeighborhood{\Graph}{0}{\Node}} = 4$,\,
$\CardS{\StructNeighborhood{\Graph}{1}{\Node}} = 8$,
and
$\StructNeighborhood{\Graph}{2}{\Node} = \StructRepr{\Graph}$.

\AP
A \intro{structure property}
is a set $\Property$ of \kl{structures}
that is closed under isomorphism.
In particular,
since we identify \kl{graphs} with their \kl{structural representations},
every \kl{graph property} is also a \kl{structure property}.
We will often restrict
a given class~$\Class$ of \kl{structure properties}
(e.g.,~$\NLP$)
to \kl{structures} that have
some presupposed \kl(structure){property}~$\BaseProperty$
(e.g.,~$\NODE$).
In such cases,
we write $\Class\On{\BaseProperty}$
for the restriction of~$\Class$ to~$\BaseProperty$, i.e.,
$\intro*\Class\On{\BaseProperty} =
 \SetBuilder{\Property \cap \BaseProperty}{\Property \in \Class}$.

%%% Local Variables:
%%% mode: latex
%%% TeX-master: "../lph-paper"
%%% End:

\section{Distributed Turing machines}
\label{sec:turing-machines}

\AP
We formalize synchronous distributed algorithms
using the notion of \kl{distributed Turing machines}.
As illustrated in Figure~\ref{fig:turing-machine},
such \kl{machines} are equipped with three one-way infinite \intro{tapes}:
a \intro[receiving tape]{receiving} \kl{tape}
to store incoming messages,
an \intro[internal tape]{internal} \kl{tape}
to store the \kl{machine}'s internal state
and perform local computations, and
a \intro[sending tape]{sending} \kl{tape}
to store outgoing messages.

\begin{figure}[htb]
  \centering
  \begin{tikzpicture}[
    semithick,on grid,
    tape/.style={
      matrix of math nodes,inner sep=0ex,
      column sep={3ex,between origins},
      nodes={draw,minimum size=3ex,inner sep=0ex,anchor=center}
    },
    continuation/.style={
      circle,ultra thick,draw=white,fill=white,xshift=-2ex
    },
    rw head/.style={
      rounded corners,shorten >=0.1ex,
      -{Stealth[round,width=1.4ex,length=0.8ex,inset=0.4ex]}
    },
  ]
  \matrix (receiving tape) [tape] {
    \Endmarker & 0 & 1 & \Separator & 1 & 0 & 0 & 1 &
    0 & 1 & 1 & \Separator & 0 & 0 & 0 & \Separator &
    {} & |[continuation]| \dots \\
  };
  \matrix (internal tape) [tape,below=7ex of receiving tape] {
    \Endmarker & 0 & 0 & 1 & 0 & 1 & 0 & 0 &
    1 & 1 & 0 & 1 & 1 & 1 & 1 & 0 &
    {} & |[continuation]| \dots \\
  };
  \matrix (sending tape) [tape,below=7ex of internal tape] {
    \Endmarker & 1 & 1 & 1 & 0 & 0 & \Separator & 1 &
    0 & 0 & 1 & 1 & 0 & \Separator & 0 & 1 &
    {} & |[continuation]| \dots \\
  };
  \node (controller) [draw,circle,fill=black,text=white]
        at ([shift={(-8ex,-3.5ex)}]internal tape.west) {$\State$};
  \draw[rw head] (controller.20) -|
                 ([shift={(-4ex,-3.5ex)}]receiving tape.west) -|
                 (receiving tape-1-9.south);
  \draw[rw head] (controller.0) -| (internal tape-1-6.south);
  \draw[rw head] (controller.-20) -|
                 ([shift={(-4ex,-3.5ex)}]sending tape.west) -|
                 (sending tape-1-2.south);
  \node[above=0.3ex of controller.north] {\kl{state}};
  \node[right=1ex of receiving tape.east] {\kl{receiving tape}};
  \node[right=1ex of internal tape.east] {\kl{internal tape}};
  \node[right=1ex of sending tape.east] {\kl{sending tape}};
  \draw[decorate,decoration={brace,raise=0.6ex,amplitude=1.2ex}]
      (receiving tape-1-5.north west) --
      node[above=1.4ex] {message from second \kl{neighbor}}
      (receiving tape-1-11.north east);
  \draw[decorate,decoration={brace,mirror,raise=0.6ex,amplitude=1.2ex}]
      (sending tape-1-8.south west) --
      node[below=1.4ex] {message to second \kl{neighbor}}
      (sending tape-1-13.south east);
\end{tikzpicture}

%%% Local Variables:
%%% mode: latex
%%% TeX-master: "../lph-paper"
%%% End:
  \caption{
    Local copy of a \kl{distributed Turing machine}
    being \kl{executed} by a \kl{node}.
  }
  \label{fig:turing-machine}
\end{figure}

\subparagraph*{Formal representation.}

\AP
A \intro{distributed Turing machine}
is represented by a tuple
$\Machine = \Tuple{\StateSet, \TransFunc}$
consisting of
a finite set~$\StateSet$ of \intro{states}
and a \intro{transition function}
$\TransFunc \colon
\StateSet \times \Alphabet^3 \to
\StateSet \times \Alphabet^3 \times \Set{\Left, \Stay, \Right}^3$.
Here,
$\Alphabet$ is the tape alphabet
$\Set{\Endmarker, \Blank, \Separator, 0, 1}$
with
the \intro{left-end marker}\,~$\intro*\Endmarker$,
the \intro{blank symbol}\,~$\intro*\Blank$, and
the \intro{separator}\,~$\intro*\Separator$.
We assume
that~$\StateSet$ always contains
the designated \kl{states}
$\intro*\StartState$,\, $\intro*\PauseState$, and~$\intro*\StopState$.

\AP
When we refer to the \intro[tape content]{content} of a \kl{tape},
we mean the sequence of symbols written on the \kl{tape}
ignoring any leading or trailing occurrences of $\Endmarker$ and~$\Blank$.
In particular,
if the first cell of the \kl{tape} contains~$\Endmarker$
and the remaining cells contain~$\Blank$,
we consider the \kl{tape} to be empty.

\subparagraph*{Execution.}

\AP
A \kl{distributed Turing machine}
$\Machine = \Tuple{\StateSet, \TransFunc}$
can be \kl{executed}
on any \kl{graph}~$\Graph$,
under any \kl{identifier assignment}~$\IdMap$ of~$\Graph$
and any \kl{certificate-list assignment}~$\CertifListMap$
of $\Tuple{\Graph, \IdMap}$,
provided that~$\IdMap$ is at least $1$-\kl{locally unique}.
An \intro{execution} consists of
a sequence of synchronous \intro{communication rounds},
where
all \kl{nodes} start at the same time
and run their own copy of~$\Machine$.
In every \kl{round},
each \kl{node}~$\Node[1] \inG \Graph$
goes through three \intro{phases}:
(\ref{itm:receive})~it receives messages from its \kl{neighbors},
(\ref{itm:compute})~it performs local computations, and
(\ref{itm:send})~it sends messages to its \kl{neighbors}.
We now describe these \kl{phases} in~detail.

\begin{enumerate}
\item \label{itm:receive}
  In the first \kl{phase},
  the messages
  $\Message_1, \dots, \Message_{\Degree}
  \in \Set{0, 1}^\KleeneStar$
  that~$\Node[1]$ receives from its \kl{neighbors}
  $\Node[2]_1, \dots, \Node[2]_{\Degree}$
  are concatenated using the \kl{separator}~$\Separator$
  (including a trailing~$\Separator$)
  and written on~$\Node[1]$'s \kl{receiving tape}.
  Any previous \kl[tape content]{content} is discarded
  so that the new \kl[tape content]{content}
  of the \kl{receiving tape} is the string
  $\Message_1 \, \Separator \cdots \Separator \, \Message_{\Degree} \, \Separator$.
  In particular,
  if we are in the first \kl{round},
  the \kl[tape content]{content} is~$\Separator^{\Degree}$,
  which indicates
  that~$\Node[1]$ has not yet received any (nonempty) messages.
  In later \kl{rounds},
  $\Message_1, \dots, \Message_{\Degree}$
  correspond to the messages
  that were sent by the \kl{neighbors} in the previous \kl{round},
  sorted in ascending \kl{identifier order}.
  That is,
  we assume
  $\IdMap(\Node[2]_1) \SmallerThan \dots \SmallerThan
   \IdMap(\Node[2]_{\Degree})$.
\item \label{itm:compute}
  In the second \kl{phase},
  $\Node[1]$'s copy of~$\Machine$ behaves like
  a standard Turing machine with three \kl{tapes}.
  The \kl{receiving tape} is initialized as stated above,
  while the \kl{sending tape} is initially empty,
  meaning that any \kl[tape content]{content}
  from the previous \kl{round} is erased.
  In case we are in the first \kl{round},
  the \kl{internal tape} is initialized to the string
  $\Labeling{\Graph}(\Node[1]) \, \Separator \,
   \IdMap(\Node[1])   \, \Separator \,
   \CertifListMap(\Node[1])$,
  i.e., the \kl{node} gets a copy of
  its \kl{label}, \kl{identifier}, and \kl{certificates}.
  Otherwise,
  the \kl[tape content]{content} of the \kl{internal tape}
  remains the same as at the end of the previous \kl{round}.
  Now,
  if the \kl{machine} ended up in \kl{state}~$\StopState$
  in the previous \kl{round},
  then it remains in that \kl{state}
  and immediately goes to \kl{phase}~\ref{itm:send}.
  Otherwise,
  it starts its local computation in \kl{state}~$\StartState$
  with all three tape heads on the leftmost cell of the tapes,
  and then goes through a sequence of \intro{computation steps}.
  In each \kl{step},
  depending on the current \kl{state}
  and the symbols currently scanned on the three \kl{tapes},
  the \kl{transition function}~$\TransFunc$
  tells~$\Machine$ how to update its \kl{state}
  and the symbols on the \kl{tapes},
  and also in which directions to move the three tape heads.
  The local computation halts
  as soon as the \kl{machine} reaches one of the \kl{states}
  $\PauseState$ or~$\StopState$.
\item \label{itm:send}
  In the third \kl{phase},
  the messages
  $\Message'_1, \dots, \Message'_{\Degree}
  \in \Set{0, 1}^\KleeneStar$
  sent to the \kl{neighbors}
  $\Node[2]_1, \dots, \Node[2]_{\Degree}$
  correspond to the first~$\Degree$ bit strings
  stored on the \kl{sending tape},
  using the symbol~$\Separator$ as a separator
  and ignoring any~$\Blank$'s.
  The order of the \kl{neighbors}
  is the same as in \kl{phase}~\ref{itm:receive},
  i.e., the ascending order of \kl{identifiers}.
  In case there are not enough messages on the \kl{sending tape},
  the missing ones default to the empty string.
  In particular,
  if~$\Node[1]$ has already reached \kl{state}~$\StopState$
  in the previous \kl{round},
  then its \kl{sending tape} remains empty,
  so all \kl{neighbors} receive an empty message from~$\Node[1]$.
\end{enumerate}

\AP
The \kl{execution} \intro{terminates}
in $\Radius \in \Positives$ \kl{rounds}
if all \kl{nodes} have reached \kl{state}~$\StopState$
by the end of the~$\Radius$-th \kl{round}.
Note that this implies that
the local computations of all \kl{nodes} halt in every \kl{round}.
Throughout this paper,
we will restrict ourselves to \kl{distributed Turing machines}
whose \kl{executions} \kl{terminate} on every \kl{graph}
under all \kl[identifier assignments]{identifier}
and \kl{certificate-list assignments}.

\subparagraph*{Result and decision.}

\AP
The \intro{result}
$\intro*\Result{\Machine}{\Graph, \IdMap, \CertifListMap}$
computed by~$\Machine$
on \kl{graph}~$\Graph$
under \kl{identifier assignment}~$\IdMap$
and \kl{certificate-list assignment}~$\CertifListMap$
is the \kl{graph}~$\Graph'$
whose \kl{nodes} and \kl{edges} are the same as those of~$\Graph$,
and whose \kl{labeling} function~$\Labeling{\Graph'}$
assigns to each \kl{node}~$\Node$
the bit string written on $\Node$'s \kl{internal tape}
after $\Machine$'s \kl{execution} has \kl{terminated}.
To guarantee that this is indeed a bit string,
any symbols other than $0$ and $1$ are ignored.
In case we do not need any \kl{certificate assignments},
we simply write
$\reintro*\Result{\Machine}{\Graph, \IdMap}$
to denote the \kl{result}
computed by~$\Machine$
on~$\Graph$ under~$\IdMap$
and the trivial \kl{certificate-list assignment}
that assigns the empty string to every \kl{node} of~$\Graph$.

\AP
A \kl{distributed Turing machine} can act as
a consensus-based decision procedure,
where all \kl{nodes} must agree
in order for a given input to be accepted.
More precisely,
when \kl{executing}~$\Machine$
on~$\Graph$ under~$\IdMap$ and~$\CertifListMap$,
the individual \intro{verdict} of \kl{node} $\Node \in \Graph$
is the string~$\String$
with which $\Node$ is \kl{labeled}
in the \kl{result}~$\Result{\Machine}{\Graph, \IdMap, \CertifListMap}$.
We say that
$\Node$~\intro{accepts}
in~$\Result{\Machine}{\Graph, \IdMap, \CertifListMap}$
if $\String = 1$,
and that $\Node$~\intro{rejects}
in~$\Result{\Machine}{\Graph, \IdMap, \CertifListMap}$
otherwise.
Based on that,
$\Graph$ is \reintro{accepted} by~$\Machine$
under~$\IdMap$ and~$\CertifListMap$,
written
$\Result{\Machine}{\Graph, \IdMap, \CertifListMap} \equiv
 \intro*\Accept$,
if every \kl{node} \kl{accepts}
in~$\Result{\Machine}{\Graph, \IdMap, \CertifListMap}$.
Conversely,
$\Graph$ is \reintro{rejected} by~$\Machine$
under~$\IdMap$ and~$\CertifListMap$
if at least one \kl{node} \kl{rejects}
in~$\Result{\Machine}{\Graph, \IdMap, \CertifListMap}$.

\subparagraph*{Running time.}

In order to measure the running time of \kl{distributed Turing machines},
we will use two different metrics:
\kl{round time},
which corresponds to the number of \kl{communication rounds}
in an \kl{execution},
and \kl{step time},
which gives the number of \kl{computation steps}
made by a single \kl{node} in one \kl{round}.

\AP
More precisely,
for any \kl{execution} of a \kl{distributed Turing machine}~$\Machine$,
the \intro{round running time}
is the number of \kl{rounds}
until all \kl{nodes} have reached \kl{state}~$\StopState$.
Given some function $f \colon \Naturals \to \Naturals$,
we say that~$\Machine$ runs in \kl{round time}~$f$
if the \kl{round running time} is bounded
by~$f$ with respect to the \kl(graph){cardinality} of the \kl{graph}
on which~$\Machine$ is \kl{executed}.
This means that
for every \kl{graph}~$\Graph$,
every \kl{identifier assignment}~$\IdMap$ of~$\Graph$,
and every \kl{certificate-list assignment}~$\CertifListMap$
of $\Tuple{\Graph, \IdMap}$,
all \kl{nodes} of~$\Graph$ reach \kl{state}~$\StopState$
after at most~$f(\CardG{\Graph})$ \kl{rounds}
in the corresponding \kl{execution} of~$\Machine$.
Accordingly,
$\Machine$ runs in \intro{constant round time}
if this holds for some constant function~$f$.

\AP
On the other hand,
the \intro{step running time} of \kl{node}~$\Node$
in \kl{round} $\Round \in \Positives$
of an \kl{execution} of~$\Machine$
is the number of local \kl{computation steps}
that~$\Node$ makes during (\kl{phase}~\ref{itm:compute} of) \kl{round}~$\Round$.
For $f \colon \Naturals \to \Naturals$,
we say that~$\Machine$
runs in \kl{step time}~$f$
if in every \kl{execution},
the \kl{step running time}
of every \kl{node}~$\Node$ in every \kl{round}~$\Round$
is bounded
by~$f$ with respect to
the length of $\Node$'s initial \kl{tape contents} in \kl{round}~$\Round$.
This means that
if~$\Machine$ starts in \kl{state}~$\StartState$
with some arbitrary strings
$\String[1] \in \Set{0, 1, \Separator}^\KleeneStar$
and
$\String[2] \in \Set{0, 1, \Separator, \Blank}^\KleeneStar$
written on its \kl[receiving tape]{receiving} and \kl{internal tapes},
then $\Machine$ reaches $\PauseState$~or~$\StopState$
after at most
$f(\Length{\String[1]} + \Length{\String[2]})$
\kl{steps}.
Accordingly,
$\Machine$~runs in \intro{polynomial step time}
if this holds for some polynomial function~$f$.

\AP
A \intro{locally polynomial machine} is a \kl{distributed Turing machine}
that runs in \kl{constant round time} and \kl{polynomial step time}.

\subparagraph*{Arbiters and the locally polynomial hierarchy.}

\AP
As explained in Section~\ref{ssec:overview-complexity-classes},
each \kl{graph property}~$\Property$
in the \kl{locally polynomial hierarchy}
corresponds to a game between two players:
\intro(certificate){Eve},
who tries to prove that
a given \kl{graph}~$\Graph$~has \kl(graph){property}~$\Property$,
and \intro(certificate){Adam},
who tries to prove the opposite.
The players take turns
labeling the \kl{nodes} of~$\Graph$ with \kl{certificates},
which serve as
proofs (in \kl(certificate){Eve}'s case)
and
counterproofs (in~\kl(certificate){Adam}'s case).
After a fixed number~$\Level$ of moves,
the winner is determined by a \kl{locally polynomial machine}~$\Machine$,
and the \kl{graph}~$\Graph$ has \kl(graph){property}~$\Property$
if and only if \kl(certificate){Eve} has a winning strategy in this game.
Depending on who makes the first move,
$\Property$~is classified
as a~$\SigmaLP{\Level}$-\kl(graph){property}
(if~\kl(certificate){Eve} starts)
or a~$\PiLP{\Level}$-\kl(graph){property}
(if~\kl(certificate){Adam}~starts).

\AP
Formally,
\label{def:locally-polynomial-hierarchy}%
we represent
\kl(certificate){Eve}'s and \kl(certificate){Adam}'s choices
by quantifying existentially and universally,
respectively,
over the \kl{certificate assignments} chosen by each player.
More precisely,
for $\Level \in \Naturals$,
a \kl{graph property}~$\Property$ belongs to~$\intro*\SigmaLP{\Level}$
if there exists
a \kl{locally polynomial machine}~$\Machine$,
constants $\IdentRadius, \CertifRadius \in \Positives$,
and a polynomial function~$\CertifPolynomial$
such that the following equivalence holds
for every \kl{graph}~$\Graph$
and every $\IdentRadius$-\kl{locally unique}
\kl{identifier assignment}~$\IdMap$~of~$\Graph$:
\begin{equation*}
  \Graph \in \Property
  \; \iff \;
  \exists \CertifMap_1 \,
  \forall \CertifMap_2
  \dots
  \Quantifier \CertifMap_{\Level}:
  \Result{\Machine}{
    \Graph, \,
    \IdMap, \,
    \CertifMap_1 \CertifConcat
    \CertifMap_2 \CertifConcat \dots \CertifConcat
    \CertifMap_{\Level}
  } \equiv
  \Accept,
\end{equation*}
where
$\Quantifier$ is~$\forall$ if $\Level$ is even
and $\exists$ otherwise,
and all quantifiers range over
$\Tuple{\CertifRadius, \CertifPolynomial}$-\kl{bounded certificate assignments}
of $\Tuple{\Graph, \IdMap}$.
We say that $\Machine$~\intro{arbitrates}~$\Property$
with respect to~$\SigmaLP{\Level}$
and call it
a $\SigmaLP{\Level}$-\intro{arbiter} for~$\Property$
under $\IdentRadius$-\kl{locally unique} \kl{identifiers}
and $\Tuple{\CertifRadius, \CertifPolynomial}$-\kl{bounded certificates}.

\AP
The class~$\intro*\PiLP{\Level}$
and the notion of $\PiLP{\Level}$-\reintro{arbiters}
are defined analogously,
with the only difference that
quantifier alternation starts with a universal quantifier
instead of an existential one.
That is,
for $\PiLP{\Level}$,
we modify the above equivalence
to read
“$\forall \CertifMap_1 \,
  \exists \CertifMap_2
  \dots
  \Quantifier \CertifMap_{\Level}$”,
where
$\Quantifier$ is~$\exists$ if $\Level$ is even
and $\forall$ otherwise.
We refer to the family of classes
$\Set{\SigmaLP{\Level}\!, \PiLP{\Level}}_{\Level \in \Naturals}$
as the \intro{locally polynomial hierarchy}.

Note that the \kl{certificate assignments}
$\CertifMap_1, \dots, \CertifMap_{\Level}$
may depend on the \kl{identifier assignment}~$\IdMap$.
Moreover,
the individual \kl{verdict} of a single \kl{node}
may vary depending on
the \kl{identifiers} and \kl{certificates} in its \kl{neighborhood}.
However,
$\Graph$'s membership in~$\Property$
(and thus whether \kl(certificate){Eve} has a winning strategy)
must be independent of the particular \kl{identifier assignment}.

\AP
Two classes at the lowest levels of the hierarchy
are of particular interest:
$\intro*\LP = \SigmaLP{0}$
(for \emph{locally polynomial time})
is the class of \kl{graph properties}
that can be \intro{decided} by a \kl{locally polynomial machine},
and
$\intro*\NLP = \SigmaLP{1}$
(for \emph{nondeterministic locally polynomial time})
is the class of \kl{graph properties}
that can be \intro{verified} by a \kl{locally polynomial machine}.
Accordingly,
$\SigmaLP{0}$-\kl{arbiters} and $\SigmaLP{1}$-\kl{arbiters}
are also called
$\LP$-\intro{deciders} and $\NLP$-\intro{verifiers},
respectively.

\subparagraph*{Complement hierarchy.}

\AP
The \intro{complement class}
of a class~$\Class$ of \kl{graph properties}
is the class
$\SetBuilder{\Complement{\Property}}{\Property \in \Class}$,
where
$\Complement{\Property}$~denotes
the \intro{complement} of a \kl{graph property}~$\Property$, i.e.,
$\intro*\Complement{\Property} = \GRAPH \setminus \Property$.
For $\Level \in \Naturals$,
we denote the \kl{complement classes} of~$\SigmaLP{\Level}$ and~$\PiLP{\Level}$
by~$\intro*\coSigmaLP{\Level}$ and~$\intro*\coPiLP{\Level}$,
and also often denote
$\coSigmaLP{0}$ and~$\coSigmaLP{1}$
by~$\intro*\coLP$ and~$\intro*\coNLP$,
respectively.
As we shall see in Corollary~\ref{cor:complement-classes},
classes on the same level of the \kl{locally polynomial hierarchy}
are neither \kl{complement classes} of each other,
nor are they closed under \kl{complementation},
so it makes sense to consider their \kl{complement classes} in their own right.
We will refer to the family of classes
$\Set{\coSigmaLP{\Level}\!, \coPiLP{\Level}}_{\Level \in \Naturals}$
as the \intro{complement hierarchy} of the \kl{locally polynomial hierarchy}.

\subparagraph*{Connection to standard complexity classes.}

\AP
On \kl{single-node graphs},
\kl{distributed Turing machines}
are equivalent to standard Turing machines
that take as input
the \kl{label} and \kl{certificates} of the unique \kl{node}.
The \kl{node}'s \kl{identifier} is irrelevant
and can therefore be assumed empty,
so the condition of the \kl{certificates}
being $\Tuple{\CertifRadius, \CertifPolynomial}$-\kl(certificate){bounded}
reduces to them
being polynomially bounded in the length of the \kl{label}.
Hence,
by restricting the classes $\SigmaLP{\Level}$~and~$\PiLP{\Level}$ to~$\NODE$
and identifying \kl{single-node graphs} with strings,
we obtain the corresponding classes
$\intro*\SigmaP{\Level}$~and~$\intro*\PiP{\Level}$
of the classical \intro{polynomial hierarchy}
introduced by Meyer and Stockmeyer~\cite{DBLP:conf/focs/MeyerS72}
(see, e.g.,~\cite[\S\,5.2]{DBLP:books/daglib/0023084}).
In particular,
$\intro*\PTIME = \LP\On\NODE = \coLP\On\NODE$
and
\phantomintro{\coNP}%
$\intro*\NP = \NLP\On\NODE = \coPiLP{1}\On\NODE$.

%%% Local Variables:
%%% mode: latex
%%% TeX-master: "../lph-paper"
%%% End:

\section{Logic with bounded quantifiers}
\label{sec:logic}

We now introduce a logical formalism
that will provide a purely syntactic characterization
of most of the complexity classes defined in the previous section.
This characterization will be presented in Section~\ref{sec:fagin}.

%–––––––––––––––––––––––––––––––––––––––––––––––––––––––––––––––––––––––––––––––
\subsection{Definitions}
\label{sec:logic-definitions}

We begin with the necessary formal definitions,
and then illustrate them with a series of examples
in Section~\ref{ssec:example-formulas},
using standard \kl{graph properties}
such as $3$-\kl{colorability} and \kl{Hamiltonicity}.
The reader may wish to skip ahead to the examples
and refer back to this subsection as needed.

\subparagraph*{Variables and interpretations.}

\AP
Let~$\intro*\FOVarSet$
be an infinite supply of \intro{first-order variables}
and
$\intro*\SOVarSet = {\bigcup_{\Arity \geq 1}}\SOVarSet[\Arity]$
be an infinite supply of \intro{second-order variables},
also called \reintro{relation variables},
where~$\reintro*\SOVarSet[\Arity]$ contains
the \kl{second-order variables} of \intro{arity}~$\Arity$
and $\SOVarSet[\Arity] \cap \SOVarSet[\Arity'] = \EmptySet$
for~$\Arity \neq \Arity'$.
We collectively refer to
the elements of~$\FOVarSet$ and~$\SOVarSet$
as \intro{variables}.

\AP
A \intro{variable assignment}~$\Assignment$
of a set of \kl{variables}
$\VarSet \subseteq \FOVarSet \cup \SOVarSet$
on a \kl{structure}~$\Struct$
is a function that maps
each \kl{first-order variable} of~$\VarSet$
to an \kl{element} of~$\Domain{\Struct}$
and
each \kl{second-order variable} of~$\VarSet$
to a relation of matching \kl{arity} over~$\Domain{\Struct}$.
The value~$\Assignment(\SOVar)$
assigned to a \kl{variable}~$\SOVar$
is called the \intro{interpretation} of~$\SOVar$ under~$\Assignment$.
We sometimes write
$\intro*\Version{\Assignment}{\SOVar}{\Relation}$
to denote the \kl{variable assignment}
that is identical to~$\Assignment$
except for mapping~$\SOVar$ to~$\Relation$.
Moreover,
if $\Assignment$~is irrelevant or clear from context,
we may also omit it to simplify the exposition,
and refer directly to~$\SOVar$ as an \kl{element} or a relation
when we really mean~$\Assignment(\SOVar)$.

\subparagraph*{Syntax and semantics.}

To avoid repetitions,
we first define the syntax and semantics
of a generalized class of \kl{logical formulas},
and then specify which particular subclasses we are interested~in.

\AP
Table~\ref{tab:syntax-semantics} shows how
\intro{logical formulas}, or simply \reintro{formulas},
are built up inductively
(in the first column),
and what they mean
(in~the third column).
It also specifies the set $\intro*\Free(\Formula)$
of \kl{variables} that occur \intro{freely}
in a given \kl{formula}~$\Formula$,
i.e., outside the scope of any \kl{quantifier}.
When we need to distinguish between
\kl[first-order variables]{first-order} and \kl{second-order variables},
we use the notations
$\intro*\FreeFO(\Formula) = \Free(\Formula) \cap \FOVarSet$
and
$\intro*\FreeSO(\Formula) = \Free(\Formula) \cap \SOVarSet$.
If $\Free(\Formula) = \EmptySet$,
then $\Formula$~is called a~\intro{sentence}.

\newcounter{rowcount}
\NewCommand{\RowLabel}[m]{%
  \refstepcounter{rowcount}\textit{\therowcount.}\label{#1}}
\begin{table}[htb]
  \phantomintro{\BitTrue}%
  \phantomintro{\Linked}%
  \phantomintro{\Equal}%
  \phantomintro{\InRel}%
  \phantomintro{\NOT}%
  \phantomintro{\OR}%
  \phantomintro{\Exists}%
  \phantomintro{\ExistsNb}%
  \phantomintro{\ExistsRel}%
  \centering
  \caption{Syntax and semantics of all logics considered in this paper.}
  \label{tab:syntax-semantics}
  \begin{tabular}{llll}
    \toprule
    & \textit{Syntax} & \textit{\kl{Free} \kl{variables}} & \textit{Semantics}
      \\
    & \kl{Formula} $\Formula[2]$
      & Set $\Free(\Formula[2])$
      & Necessary and sufficient condition for
        $\Struct, \Assignment \Satisfies \Formula[2]$
      \\
    \midrule \addlinespace
    \RowLabel{row:bit-set}
      & $\reintro*\BitTrue{i}{\FOVar[1]}$
      & $\Set{\FOVar[1]}$
      & $\Assignment(\FOVar[1]) \in \BitSet{i}{\Struct}$
      \\ \addlinespace
    \RowLabel{row:adjacency}
      & $\FOVar[1] \reintro*\Linked{i} \FOVar[2]$
      & $\Set{\FOVar[1], \FOVar[2]}$
      & $\Assignment(\FOVar[1]) \LinkRel{i}{\Struct} \Assignment(\FOVar[2])$
      \\ \addlinespace
    \RowLabel{row:equality}
      & $\FOVar[1] \reintro*\Equal \FOVar[2]$
      & $\Set{\FOVar[1], \FOVar[2]}$
      & $\Assignment(\FOVar[1]) = \Assignment(\FOVar[2])$
      \\ \addlinespace
    \RowLabel{row:so-variable}
      & $\reintro*\InRel{\SOVar}{\FOVar[1]_1,\dots,\FOVar[1]_{\Arity}}$
      & $\Set{\SOVar, \FOVar[1]_1, \dots, \FOVar[1]_{\Arity}}$
      & $\bigTuple{\Assignment(\FOVar[1]_1), \dots, \Assignment(\FOVar[1]_{\Arity})}
         \in \Assignment(\SOVar)$
      \\ \addlinespace
    \RowLabel{row:negation}
      & $\reintro*\NOT \Formula[1]$
      & $\Free(\Formula[1])$
      & not\, $\Struct, \Assignment \Satisfies \Formula[1]$
      \\ \addlinespace
    \RowLabel{row:disjunction}
      & $\Formula[1]_1 \reintro*\OR \Formula[1]_2$
      & $\Free(\Formula[1]_1) \cup \Free(\Formula[1]_2)$
      & $\Struct, \Assignment \Satisfies \Formula[1]_1$
        \:or\:
        $\Struct, \Assignment \Satisfies \Formula[1]_2$
      \\ \addlinespace
    \RowLabel{row:fo-quantification}
      & $\reintro*\Exists{\FOVar[1]} \, \Formula[1]$
      & $\Free(\Formula[1]) \setminus \Set{\FOVar[1]}$
      & $\Struct, \Version{\Assignment}{\FOVar[1]}{\Element}
         \Satisfies \Formula[1]$\:
        for some $\Element \in \Domain{\Struct}$
      \\ \addlinespace
    \RowLabel{row:bf-quantification}
      & $\underbrace{%
          \reintro*\ExistsNb{\FOVar[1]}{\FOVar[2]} \; \Formula[1]
        }_{\mathclap{\text{where $\FOVar[1] \neq \FOVar[2]$}}}$
      & $\Set{\FOVar[2]} \cup \Free(\Formula[1]) \setminus \Set{\FOVar[1]}$
      & $\Struct, \Version{\Assignment}{\FOVar[1]}{\Element}
         \Satisfies \Formula[1]$\:
        for some $\Element \in \Domain{\Struct}$
        s.t.\ $\Element \NeighborRel{\Struct} \! \Assignment(\FOVar[2])$
      \\ \addlinespace
    \RowLabel{row:so-quantification}
      & $\reintro*\ExistsRel{\SOVar} \, \Formula[1]$
      & $\Free(\Formula[1]) \setminus \Set{\SOVar}$
      & $\Struct, \Version{\Assignment}{\SOVar}{\Relation}
         \Satisfies \Formula[1]$\:
        for some $\Relation \subseteq (\Domain{\Struct})^{\Arity}$
      \\ \addlinespace
    \midrule
    \multicolumn{4}{l}{
      Here,\,
      $i, \Arity \in \Positives$,\;
      $\FOVar[1], \FOVar[1]_1, \dots, \FOVar[1]_{\Arity}, \FOVar[2]
      \in \FOVarSet$,\;
      $\SOVar \in \SOVarSet[\Arity]$,\,
      and\,
      $\Formula[1], \Formula[1]_1, \Formula[1]_2$ are \kl{formulas}.}
      \\
    \bottomrule
  \end{tabular}
\end{table}

\AP
The truth of a \kl{formula}~$\Formula$ can be evaluated
on a \kl{structure}~$\Struct$ of \kl{signature} $\Tuple{m, n}$
under a \kl{variable assignment}~$\Assignment$ of $\Free(\Formula)$ on~$\Struct$,
provided that
$\Formula$~does not contain any expressions of the form
$\BitTrue{i}{\FOVar[1]}$ or
$\FOVar[1] \Linked{j} \FOVar[2]$
for $i > m$ and~$j > n$.
Assuming this basic requirement is met,
the third column of Table~\ref{tab:syntax-semantics}
specifies in which cases
$\Struct$~\intro{satisfies} $\Formula$ under~$\Assignment$,
written $\Struct, \Assignment \intro*\Satisfies \Formula$.
If $\Formula$~is a \kl{sentence},
$\Assignment$~is irrelevant,
so we simply say that
$\Struct$~\reintro{satisfies} $\Formula$
and write $\Struct \reintro*\Satisfies \Formula$.
The \kl(structure){property}
\intro{defined} by a \kl{sentence}~$\Formula$
on a class of \kl{structures}~$\BaseProperty$
is the set
$\SetBuilder{\Struct \in \BaseProperty}{\Struct \Satisfies \Formula}$.

\AP
Lines~\ref{row:bit-set} to~\ref{row:so-variable}
of Table~\ref{tab:syntax-semantics}
correspond to \intro{atomic} \kl{formulas}.
An \kl{atomic} \kl{formula} of the form
$\BitTrue{i}{\FOVar[1]}$ or
$\FOVar[1] \Linked{i} \FOVar[2]$
refers to the corresponding
set~$\BitSet{i}{\Struct}$ or binary relation~$\LinkRel{i}{\Struct}$
given by the \kl{structure}~$\Struct$,
while an \kl{atomic} \kl{formula} of the form
$\InRel{\SOVar}{\FOVar[1]_1,\dots,\FOVar[1]_{\Arity}}$
refers to an additional relation~$\Assignment(\SOVar)$
given by the \kl{variable assignment}~$\Assignment$.
Lines~\ref{row:negation} and~\ref{row:disjunction}
describe the usual \intro{Boolean connectives},
and the remaining lines correspond to \intro{quantifiers} over different scopes:
\intro{first-order quantification}
on lines~\ref{row:fo-quantification} and~\ref{row:bf-quantification}
ranges over \kl{elements},
and
\intro{second-order quantification}
on line~\ref{row:so-quantification}
ranges over relations.

\AP
Of particular interest for this paper is
the \intro(quantifier){bounded} version of \kl{first-order quantification}
shown on line~\ref{row:bf-quantification}.
Intuitively,
$\ExistsNb{\FOVar[1]}{\FOVar[2]} \; \Formula[1]$
can be read as
“there exists an \kl{element}~$\FOVar[1]$ connected to~$\FOVar[2]$
such that $\Formula[1]$~is \kl{satisfied}”.
Here,
“connected” means that
the \kl{elements} of~$\Struct$ represented by~$\FOVar[1]$ and~$\FOVar[2]$
are related by some relation~$\LinkRel{i}{\Struct}$ or its inverse.
Thus,
\kl(quantifier){bounded} \kl{first-order quantification}
is relative to an already fixed \kl{element},
represented here by the \kl{free} \kl{variable}~$\FOVar[2]$.

\subparagraph*{Syntactic sugar.}

\AP
By nesting \kl(quantifier){bounded} \kl{first-order quantifiers},
we can \kl{quantify} over \kl{elements}
that lie within a given distance $\Radius \in \Naturals$
from the fixed \kl{element}.
To simplify this,
we introduce the shorthand notation
$\intro*\ExistsLoc{\FOVar[1]}{\Radius}{\FOVar[2]} \; \Formula[1]$,
which is defined inductively as follows,
for any $\FOVar[1], \FOVar[2] \in \FOVarSet$ and \kl{formula}~$\Formula[1]$:
\begin{align*}
  \ExistsLoc{\FOVar[1]}{0}{\FOVar[2]} \; \Formula[1]
  &\quad \text{is equivalent to} \quad
  \Substitution{\Formula[1]}{\FOVar[1]}{\FOVar[2]},
  \quad \text{and} \\
  \ExistsLoc{\FOVar[1]}{\Radius + 1}{\FOVar[2]} \; \Formula[1]
  &\quad \text{is equivalent to} \quad
  \ExistsLoc{\FOVar[1]}{\Radius}{\FOVar[2]} \,
  \bigl(
    \Formula[1] \,\OR\,
    \ExistsNb{\FOVar[1]'}{\FOVar[1]} \;
    \Substitution{\Formula[1]}{\FOVar[1]}{\FOVar[1]'}
  \bigr),
\end{align*}
\AP
where
$\intro*\Substitution{\Formula[1]}{\FOVar[1]}{\FOVar[2]}$
denotes the \kl{formula} obtained from~$\Formula[1]$
by substituting every \kl{free} occurrence of~$\FOVar[1]$ by~$\FOVar[2]$,
and~$\FOVar[1]'$ is a fresh \kl{first-order variable}
that does not occur in~$\Formula[1]$.

\AP
For additional convenience,
we will make liberal use of
truth constants
(i.e., $\intro*\False$, $\intro*\True$)
and the remaining operators of predicate logic
\phantomintro{\ForAllLoc}%
\phantomintro{\ForAllRel}%
(i.e., $\intro*\AND$, $\intro*\IMP$, $\intro*\IFF$, $\intro*\ForAll{}$),
use shorthand notations such as
$\FOVar[1] \intro*\NotEqual \FOVar[2]$,
and we may leave out some parentheses,
assuming that~$\OR$ and $\AND$
take precedence over~$\IMP$ and~$\IFF$.
Moreover,
a sequence consisting
solely of \intro{existential}~($\Exists{}$) or
solely of \intro{universal}~($\ForAll{}$) \kl{quantifiers}
may be combined into
a single \kl{quantifier} that binds a tuple of \kl{variables}.
For instance,
we may write
$\ForAllRel{\Vector{\SOVar}} \, \Formula$
instead of
$\ForAllRel{\SOVar_1} \dots \ForAllRel{\SOVar_n} \, \Formula$,
where
$\Vector{\SOVar} = \Tuple{\SOVar_1, \dots, \SOVar_n}$.

\subparagraph*{Formulas expressing relations.}

\AP
Given a \kl{formula}~$\Formula$
with
$\FreeFO(\Formula) = \Set{\FOVar[1]_1, \dots, \FOVar[1]_n}$,
we often write
$\Formula\intro*\Of{\FOVar[1]_1, \dots, \FOVar[1]_n}$
instead of simply~$\Formula$
to convey the intention that $\Formula$~expresses
some relation between the \kl{elements} represented by
$\FOVar[1]_1, \dots, \FOVar[1]_n$.
If we then want to express that the same relation holds
between some other \kl{variables}
$\FOVar[2]_1, \dots, \FOVar[2]_n$
that do not occur in the scope of any \kl{quantifier} in~$\Formula$,
we write
$\Formula\reintro*\Of{\FOVar[2]_1, \dots, \FOVar[2]_n}$
to denote the \kl{formula} obtained from~$\Formula$
by simultaneously replacing all \kl{free} occurrences of
$\FOVar[1]_1, \dots, \FOVar[1]_n$
by
$\FOVar[2]_1, \dots, \FOVar[2]_n$,
respectively.

\subparagraph*{Fragments of first-order logic.}

\AP
For our purposes,
the class~$\intro*\FOL$ of \kl{formulas} of \intro{first-order logic}
is generated by the grammar%
\begin{equation}
  \Formula \Coloneqq \BitTrue{i}{\FOVar[1]}
                \mid \FOVar[1] \Linked{i} \FOVar[2]
                \mid \FOVar[1] \Equal \FOVar[2]
                \mid \InRel{\SOVar}{\FOVar[1]_1,\dots,\FOVar[1]_{\Arity}}
                \mid \NOT \Formula
                \mid \Formula \OR \Formula
                \mid \Exists{\FOVar[1]} \, \Formula,
  \label{eq:first-order-logic} \tag{$\FOL$}
\end{equation}
where
$i, \Arity \in \Positives$,\,
$\FOVar[1], \FOVar[1]_1, \dots, \FOVar[1]_{\Arity}, \FOVar[2]
 \in \FOVarSet$,
and $\SOVar \in \SOVarSet[\Arity]$.

\AP
The class~$\intro*\BFL$ of \kl{formulas} of
the \intro{bounded fragment} of \kl{first-order logic}
is defined by a similar grammar,
the only difference being
that \kl{first-order quantification} is \kl(quantifier){bounded}:
\begin{equation}
  \Formula \Coloneqq \BitTrue{i}{\FOVar[1]}
                \mid \FOVar[1] \Linked{i} \FOVar[2]
                \mid \FOVar[1] \Equal \FOVar[2]
                \mid \InRel{\SOVar}{\FOVar[1]_1,\dots,\FOVar[1]_{\Arity}}
                \mid \NOT \Formula
                \mid \Formula \OR \Formula
                \mid \underbrace{%
                       \ExistsNb{\FOVar[1]}{\FOVar[2]} \; \Formula
                     }_{\text{where $\FOVar[1] \neq \FOVar[2]$}}
  \label{eq:bounded-fragment} \tag{$\BFL$}
\end{equation}

\AP
To give some basic examples,
when evaluated on (the \kl{structural representation} of) a \kl{graph},
the following $\BFL$-\kl{formulas} state that
the \kl{element} represented by the \kl{first-order variable}~$\FOVar[1]$
corresponds
to a \kl{node},
to a \kl{labeling bit} of value $0$, and
to a \kl{labeling bit} of value $1$,
respectively:
\phantomintro{\IsNode}%
\phantomintro{\IsBit}%
\begin{align*}
  \reintro*\IsNode\Of{\FOVar[1]} =
  \NOT \ExistsNb{\FOVar[2]}{\FOVar[1]} \, (\FOVar[2] \Linked{2} \FOVar[1])
  \qquad
  \reintro*\IsBit{0}\Of{\FOVar[1]} &=
  \NOT \IsNode\Of{\FOVar[1]} \AND \NOT \BitTrue{1}{\FOVar[1]}
  \\
  \reintro*\IsBit{1}\Of{\FOVar[1]} &=
  \NOT \IsNode\Of{\FOVar[1]} \AND \BitTrue{1}{\FOVar[1]}
\end{align*}
\AP
The first \kl{formula} is particularly useful,
as we will often restrict \kl{quantification} to \kl{nodes}.
To simplify this,
we introduce the notation
\phantomintro{\ForAllNode}%
$\intro*\ExistsNode{\FOVar[1]} \; \Formula$
to abbreviate
$\Exists{\FOVar[1]} \, (\IsNode\Of{\FOVar[1]}  \AND \Formula)$,
and
\phantomintro{\ExistsNbNode}%
\phantomintro{\ForAllNbNode}%
\phantomintro{\ForAllLocNode}%
$\intro*\ExistsLocNode{\FOVar[1]}{\Radius}{\FOVar[2]} \; \Formula$
to abbreviate
$\ExistsLoc{\FOVar[1]}{\Radius}{\FOVar[2]} \,
 (\IsNode\Of{\FOVar[1]}  \AND \Formula)$,
and similarly for \kl{universal quantifiers}.

\AP
Since every $\BFL$-\kl{formula} contains
at least one \kl{free} \kl{first-order variable},
evaluating such a \kl{formula}
always requires a \kl{variable assignment}
that provides an \kl{element} as a “starting point”.
To consider \kl{structures} without \kl{variable assignments},
we introduce~$\intro*\LFOL$,
the class of \kl{formulas} of \intro{local first-order logic},
which are $\BFL$-\kl{formulas} prefixed by a single
\kl{universal} \kl{first-order quantifier}.
That is,
$\LFOL$ consists of \kl{formulas} of the form
$\ForAll{\FOVar} \, \Formula$,
where $\FOVar \in \FOVarSet$ and $\Formula \in \BFL$.

\subparagraph*{Second-order hierarchies.}

\AP
$\FOL$ and $\LFOL$ form the basis
of two hierarchies of alternating \kl{second-order quantifiers}.
The first,
called the \intro{second-order hierarchy},
starts with the base class
$\intro*\SigmaFOL{0} = \intro*\PiFOL{0} = \FOL$,
and continues for $\Level > 0$
with the classes
$\reintro*\SigmaFOL{\Level}$ and $\reintro*\PiFOL{\Level}$
that are obtained by prepending
blocks of \kl{existential} and \kl{universal} \kl{second-order quantifiers}
to \kl{formulas} of $\PiFOL{\Level - 1}$ and $\SigmaFOL{\Level - 1}$,
respectively.
That is,
$\SigmaFOL{\Level}$
consists of \kl{formulas} of the form
$\ExistsRel{\SOVar_1} \dots \ExistsRel{\SOVar_n} \, \Formula$,
where $\SOVar_1, \dots, \SOVar_n \in \SOVarSet$
and $\Formula \in \PiFOL{\Level - 1}$,
whereas
$\PiFOL{\Level}$
consists of \kl{formulas} of the form
$\ForAllRel{\SOVar_1} \dots \ForAllRel{\SOVar_n} \, \Formula$,
where~$\Formula \in \SigmaFOL{\Level - 1}$.

\AP
The other hierarchy,
called the \intro{local second-order hierarchy},
is defined the same way,
except that it starts with the base class~$\LFOL$ instead of~$\FOL$.
That is,
$\intro*\SigmaLFOL{0} = \intro*\PiLFOL{0} = \LFOL$,
and for $\Level > 0$,
the classes
$\reintro*\SigmaLFOL{\Level}$ and $\reintro*\PiLFOL{\Level}$
are obtained by prepending
blocks of \kl{existential} and \kl{universal} \kl{second-order quantifiers}
to \kl{formulas} of $\PiLFOL{\Level - 1}$ and $\SigmaLFOL{\Level - 1}$,
respectively.
Notice that it would \emph{not} be equivalent to define~$\PiLFOL{\Level}$ as
the set of negations of \kl{formulas} in~$\SigmaLFOL{\Level}$
because $\LFOL$~is not closed under negation.

\AP
As with the \kl{locally polynomial hierarchy},
it can be helpful to think of \kl{formulas}
of the (\kl[local second-order hierarchy]{local}) \kl{second-order hierarchy}
as a two-player game between \intro{Eve} and \intro{Adam},
who choose
the \kl{existentially} and \kl{universally} \kl{quantified} relations,
respectively.
From this point of view,
the referee of the game corresponds to the $\FOL$- or $\LFOL$-\kl{subformula}
nested inside the \kl{second-order quantifications},
and the whole \kl{formula} is \kl{satisfied} by the input \kl{structure}
precisely if \kl{Eve} has a winning strategy.

\AP
\intro{Second-order logic} is
the union of all classes of the \kl{second-order hierarchy},
and similarly
\intro{local second-order logic} is
the union of all classes of the \kl{local second-order hierarchy}.
The classes $\SigmaFOL{1}$ and~$\SigmaLFOL{1}$ will be referred to
as the the \intro{existential fragments}
of \kl{second-order logic} and \kl{local second-order logic},
respectively.

\subparagraph*{Classes of definable properties.}

\AP
For any class of \kl{formulas}
$\ClassL \in
 \Set{
   \SigmaFOL{\Level},
   \PiFOL{\Level},
   \SigmaLFOL{\Level},
   \PiLFOL{\Level}
 }$
with $\Level \in \Naturals$,
we use the corresponding boldface notation
$\Class \in
 \Set{
   \intro*\SigmaFO{\Level},
   \intro*\PiFO{\Level},
   \intro*\SigmaLFO{\Level},
   \intro*\PiLFO{\Level}
 }$
to denote the class of \kl{structure properties}
that can be \kl{defined} by a \kl{formula} of~$\ClassL$.
It is worth noting that
$\SigmaFO{\Level}\On{\NODE} = \SigmaLFO{\Level}\On{\NODE}$
and
$\PiFO{\Level}\On{\NODE} = \PiLFO{\Level}\On{\NODE}$,
since the distinction between
\kl(quantifier){bounded} and \kl(quantifier){unbounded} \kl{quantification}
is irrelevant on \kl{single-node graphs}.

%–––––––––––––––––––––––––––––––––––––––––––––––––––––––––––––––––––––––––––––––
\subsection{Example formulas}
\label{ssec:example-formulas}

\AP
We now show how to express a number of \kl{graph properties}
in \kl{local second-order logic},
starting with a very simple \kl(graph){property}
that the \kl{nodes} can check locally:
$\intro*\ALLSELECTED$,
the set of \kl{labeled graphs}
in which all \kl{nodes} are assigned the \kl{label}~$1$
(i.e., they are all “selected”).

\begin{example}
  \label{ex:all-selected}
  We can easily \kl{define}~$\ALLSELECTED$ on~$\GRAPH$
  with the $\LFOL$-\kl{formula}
  $\ForAllNode{\FOVar[1]} \, \IsSelected\Of{\FOVar[1]}$,
  where
  $\FOVar[1]$ is a \kl{first-order variable},
  and
  \begin{equation*}
    \intro*\IsSelected\Of{\FOVar[1]}
    \, = \,
    \ExistsNb{\FOVar[2]}{\FOVar[1]} \,
    \bigl(
      \IsBit{1}\Of{\FOVar[2]}
      \, \AND \,
      \NOT \ExistsNb{\FOVar[3]}{\FOVar[2]} \,
      (
        \FOVar[3] \Linked{1} \FOVar[2]
        \, \OR \,
        \FOVar[2] \Linked{1} \FOVar[3]
      )
    \bigr)
  \end{equation*}
  is a $\BFL$-\kl{formula} that states that
  the \kl{node} represented by~$\FOVar[1]$ is \kl{labeled} with the string~$1$.
  Here,
  the \kl{first-order variable}~$\FOVar[2]$
  is used to represent $\FOVar[1]$'s unique \kl{labeling bit}.
  \lipicsEnd
\end{example}

\AP
Next,
we more formally revisit $3$-\intro{colorability},
a \kl(graph){property}
for which we already sketched a \kl{formula}
in Example~\ref{ex:3-colorable-overview}.
For $k \in \Positives$,
the set $\intro*\COLORABLE{k}$ consists of
all \kl{graphs}~$\Graph$ for which there exists a function
$f \colon \NodeSet{\Graph} \to \Range{k}*$
such that
$f(\Node[1]) \neq f(\Node[2])$
for all $\Set{\Node[1], \Node[2]} \in \EdgeSet{\Graph}$.

\begin{example}
  \label{ex:3-colorable}
  We can \kl{define} $\COLORABLE{3}$ on~$\GRAPH$
  with the $\SigmaLFOL{1}$-\kl{formula}
  \begin{equation*}
    \ExistsRel{C_0, C_1, C_2} \;
    \ForAllNode{\FOVar[1]} \, \WellColored\Of{\FOVar[1]},
  \end{equation*}
  where
  $C_0$, $C_1$ and~$C_2$ are \kl{unary} \kl{second-order variables}
  intended to represent
  the sets of \kl{nodes} colored with $0$, $1$ and~$2$,
  respectively,
  $\FOVar[1]$ is a \kl{first-order variable},
  and
  \begin{equation*}
    \intro*\WellColored\Of{\FOVar[1]} =
    \Bigl( \,
      \BigOR_{\lalign{i \in \Range{3}*}} \InRel{C_i}{\FOVar[1]}
    \, \Bigr)
    \, \AND \,
    \Bigl( \,
      \BigAND_{\lalign{i,j \in \Range{3}*:\, i \neq j}}
      \NOT
      \bigl(
        \InRel{C_i}{\FOVar[1]}
        \AND
        \InRel{C_j}{\FOVar[1]}
      \bigr)
    \, \Bigr)
    \, \AND \:
    \ForAllNbNode{\FOVar[2]}{\FOVar[1]}
    \Bigl( \,
      \BigAND_{\lalign{i \in \Range{3}*}}
      \NOT
      \bigl(
        \InRel{C_i}{\FOVar[1]}
        \AND
        \InRel{C_i}{\FOVar[2]}
      \bigr)
    \, \Bigr)
  \end{equation*}
  states that
  the \kl{node} represented by~$\FOVar[1]$ is correctly colored.
  More precisely,
  the first two conjuncts express that
  $\FOVar[1]$~is assigned one color and one color only,
  while the third conjunct expresses that
  $\FOVar[1]$'s color is different from its \kl{neighbors}' colors.
  \lipicsEnd
\end{example}

\AP
To make things a little more challenging,
let us now consider the \kl{complement} of $\ALLSELECTED$,
which we denote by $\intro*\NOTALLSELECTED$.
This \kl(graph){property} is more difficult
to express in \kl{local second-order logic}.
In fact,
as we will see in the proof of Proposition~\ref{prp:colp-vs-nlp},
it is not $\SigmaLFOL{1}$-\kl{definable}.

\begin{example}
  \label{ex:not-all-selected}
  A straightforward way to \kl{define}~$\NOTALLSELECTED$ on~$\GRAPH$
  would be to negate the \kl{formula} from Example~\ref{ex:all-selected},
  yielding the $\FOL$-\kl{formula}
  $\ExistsNode{\FOVar[1]} \, \NOT \IsSelected\Of{\FOVar[1]}$.
  But this \kl{formula} does not belong to \kl{local second-order logic}
  because of the
  \kl(quantifier){unbounded} \kl{existential} \kl{first-order quantification}
  over~$\FOVar[1]$.
  To remedy this,
  we can rewrite it as an equivalent
  $\SigmaLFOL{3}$\nobreakdash-\kl{formula} $\Fml{ExistsUnselectedNode}$,
  which intuitively describes the following game:
  First,
  \kl{Eve} tries to cover the input \kl{graph}
  with a spanning forest
  whose roots include only unselected \kl{nodes}.
  She represents this forest by
  a \kl{binary} \kl{relation variable}~$P$,
  where
  $\InRel{P}{\FOVar[1], \FOVar[2]}$
  is intended to mean
  “the parent of~$\FOVar[1]$ is~$\FOVar[2]$”.
  If she succeeds,
  one should thus always reach an unselected \kl{node}
  by following parent pointers.
  Then,
  \kl{Adam} tries to disprove \kl{Eve}'s claim that $P$~represents a forest
  by showing that the relation contains a directed cycle.
  To do so,
  he chooses a subset~$X$ of \kl{nodes},
  and then asks \kl{Eve}
  to assign a charge (positive or negative) to each \kl{node}
  such that
  roots are positive,
  children outside~$X$ have the same charge as their parent, and
  children in~$X$ have the opposite charge of their parent.
  Now,
  if $P$~is cycle-free,
  then \kl{Eve} can charge the \kl{nodes} as requested
  by simply traversing the \kl{paths} of each tree from top to bottom,
  starting with a positive charge at the root,
  and inverting the charge every time she encounters a \kl{node} in~$X$.
  However,
  if $P$~contains a cycle,
  then \kl{Adam} can choose~$X$ to be a singleton set
  containing exactly one \kl{node} of the cycle.
  By doing so,
  he prevents \kl{Eve} from winning
  because she will either have to
  charge the \kl{node} in~$X$ like its parent,
  or charge another \kl{node} of the cycle differently than its parent.

  Formally,
  we represent the positive and negative charges
  by a \kl{unary} \kl{relation variable}~$Y$
  (interpreted as the set of positive \kl{nodes}),
  and write
  \begin{equation*}
    \Fml{ExistsUnselectedNode}
    \, = \,
    \ExistsRel{P} \:
    \ForAllRel{X} \:
    \ExistsRel{Y} \:
    \ForAllNode{\FOVar[1]} \,
    \bigl(
      \PointsTo{\NOT \IsSelected}\Of{\FOVar[1]}
    \bigr).
  \end{equation*}
  \AP
  Here,
  the \kl{subformula} $\PointsTo{\NOT \IsSelected}\Of{\FOVar[1]}$
  basically states that $\FOVar[1]$'s parent pointer
  points in the direction of an unselected \kl{node},
  assuming that both players play optimally
  and that \kl{Eve} wins the game described above.
  Since the same idea will be useful later
  for conditions other than $\NOT \IsSelected\Of{\FOVar[1]}$,
  we present this \kl{subformula} as a \kl{formula} schema
  that can be instantiated with
  any $\BFL$-\kl{formula}~$\Formula[3]\Of{\FOVar[1]}$:
  \begin{equation*}
    \intro*\PointsTo{\Formula[3]}\Of{\FOVar[1]}
    \, = \,
    \Fml{UniqueParent}\Of{\FOVar[1]}
    \, \AND \,
    \RootCase{\Formula[3]}\Of{\FOVar[1]}
    \, \AND \,
    \Fml{ChildCase}\Of{\FOVar[1]},
  \end{equation*}
  where
  \begin{equation*}
    \Fml{UniqueParent}\Of{\FOVar[1]}
    \, = \,
    \ExistsLocNode{\FOVar[2]}{1}{\FOVar[1]} \:\!
    \Bigl(
      \InRel{P}{\FOVar[1], \FOVar[2]}
      \, \AND \,
      \ForAllLocNode{\FOVar[3]}{1}{\FOVar[1]} \:\!
      \bigl(
        \InRel{P}{\FOVar[1], \FOVar[3]}
        \IMP
        \FOVar[3] \Equal \FOVar[2]
      \bigr)
    \Bigr)
  \end{equation*}
  \AP
  states that $\FOVar[1]$~has exactly one parent
  (possibly itself, in which case it is a root);
  using the alias~%
  $\intro*\Root\Of{\FOVar[1]} = \InRel{P}{\FOVar[1], \FOVar[1]}$,
  \begin{equation*}
    \RootCase{\Formula[3]}\Of{\FOVar[1]}
    \, = \,
    \Root\Of{\FOVar[1]}
    \, \IMP \,
    \bigl( \,
      \Formula[3]\Of{\FOVar[1]}
      \AND
      \InRel{Y}{\FOVar[1]}
    \, \bigr)
  \end{equation*}
  states that if $\FOVar[1]$~is a root,
  then it \kl{satisfies} the target condition~$\Formula[3]$
  and is positively charged;
  and
  \begin{equation*}
    \Fml{ChildCase}\Of{\FOVar[1]}
    \, = \,
    \NOT \Root\Of{\FOVar[1]}
    \, \IMP \,
    \ExistsNbNode{\FOVar[2]}{\FOVar[1]} \:\!
    \Bigl(
      \InRel{P}{\FOVar[1], \FOVar[2]}
      \, \AND \,
      \bigl( \,
        \InRel{Y}{\FOVar[1]}
        \IFF
        \NOT (\InRel{Y}{\FOVar[2]} \IFF \InRel{X}{\FOVar[1]})
      \, \bigr)
    \Bigr)
  \end{equation*}
  states that if $\FOVar[1]$~is a child,
  then
  it has the same charge as its parent
  if it lies outside~$X$,
  and the opposite charge of its parent
  if it belongs to~$X$.
  \lipicsEnd
\end{example}

\AP
The spanning-forest construction described in Example~\ref{ex:not-all-selected}
can be generalized to express the \kl{complement} of any \kl{graph property}
\kl{definable} in \kl{local second-order logic}.
We now demonstrate this with the \kl{complement} of $\COLORABLE{3}$,
which we denote by $\intro*\NONCOLORABLE{3}$.
As we will show in Corollary~\ref{cor:non-three-colorable-not-in-nlp},
this \kl(graph){property} is not $\SigmaLFOL{1}$-\kl{definable} either.

\begin{example}
  \label{ex:not-3-colorable}
  To \kl{define} $\NONCOLORABLE{3}$,
  we could simply negate the \kl{formula} from Example~\ref{ex:3-colorable},
  yielding the $\PiFOL{1}$-\kl{formula}
  $\ForAllRel{C_0, C_1, C_2} \,
   \ExistsNode{\FOVar[1]} \: \NOT \WellColored\Of{\FOVar[1]}$.
  But just as in Example~\ref{ex:not-all-selected},
  this \kl{formula} does not belong to \kl{local second-order logic}
  because of the
  \kl(quantifier){unbounded} \kl{existential} \kl{first-order quantification}
  over~$\FOVar[1]$.
  Fortunately,
  the solution is also very similar:
  we can rewrite our initial attempt
  as the equivalent $\PiLFOL{4}$\nobreakdash-\kl{formula}
  $\ForAllRel{C_0, C_1, C_2} \, \Fml{ExistsBadNode}$,
  using the \kl{subformula}
  \begin{equation*}
    \Fml{ExistsBadNode}
    \, = \,
    \ExistsRel{P} \:
    \ForAllRel{X} \:
    \ExistsRel{Y} \:
    \ForAllNode{\FOVar[1]} \,
    \bigl(
      \PointsTo{\NOT \WellColored}\Of{\FOVar[1]}
    \bigr),
  \end{equation*}
  where
  $P$~is a \kl{binary} relation,
  $X$~and~$Y$ are sets,
  and the \kl{subformula} $\PointsTo{\NOT \WellColored}\Of{\FOVar[1]}$
  is an instantiation of
  the \kl{formula} schema from Example~\ref{ex:not-all-selected}.
  \lipicsEnd
\end{example}

\AP
In the previous two examples,
we needed to express the existence of a particular \kl{node}.
But what if we also want to make sure that the \kl{node} is unique?
It turns out that
the spanning-forest construction from Example~\ref{ex:not-all-selected}
can be easily extended to enforce that
\kl{Eve}'s spanning forest is actually a spanning tree.
Let us illustrate this with
the \kl(graph){property} $\intro*\ONESELECTED$,
which requires that
\emph{exactly} one \kl{node} has the \kl{label}~$1$.

\begin{example}
  \label{ex:one-selected}
  To \kl{define} $\ONESELECTED$ on~$\GRAPH$,
  we intuitively let \kl{Eve} and \kl{Adam} play
  a similar $3$-move game as in Example~\ref{ex:not-all-selected},
  but now we tighten the rules for \kl{Eve}
  so that her spanning forest~$P$
  must be a spanning tree rooted at the unique selected \kl{node}.
  Again,
  \kl{Adam} tries to refute her choice,
  and this time he has two possible lines of attack:
  either find a directed cycle in~$P$
  (as before),
  or show that there are at least two selected \kl{nodes}.
  He launches both attacks simultaneously
  by choosing a single subset~$X$ of \kl{nodes},
  and then presenting \kl{Eve} with two challenges:
  to charge each \kl{node} positively or negatively according to~$X$
  (as~described in Example~\ref{ex:not-all-selected}),
  and to tell each \kl{node}
  whether the selected \kl{node} belongs to~$X$.
  As before,
  if $P$~contains a directed cycle,
  he can make \kl{Eve} fail the first challenge.
  But now,
  if there are at least two selected \kl{nodes},
  he can make her fail the second challenge
  by choosing a set~$X$ that contains only one selected \kl{node}.
  In this way,
  no matter what \kl{Eve} tells each \kl{node}
  about the supposedly unique selected \kl{node}'s membership in~$X$,
  there will necessarily be an inconsistency
  between two adjacent \kl{nodes}
  that receive conflicting information.

  Formally,
  we represent \kl{Eve}'s responses to \kl{Adam}'s challenges
  by two \kl{unary} \kl{relation variables}~$Y$ and~$Z$,
  the former corresponding to the set of positive \kl{nodes}
  (as before),
  and the latter intended as a Boolean predicate indicating
  whether the selected \kl{node} belongs to~$X$.
  We then \kl{define} $\ONESELECTED$
  with the $\SigmaLFOL{3}$-\kl{formula}
  \begin{equation*}
    \ExistsRel{P} \:
    \ForAllRel{X} \:
    \ExistsRel{Y, Z} \:
    \ForAllNode{\FOVar[1]} \,
    \bigl(
      \PointsToUnique{\IsSelected}\Of{\FOVar[1]}
    \bigr),
  \end{equation*}
  where
  $\PointsToUnique{\Formula[3]}\Of{\FOVar[1]}$,
  for $\Formula[3] \in \BFL$,
  is an extension of the \kl{formula} schema
  $\PointsTo{\Formula[3]}\Of{\FOVar[1]}$
  from Example~\ref{ex:not-all-selected}.
  It states that
  in addition to $\FOVar[1]$~pointing towards
  a \kl{node} that \kl{satisfies} the target condition~$\Formula[3]$,
  that \kl{node} is also unique
  (again,
  assuming that both players play optimally
  and that \kl{Eve} wins the game).
  We can write this as
  \AP
  \begin{equation*}
    \intro*\PointsToUnique{\Formula[3]}\Of{\FOVar[1]}
    \, = \,
    \PointsTo{\Formula[3]}\Of{\FOVar[1]}
    \, \AND \,
    \BelievesInOne{\Formula[3]}\Of{\FOVar[1]},
  \end{equation*}
  where
  \begin{equation*}
    \BelievesInOne{\Formula[3]}\Of{\FOVar[1]}
    \, = \,
    \Bigl(
      \ForAllNbNode{\FOVar[2]}{\FOVar[1]} \:\!
      \bigl(
        \InRel{Z}{\FOVar[1]} \IFF \InRel{Z}{\FOVar[2]}
      \bigr)
    \Bigr)
    \, \AND \,
    \Bigl(
      \Formula[3]\Of{\FOVar[1]}
      \, \IMP \,
      \bigl(
        \InRel{Z}{\FOVar[1]} \IFF \InRel{X}{\FOVar[1]}
      \bigr)
    \Bigr)
  \end{equation*}
  states that
  the information available to~$\FOVar[1]$
  does not contradict
  the uniqueness of the \kl{node} \kl{satisfying}~$\Formula[3]$.
  The first conjunct ensures that
  all \kl{nodes} agree on the value of the Boolean predicate~$Z$,
  while the second conjunct ensures that
  any \kl{node} \kl{satisfying}~$\Formula[3]$
  sets $Z$ to true
  if and only if
  the \kl{node} itself belongs to~$X$.
  \lipicsEnd
\end{example}

\AP
The refined spanning-tree construction shown in Example~\ref{ex:one-selected}
is a useful building block
that allows us to express
many natural \kl{graph properties} in \kl{local second-order logic}.
We demonstrate this first for $\intro*\HAMILTONIAN$,
the \kl(graph){property} of containing a \intro{Hamiltonian cycle},
i.e., a \kl{cycle} that passes through each \kl{node} exactly once.

\begin{example}
  \label{ex:hamiltonian}
  A \kl{Hamiltonian cycle} can be seen as
  a Hamiltonian path extended with an additional \kl{edge}
  connecting the endpoints of the path.
  In turn,
  a Hamiltonian path
  can be seen as a special case of a spanning tree,
  where each \kl{node} has at most one child.
  Using the previous spanning-tree construction,
  this allows us to \kl{define} $\HAMILTONIAN$ on~$\GRAPH$
  with the following $\SigmaLFOL{3}$\nobreakdash-\kl{formula}:
  \begin{equation*}
    \ExistsRel{P} \:
    \ForAllRel{X} \:
    \ExistsRel{Y, Z} \:
    \ForAllNode{\FOVar[1]} \,
    \bigl(
      \PointsToUnique{\Root}\Of{\FOVar[1]}
      \, \AND \,
      \Fml{MaxOneChild}\Of{\FOVar[1]}
      \, \AND \,
      \Fml{SeesLeafIfRoot}\Of{\FOVar[1]}
    \bigr)
  \end{equation*}
  Here,
  $\PointsToUnique{\Root}\Of{\FOVar[1]}$
  is an instantiation of
  the \kl{formula} schema from Example~\ref{ex:one-selected}
  with the \kl{subformula} $\Root$ from Example~\ref{ex:not-all-selected},
  which ensures that $P$~represents a spanning tree;
  \begin{equation*}
    \Fml{MaxOneChild}\Of{\FOVar[1]}
    \, = \,
    \ForAllNbNode{\FOVar[2],\FOVar[3]}{\FOVar[1]} \:\!
    \bigl(
      \InRel{P}{\FOVar[2], \FOVar[1]}
      \AND
      \InRel{P}{\FOVar[3], \FOVar[1]}
      \, \IMP \,
      \FOVar[2] \Equal \FOVar[3]
    \bigr)
  \end{equation*}
  states that $\FOVar[1]$~has at most one child in the spanning tree;
  and
  \begin{equation*}
    \Fml{SeesLeafIfRoot}\Of{\FOVar[1]}
    \, = \,
    \Root\Of{\FOVar[1]}
    \, \IMP \,
    \ExistsNbNode{\FOVar[2]}{\FOVar[1]} \,
    \bigl(
      \NOT \InRel{P}{\FOVar[2], \FOVar[1]}
      \, \AND \,
      \ForAllNbNode{\FOVar[3]}{\FOVar[2]} \:\!
      (
        \NOT \InRel{P}{\FOVar[3], \FOVar[2]}
      )
    \bigr)
  \end{equation*}
  states that if $\FOVar[1]$~is the root,
  then it is \kl{adjacent} to the (unique) leaf~$\FOVar[2]$;
  this leaf cannot be the root's child,
  since a \kl{Hamiltonian cycle} requires at least three \kl{nodes}.
  \lipicsEnd
\end{example}

As noted by Göös and Suomela~\cite[\S\,5.1]{DBLP:journals/toc/GoosS16},
“a spanning tree equipped with a locally checkable proof is a versatile tool”.
Adapting their examples to our framework,
we can use
the spanning-tree construction from Example~\ref{ex:one-selected}
to also express the following \kl{graph properties}
as $\SigmaLFOL{3}$-\kl{formulas}:
\begin{itemize}
  \item \AP
        \phantomintro{acyclic}%
        $\intro*\ACYCLIC$,
        the \kl(graph){property} of not containing any cycles.
        To convince the \kl{nodes} of this \kl(graph){property},
        \kl{Eve} provides a spanning tree
        (which \kl{Adam} tries to refute as before),
        and each \kl{node} checks that
        all its \kl{incident} \kl{edges} belong to this tree.
  \item \AP
        $\intro*\ODD$,
        the \kl(graph){property} of having an odd number of \kl{nodes}.
        Again,
        \kl{Eve} provides a spanning tree,
        and this time she aggregates modulo-two counters
        from the leaves to the root.
        Each \kl{node} then checks that its own counter value
        equals one plus the sum of its children's values,
        modulo two,
        and the root also checks that its own value is one.
        On an intuitive level,
        the modulo-two summation at each \kl{node}
        can be implemented by a finite automaton
        that traverses the values of the \kl{node}'s children
        in an arbitrary order.
        More formally,
        we can represent the counter values and the states of the automaton
        by unary relations,
        and the order of a \kl{node}'s children
        by a binary relation.
        All of these are chosen by \kl{Eve} in her first move,
        and subsequently verified by the \kl{nodes}.
  \item $\NONCOLORABLE{2}$,
        the \kl(graph){property} of not being $2$-\kl{colorable}.
        A \kl{graph} is non-$2$-\kl{colorable}
        if and only if
        it contains a cycle of odd length
        (see, e.g., \cite[Prp.~1.6.1]{DBLP:books/daglib/0030488}).
        To prove the existence of such a cycle,
        \kl{Eve} retraces it with a binary relation~$\SOVar$
        (choosing an arbitrary but consistent orientation),
        constructs a spanning tree
        whose root belongs to that cycle,
        and propagates a modulo-two counter around the cycle.
        Again,
        the counter values are represented by a unary relation.
        The root checks that it has
        the same value as its $\SOVar$-predecessor,
        while every other \kl{node} on the cycle checks
        that it has the opposite value of its $\SOVar$-predecessor.
        Since the root is guaranteed to be unique
        by the spanning-tree construction,
        this implies that the cycle has odd length.
        (Notice that $\SOVar$~may retrace several disjoint cycles,
        but only the one containing the root will be odd.)
\end{itemize}

\AP
To finish with a slightly more involved example,
let us consider $\intro*\NONHAMILTONIAN$,
the \kl{complement} \kl(graph){property} of $\HAMILTONIAN$.
Since we already have
the $\SigmaLFOL{3}$-\kl{formula} for $\HAMILTONIAN$
from Example~\ref{ex:hamiltonian},
we could simply use
the complementation technique from Example~\ref{ex:not-3-colorable}
to get a $\PiLFOL{6}$-\kl{formula} for $\NONHAMILTONIAN$.
But we can do better;
below we provide a $\PiLFOL{4}$-\kl{formula}.

\begin{example}
  \label{ex:non-hamiltonian}
  We \kl{define} $\NONHAMILTONIAN$
  based on the following characterization:
  a \kl{graph} is \kl{Hamiltonian}
  if and only if
  it contains a \kl{spanning subgraph}
  (i.e., a \kl{subgraph} containing all \kl{nodes})
  that is $2$-\kl{regular}
  (i.e., all \kl{nodes} have \kl{degree}~$2$)
  and \kl{connected}
  (i.e., any two \kl{nodes} are linked by a \kl{path}).

  Intuitively,
  the \kl{complement} of this \kl(graph){property}
  can be tested through the following game,
  where
  \kl{Adam} starts by proposing a supposed \kl{Hamiltonian cycle},
  and then \kl{Eve} tries to prove that his proposition is incorrect:
  First,
  \kl{Adam} chooses a $2$-\kl{regular} \kl{spanning subgraph},
  which he represents by a \kl{binary} relation~$H$.
  The intended meaning of $\InRel{H}{\FOVar[1], \FOVar[2]}$ is
  “the \kl{edge} $\Set{\FOVar[1], \FOVar[2]}$ belongs to the \kl{subgraph}”.
  He claims that the chosen \kl{subgraph} is a \kl{Hamiltonian cycle}.
  Next,
  \kl{Eve} tries to disprove this claim.
  If she is right,
  there are two possible cases:
  either \kl{Adam} did not really propose a $2$-\kl{regular} \kl{spanning subgraph},
  or his \kl{subgraph} is not \kl{connected}
  i.e., it consists of multiple disjoint \kl{cycles}.
  To tell the \kl{nodes} which of the two cases applies
  (say, giving priority to the first if both apply),
  \kl{Eve} assigns each \kl{node} a bit,
  which she represents by a \kl{unary} relation~$C$.
  In the first case,
  represented by~$\NOT \InRel{C}{\FOVar[1]}$,
  she constructs a spanning forest~$P$ whose roots are \kl{nodes}
  at which the $2$-\kl{regularity} condition is violated.
  In the second case,
  represented by~$\InRel{C}{\FOVar[1]}$,
  she chooses a subset~$S$ of \kl{nodes} containing
  some but not all \kl{components} of \kl{Adam}'s \kl{subgraph},
  without cutting any \kl{components},
  and then constructs a spanning forest~$P$ whose roots
  witness that~$S$ does indeed partition the \kl{subgraph} into two nonempty parts.
  (There must be at least one witness,
  since we require by definition that the input \kl{graph} is \kl{connected}.)
  In both cases,
  the spanning forest is validated
  using the technique from Example~\ref{ex:not-all-selected}.
  This adds two more steps to the game,
  where \kl{Adam} challenges her with a set~$X$
  and she responds with a set~$Y$.

  Formally,
  we \kl{define} $\NONHAMILTONIAN$ on~$\GRAPH$
  with the $\PiLFOL{4}$\nobreakdash-\kl{formula}
  \begin{equation*}
    \ForAllRel{H} \:
    \ExistsRel{C, S, P} \:
    \ForAllRel{X} \:
    \ExistsRel{Y} \:
    \ForAllNode{\FOVar[1]} \,
    \bigl(
      \InAgreementOn{C}\Of{\FOVar[1]}
      \, \AND \,
      \Fml{InvalidCase}\Of{\FOVar[1]}
      \, \AND \,
      \Fml{DisjointCase}\Of{\FOVar[1]}
    \bigr),
  \end{equation*}
  where
  \begin{equation*}
    \InAgreementOn{C}\Of{\FOVar[1]}
    \, = \,
    \ForAllNbNode{\FOVar[2]}{\FOVar[1]} \,
    \bigl(
      \InRel{C}{\FOVar[1]} \IFF \InRel{C}{\FOVar[2]}
    \bigr)
  \end{equation*}
  ensures that all \kl{nodes} agree on the type of mistake \kl{Adam} has made,
  $\Fml{InvalidCase}\Of{\FOVar[1]}$
  covers the case where \kl{Adam}
  has violated the $2$-\kl{regularity} condition,
  and
  $\Fml{DisjointCase}\Of{\FOVar[1]}$
  covers the case where his \kl{subgraph} consists of multiple \kl{components}.

  The \kl{subformula} for the first case can be written as
  \begin{equation*}
    \Fml{InvalidCase}\Of{\FOVar[1]}
    \, = \,
    \NOT \InRel{C}{\FOVar[1]}
    \, \IMP \,
    \PointsTo{\NOT \Fml{DegreeTwo}}\Of{\FOVar[1]},
  \end{equation*}
  where
  $\PointsTo{\NOT \Fml{DegreeTwo}}\Of{\FOVar[1]}$
  is an instantiation of
  the \kl{formula} schema from Example~\ref{ex:not-all-selected}
  with the negation of the $\BFL$-\kl{formula}
  \begin{equation*}
    \Fml{DegreeTwo}\Of{\FOVar[1]}
    \, = \,
    \ExistsNbNode{\FOVar[2]_1, \FOVar[2]_2}{\FOVar[1]}
    \begin{pmatrix*}[l]
      \,
      \FOVar[2]_1 \NotEqual \FOVar[2]_2
      \: \AND \:
      \BigAND_{i \in \Set{1,2}}
      \bigl(
        \InRel{H}{\FOVar[1], \FOVar[2]_i}
        \AND
        \InRel{H}{\FOVar[2]_i, \FOVar[1]}
      \bigr)
      \: \AND {} \\[0.5ex]
      \,
      \ForAllNbNode{\FOVar[3]}{\FOVar[1]} \,
      \bigl(
        \InRel{H}{\FOVar[1], \FOVar[3]} \OR \InRel{H}{\FOVar[3], \FOVar[1]}
        \, \IMP \,
        \BigOR_{i \in \Set{1,2}} (\FOVar[3] \Equal \FOVar[2]_i)
      \bigr)
    \end{pmatrix*}\!.
  \end{equation*}
  This \kl{formula} states that $\FOVar[1]$~has exactly two \kl{neighbors}
  in the \kl{spanning subgraph} represented by~$H$,
  which must be a symmetric relation.

  For the second case,
  we can write
  \begin{equation*}
    \Fml{DisjointCase}\Of{\FOVar[1]}
    \, = \,
    \InRel{C}{\FOVar[1]}
    \, \IMP \,
    \NOT \Fml{CutAt}\Of{\FOVar[1]}
    \, \AND \,
    \PointsTo{\Fml{SeparationAt}}\Of{\FOVar[1]},
  \end{equation*}
  where
  the first conjunct ensures that
  \kl{Eve}'s partition does not cut any \kl{component},
  while the second conjunct ensures that her partition is nontrivial.
  This makes use of the \kl{subformulas}
  \begin{equation*}
    \Fml{CutAt}\Of{\FOVar[1]}
    \, = \,
    \ExistsNbNode{\FOVar[2]}{\FOVar[1]} \,
    \bigl(
      \InRel{H}{\FOVar[1], \FOVar[2]}
      \, \AND \,
      (\InRel{S}{\FOVar[1]} \IFF \NOT \InRel{S}{\FOVar[2]})
    \bigr),
  \end{equation*}
  which states that
  $\FOVar[1]$ and one of its \kl{component} \kl{neighbors}
  are on opposite sides of the partition,
  and
  \begin{equation*}
    \Fml{SeparationAt}\Of{\FOVar[1]}
    \, = \,
    \NOT \InAgreementOn{S}\Of{\FOVar[1]},
  \end{equation*}
  which states that $\FOVar[1]$~sees \kl{nodes} both inside and outside of~$S$.
  \lipicsEnd
\end{example}

%%% Local Variables:
%%% mode: latex
%%% TeX-master: "../lph-paper"
%%% End:

\section{Restrictive arbiters}
\label{sec:restrictive-arbiters}

The notion of \kl{arbiters} defined in Section~\ref{sec:turing-machines}
was kept simple for the sake of presentation,
but it can be cumbersome
when constructing \kl{arbiters} for specific \kl{graph properties}.
In this section,
we provide a more flexible definition
that allows us to impose additional restrictions on
the input \kl{graphs} and \kl{certificates}.
We then prove its equivalence to the original definition.

\subparagraph*{Certificate restrictors.}

\AP
Let $\IdentRadius$ and~$\CertifRadius$ be positive integers,
and $\CertifPolynomial$~be a polynomial function.
A \intro{certificate restrictor}
for $\Tuple{\CertifRadius, \CertifPolynomial}$-\kl{bounded certificates}
under $\IdentRadius$-\kl{locally unique} \kl{identifiers}
is a \kl{locally polynomial machine}~$\Machine$
that satisfies the following property
for every \kl{graph}~$\Graph$,
every $\IdentRadius$\nobreakdash-\kl{locally unique}
\kl{identifier assignment}~$\IdMap$ of~$\Graph$,
every
$\Tuple{\CertifRadius, \CertifPolynomial}$-\kl(certificate-list){bounded}
\kl{certificate-list assignment}~$\CertifListMap$
of $\Tuple{\Graph, \IdMap}$,
and every
$\Tuple{\CertifRadius, \CertifPolynomial}$-\kl(certificate){bounded}
\kl{certificate assignment}~$\CertifMap$
of $\Tuple{\Graph, \IdMap}$:
if some \kl{node} $\Node[1] \inG \Graph$ \kl{rejects} in~%
$\Result{\Machine}{\Graph, \IdMap, \CertifListMap \CertifConcat \CertifMap}$,
then there exists an
$\Tuple{\CertifRadius, \CertifPolynomial}$-\kl(certificate){bounded}
\kl{certificate assignment}~$\CertifMap'$
differing from~$\CertifMap$ only in the \kl{certificate} assigned to~$\Node[1]$
such that
$\Node[1]$~\kl{accepts} in~%
$\Result{\Machine}{\Graph, \IdMap, \CertifListMap \CertifConcat \CertifMap'}$
while the \kl{verdict} of all other \kl{nodes} remains the same as in~%
$\Result{\Machine}{\Graph, \IdMap, \CertifListMap \CertifConcat \CertifMap}$.
We refer to this property as \intro{local repairability}.
Moreover,
we say that $\Machine$~is \intro{trivial} if
$\Result{\Machine}{\Graph, \IdMap, \CertifListMap \CertifConcat \CertifMap}
 \equiv \Accept$
for all choices of
$\Graph$, $\IdMap$, $\CertifListMap$,~$\CertifMap$.

\subparagraph*{Restrictive arbiters.}

\AP
\label{def:restrictive-arbiter}%
Let $\Level$ be a nonnegative integer,
$\IdentRadius$ and~$\CertifRadius$ be positive integers,
$\CertifPolynomial$~be a polynomial function,
$\BaseProperty$~be an $\LP$-\kl(graph){property},
and
$\Machine_1, \dots, \Machine_{\Level}$
be \kl{certificate restrictors}
for $\Tuple{\CertifRadius, \CertifPolynomial}$\nobreakdash-\kl{bounded certificates}
under $\IdentRadius$-\kl{locally unique} \kl{identifiers}.
A \intro{restrictive} $\SigmaLP{\Level}$\nobreakdash-\reintro{arbiter}
for a \kl{graph property}~$\Property$
on~$\BaseProperty$
under
$\IdentRadius$-\kl{locally unique} \kl{identifiers}
and
$\Tuple{\CertifRadius, \CertifPolynomial}$\nobreakdash-\kl{bounded certificates}
restricted by
$\Machine_1, \dots, \Machine_{\Level}$
is a \kl{locally polynomial machine}~$\Machine$
that satisfies the following equivalence
for every \kl{graph}~$\Graph \in \BaseProperty$
and every
$\IdentRadius$-\kl{locally unique}
\kl{identifier assignment}~$\IdMap$ of~$\Graph$:
\begin{equation*}
  \Graph \in \Property
  \; \iff \;
  \exists \CertifMap_1 \,
  \forall \CertifMap_2
  \dots
  \Quantifier \CertifMap_{\Level}:
  \Result{\Machine}{
    \Graph, \,
    \IdMap, \,
    \CertifMap_1 \CertifConcat
    \CertifMap_2 \CertifConcat \dots \CertifConcat
    \CertifMap_{\Level}
  } \equiv
  \Accept,
\end{equation*}
where
$\Quantifier$ is~$\forall$ if $\Level$ is even
and $\exists$ otherwise,
and all quantifiers range over
$\Tuple{\CertifRadius, \CertifPolynomial}$-\kl{bounded certificate assignments}
of $\Tuple{\Graph, \IdMap}$
with the additional restriction that
$\Result{\Machine_i}{
   \Graph, \,
   \IdMap, \,
   \CertifMap_1 \CertifConcat \dots \CertifConcat
   \CertifMap_i
 } \equiv
\Accept$
for all $i \in \Range[1]{\Level}$.
If all \kl{certificate restrictors} are \kl{trivial},
we say that $\Machine$~operates under \intro{unrestricted}
$\Tuple{\CertifRadius, \CertifPolynomial}$\nobreakdash-\kl{bounded certificates}.
We analogously define
\reintro{restrictive} $\PiLP{\Level}$\nobreakdash-\reintro{arbiters}.

\AP
Notice that the notion of
$\SigmaLP{\Level}$- and $\PiLP{\Level}$-\kl{arbiters} for~$\Property$
introduced on page~\pageref{def:locally-polynomial-hierarchy}
coincides with the notion of
\kl{restrictive} $\SigmaLP{\Level}$- and $\PiLP{\Level}$-\kl{arbiters}
for~$\Property$ on~$\GRAPH$
under \kl{unrestricted} \kl{certificates}.
We will refer to
such \kl{arbiters} as \intro{permissive arbiters}
when we want to emphasize the distinction from other \kl{restrictive arbiters}.
Although not every \kl{restrictive arbiter} is \kl{permissive},
we can prove the following lemma,
which allows us to use arbitrary \kl{restrictive arbiters}
whenever it is more convenient.

\begin{lemma}
  \label{lem:restrictive-arbiters}
  Let $\Level \in \Naturals$
  and $\BaseProperty, \Property \subseteq \GRAPH$
  with $\BaseProperty \in \LP$.
  The \kl{graph property} $\Property \cap \BaseProperty$
  belongs to~$\SigmaLP{\Level}\On{\BaseProperty}$
  if and only if
  $\Property$~has
  a \kl{restrictive} $\SigmaLP{\Level}$-\kl{arbiter} on~$\BaseProperty$.
  The analogous statement holds for~$\PiLP{\Level}\On{\BaseProperty}$.
\end{lemma}

\begin{proof}
  We prove only the first statement,
  since the proof for~$\PiLP{\Level}\On{\BaseProperty}$ is completely analogous.
  By definition,
  if $\Property \cap \BaseProperty$ belongs to~$\SigmaLP{\Level}\On{\BaseProperty}$,
  then there exists
  a \kl{permissive} $\SigmaLP{\Level}$\nobreakdash-\kl{arbiter}~$\Machine$
  for a \kl{graph property}~$\Property'$
  such that $\Property' \cap \BaseProperty = \Property \cap \BaseProperty$,
  and thus
  $\Machine$~is also a \kl{restrictive} $\SigmaLP{\Level}$-\kl{arbiter}
  for~$\Property$ on~$\BaseProperty$.

  For the converse,
  we have to convert
  an arbitrary \kl{restrictive arbiter} for~$\Property$ on~$\BaseProperty$
  into
  a \kl{permissive arbiter} for some \kl{graph property}~$\Property'$
  such that $\Property' \cap \BaseProperty = \Property \cap \BaseProperty$.
  We do this for $\Property' = \Property \cap \BaseProperty$,
  proceeding in two steps
  by first removing the restrictions on the input \kl{graphs}
  and then on the \kl{certificates}.

  Let $\Machine^a$ be
  a \kl{restrictive} $\SigmaLP{\Level}$\nobreakdash-\kl{arbiter}
  for~$\Property$ on~$\BaseProperty$
  under $\IdentRadius^a$-\kl{locally unique} \kl{identifiers}
  and $\Tuple{\CertifRadius^a, \CertifPolynomial^a}$-\kl{bounded certificates}
  restricted by
  $\Machine_1^a, \dots, \Machine_{\Level}^a$.
  \begin{enumerate}
  \item We start by converting~$\Machine^a$ into
    a \kl{restrictive} $\SigmaLP{\Level}$\nobreakdash-\kl{arbiter}~$\Machine^b$
    for $\Property \cap \BaseProperty$ on arbitrary \kl{graphs}.
    Since $\BaseProperty$ is in~$\LP$,
    there exists an $\LP$-\kl{decider}~$\Machine^{\BaseProperty}$
    for that \kl(graph){property}
    and an integer $\IdentRadius^{\BaseProperty} \in \Positives$
    such that
    $\Machine^{\BaseProperty}$~operates under
    $\IdentRadius^{\BaseProperty}$\nobreakdash-\kl{locally unique} \kl{identifiers}.
    When \kl{executing}~$\Machine^b$
    on a \kl{graph}~$\Graph$
    under an \kl{identifier assignment}~$\IdMap$
    and a \kl{certificate-list assignment}~$\CertifListMap$,
    the \kl{nodes} first simulate~$\Machine^{\BaseProperty}$
    to check whether $\Graph$ belongs to~$\BaseProperty$.
    Any \kl{node} that \kl{rejects} in the simulation
    also immediately \kl{rejects} in
    $\Result{\Machine^b}{\Graph, \IdMap, \CertifListMap}$,
    so $\Graph$~can only be \kl{accepted} if it belongs to~$\BaseProperty$.
    Then,
    the \kl{nodes} that have not \kl{rejected} simulate~$\Machine^a$
    and return the \kl{verdict} reached in that second simulation
    (unless they learn about some \kl{node} that has previously \kl{rejected},
    in which case they also \kl{reject}).
    The \kl{machine}~$\Machine^b$ obtained this way
    operates
    on arbitrary \kl{graphs}
    under $\IdentRadius^b$-\kl{locally unique} \kl{identifiers}
    and $\Tuple{\CertifRadius^b, \CertifPolynomial^b}$-\kl{bounded certificates}
    restricted by
    $\Machine_1^b, \dots, \Machine_{\Level}^b$,
    where
    $\IdentRadius^b = \max \Set{\IdentRadius^a,\, \IdentRadius^{\BaseProperty}}$,\:
    $\CertifRadius^b = \CertifRadius^a$,\:
    $\CertifPolynomial^b = \CertifPolynomial^a$,\,
    and
    $\Machine_i^b = \Machine_i^a$
    for $i \in \Range[1]{\Level}$.
  \item
    Now we convert~$\Machine^b$ into
    a \kl{permissive} $\SigmaLP{\Level}$\nobreakdash-\kl{arbiter}~$\Machine^c$
    for $\Property \cap \BaseProperty$.
    By definition,
    for every \kl{graph}~$\Graph$
    and every $\IdentRadius^b$-\kl{locally unique}
    \kl{identifier assignment}~$\IdMap$ of~$\Graph$,
    \begin{equation*}
      \Graph \in \Property \cap \BaseProperty
      \; \iff \;
      \exists \CertifMap_1 \,
      \forall \CertifMap_2
      \dots
      \Quantifier \CertifMap_{\Level}:
      \Result{\Machine^b}{
        \Graph, \,
        \IdMap, \,
        \CertifMap_1 \CertifConcat
        \CertifMap_2 \CertifConcat \dots \CertifConcat
        \CertifMap_{\Level}
      } \equiv
      \Accept,
    \end{equation*}
    where all quantifiers range over
    $\Tuple{\CertifRadius^b, \CertifPolynomial^b}$-\kl{bounded certificate assignments}
    of $\Tuple{\Graph, \IdMap}$
    with the additional restriction that
    $\Result{\Machine_i^b}{
       \Graph, \,
       \IdMap, \,
       \CertifMap_1 \CertifConcat \dots \CertifConcat
       \CertifMap_i
     } \equiv
    \Accept$
    for all $i \in \Range[1]{\Level}$.
    The new \kl{machine}~$\Machine^c$
    has to satisfy the analogous equivalence
    without the additional restriction
    on the \kl{certificate assignments}.

    When \kl{executing}~$\Machine^c$
    on~$\Graph$ under~$\IdMap$ and
    $\CertifMap_1 \CertifConcat \dots \CertifConcat
     \CertifMap_{\Level}$,
    the \kl{nodes} first simulate
    $\Machine_1^b, \dots, \Machine_{\Level}^b$
    to check if the given \kl{certificates} satisfy
    the imposed restrictions.
    As a result of this first phase,
    each \kl{node}~$\Node[1]$ stores a flag~$\Tag{ok}_i$
    for each $i \in \Range[1]{\Level}$
    to indicate whether $\Node[1]$~\kl{accepts} in
    $\Result{\Machine_i^b}{
       \Graph, \,
       \IdMap, \,
       \CertifMap_1 \CertifConcat \dots \CertifConcat
       \CertifMap_i
     }$.
    Then,
    the \kl{nodes} simulate~$\Machine^b$
    while simultaneously updating their flags to propagate errors.
    That is,
    if a \kl{node} sees
    that the $\Tag{ok}_i$~flag of one of its \kl{neighbors} is false,
    then it also sets its own $\Tag{ok}_i$~flag to false.
    Once the simulation of~$\Machine^b$ has terminated,
    each \kl{node}~$\Node[1]$ goes \emph{sequentially} through its flags
    $\Tag{ok}_1, \dots, \Tag{ok}_{\Level}$
    to verify that they are all true.
    If it encounters an $\Tag{ok}_i$~flag that is false,
    $\Node[1]$~aborts, writes a \kl{verdict} on its \kl{internal tape},
    and enters \kl{state}~$\StopState$.
    The \kl{verdict} depends on whether
    the \kl{certificate assignment}~$\CertifMap_i$
    is quantified existentially or universally:
    in the first case,
    the \kl{verdict} is~$0$ (\kl{reject}),
    whereas in the second
    it is~$1$ (\kl{accept}).
    Finally,
    if it did not stop before,
    $\Node[1]$~returns the same \kl{verdict} it would have returned
    when \kl{executing}~$\Machine^b$.

    Note that
    the sequential verification and early termination described above
    ensures that quantifications are relativized
    in the same way as for~$\Machine^b$.
    More precisely,
    if the first \kl{certificate assignment}
    violating the restrictions
    is quantified existentially,
    then the input \kl{graph} is \kl{rejected}
    because
    all \kl{nodes} that know about the violation
    return~$0$.
    If instead the first \kl{certificate assignment}~$\CertifMap_i$
    violating the restrictions of the corresponding \kl{machine}~$\Machine_i^b$
    is quantified universally,
    then there are two possibilities:
    either the input \kl{graph} is \kl{accepted}
    (the desired outcome),
    or it is \kl{rejected}
    because of some \kl{node}~$\Node[1]$
    that does not know about the violation.
    However,
    in the latter case,
    there exists an
    $\Tuple{\CertifRadius^b, \CertifPolynomial^b}$%
    -\kl{bounded certificate assignment}~$\CertifMap'_i$
    that does not violate the restrictions of~$\Machine_i^b$
    but for which $\Node[1]$~still \kl{rejects}.
    This is because
    $\Machine_i^b$~satisfies \kl{local repairability},
    so all defects in~$\CertifMap_i$ can be fixed
    without affecting~$\Node[1]$.
    Intuitively speaking,
    $\Node[1]$'s \kl{verdict} is legitimate
    since it is independent of the violation.
    \qedhere
  \end{enumerate}
\end{proof}

%%% Local Variables:
%%% mode: latex
%%% TeX-master: "../lph-paper"
%%% End:
\section{A generalization of Fagin's theorem}
\label{sec:fagin}

The founding result of descriptive complexity theory is Fagin's theorem,
which provides a logical, and thus machine-independent,
characterization of the complexity class~$\NP$
(see, e.g., \cite[Thm.~3.2.4]{DBLP:series/txtcs/GradelKLMSVVW07}).
In the context of this paper,
we can state it as follows.%
\footnote{%
  In the literature,
  Fagin's theorem is usually stated
  in terms of arbitrary \kl{graphs} (or even arbitrary \kl{structures})
  instead of \kl[single-node graphs]{labeled single-node graphs}.
  More specifically,
  a \kl{graph property} can be \kl{verified} by a (centralized) Turing machine
  operating in polynomial time on \kl{encodings} of \kl{graphs}
  if and only if
  it can be \kl{defined} by a \kl{formula} of
  the \kl{existential fragment} of \kl{second-order logic}.
  In symbols,
  $\NP\On{\Encoding(\GRAPH)} =
  \SetBuilder{
    \Encoding(\Property)
  }{
    \Property \in \SigmaFO{1}\On{\GRAPH}
  }$,
  where $\Encoding \colon \GRAPH \to \NODE$
  is some \kl{encoding} of \kl{graphs} as binary strings.
  However,
  the statement presented here is equivalent,
  since it is immaterial
  whether we \kl{encode} \kl{graphs} as strings or vice versa
  (see, e.g., Problem~8.4.12
  in Papadimitriou's book \cite{DBLP:books/daglib/0072413}).
}

\begin{theorem}[Fagin~\cite{Fagin74}]
  \label{thm:fagin}
  On \kl{single-node graphs},
  a \kl(graph){property} can be \kl{verified}
  by a \kl{locally polynomial machine}
  if and only if
  it can be \kl{defined} by a \kl{formula} of
  the \kl{existential fragment} of \kl{local second-order logic}.
  In symbols,
  $\NLP\On{\NODE} = \SigmaLFO{1}\On{\NODE}$,
  or equivalently,
  $\NP = \SigmaFO{1}\On{\NODE}$.
\end{theorem}

The inclusion from right to left is straightforward
because any $\SigmaFOL{1}$-\kl{formula}
$\ExistsRel{\SOVar_1} \dots \ExistsRel{\SOVar_n} \, \Formula$
can be evaluated in polynomial time
by a Turing machine
that is given some \kl{interpretation} of
$\SOVar_1, \dots, \SOVar_n$
(\kl{encoded} in the \kl{certificate} chosen by \kl(certificate){Eve}).
The machine can use brute force to check whether
the \kl{first-order} \kl{formula}~$\Formula$ is \kl{satisfied}
under the given \kl{interpretation} of
$\SOVar_1, \dots, \SOVar_n$
by simply iterating over all possible \kl{interpretations}
of the \kl{first-order variables} in~$\Formula$.
The reverse inclusion, however, is more intricate,
as it involves encoding
the space-time diagram of a Turing machine
using a collection of relations over the input \kl{structure}.
The key insight that makes this possible is the following:
since the machine's running time is polynomially bounded
in the \kl{structure}'s \kl(structure){cardinality},
each cell of the space-time diagram
can be represented by a tuple of \kl{elements}
whose length depends on the degree of the bounding~polynomial.

In this section,
we generalize Theorem~\ref{thm:fagin}
from \kl{single-node graphs} to arbitrary \kl{graphs},
thereby providing a logical characterization of the complexity class~$\NLP$.
We then further generalize this result to obtain similar characterizations
of the higher levels of the \kl{locally polynomial hierarchy}.

\AP
Our proofs make use of the following lemma,
which basically states that
in the \kl{execution} of a \kl{locally polynomial machine},
the space-time diagram of each \kl{node} in each \kl{round}
is polynomially bounded in the \kl(structure){cardinality}
of a constant-radius \kl{neighborhood} of the \kl{node}.
For a given \kl{execution} of a \kl{machine}~$\Machine$,
the \intro{space usage}
of \kl{node}~$\Node$ in \kl{round}~$\Round$
is the maximum number of \kl{tape} cells
that~$\Node$ occupies in \kl{round}~$\Round$.
More precisely,
if we denote by~$t$
the \kl{step running time} of~$\Node$ in \kl{round}~$\Round$
and by~$\ell_j$
the total length of~$\Node$'s \kl{tape contents}
after its $j$-th \kl{computation step},
then $\Node$'s \kl{space usage} in \kl{round}~$\Round$ is
$\max \SetBuilder{\ell_j}{0 \leq j \leq t}$.

\begin{lemma}
  \label{lem:polynomial-space-time}
  Let
  $\Level$ and~$\Radius$ be positive integers,
  $\Polynomial$ be a polynomial function,
  and $\Machine$ be a \kl{locally polynomial machine}
  running in \kl{round time}~$\Radius$
  and \kl{step time}~$\Polynomial$.
  There exists a polynomial function~$f$
  such that the following holds for
  every \kl{labeled graph}~$\Graph$,
  every \kl{small} $\Radius$-\kl{locally unique}
  \kl{identifier assignment}~$\IdMap$ of~$\Graph$,
  and all
  $\Tuple{\Radius, \Polynomial}$-\kl{bounded certificate assignments}
  $\CertifMap_1, \dots, \CertifMap_{\Level}$
  of $\Tuple{\Graph, \IdMap}$:
  in the \kl{execution} of~$\Machine$
  on~$\Graph$ under~$\IdMap$ and
  $\CertifMap_1 \CertifConcat \dots \CertifConcat \CertifMap_{\Level}$,
  the \kl{step running time} and \kl{space usage}
  of each \kl{node} $\Node \inG \Graph$
  in each \kl{round} $\Round \in \Range[1]{\Radius}$
  are at most
  $f \bigl(
    \CardS{\StructNeighborhood{\Graph}{4\Radius}{\Node}}
  \bigr)$,
  i.e., $f$~applied to
  the number of \kl{nodes} and \kl{labeling bits}
  of~$\Node$'s $4\Radius$-\kl{neighborhood} in~$\Graph$.
\end{lemma}

\begin{proof}
  By definition,
  if $\IdMap$~is a \kl{small} $\Radius$-\kl{locally unique}
  \kl{identifier assignment}~$\IdMap$ of~$\Graph$,
  then
  $\Length{\IdMap(\Node[1])} \leq
   \Ceiling{\log_2\CardG{\Neighborhood{\Graph}{2\Radius}{\Node[1]}}} \leq
   \Ceiling{\log_2\CardS{\StructNeighborhood{\Graph}{2\Radius}{\Node[1]}}}$
  for every \kl{node} $\Node \inG \Graph$.
  Consider arbitrary
  $\Tuple{\Radius, \Polynomial}$-\kl{bounded certificate assignments}
  $\CertifMap_1, \dots, \CertifMap_{\Level}$
  of $\Tuple{\Graph, \IdMap}$,
  let
  $\CertifListMap =
   \CertifMap_1 \CertifConcat \dots \CertifConcat \CertifMap_{\Level}$,
  and let us denote by~$b_{\Round}(\Node[1])$ the maximum
  of~$\Node[1]$'s \kl{step running time} and \kl{space usage}
  in \kl{round} $\Round \in \Range[1]{\Radius}$
  of $\Machine$'s~\kl{execution}
  on~$\Graph$ under~$\IdMap$ and~$\CertifListMap$.
  Furthermore,
  for $k \in \Range[3\Radius]{4\Radius - 1}$,
  let us write $n_k(\Node[1])$ as a shorthand for
  $\CardS{\StructNeighborhood{\Graph}{k}{\Node[1]}}$.
  We now show by induction
  that for each \kl{round} $\Round \in \Range[1]{\Radius}$,
  there is a polynomial~$f_{\Round}$
  (independent of~$\Graph$, $\IdMap$, and~$\CertifListMap$)
  such that
  $b_{\Round}(\Node[1]) \leq f_{\Round}(n_{3\Radius + \Round - 1}(\Node[1]))$.

  Assuming \kl{node}~$\Node[1]$ has~$\Degree$ \kl{neighbors}
  $\Node[2]_1, \dots, \Node[2]_{\Degree}$,
  its initial \kl{tape contents} in \kl{round}~$1$ consist of
  the string~$\Separator^{\Degree}$ on its \kl{receiving tape}
  and the string
  $\Labeling{\Graph}(\Node[1]) \, \Separator \,
   \IdMap(\Node[1])   \, \Separator \,
   \CertifListMap(\Node[1])$
  on its \kl{internal tape}.
  Since $\Machine$~runs in \kl{step time}~$\Polynomial$,
  \kl{node}~$\Node[1]$'s \kl{step running time} in \kl{round}~$1$
  cannot exceed
  $\Polynomial \bigl(
    \Degree +
    \Length{
      \Labeling{\Graph}(\Node[1]) \, \Separator \,
      \IdMap(\Node[1])   \, \Separator \,
      \CertifListMap(\Node[1])}
  \bigr)$.
  This also implies that
  $\Node[1]$'s \kl{space usage} in \kl{round}~$1$,
  and thus~$b_1(\Node[1])$,
  cannot exceed
  $3 \cdot \Polynomial \bigl(
    \Degree +
    \Length{
      \Labeling{\Graph}(\Node[1]) \, \Separator \,
      \IdMap(\Node[1])   \, \Separator \,
      \CertifListMap(\Node[1])}
  \bigr)$,
  since all three \kl{tape} heads
  start on the leftmost cell of their respective \kl{tape}
  and can advance by at most one cell in each \kl{computation step}.
  Moreover,
  we know that
  $\Degree \leq n_1(\Node[1])$,\;
  $\Length{\Labeling{\Graph}(\Node[1])} \leq n_0(\Node[1])$,\;
  $\Length{\IdMap(\Node[1])} \leq \Ceiling{\log_2 n_{2\Radius}(\Node[1])}$,\,
  and
  (since each $\CertifMap_i$ is
  $\Tuple{\Radius, \Polynomial}$-\kl(certificate){bounded}),
  \begin{equation*}
    \Length{\CertifListMap(\Node[1])}
    \, \leq \,
    \Level +
    \Level \cdot
    \Polynomial
    \Bigl( \,
      \smashoperator{
        \sum_{
          \hspace{4ex}
          \Node[2] \inG \Neighborhood{\Graph}{\Radius}{\Node[1]}
        }
      }
      \Length{
        \Labeling{\Graph}(\Node[2])
        \, \Separator \,
        \IdMap(\Node[2])
      } \,
    \Bigr)
    \, \leq \,
    \Polynomial'(n_{3\Radius}(\Node[1])),
  \end{equation*}
  where
  $\Polynomial'$ is a polynomial
  that depends only on $\Level$ and~$\Polynomial$.
  The last inequality stems from the fact that
  the $2\Radius$-\kl{neighborhood} of any \kl{node}
  $\Node[2] \inG \Neighborhood{\Graph}{\Radius}{\Node[1]}$
  is included in $\Node[1]$'s $3\Radius$-\kl{neighborhood}.
  We can therefore conclude that
  $b_1(\Node[1]) \leq f_1(n_{3\Radius}(\Node[1]))$
  for some polynomial~$f_1$ that can be easily derived from~$\Polynomial'$.

  Now,
  let us assume by induction
  that there exists a polynomial~$f_{\Round}$
  such that
  $b_{\Round}(\Node[1]) \leq f_{\Round}(n_{3\Radius + \Round - 1}(\Node[1]))$
  for every \kl{node}~$\Node[1] \inG \Graph$.
  At the beginning of \kl{round}~${\Round + 1}$,
  $\Node[1]$'s \kl{internal tape} contains a string~$\String[1]$
  of length less than~$b_{\Round}(\Node[1])$
  and its \kl{receiving tape} contains a string of the form
  $\Message_1 \, \Separator \dots \Separator \, \Message_{\Degree} \, \Separator$,
  where~$\Message_j$ is a message
  of length less than~$b_{\Round}(\Node[2]_j)$
  that was sent by \kl{neighbor}~$\Node[2]_j$ in \kl{round}~$\Round$.
  Again,
  since we know that~$\Machine$ runs in \kl{step time}~$\Polynomial$,
  this gives us an upper bound on~$\Node[1]$'s
  \kl{step running time} and \kl{space usage}:
  $b_{\Round + 1}(\Node[1]) \leq
  3 \cdot \Polynomial \bigl(
    \Length{\String[1]} +
    \Length{
      \Message_1 \, \Separator \dots \Separator \, \Message_{\Degree} \, \Separator
    }
  \bigr)$.
  From the induction hypothesis we obtain that
  $\Length{\String[1]} +
  \Length{
    \Message_1 \, \Separator \dots \Separator \, \Message_{\Degree} \, \Separator
  }
  \leq \Degree +
   \sum_{\Node[3] \in \Set{\Node[1], \Node[2]_1, \dots \Node[2]_{\Degree}}}
     f_{\Round}(n_{3\Radius + \Round - 1}(\Node[3]))$,
  which cannot exceed
  $n_1(\Node[1]) \cdot
   \bigl( f_{\Round}(n_{3\Radius + \Round}(\Node[1])) + 1 \bigr)$,
  given that the $(3\Radius + \Round - 1)$-\kl{neighborhood}
  of every \kl{neighbor}~$\Node[2]_j$ of~$\Node[1]$
  is included in $\Node[1]$'s $(3\Radius + \Round)$-\kl{neighborhood}.
  We thus obtain the bound
  $b_{\Round+1}(\Node[1]) \leq f_{\Round+1}(n_{3\Radius + \Round}(\Node[1]))$,
  where $f_{\Round+1}$ is a polynomial
  that can be easily derived from~$\Polynomial$ and~$f_{\Round}$.

  Ultimately,
  we have
  $b_{\Round}(\Node[1]) \leq f_{\Radius}(n_{4\Radius - 1}(\Node[1]))$
  for all $\Round \in \Range[1]{\Radius}$,
  which implies our claim.
\end{proof}

We now generalize Fagin's theorem from~$\NP$ to $\NLP$,
resulting in the following statement.
Notice that the original result (Theorem~\ref{thm:fagin}) can be recovered
by restricting both sides of the equivalence to \kl{single-node graphs}.

\begin{theorem}
  \label{thm:local-fagin}
  On arbitrary \kl{graphs},
  a \kl(graph){property} can be \kl{verified}
  by a \kl{locally polynomial machine}
  if and only if
  it can be \kl{defined} by a \kl{formula} of
  the \kl{existential fragment} of \kl{local second-order logic}.
  In symbols,
  $\NLP = \SigmaLFO{1}\On{\GRAPH}$.
\end{theorem}

As this result will be further generalized below,
we do not explicitly prove the backward direction,
which is a simple special case of the backward direction
of Theorem~\ref{thm:local-hierarchy-equivalence}
on page~\pageref{prf:local-hierarchy-equivalence-backward}.
However,
we do explicitly prove the forward direction
to provide a more accessible introduction to the general case
presented on page~\pageref{prf:local-hierarchy-equivalence-forward}.
The key idea of encoding a space-time diagram as a collection of relations
remains the same as in Fagin's original proof,
but we have to deal with additional issues such as
\kl{locally unique} \kl{identifiers} and
the exchange of messages between \kl{adjacent} \kl{nodes}.

\begin{proof}[Proof of Theorem~\ref{thm:local-fagin} -- Forward direction]
  Let $\Property$ be a \kl{graph property} in~$\NLP$,
  and let
  $\Machine = \Tuple{\StateSet, \TransFunc}$
  be an $\NLP$-\kl{verifier} for~$\Property$
  that operates under $\IdentRadius$-\kl{locally unique} \kl{identifiers}
  and $\Tuple{\CertifRadius_1, \CertifPolynomial_1}$\nobreakdash-\kl{bounded certificates},
  and runs in \kl{round time}~$\Radius_2$ and \kl{step time}~$\Polynomial_2$.
  Moreover,
  let $\Radius = \max\Set{\IdentRadius, \CertifRadius_1, \Radius_2}$,
  and let $\Polynomial$ be a polynomial
  that bounds both $\CertifPolynomial_1$ and~$\Polynomial_2$.
  By the proof of Lemma~\ref{lem:restrictive-arbiters},
  we may assume without loss of generality
  that $\Machine$ \kl{rejects}
  under any \kl{certificate assignment} violating the
  $\Tuple{\CertifRadius_1, \CertifPolynomial_1}$\nobreakdash-\kl(certificate){boundedness}
  condition,
  so it does not matter if
  \kl(certificate){Eve} chooses \kl{certificates} that are too large.
  Now,
  we fix $k \in \Naturals$ such that
  the polynomial~$f$ described in Lemma~\ref{lem:polynomial-space-time}
  (for $\Level = 1$,
  and our choices of $\Radius$, $\Polynomial$, and $\Machine$)
  satisfies $f(n) < n^k$ for $n > 1$.
  By~Lemma~\ref{lem:polynomial-space-time},
  for every \kl{graph}~$\Graph$
  whose \kl{structural representation} has at least two \kl{elements},%
  \footnote{\label{ftn:local-fagin-forward}%
    For the \kl{single-node graph}~$\Graph$
    whose \kl{node}~$\Node$ is \kl{labeled} with the empty string,
    we have
    $(\CardS{\StructNeighborhood{\Graph}{4\Radius}{\Node}})^k = 1$
    for all $k \in \Naturals$,
    since the \kl{structural representation}~$\StructRepr{\Graph}$
    has only one \kl{element}.
    If $\Graph$ belongs to~$\Property$,
    we can easily treat it as a special case
    in the \kl{formula}~$\Formula[1]_{\Machine}$ described~here.}
  every \kl{small} $\Radius$-\kl{locally unique}
  \kl{identifier assignment}~$\IdMap$ of~$\Graph$,
  and every
  $\Tuple{\Radius, \Polynomial}$-\kl{bounded certificate assignment}~$\CertifMap$
  of $\Tuple{\Graph, \IdMap}$,
  the \kl{step running time} and \kl{space usage}
  of each \kl{node} $\Node \inG \Graph$
  are bounded by~%
  $\bigl(\CardS{\StructNeighborhood{\Graph}{4\Radius}{\Node}}\bigr)^k$
  in each \kl{round} $\Round \in \Range[1]{\Radius}$
  of the corresponding \kl{execution} of~$\Machine$.
  Intuitively,
  this gives us a bound on the amount of information required to describe
  \kl(certificate){Eve}'s choice of \kl{certificates}
  and the subsequent \kl{execution} of the \kl{verifier}~$\Machine$,
  assuming that the \kl{nodes} of the input \kl{graph}
  are assigned \kl{small} \kl{identifiers}.
  Although our description below does not explicitly state this assumption,
  it may run out of “space” if the provided \kl{identifiers} are too large.
  However,
  this is not a problem because incomplete \kl{executions} are simply ignored.

  We convert~$\Machine$ into a $\SigmaLFOL{1}$-\kl{sentence}
  \kl{defining}~$\Property$
  that is of the form
  $\Formula[1]_{\Machine} =
  \ExistsRel{\Vector{\SOVar}} \:
  \ForAllNode{\FOVar[1]} \:
  \Formula[2]\Of{\FOVar[1]}$,
  where
  $\Vector{\SOVar} =
   \Tuple{
     \SOVarTotOrd,
     \SOVarCorrInter,
     \SOVarCorrIntra,
     \SOVarIdent{0},
     \SOVarIdent{1},
     \dots
   }$
  is a collection of \kl{second-order variables}
  intended to represent an \kl{accepting} \kl{execution} of~$\Machine$,
  and~$\Formula[2]\Of{\FOVar[1]}$ is a $\BFL$-\kl{formula}
  stating that,
  from the point of view of \kl{node}~$\FOVar[1]$,
  the \kl{variables} in~$\Vector{\SOVar}$ do indeed represent
  a valid \kl{execution} of~$\Machine$
  in which~$\FOVar[1]$ \kl{accepts}.
  We start by introducing the \kl{relation variables} in~$\Vector{\SOVar}$,
  together with their intended \kl{interpretations}.

  \begin{itemize}
  \item\AP $\intro*\SOVarTotOrd$:
    a $(2k + 1)$-\kl{ary} relation
    that associates with each \kl{node}~$\FOVar[1]$
    a linear order on the $k$-tuples of \kl{elements}
    in $\FOVar[1]$'s $4\Radius$-\kl{neighborhood}.
    The intended meaning of
    $\InRel{\SOVarTotOrd}{\FOVar[1], \FOVarPos_1, \FOVarPos_2}$
    is:
    “from $\FOVar[1]$'s point of view,
    tuple~$\FOVarPos_1$ is strictly smaller than tuple~$\FOVarPos_2$”.
  \item\AP $\intro*\SOVarCorrInter$:
    a $(4k + 2)$-\kl{ary} relation
    that establishes a correspondence
    between the linear orders
    of two \kl{nodes}~$\FOVar[1]_1$ and~$\FOVar[1]_2$
    that lie at \kl{distance} at most~$2\Radius$ of each other.
    The intended meaning of
    $\InRel{\SOVarCorrInter}{
      \FOVar[1]_1, \FOVarPos_1, \FOVarPos_1',
      \FOVar[1]_2, \FOVarPos_2, \FOVarPos_2'}$
    is:
    “The number of steps from~$\FOVarPos_1$ to~$\FOVarPos_1'$
    in the linear order of~$\FOVar[1]_1$
    is the same as
    the number of steps from~$\FOVarPos_2$ to~$\FOVarPos_2'$
    in the linear order of~$\FOVar[1]_2$”.
    Or, to put it more loosely,
    “$\FOVarPos_1' - \FOVarPos_1$ for~$\FOVar[1]_1$
    is the same as
    $\FOVarPos_2' - \FOVarPos_2$ for~$\FOVar[1]_2$”.
  \item\AP $\intro*\SOVarCorrIntra$:
    a $(k + 2)$-\kl{ary} relation
    that establishes a correspondence
    between the linear order of \kl{node}~$\FOVar[1]$
    (defined by~$\SOVarTotOrd$)
    and the order of the bits in $\FOVar[1]$'s \kl{label}
    (represented by \kl{elements} of the input \kl{structure}).
    The intended meaning of
    $\InRel{\SOVarCorrIntra}{\FOVar[1], \FOVar[3], \FOVarPos}$
    is:
    “the position of bit~$\FOVar[3]$ in~$\FOVar[1]$'s \kl{label}
    corresponds to the position of~$\FOVarPos$ in~$\FOVar[1]$'s linear order”.
  \item\AP $\intro*\SOVarIdent{\Symbol}$:
    a family of $(k + 1)$-\kl{ary} relations,
    for $\Symbol \in \Set{0, 1}$,
    whose purpose is to represent
    the $\Radius$-\kl{locally unique} \kl{identifier}
    of each \kl{node}~$\FOVar[1]$.
    The intended meaning of
    $\InRel{\SOVarIdent{\Symbol}}{\FOVar[1], \FOVarPos}$
    is:
    “the~$\FOVarPos$-th bit of $\FOVar[1]$'s \kl{identifier}
    is an~$\Symbol$”.
  \item\AP $\intro*\SOVarIdOrd$:
    a binary order relation
    that compares, with respect to their \kl{identifiers},
    two \kl{nodes}~$\FOVar[1]_1$ and~$\FOVar[1]_2$
    that have some common \kl{neighbor}.
    The intended meaning of
    $\InRel{\SOVarIdOrd}{\FOVar[1]_1, \FOVar[1]_2}$
    is:
    “$\FOVar[1]_1$'s \kl{identifier} is smaller than
    $\FOVar[1]_2$'s \kl{identifier}”.
  \item\AP $\intro*\SOVarCertif{\Symbol}$:
    a family of $(k + 1)$-\kl{ary} relations,
    for $\Symbol \in \Set{0, 1}$,
    whose purpose is to represent
    the $\Tuple{\Radius, \Polynomial}$-\kl{bounded certificate}
    of each \kl{node}~$\FOVar[1]$.
    The intended meaning of
    $\InRel{\SOVarCertif{\Symbol}}{\FOVar[1], \FOVarPos}$
    is:
    “the~$\FOVarPos$-th bit of $\FOVar[1]$'s \kl{certificate}
    is an~$\Symbol$”.
  \item\AP $\intro*\SOVarState{\Round}{\State}$:
    a family of $(k + 1)$-\kl{ary} relations,
    for $\Round \in \Range[1]{\Radius}$ and $\State \in \StateSet$,
    that indicate the \kl{state} of each \kl{node}
    in every \kl{communication round} and \kl{computation step}.
    The intended meaning of~%
    $\InRel{\SOVarState{\Round}{\State}}{\FOVar[1], \FOVarTime}$
    is:
    “in \kl{round}~$\Round$, at \kl{step}~$\FOVarTime$,
    \kl{node}~$\FOVar[1]$ is in \kl{state}~$\State$”.
  \item\AP $\intro*\SOVarHead{\Round}{\TapeIdx}$:
    a family of $(2k + 1)$-\kl{ary} relations,
    for $\Round \in \Range[1]{\Radius}$
    and $\TapeIdx \in \Set{\Receiving, \Internal, \Sending}$,
    that indicate the positions of the three tape heads
    of each \kl{node}
    in every \kl{communication round} and \kl{computation step}.
    The intended meaning of
    $\InRel{\SOVarHead{\Round}{\TapeIdx}}{\FOVar[1], \FOVarTime, \FOVarPos}$
    is:
    “in \kl{round}~$\Round$, at \kl{step}~$\FOVarTime$,
    \kl{node}~$\FOVar[1]$'s head on \kl{tape}~$\TapeIdx$ is at position~$\FOVarPos$”.
    Here,
    we adopt the convention that $\Receiving$, $\Internal$, and $\Sending$
    refer to the
    \kl[receiving tape]{receiving}, \kl[internal tape]{internal}, and \kl{sending tapes},
    respectively.
  \item\AP $\intro*\SOVarTape{\Round}{\TapeIdx}{\Symbol}$:
    a family of $(2k + 1)$-\kl{ary} relations,
    for $\Round \in \Range[1]{\Radius}$,\,
    $\TapeIdx \in \Set{\Receiving, \Internal, \Sending}$, and
    $\Symbol \in \Set{\Endmarker, \Blank, \Separator, 0, 1}$,
    that indicate the \kl{tape} contents
    of each \kl{node}
    in every \kl{communication round} and \kl{computation step}.
    The intended meaning of
    $\InRel{\SOVarTape{\Round}{\TapeIdx}{\Symbol}}{\FOVar[1], \FOVarTime, \FOVarPos}$
    is:
    “in \kl{round}~$\Round$, at \kl{step}~$\FOVarTime$,
    \kl{node}~$\FOVar[1]$'s \kl{tape}~$\TapeIdx$ contains an~$\Symbol$
    at position~$\FOVarPos$”.
  \item\AP $\intro*\SOVarExch{\Round}{\TapeIdx}$:
    a family of $(k + 2)$-\kl{ary} relations,
    for $\Round \in \Range[1]{\Radius}$ and $\TapeIdx \in \Set{\Receiving, \Sending}$,
    that indicate the tape positions of the incoming and outgoing messages
    that each \kl{node} exchanges with its \kl{neighbors}
    in every \kl{communication round}.
    The intended meaning of
    $\InRel{\SOVarExch{\Round}{\TapeIdx}}{\FOVar[1], \FOVar[2], \FOVarPos}$
    is:
    “in~\kl{round}~$\Round$,
    the message exchanged between~$\FOVar[1]$ and~$\FOVar[2]$
    is written immediately after position~$\FOVarPos$
    on $\FOVar[1]$'s \kl{tape}~$\TapeIdx$”
    (which implies that the symbol at position~$\FOVarPos$
    is either a~$\Endmarker$ or a~$\Separator$).
    For $\TapeIdx = \Receiving$,
    this corresponds to the incoming message
    that~$\FOVar[1]$ receives from~$\FOVar[2]$
    at the beginning of \kl{round}~$\Round$,
    whereas for $\TapeIdx = \Sending$,
    this corresponds to the outgoing message
    that~$\FOVar[1]$ sends to~$\FOVar[2]$
    at the end of \kl{round}~$\Round$.
  \end{itemize}

  It remains to specify
  the $\BFL$-\kl{formula}~$\Formula[2]\Of{\FOVar[1]}$,
  which is evaluated from the point of view of the \kl{node}
  represented by the \kl{first-order variable}~$\FOVar[1]$.
  We define~$\Formula[2]$ as the conjunction
  of the following $\BFL$-\kl{formulas},
  most of which are described only on an intuitive level
  to keep the exposition readable.

  \begin{itemize}
  \item\AP $\intro*\LinearTupleOrder\Of{\FOVar[1]}$
    states that the relation~$\SOVarTotOrd$ does indeed yield a linear order over
    the $k$-tuples in~$\FOVar[1]$'s $4\Radius$-\kl{neighborhood},
    as described above.
    We can write this \kl{formula} as follows:
    \begin{equation*}
      \ForAllLoc{\FOVarPos_1, \FOVarPos_2, \FOVarPos_3}{4\Radius + 1}{\FOVar[1]}
      \begin{pmatrix*}[l]
        \NOT \InRel{\SOVarTotOrd}{\FOVar[1], \FOVarPos_1, \FOVarPos_1}
        \; \AND \;
        \Bigl(
          \FOVarPos_1 \NotEqual \FOVarPos_2 \,\IMP\,
          \InRel{\SOVarTotOrd}{\FOVar[1], \FOVarPos_1, \FOVarPos_2} \OR
          \InRel{\SOVarTotOrd}{\FOVar[1], \FOVarPos_2, \FOVarPos_1}
        \Bigr) \,\, \\[0.5ex]
        {} \, \AND \;
        \Bigl(
          \InRel{\SOVarTotOrd}{\FOVar[1], \FOVarPos_1, \FOVarPos_2} \AND
          \InRel{\SOVarTotOrd}{\FOVar[1], \FOVarPos_2, \FOVarPos_3} \,\IMP\,
          \InRel{\SOVarTotOrd}{\FOVar[1], \FOVarPos_1, \FOVarPos_3}
        \Bigr)
      \end{pmatrix*}
    \end{equation*}
    The \kl{formula} “looks” up to distance $4\Radius + 1$
    because
    the \kl{labeling bits} of a \kl{node} at \kl{distance}~$4\Radius$
    lie $4\Radius + 1$ steps away
    in the \kl{structural representation} of the input \kl{graph}.
  \item\AP $\intro*\Correspondences\Of{\FOVar[1]}$
    enforces that the relation~$\SOVarCorrInter$ establishes the desired correspondence
    between the linear order of~$\FOVar[1]$
    and the linear order of every other \kl{node}
    in $\FOVar[1]$'s $2\Radius$-\kl{neighborhood},
    and that the relation~$\SOVarCorrIntra$ establishes the desired correspondence
    between the \kl{label} and linear order of~$\FOVar[1]$.
    Both properties can be easily specified inductively.
    For~$\SOVarCorrInter$,
    the base case states that
    $\InRel{\SOVarCorrInter}{
      \FOVar[1], \FOVarPos_1, \FOVarPos_1,
      \FOVar[2], \FOVarPos_2, \FOVarPos_2}$
    must hold for
    every \kl{node}~$\FOVar[2]$
    in $\FOVar[1]$'s $2\Radius$-\kl{neighborhood}
    and all $k$-tuples~$\FOVarPos_1$ and~$\FOVarPos_2$
    in the $4\Radius$-\kl{neighborhoods} of~$\FOVar[1]$ and~$\FOVar[2]$,
    respectively.
    The induction step then states that
    $\InRel{\SOVarCorrInter}{
      \FOVar[1], \FOVarPos_1, \FOVarPos_1',
      \FOVar[2], \FOVarPos_2, \FOVarPos_2'}$
    implies
    $\InRel{\SOVarCorrInter}{
      \FOVar[1], \FOVarPos_1, \text{“$\FOVarPos_1' + 1$”},
      \FOVar[2], \FOVarPos_2, \text{“$\FOVarPos_2' + 1$”}}$,
    where “$\FOVarPos_1' + 1$” and “$\FOVarPos_2' + 1$”
    represent the direct successors of~$\FOVarPos_1'$ and~$\FOVarPos_2'$
    with respect to~$\FOVar[1]$ and~$\FOVar[2]$,
    respectively.
    For~$\SOVarCorrIntra$,
    the specification is very similar.
  \item\AP $\intro*\UniqueIdentifier\Of{\FOVar[1]}$
    ensures that the relations~$\SOVarIdent{0}$ and~$\SOVarIdent{1}$
    represent a binary string for~$\FOVar[1]$
    and that this string constitutes
    an $\Radius$-\kl{locally unique} \kl{identifier}.
    In detail,
    this means that, with respect to~$\FOVar[1]$,
    each position~$\FOVarPos$
    can be either
    unlabeled, labeled with~$0$, or labeled with~$1$
    (but not both).
    If~$\FOVarPos$ is labeled,
    then so must be its predecessor “$\FOVarPos - 1$”
    in the order defined by~$\SOVarTotOrd$.
    Furthermore,
    for every \kl{node}~$\FOVar[2]$
    in~$\FOVar[1]$'s $2\Radius$-\kl{neighborhood},
    there exists a position~$\FOVarPos$
    at which the \kl{labelings} of~$\FOVar[1]$ and~$\FOVar[2]$ differ,
    entailing that~$\FOVar[1]$'s \kl{identifier} is $\Radius$-\kl{locally unique}.
    In order to refer to
    “the \kl{labeling} of~$\FOVar[2]$ at position~$\FOVarPos$”,
    we use~$\SOVarCorrInter$ to relate~$\FOVarPos$ to another $k$-tuple~$\FOVarPos'$
    that represents the same position as~$\FOVarPos$
    from the point of view of \kl{node}~$\FOVar[2]$.
  \item\AP $\intro*\NeighborOrder\Of{\FOVar[1]}$
    states that in~$\FOVar[1]$'s $1$-\kl{neighborhood},
    the relation~$\SOVarIdOrd$ agrees with
    the \kl{identifier order} of the \kl{nodes}.
    That is,
    for all \kl{neighbors}~$\FOVar[2]$ and~$\FOVar[3]$ of~$\FOVar[1]$,
    we have $\InRel{\SOVarIdOrd}{\FOVar[2], \FOVar[3]}$
    precisely if the \kl{identifier} of~$\FOVar[2]$
    is smaller than the \kl{identifier} of~$\FOVar[3]$
    with respect to the \kl{identifier order}.
    This is the case if
    at the first position~$\FOVarPos$
    where the \kl{identifiers} of~$\FOVar[2]$ and~$\FOVar[3]$ differ,
    either $\FOVar[2]$'s bit is smaller than $\FOVar[3]$'s bit
    or we have reached the end of $\FOVar[2]$'s \kl{identifier}
    while $\FOVar[3]$'s \kl{identifier} still goes on.
    Again,
    we make use of the correspondence relation~$\SOVarCorrInter$
    to relate matching positions of~$\FOVar[2]$ and~$\FOVar[3]$.
  \item\AP $\intro*\Certificate\Of{\FOVar[1]}$
    is very similar to $\UniqueIdentifier\Of{\FOVar[1]}$.
    It ensures that the relations~$\SOVarCertif{0}$ and~$\SOVarCertif{1}$
    represent a \kl{certificate} of~$\FOVar[1]$
    consisting of $0$'s and $1$'s.
  \item\AP $\intro*\ExecGroundRules\Of{\FOVar[1]}$
    formalizes some basic properties that
    any \kl{execution} of~$\Machine$
    must satisfy at \kl{node}~$\FOVar[1]$.
    In particular,
    in each~\kl{round} $\Round \in \Range[1]{\Radius}$,
    at every \kl{step}~$\FOVarTime$,
    the \kl{machine} must be
    in exactly one \kl{state} $\State \in \StateSet$,
    there must be exactly one symbol
    $\Symbol \in \Set{\Endmarker, \Blank, \Separator, 0, 1}$
    written at each position~$\FOVarPos$ of each of the three \kl{tapes}
    (the symbol of the first position always being~$\Endmarker$),
    and each tape head must be located at exactly one position.
    Moreover,
    in each~\kl{round} there must be some \kl{step}~$\FOVarTime$
    at which the \kl{machine} halts by reaching one of the \kl{states}
    $\PauseState$ or~$\StopState$.
    Intuitively speaking,
    $\FOVar[1]$~must respect the basic “mechanics” of Turing machines
    and may not run out of space or time
    (both of which are bounded by the number of $k$-tuples
    in $\FOVar[1]$'s $4\Radius$-\kl{neighborhood}).
  \item\AP $\intro*\OwnInput\Of{\FOVar[1]}$
    ensures that at \kl{step}~$0$ of \kl{round}~$1$,
    \kl{node}~$\FOVar[1]$'s \kl{internal tape}
    contains the string
    $\Labeling{\Graph}(\FOVar[1]) \, \Separator \,
     \IdMap(\FOVar[1])   \, \Separator \,
     \CertifMap(\FOVar[1])$,
    where $\Labeling{\Graph}(\FOVar[1])$
    is the \kl{label} of $\FOVar[1]$
    that is represented by
    the unary relation~$\BitSet{1}{\StructRepr{\Graph}}$
    of the input \kl{structure}~$\StructRepr{\Graph}$,
    $\IdMap(\FOVar[1])$ is the \kl{identifier} of~$\FOVar[1]$
    that is represented by
    the relations~$\SOVarIdent{0}$ and~$\SOVarIdent{1}$,
    and $\CertifMap(\FOVar[1])$ is the \kl{certificate} of~$\FOVar[1]$
    that is represented by
    the relations~$\SOVarCertif{0}$ and~$\SOVarCertif{1}$.
    To check that $\Labeling{\Graph}(\FOVar[1])$
    is written at the beginning of the \kl{internal tape},
    i.e., just after the \kl{left-end marker}~$\Endmarker$,
    we make use of the relation~$\SOVarCorrIntra$ as follows:
    for every \kl{labeling bit}~$\FOVar[3]$
    and every position~$\FOVarPos$,
    if
    $\InRel{\SOVarCorrIntra}{\FOVar[1], \FOVar[3], \FOVarPos}$
    and
    $\IsBit{\Symbol}\Of{\FOVar[3]}$,
    for $\Symbol \in \Set{0, 1}$,
    then we must have
    $\InRel{\SOVarTape{1}{\Internal}{\Symbol}}
     {\FOVar[1], \FOVarTime_0, \text{“$\FOVarPos + 1$”}}$,
    where~$\FOVarTime_0$ is the tuple representing \kl{computation step}~$0$
    and “$\FOVarPos + 1$” denotes the direct successor of~$\FOVarPos$
    (according to the order defined by~$\SOVarTotOrd$ with respect to~$\FOVar[1]$).
    Based on that,
    the position of the first \kl{separator}~$\Separator$
    must be $\FOVar[1]$'s smallest position~$\FOVarPos[2]$
    such that
    $\InRel{\SOVarCorrIntra}{\FOVar[1], \FOVar[3], \text{“$\FOVarPos[2] - 1$”}}$
    does not hold for any \kl{labeling bit}~$\FOVar[3]$.
    Next,
    we check that
    the first \kl{separator} is followed by $\IdMap(\FOVar[1])$,
    using the relation~$\SOVarCorrInter$
    to express that~$\IdMap(\FOVar[1])$
    is shifted by “$\FOVarPos[2] + 1$” positions to the right
    of the initial position~$\FOVarPos_0$:
    for all positions~$\FOVarPos$ and~$\FOVarPos'$,
    if
    $\InRel{\SOVarCorrInter}{
      \FOVar[1], \FOVarPos_0, \FOVarPos,
      \FOVar[1], \text{“$\FOVarPos[2] + 1$”}, \FOVarPos'}$
    and
    $\InRel{\SOVarIdent{\Symbol}}{\FOVar[1], \FOVarPos}$,
    for $\Symbol \in \Set{0, 1}$,
    then we must have
    $\InRel{\SOVarTape{1}{\Internal}{\Symbol}}{\FOVar[1], \FOVarTime_0, \FOVarPos'}$.
    Finally,
    we check that the second \kl{separator}~$\Separator$
    is followed by $\CertifMap(\FOVar[1])$,
    proceeding completely analogously with~$\SOVarCertif{\Symbol}$
    instead of~$\SOVarIdent{\Symbol}$.
  \item\AP $\intro*\ReceiveMessages{\Round}\Of{\FOVar[1]}$,
    for $\Round \in \Range[1]{\Radius}$,
    states that the messages received by~$\FOVar[1]$
    at the beginning of \kl{round}~$\Round$
    are written on~$\FOVar[1]$'s \kl{receiving tape},
    each followed by the \kl{separator}~$\Separator$,
    and sorted according to
    the \kl{identifier order} of the senders.
    To accomplish this,
    the \kl{formula} also guarantees by induction
    that the relation~$\SOVarExch{\Round}{\Receiving}$
    correctly represents the starting positions
    of the messages that~$\FOVar[1]$ receives
    from each \kl{neighbor}.
    \begin{itemize}
    \item The base case of the definition of~$\SOVarExch{\Round}{\Receiving}$
      is straightforward,
      since the message from~$\FOVar[1]$'s first \kl{neighbor}~$\FOVar[2]_1$
      (with respect to the order relation~$\SOVarIdOrd$)
      must start immediately after the \kl{left-end marker}~$\Endmarker$
      on~$\FOVar[1]$'s \kl{receiving tape}.
      That is,
      $\InRel{\SOVarExch{\Round}{\Receiving}}{\FOVar[1], \FOVar[2]_1, \FOVarPos_1}$,
      where~$\FOVarPos_1$ is~$\FOVar[1]$'s first position
      (with respect to the relation~$\SOVarTotOrd$).
    \item Now,
      we need to distinguish two cases.
      First,
      suppose that for a given \kl{neighbor}~$\FOVar[2]$,
      we have
      $\InRel{\SOVarExch{\Round - 1}{\Sending}}{\FOVar[2], \FOVar[1], \FOVarPos'}$
      and
      $\InRel{\SOVarExch{\Round}{\Receiving}}{\FOVar[1], \FOVar[2], \FOVarPos}$,
      i.e.,
      the message that~$\FOVar[2]$ sends to~$\FOVar[1]$
      at the end of \kl{round}~$\Round - 1$
      is stored right after position~$\FOVarPos'$
      on $\FOVar[2]$'s \kl{sending tape},
      and the message that~$\FOVar[1]$ receives from~$\FOVar[2]$
      at the beginning of \kl{round}~$\Round$
      is stored right after position~$\FOVarPos$
      on $\FOVar[1]$'s \kl{receiving tape}.
      Based on this information,
      the \kl{formula}~$\ReceiveMessages{\Round}\Of{\FOVar[1]}$
      ensures that the two messages are indeed the same
      by stating that on $\FOVar[2]$'s \kl{sending tape},
      at the end of \kl{round}~$\Round - 1$,
      every position~$\FOVarPos[2]'$ located between~$\FOVarPos'$
      and the next occurrence of~$\Separator$
      (including the latter position)
      contains the same symbol as
      the corresponding position~$\FOVarPos[2]$ on $\FOVar[1]$'s \kl{receiving tape}
      at the beginning of \kl{round}~$\Round$.
      (Without loss of generality,
      we may assume that the messages on $\FOVar[2]$'s \kl{sending tape}
      are always followed by a~$\Separator$
      and do not contain any useless~$\Blank$'s.)
      The fact that position~$\FOVarPos[2]$
      corresponds to position~$\FOVarPos[2]'$
      is expressed by the relation
      $\InRel{\SOVarCorrInter}{
        \FOVar[1], \FOVarPos, \FOVarPos[2],
        \FOVar[2], \FOVarPos', \FOVarPos[2]'}$,
      which states that
      the distance from~$\FOVarPos$ to~$\FOVarPos[2]$
      on $\FOVar[1]$'s \kl{receiving tape}
      is the same as
      the distance from~$\FOVarPos'$ to~$\FOVarPos[2]'$
      on $\FOVar[2]$'s \kl{sending tape}.

      The second case is when we have
      $\InRel{\SOVarExch{\Round}{\Receiving}}{\FOVar[1], \FOVar[2], \FOVarPos}$
      but there is no~$\FOVarPos'$ such that
      $\InRel{\SOVarExch{\Round - 1}{\Sending}}{\FOVar[2], \FOVar[1], \FOVarPos'}$.
      This means that
      $\FOVar[2]$ has not written any message for~$\FOVar[1]$
      on its \kl{sending tape}
      at the end of \kl{round} $\Round - 1$,
      and therefore that
      $\FOVar[1]$ receives the empty string from~$\FOVar[2]$
      in \kl{round}~$\Round$.
      (Note that this happens in particular for $\Round = 1$.)
      In this case,
      our \kl{formula} simply states
      that at the beginning of \kl{round}~$\Round$,
      the \kl{receiving tape} of~$\FOVar[1]$
      contains the \kl{separator}~$\Separator$
      at position~“$\FOVarPos + 1$”
      (which, as before,
      represents the direct successor of~$\FOVarPos$ with respect to~$\SOVarTotOrd$).
    \item Finally,
      to complete the inductive definition of~$\SOVarExch{\Round}{\Receiving}$,
      we state that
      $\InRel{\SOVarExch{\Round}{\Receiving}}{\FOVar[1], \FOVar[2], \FOVarPos}$
      implies
      $\InRel{\SOVarExch{\Round}{\Receiving}}{\FOVar[1], \FOVar[3], \FOVarPos[2]}$
      if $\FOVar[3]$ is the smallest \kl{neighbor} of~$\FOVar[1]$
      strictly greater than~$\FOVar[2]$
      (with~respect to~$\SOVarIdOrd$)
      and~$\FOVarPos[2]$ is the smallest position of~$\FOVar[1]$
      that is strictly greater than~$\FOVarPos$
      (with respect to~$\SOVarTotOrd$)
      and contains the symbol~$\Separator$.
      To make sure that the relation $\SOVarExch{\Round}{\Receiving}$ is minimal,
      we also require that
      for every \kl{neighbor}~$\FOVar[2]$ of~$\FOVar[1]$,
      there is only one position~$\FOVarPos$ such that
      $\InRel{\SOVarExch{\Round}{\Receiving}}{\FOVar[1], \FOVar[2], \FOVarPos}$.
    \end{itemize}
  \item\AP $\intro*\InitLocalConfig{\Round}\Of{\FOVar[1]}$,
    for $\Round \in \Range[1]{\Radius}$,
    provides the missing parts of
    the description of~$\FOVar[1]$'s local configuration
    at time~$0$ in \kl{round}~$\Round$:
    First,
    if~$\Round \geq 2$,
    the \kl{formula} stipulates that the \kl{node}'s \kl{internal tape}
    contains the same string as at the end of \kl{round}~$\Round - 1$;
    for $\Round = 1$,
    the initial content of the \kl{internal tape} in \kl{round}~$1$
    has already been specified by
    the above \kl{formula}~$\OwnInput\Of{\FOVar[1]}$.
    Second,
    the \kl{sending tape} must initially be completely empty.
    Third,
    the \kl{machine}'s \kl{state} must be reset to~$\StartState$,
    unless it has reached~$\StopState$ in \kl{round}~$\Round - 1$,
    in which case the \kl{state} remains unchanged.
  \item\AP $\intro*\ComputeLocally{\Round}\Of{\FOVar[1]}$,
    for $\Round \in \Range[1]{\Radius}$,
    describes how~$\Machine$ performs
    a single \kl{computation step}
    at \kl{node}~$\FOVar[1]$ in \kl{round}~$\Round$.
    The \kl{formula} essentially states the following:
    On the one hand,
    for the \kl{receiving tape},
    the cell contents remain the same at all \kl{steps}~$\FOVarTime$,
    while for the \kl[internal tape]{internal} and \kl{sending tapes},
    the symbol at position~$\FOVarPos$ remains unchanged
    between \kl{steps}~$\FOVarTime$ and~“$\FOVarTime + 1$”
    if the corresponding tape head is not located at position~$\FOVarPos$
    at \kl{step}~$\FOVarTime$.
    On the other hand,
    if at \kl{step}~$\FOVarTime$
    the \kl{machine} is in \kl{state}~$\State$
    and reads the symbols $\Symbol_1$,~$\Symbol_2$,~$\Symbol_3$
    at positions $\FOVarPos_1$,~$\FOVarPos_2$,~$\FOVarPos_3$
    of the
    \kl[receiving tape]{receiving}, \kl[internal tape]{internal}, and \kl{sending tapes},
    respectively,
    then the configuration at time~“$\FOVarTime + 1$”
    must take into account the updates specified by
    the \kl{transition function}~$\TransFunc$.
    That is,
    if
    $\TransFunc(\State, \Symbol_1, \Symbol_2, \Symbol_3)
    = \Tuple{\State', \Symbol'_2, \Symbol'_3, \Move_1, \Move_2, \Move_3}$,
    then at time “$\FOVarTime + 1$”,,
    the \kl{machine} is in \kl{state}~$\State'$,
    position~$\FOVarPos_2$ of the \kl{internal tape}
    contains the symbol~$\Symbol'_2$,
    position~$\FOVarPos_3$ of the \kl{sending tape}
    contains the symbol~$\Symbol'_3$,
    and the three tape heads are located at positions
    “$\FOVarPos_1 + \Move_1$”,
    “$\FOVarPos_2 + \Move_2$”, and
    “$\FOVarPos_3 + \Move_3$”.
  \item\AP $\intro*\SendMessages{\Round}\Of{\FOVar[1]}$,
    for $\Round \in \Range[1]{\Radius}$,
    guarantees that the relation~$\SOVarExch{\Round}{\Sending}$
    correctly represents the starting positions
    of the messages that~$\FOVar[1]$ sends to each \kl{neighbor}
    at the end of \kl{round}~$\Round$.
    Since the inductive definition of~$\SOVarExch{\Round}{\Sending}$
    is very similar to that of~$\SOVarExch{\Round}{\Receiving}$
    given in~$\ReceiveMessages{\Round}\Of{\FOVar[1]}$,
    we do not further elaborate on it.
    Let us only note that in contrast to $\SOVarExch{\Round}{\Receiving}$,
    for some \kl{neighbors}~$\FOVar[2]$
    there might be no~$\FOVarPos$
    such that
    $\InRel{\SOVarExch{\Round}{\Sending}}{\FOVar[1], \FOVar[2], \FOVarPos}$.
    This happens when~$\FOVar[1]$ does not write enough messages
    on its \kl{sending tape} for all its \kl{neighbors}
    (in which case the missing messages default to the empty string).
  \item\AP $\intro*\AcceptFml\Of{\FOVar[1]}$
    states that
    $\FOVar[1]$ must eventually reach \kl{state}~$\StopState$
    in some \kl{round} $\Round \in \Range[1]{\Radius}$,
    with the string~$1$ written on its \kl{internal tape}.
    \qedhere
  \end{itemize}
\end{proof}

Fagin's theorem was extended
by Stockmeyer~\cite{DBLP:journals/tcs/Stockmeyer76}
to the higher levels of the \kl{polynomial hierarchy},
thus establishing a levelwise correspondence with the \kl{second-order hierarchy}
(see, e.g., \cite[Cor.~9.9]{DBLP:books/sp/Libkin04}).
In the next theorem,
we show that this extension also carries over to the distributed setting.
Again,
the original result can be recovered
by restricting both sides of the equivalences to \kl{single-node graphs}.

\begin{theorem}
  \label{thm:local-hierarchy-equivalence}
  The \kl{locally polynomial hierarchy}
  and the \kl{local second-order hierarchy} on \kl{graphs}
  are levelwise equivalent from level~$1$ onwards.
  More precisely,
  $\SigmaLP{\Level} = \SigmaLFO{\Level}\On{\GRAPH}$
  and
  $\PiLP{\Level} = \PiLFO{\Level}\On{\GRAPH}$
  for all $\Level \in \Positives$.
\end{theorem}

\begin{proof}[Proof of Theorem~\ref{thm:local-hierarchy-equivalence} -- Backward direction]
  \phantomsection\label{prf:local-hierarchy-equivalence-backward}%
  Let us first assume that
  we are given a $\SigmaLFOL{\Level}$-\kl{sentence}
  $\Formula[1] =
   \ExistsRel{\Vector{\SOVar}_1}
   \ForAllRel{\Vector{\SOVar}_2} \dots
   \QuantifierRel{\Vector{\SOVar}_{\Level}} \,
   \ForAll{\FOVar[1]} \, \Formula[2]\Of{\FOVar[1]}$,
  where each $\Vector{\SOVar}_i$
  is a tuple of \kl{second-order variables},
  and $\QuantifierRel{}$ is~$\ForAllRel{}$ if $\Level$ is even
  and $\ExistsRel{}$ otherwise.
  Let $\Property$ be
  the \kl(graph){property} \kl{defined} by~$\Formula[1]$ on \kl{graphs},
  and let $\Radius$ be the maximum nesting depth
  of \kl(quantifier){bounded} \kl{first-order quantifiers}
  in the $\BFL$-\kl{formula}~$\Formula[2]$
  (intuitively,
  the distance up to which $\Formula[2]$ can “see”).
  We construct
  a \kl{restrictive} $\SigmaLP{\Level}$-\kl{arbiter}~$\Machine_{\Formula[1]}$
  for~$\Property$
  under $\Radius$-\kl{locally unique} \kl{identifiers}
  and $\Tuple{\Radius, \CertifPolynomial}$\nobreakdash-\kl{bounded certificates}
  restricted by \kl{machines}
  $\Machine_1, \dots, \Machine_{\Level}$.
  Here,
  $\CertifPolynomial$~is a polynomial chosen based on
  the number and \kl{arities} of the \kl{variables} in
  $\Vector{\SOVar}_1, \dots, \Vector{\SOVar}_{\Level}$
  such that
  the \kl{certificates} described below satisfy the
  $\Tuple{\Radius, \CertifPolynomial}$\nobreakdash-\kl(certificate){boundedness}
  condition.
  The intention is that
  for each $i \in \Range[1]{\Level}$
  and each \kl{node}~$\Node[1]$ of the input \kl{graph}~$\Graph$,
  the \kl{certificate}~$\CertifMap_i(\Node[1])$ \kl{encodes}
  part of a \kl{variable assignment} on~$\StructRepr{\Graph}$
  that assigns \kl{interpretations}
  to the \kl{relation variables} in~$\Vector{\SOVar}_i$.
  More precisely,
  each \kl{machine}~$\Machine_i$
  restricts quantification over~$\CertifMap_i$
  such that
  for each \kl{relation variable}~$\SOVar$ in~$\Vector{\SOVar}_i$
  of \kl{arity}~$\Arity$,
  the \kl{certificate}~$\CertifMap_i(\Node[1])$ must \kl{encode}
  a set of $\Arity$-tuples
  whose first \kl{element} represents
  either~$\Node[1]$ or one of its \kl{labeling bits},
  and whose remaining \kl{elements} all represent \kl{nodes} or \kl{labeling bits}
  that lie in the $2\Radius$-\kl{neighborhood} of~$\Node[1]$.
  In order to refer to the \kl{elements}
  corresponding to a particular \kl{node}~$\Node[2]$,
  the \kl{certificate} makes use of
  $\Node[2]$'s \kl{locally unique} \kl{identifier}~$\IdMap(\Node[2])$.
  Since these additional restrictions on the \kl{certificates}
  clearly satisfy \kl{local repairability},
  Lemma~\ref{lem:restrictive-arbiters}
  allows us to subsequently convert~$\Machine_{\Formula[1]}$
  into an equivalent $\SigmaLP{\Level}$-\kl{arbiter} without restrictions.

  The \kl{machine}~$\Machine_{\Formula[1]}$
  proceeds in $\Radius + 1$ \kl{rounds}.
  In the first~$\Radius$ \kl{rounds},
  each \kl{node}~$\Node[1]$ collects information
  about its $\Radius$-\kl{neighborhood},
  which allows it to reconstruct
  $\Neighborhood{\Graph}{\Radius}{\Node[1]}$
  and all the \kl{identifiers} and \kl{certificates}
  in that \kl{subgraph}.
  Then,
  in the last \kl{round},
  each \kl{node}~$\Node[1]$ evaluates~$\Formula[2]$ locally
  on~$\StructNeighborhood{\Graph}{\Radius}{\Node[1]}$
  and \kl{accepts} if and only if
  $\Formula[2]$ is \kl{satisfied} at
  the \kl{elements} representing~$\Node[1]$ and its \kl{labeling bits}.
  Since~$\Formula[2]$ only makes use of \kl{first-order quantification},
  this can be done in polynomial~time
  (simply by exhaustive search).

  Note that the \kl{certificate assignments}
  $\CertifMap_1, \dots, \CertifMap_{\Level}$
  \kl{encode} relations
  $\Vector{\SOVar}_1, \dots, \Vector{\SOVar}_{\Level}$
  on~$\StructRepr{\Graph}$
  that relate only \kl{elements}
  whose associated \kl{nodes} lie
  at \kl{distance} at most~$2\Radius$ from each other.
  However,
  this does not entail any loss of generality
  because the \kl{formula}~$\Formula[2]$
  (which belongs to~$\BFL$)
  can only make statements about \kl{elements}
  that lie this close together anyway.

  The construction is completely analogous
  if we are given a $\PiLFOL{\Level}$-\kl{sentence}
  instead of a $\SigmaLFOL{\Level}$-\kl{sentence}.
\end{proof}

\begin{proof}[Proof of Theorem~\ref{thm:local-hierarchy-equivalence} -- Forward direction]
  \phantomsection\label{prf:local-hierarchy-equivalence-forward}%
  We start by showing that
  $\SigmaLP{\Level} \subseteq \SigmaLFO{\Level}\On{\GRAPH}$.
  Let $\Property$ be a \kl{graph property} in~$\SigmaLP{\Level}$,
  and let $\Machine$ be a $\SigmaLP{\Level}$-\kl{arbiter} for~$\Property$
  that operates under
  $\IdentRadius$\nobreakdash-\kl{locally unique} \kl{identifiers}
  and $\Tuple{\CertifRadius_1, \CertifPolynomial_1}$-\kl{bounded certificates},
  and runs in \kl{round time}~$\Radius_2$ and \kl{step time}~$\Polynomial_2$.
  Moreover,
  let $\Radius = \max\Set{\IdentRadius, \CertifRadius_1, \Radius_2}$,
  and let $\Polynomial$ be a polynomial
  that bounds both $\CertifPolynomial_1$ and~$\Polynomial_2$.
  By the proof of Lemma~\ref{lem:restrictive-arbiters},
  we may assume without loss of generality
  that $\Machine$ relativizes quantification
  to \kl{certificate assignments} that satisfy
  the $\Tuple{\CertifRadius_1, \CertifPolynomial_1}$-\kl(certificate){boundedness} condition,
  so it does not matter if
  \kl(certificate){Eve} and \kl(certificate){Adam} choose
  \kl{certificates} that are too large.
  Now,
  we fix $k \in \Naturals$ such that
  the polynomial~$f$ described in Lemma~\ref{lem:polynomial-space-time}
  (for our choices of $\Level$, $\Radius$, $\Polynomial$, and $\Machine$)
  satisfies $f(n) < n^k$ for $n > 1$.
  By~Lemma~\ref{lem:polynomial-space-time},
  for every \kl{graph}~$\Graph$
  whose \kl{structural representation} has at least two \kl{elements},%
  \footnote{See Footnote~\ref{ftn:local-fagin-forward}
    in the proof of Theorem~\ref{thm:local-fagin} (forward direction)
    on page~\pageref{ftn:local-fagin-forward}.}
  every \kl{small} $\Radius$-\kl{locally unique}
  \kl{identifier assignment}~$\IdMap$ of~$\Graph$,
  and all
  $\Tuple{\Radius, \Polynomial}$-\kl{bounded certificate assignments}
  $\CertifMap_1, \dots, \CertifMap_{\Level}$
  of $\Tuple{\Graph, \IdMap}$,
  the \kl{step running time} and \kl{space usage}
  of each \kl{node} $\Node \inG \Graph$
  are bounded by~%
  $\bigl(\CardS{\StructNeighborhood{\Graph}{4\Radius}{\Node}}\bigr)^k$
  in each \kl{round} $\Round \in \Range[1]{\Radius}$
  of the corresponding \kl{execution} of~$\Machine$.
  Intuitively,
  this gives us a bound on the amount of information
  required to describe a game between
  \kl(certificate){Eve} and \kl(certificate){Adam}
  and the subsequent \kl{execution} of the \kl{arbiter}~$\Machine$,
  assuming that the \kl{nodes} of the input \kl{graph}
  are assigned \kl{small} \kl{identifiers}.

  \AP
  To convert~$\Machine$ into a
  $\SigmaLFOL{\Level}$-\kl{sentence}~$\Formula[1]_{\Machine}$
  \kl{defining}~$\Property$,
  we use the same \kl{relation variables}
  as in the proof of Theorem~\ref{thm:local-fagin}
  (with the same intended \kl{interpretations}),
  except that instead of $\Tuple{\SOVarCertif{0}, \SOVarCertif{1}}$
  we now use
  $\Tuple{
     \intro*\SOVarCertifAlt{\LevelIdx}{0},
     \reintro*\SOVarCertifAlt{\LevelIdx}{1}
   }_{\LevelIdx \in \Range[1]{\Level}}$
  to represent the \kl{certificate assignments}
  $\CertifMap_1, \dots, \CertifMap_{\Level}$.
  More precisely,
  our \kl{formula}~$\Formula[1]_{\Machine}$ is of the form
  \begin{equation*}
    \ExistsRel{\Vector{\SOVar}_{\Tag{aux}}} \,
    \ExistsRel{\Vector{\SOVar}_1} \:
    \ForAllRel{\Vector{\SOVar}_2}
    \dots
    \QuantifierRel{\Vector{\SOVar}_{\Level}} \;
    \QuantifierRel{\Vector{\SOVar}_{\Tag{exe}}} \:
    \ForAllNode{\FOVar[1]} \:
    \Formula[2]\Of{\FOVar[1]},
  \end{equation*}
  where
  $\QuantifierRel{}$ is~$\ForAllRel{}$ if $\Level$ is even
  and $\ExistsRel{}$ otherwise,
  $\Vector{\SOVar}_{\Tag{aux}} =
   \Tuple{
     \SOVarTotOrd,
     \SOVarCorrInter,
     \SOVarCorrIntra,
     \SOVarIdent{0},
     \SOVarIdent{1},
     \SOVarIdOrd
   }$
  is a collection of auxiliary \kl{relation variables}
  that help us specify the remaining relations,
  $\SOVarIdent{0}$ and $\SOVarIdent{1}$ are intended to represent an
  $\Radius$-\kl{locally unique} \kl{identifier assignment}~$\IdMap$,
  each~$\Vector{\SOVar}_{\LevelIdx}$
  is a pair $\Tuple{\SOVarCertifAlt{\LevelIdx}{0}, \SOVarCertifAlt{\LevelIdx}{1}}$
  of \kl{variables} intended to represent
  a \kl{certificate assignment}~$\CertifMap_{\LevelIdx}$
  (for $\LevelIdx \in \Range[1]{\Level}$),
  and
  $\Vector{\SOVar}_{\Tag{exe}} =
   \Tuple{
     \dots,
     \SOVarState{\Round}{\State}, \dots,
     \SOVarHead{\Round}{\TapeIdx}, \dots,
     \SOVarTape{\Round}{\TapeIdx}{\Symbol}, \dots,
     \SOVarExch{\Round}{\TapeIdx}, \dots
   }$
  is a collection of \kl{variables} intended to represent
  the \kl{execution} of~$\Machine$ on the input \kl{graph}
  under~$\IdMap$ and
  $\CertifMap_1
   \CertifConcat \dots \CertifConcat
   \CertifMap_{\Level}$.

  Similarly to the proof of Theorem~\ref{thm:local-fagin},
  we would like to state in the $\BFL$-\kl{formula}~$\Formula[2]\Of{\FOVar[1]}$
  that from the point of view of \kl{node}~$\FOVar[1]$,
  the relations in
  $\Vector{\SOVar}_{\Tag{aux}}$,
  $\Vector{\SOVar}_1$,
  \dots,
  $\Vector{\SOVar}_{\Tag{exe}}$
  are valid
  in the sense that they correspond to their intended \kl{interpretations}.
  However,
  in order to relativize all quantifications to valid relations,
  we now have to take into account
  whether a relation is quantified existentially or universally:
  $\Formula[2]\Of{\FOVar[1]}$ must be
  false if the first invalid relation known to~$\FOVar[1]$
  is chosen existentially,
  but true if it is chosen universally.
  Intuitively speaking,
  a relation is invalid from $\FOVar[1]$'s point of view
  if it does not correctly represent the part of the \kl{execution}
  that affects~$\FOVar[1]$
  (in particular,
  $\FOVar[1]$'s local computations
  and the information $\FOVar[1]$ receives from its \kl{neighbors}).
  Unfortunately,
  this also means that for relations chosen universally,
  $\FOVar[1]$ cannot simply rely on the fact
  that other \kl{nodes} will detect invalidities;
  if $\FOVar[1]$~is affected by an invalidity within its \kl{neighborhood},
  then it must be “aware” of this itself
  to ensure that $\Formula[2]\Of{\FOVar[1]}$ holds true.
  Therefore,
  unlike in the proof of Theorem~\ref{thm:local-fagin},
  it is not sufficient for $\Formula[2]\Of{\FOVar[1]}$ to only check
  $\FOVar[1]$'s local computations and message exchanges;
  it must also check those of the \kl{nodes}
  that have a direct or indirect influence on~$\FOVar[1]$.
  This is very similar to the quantifier relativization algorithm described in
  the second part of the proof of Lemma~\ref{lem:restrictive-arbiters}.

  We now give a formal definition of~$\Formula[2]\Of{\FOVar[1]}$
  using the helper \kl{formulas}
  $\LinearTupleOrder\Of{\FOVar[1]}$,
  $\Correspondences\Of{\FOVar[1]}$, and so on,
  introduced in the proof of Theorem~\ref{thm:local-fagin}.
  The structure of~$\Formula[2]\Of{\FOVar[1]}$ depends on
  the prefix of \kl{second-order quantifiers} of~$\Formula[1]_{\Machine}$.
  At the outermost level,
  $\Formula[2]\Of{\FOVar[1]}$ checks that
  the relations in~$\Vector{\SOVar}_{\Tag{aux}}$ are valid
  from the point of view of all \kl{nodes}
  in the $\Radius$-\kl{neighborhood} of~$\FOVar[1]$,
  using the \kl{subformula}
  \begin{equation*}
    \Formula[2]_{\Tag{aux}}\Of{\FOVar[1]} =
    \ForAllLocNode{\FOVar[2]}{\Radius}{\FOVar[1]}
    \left(
      \begin{aligned}
        &\LinearTupleOrder\Of{\FOVar[2]}
        \, \AND \,
        \Correspondences\Of{\FOVar[2]}
        \, \AND {} \\
        &\UniqueIdentifier\Of{\FOVar[2]}
        \, \AND \,
        \NeighborOrder\Of{\FOVar[2]} \,
      \end{aligned}
    \right)\!.
  \end{equation*}
  Since $\Vector{\SOVar}_{\Tag{aux}}$ is quantified existentially,
  $\Formula[2]_{\Tag{aux}}\Of{\FOVar[1]}$
  is used as a conjunct in~$\Formula[2]\Of{\FOVar[1]}$, i.e.,
  \begin{equation*}
    \Formula[2]\Of{\FOVar[1]} = \,
    \Formula[2]_{\Tag{aux}}\Of{\FOVar[1]}
    \AND
    \Formula[2]_1\Of{\FOVar[1]}.
  \end{equation*}
  The second conjunct~$\Formula[2]_1\Of{\FOVar[1]}$
  checks the validity of the remaining relations,
  starting with those in~$\Vector{\SOVar}_1$.
  Depending on whether a relation is quantified existentially or universally,
  its invalidity makes the remainder of the \kl{formula} hold false or true,
  respectively.
  This leads to the inductive definition
  \begin{equation*}
    \Formula[2]_{\LevelIdx}\Of{\FOVar[1]} =
    \begin{cases*}
      \bigl(
        \ForAllLocNode{\FOVar[2]}{\Radius}{\FOVar[1]} \,
        \CertificateAlt{\LevelIdx}\Of{\FOVar[2]}
      \bigr)
      \IMP
      \Formula[2]_{\LevelIdx + 1}\Of{\FOVar[1]}
      & if $\LevelIdx$ is even,
      \\
      \bigl(
        \ForAllLocNode{\FOVar[2]}{\Radius}{\FOVar[1]} \,
        \CertificateAlt{\LevelIdx}\Of{\FOVar[2]}
      \bigr)
      \, \AND \,
      \Formula[2]_{\LevelIdx + 1}\Of{\FOVar[1]}
      & if $\LevelIdx$ is odd,
    \end{cases*}
  \end{equation*}
  for $\LevelIdx \in \Range[1]{\Level}$,
  where the helper formula
  \AP
  $\intro*\CertificateAlt{\LevelIdx}\Of{\FOVar[2]}$
  is a variant of
  $\Certificate\Of{\FOVar[2]}$
  that ensures that the
  relations~$\SOVarCertifAlt{\LevelIdx}{0}$ and~$\SOVarCertifAlt{\LevelIdx}{1}$
  correctly represent a \kl{certificate} of~$\FOVar[2]$
  consisting of $0$'s and $1$'s.
  For the last \kl{formula} $\Formula[2]_{\Level}\Of{\FOVar[1]}$,
  the \kl{subformula}~$\Formula[2]_{\Level + 1}\Of{\FOVar[1]}$
  is defined below.
  It checks that
  $\Vector{\SOVar}_{\Tag{exe}}$ is valid from $\FOVar[1]$'s point of view
  and that
  $\FOVar[1]$ eventually \kl{accepts}:
  \begin{equation*}
    \Formula[2]_{\Level + 1}\Of{\FOVar[1]} =
    \begin{cases*}
      \Formula[2]_{\Tag{exe}}\Of{\FOVar[1]}
      \IMP
      \AcceptFml\Of{\FOVar[1]}
      & if $\Level$ is even,
      \\
      \Formula[2]_{\Tag{exe}}\Of{\FOVar[1]}
      \, \AND \,
      \AcceptFml\Of{\FOVar[1]}
      & if $\Level$ is odd,
    \end{cases*}
  \end{equation*}
  where the \kl{formula}
  \begin{align*}
    \Formula[2]_{\Tag{exe}}\Of{\FOVar[1]} = \;\,
    &\ForAllLocNode{\FOVar[2]}{\Radius}{\FOVar[1]}
    \Bigl(
      \ExecGroundRules\Of{\FOVar[2]}
      \, \AND \,
      \OwnInputAlt\Of{\FOVar[2]}
    \Bigr)
    \; \AND {} \\
    &\BigAND_{\lalign{\Round \in \Range[1]{\Radius}}} \,
    \ForAllLocNode{\FOVar[2]}{\Radius - \Round}{\FOVar[1]}
    \left(
      \begin{aligned}
        &\ReceiveMessages{\Round}\Of{\FOVar[2]}
        \, \AND \,
        \InitLocalConfig{\Round}\Of{\FOVar[2]}
        \, \AND {} \\
        &\ComputeLocally{\Round}\Of{\FOVar[2]}
        \, \AND \,
        \SendMessages{\Round}\Of{\FOVar[2]}
      \end{aligned}
    \right)
  \end{align*}
  states that
  the part of the \kl{execution} that has an influence on~$\FOVar[1]$
  is correctly represented by the appropriate relations.
  (The helper \kl{formula}
  \AP
  $\intro*\OwnInputAlt\Of{\FOVar[2]}$ is a variant of $\OwnInput\Of{\FOVar[2]}$
  that takes into account the \kl{certificate assignments}
  $\CertifMap_1, \dots, \CertifMap_{\Level}$.)
  Since $\FOVar[1]$ is not influenced by the local computations
  that \kl{nodes} at \kl{distance}~$\Round$ make
  after \kl{round} $\Radius - \Round$,
  we only check those they make between \kl{rounds}~$1$ and~$\Radius - \Round$.

  To show that
  $\PiLP{\Level} \subseteq \PiLFO{\Level}\On{\GRAPH}$,
  the construction is almost the same,
  but the \kl{formula}~$\Formula[1]_{\Machine}$ is now of the form
  \begin{equation*}
    \ForAllRel{\Vector{\SOVar}_{\Tag{aux}}} \,
    \ForAllRel{\Vector{\SOVar}_1} \:
    \ExistsRel{\Vector{\SOVar}_2}
    \dots
    \QuantifierRel{\Vector{\SOVar}_{\Level}} \;
    \QuantifierRel{\Vector{\SOVar}_{\Tag{exe}}} \:
    \ForAllNode{\FOVar[1]} \:
    \Formula[2]\Of{\FOVar[1]},
  \end{equation*}
  where
  $\QuantifierRel{}$ is~$\ExistsRel{}$ if $\Level$ is even
  and $\ForAllRel{}$ otherwise.
  Since the auxiliary \kl{relation variables} in~$\Vector{\SOVar}_{\Tag{aux}}$
  are now quantified universally,
  the \kl{subformula}~$\Formula[2]\Of{\FOVar[1]}$ becomes
  \begin{equation*}
    \Formula[2]\Of{\FOVar[1]} = \,
    \Formula[2]_{\Tag{aux}}\Of{\FOVar[1]}
    \IMP
    \Formula[2]_1\Of{\FOVar[1]}.
  \end{equation*}
  Similarly,
  the roles of even and odd indices are swapped
  in the definitions of
  $\Formula[2]_{\LevelIdx}\Of{\FOVar[1]}$,
  for $\LevelIdx \in \Range[1]{\Level + 1}$,
  to account for the fact
  that the quantifier alternation now starts
  with a block of \kl{universal quantifiers}.
\end{proof}

%%% Local Variables:
%%% mode: latex
%%% TeX-master: "../lph-paper"
%%% End:

\section{Hardness and completeness results}
\label{sec:reductions}

Reductions play a fundamental role in classical complexity theory,
allowing us to compare
the computational complexities of different problems,
even when we are unable to determine any of them absolutely.
Fraigniaud, Korman, and Peleg~\cite{DBLP:journals/jacm/FraigniaudKP13}
transferred this idea to the setting of distributed decision,
where they introduced the notion of \emph{local reductions}.
A local reduction
from a \kl(graph){property}~$\Property$
to a \kl(graph){property}~$\Property'$
is performed by a distributed algorithm
that modifies the \kl{labeling} of the input \kl{graph}
in \kl{constant round time}
such that
the original \kl{graph} has \kl(graph){property}~$\Property$
if and only if
the relabeled \kl{graph} has \kl(graph){property}~$\Property'$.
While this unconstrained definition turned out to be too strong
for the types of \kl{certificates} considered,
Balliu, D'Angelo, Fraigniaud, and Olivetti~\cite{DBLP:journals/jcss/BalliuDFO18}
later refined the concept
and used it to establish completeness results
for two classes of their \kl{identifier}-independent hierarchy.
However,
they also noted that the local reductions involved were
“very much time consuming at each node”
and left as an open problem
“whether non-trivial hardness results can be established
under polynomial-time local reductions”.

In this section,
we extend Karp's~\cite{DBLP:conf/coco/Karp72} notion
of polynomial-time reductions
to computer networks.
Our definition can also be seen as a further refinement
of the aforementioned local reductions,
where we
impose bounds on the \kl{step running time}
of the algorithms performing the reductions
and additionally require
them to work under \kl{locally unique} \kl{identifiers}.
As a presentation choice,
we consider \kl{graph} transformations
that are slightly more general than just relabelings.
These transformations may implicitly \kl{encode}
parts of the topology of the new \kl{graph}
in the output computed by the \kl{nodes} of the old \kl{graph},
thus allowing \kl{reductions} to more natural \kl{graph properties}.
We then establish several \kl{hardness} and \kl{completeness} results
for the two lowest levels of the \kl{locally polynomial hierarchy}.
We start with relatively simple \kl{reductions}
showing that
\kl{Eulerianness} is $\LP$-\kl{complete},
while \kl{Hamiltonicity} is both $\LP$-\kl{hard} and $\coLP$-\kl{hard}.
Then,
building on Theorem~\ref{thm:local-fagin},
we generalize the Cook--Levin theorem from~$\NP$ to $\NLP$,
which gives us a first $\NLP$-\kl{complete} \kl{graph property}.
Finally,
using standard techniques from complexity theory,
we use this \kl(graph){property} to establish
the $\NLP$-\kl{completeness} of $3$-\kl{colorability}.

In a way,
the framework of \kl{reductions}
proves even more beneficial in the distributed setting
than in the centralized setting.
Indeed,
since we will prove
the infiniteness of the \kl{locally polynomial hierarchy}
in Section~\ref{sec:hierarchy},
all the \kl{hardness} and \kl{completeness} results presented below
immediately yield \emph{unconditional} lower bounds
on the complexity of the \kl(graph){properties} in question,
i.e., lower bounds that do not rely on any complexity-theoretic assumptions
(see Corollaries~\ref{cor:three-colorable-not-in-lp},
\ref{cor:non-three-colorable-not-in-nlp}
and~\ref{cor:noneulerian-hamiltonian-nonhamiltonian-not-in-nlp}
in Section~\ref{ssec:warm-up}).

\subparagraph*{Clusters and implementable functions.}

Before we can define an appropriate notion of reduction,
we need to generalize the idea of computable functions
to our model of computation.
Intuitively,
the \kl{result} $\Result{\Machine}{\Graph, \IdMap}$
computed by a \kl{distributed Turing machine}~$\Machine$
on a \kl{graph}~$\Graph$
under an \kl{identifier assignment}~$\IdMap$
can be interpreted as the \kl{encoding} of a new \kl{graph}~$\Graph'$.
More precisely,
for each \kl{node}~$\Node[1]$ of the original \kl{graph}~$\Graph$,
the output \kl{label} computed by~$\Node[1]$
is taken to \kl{encode} a \kl{subgraph} of~$\Graph'$,
which we refer to as $\Node[1]$'s \kl{cluster}.
Since $\Graph'$ may depend on the \kl{identifiers} provided by~$\IdMap$,
and these are not considered part of the input,
we shall regard~$\Machine$ as \kl{implementing}
a “nondeterministic” function
$f \colon \GRAPH \to \PowerSet{\GRAPH}$
instead of a function
$f \colon \GRAPH \to \GRAPH$.

\AP
Formally,
a \intro{cluster map}
from a \kl{graph}~$\Graph'$ to a \kl{graph}~$\Graph$
is a function
$g \colon \NodeSet{\Graph'} \to \NodeSet{\Graph}$
such that
$\Set{\Node[1], \Node[2]} \in \EdgeSet{\Graph'}$
implies
$g(\Node[1]) = g(\Node[2])$
or
$\Set{g(\Node[1]),\, g(\Node[2])} \in \EdgeSet{\Graph}$.
With respect to~$g$,
the \intro{cluster} of any \kl{node}
$\Node[1] \inG \Graph$
is the \kl{induced subgraph} of~$\Graph'$
whose \kl{nodes} are mapped to~$\Node[1]$,
including the \kl{labels} of those \kl{nodes}.

\AP
Let $\IdentRadius \in \Positives$.
A function $f \colon \GRAPH \to \PowerSet{\GRAPH}$
is \intro{implemented} by a \kl{distributed Turing machine}~$\Machine$
under $\IdentRadius$-\kl{locally unique} \kl{identifiers}
if
for every $\Graph \in \GRAPH$
and every $\IdentRadius$-\kl{locally unique}
\kl{identifier assignment}~$\IdMap$ of~$\Graph$,
the \kl{result} $\Result{\Machine}{\Graph, \IdMap}$
computed by~$\Machine$ on~$\Graph$ under~$\IdMap$
represents a \kl{graph}~$\Graph' \in f(\Graph)$ in the following sense:
there is a \kl{cluster map}~$g$ from~$\Graph'$ to~$\Graph$
such that
the \kl{label}
$\Labeling{\Result{\Machine}{\Graph, \IdMap}}(\Node[1])$
computed at each \kl{node} $\Node[1] \inG \Graph$
\kl{encodes} $\Node[1]$'s \kl{cluster} with respect to~$g$
and all \kl{edges} to $\Node[1]$'s \kl{neighbors}' \kl{clusters}
(i.e., all \kl{edges} between \kl{nodes} of $\Node[1]$'s \kl{cluster}
and \kl{nodes} of $\Node[1]$'s \kl{neighbors}' \kl{clusters}).
Note that since the specific \kl{graph}~$\Graph' \in f(\Graph)$
computed by~$\Machine$ may depend on~$\IdMap$,
its \kl{labels} may refer to \kl{nodes} of~$\Graph$ by their \kl{identifiers}.

\AP
A function $f \colon \GRAPH \to \PowerSet{\GRAPH}$
is called \intro{topology-preserving}
if any two \kl{graphs} $\Graph \in \GRAPH$ and $\Graph' \in f(\Graph)$
are identical except for their \kl{labeling}, i.e.,
$\NodeSet{\Graph} = \NodeSet{\Graph'}$ and
$\EdgeSet{\Graph} = \EdgeSet{\Graph'}$.

\subparagraph*{Reductions, hardness, and completeness.}

Basically,
a \kl{locally polynomial reduction}
from a \kl(graph){property}~$\Property$
to a \kl(graph){property}~$\Property'$
is a \kl{graph} transformation,
\kl{implementable} by a \kl{locally polynomial machine},
that turns an input \kl{graph}~$\Graph$ into a new \kl{graph}~$\Graph'$
such that
$\Graph$~has \kl(graph){property}~$\Property$
if and only if
$\Graph'$~has \kl(graph){property}~$\Property'$.
The existence of such a \kl{reduction} implies that
$\Property'$~is at least as hard as~$\Property$,
since it allows us to convert
a hypothetical \kl{decider}~$\Machine'$ for~$\Property'$
into a \kl{decider}~$\Machine$ for~$\Property$,
which would work as follows:
First,
$\Machine$~would simulate a \kl{machine}~$\Machine_{\Tag{red}}$
that performs the \kl{reduction},
thereby transforming~$\Graph$ into~$\Graph'$.
Then it would simulate~$\Machine'$ on~$\Graph'$,
and finally each \kl{node} of~$\Graph$ would \kl{accept}
precisely if
all \kl{nodes} of its \kl{cluster} did so in the simulation.

\AP
More formally,
let $\BaseProperty, \Property, \Property' \subseteq \GRAPH$.
A \intro{locally polynomial reduction}
from~$\Property$ to~$\Property'$ on~$\BaseProperty$
is a function $f \colon \GRAPH \to \PowerSet{\GRAPH}$
\kl{implementable} by a \kl{locally polynomial machine}
such that
for all \kl{graphs} $\Graph \in \BaseProperty$ and $\Graph' \in f(\Graph)$,
we have
$\Graph \in \Property$
if and only if
$\Graph' \in \Property'$.
If such a function exists,
we denote this fact by
$\Property \intro*\ReducesTo{\BaseProperty} \Property'$.
Given a class~$\Class$ of \kl{graph properties},
we say that $\Property'$ is
$\Class$-\intro{hard} on~$\BaseProperty$
under \kl{locally polynomial reductions}
if $\Property \ReducesTo{\BaseProperty} \Property'$
for all
$\Property \in \Class$,
and we say that $\Property'$ is
$\Class$-\intro{complete} on~$\BaseProperty$
under \kl{locally polynomial reductions}
if additionally
$\Property' \in \Class$.
Since we will not consider other types of reductions,
we usually omit mentioning “under \kl{locally polynomial reductions}”,
and to avoid specifying~$\BaseProperty$ every time,
we stipulate that
“$\Class$-\kl{hard}” and “$\Class$-\kl{complete}”
imply $\BaseProperty = \GRAPH$ if
$\Class \in
 \Set{
   \SigmaLP{\Level}\!, \PiLP{\Level}\!,
   \coSigmaLP{\Level}\!, \coPiLP{\Level}
 }_{\Level \in \Naturals}$,
and $\BaseProperty = \NODE$ if
$\Class \in
 \Set{
   \SigmaP{\Level}, \PiP{\Level}
 }_{\Level \in \Naturals}$.

\begin{remark}
  \label{rem:from-nlp-complete-to-np-complete}
  If a \kl(graph){property}~$\Property$ is $\NLP$-\kl{hard},
  then it is also $\NP$-\kl{hard},
  since $\NODE \subseteq \GRAPH$ and $\NP = \NLP\On\NODE$.
  Moreover,
  if $\Property$~is $\NLP$-\kl{complete}
  under \kl{locally polynomial reductions}
  that are also \kl{topology-preserving},
  then $\Property \cap \NODE$ is $\NP$-\kl{complete}.
  This observation generalizes to all other levels
  of the \kl[locally polynomial hierarchy]{locally polynomial}
  and \kl{polynomial hierarchies}.
\end{remark}

In the remainder of this section,
we establish a series of \kl{hardness} and \kl{completeness} results
for the two lowest levels of the \kl{locally polynomial hierarchy}.
The first is trivial,
but will be useful below.
It concerns the \kl(graph){property} $\ALLSELECTED$
introduced in Section~\ref{ssec:example-formulas}.

\begin{remark}
  \label{rem:allselected-lp-complete}
  $\ALLSELECTED$ is $\LP$-\kl{complete}.
  This holds even if we impose that
  all \kl{locally polynomial reductions} must be \kl{topology-preserving}.
\end{remark}

\begin{claimproof}
  The \kl(graph){property} obviously lies in~$\LP$,
  and its $\LP$-\kl{hardness} under \kl{topology-preserving} \kl{reductions}
  is established by the basic observation that
  any \kl{graph property}~$\Property$
  \kl{decided} by a \kl{locally polynomial machine}~$\Machine$
  can be \kl{reduced} to $\ALLSELECTED$
  simply by \kl{executing}~$\Machine$.
\end{claimproof}

\AP
By considering \kl{reductions}
that do not necessarily \kl{preserve the topology} of the input \kl{graph},
we allow \kl{reductions} to more natural \kl{graph properties}.
We now illustrate this using $\intro*\EULERIAN$,
the \kl(graph){property} of \kl{graphs} that contain an \intro{Eulerian cycle},
i.e., a \kl{cycle} that uses each \kl{edge} exactly once.
The \kl{complement} of this \kl(graph){property}
will be denoted by~$\intro*\NONEULERIAN$.

\begin{proposition}
  \label{prp:eulerian-lp-complete}
  $\EULERIAN$ is $\LP$-\kl{complete}.
\end{proposition}

\begin{proof}
  By a famous theorem due to Euler
  (see, e.g., \cite[Thm.~1.8.1]{DBLP:books/daglib/0030488})
  a \kl{connected} \kl{graph} is \kl{Eulerian}
  if and only if
  all its \kl{nodes} have even \kl{degree}.
  This characterization makes it straightforward
  to \kl{decide} $\EULERIAN$ with a \kl{locally polynomial machine}.

  To show that the \kl(graph){property} is also $\LP$-\kl{hard},
  we now describe
  a \kl{locally polynomial reduction} to it
  from the \kl(graph){property} $\ALLSELECTED$,
  which is itself $\LP$-\kl{complete} by Remark~\ref{rem:allselected-lp-complete}.
  Given an arbitrary \kl{graph}~$\Graph$,
  we construct a \kl{graph}~$\Graph'$
  whose \kl{nodes} all have even \kl{degree}
  precisely if
  all \kl{nodes} of~$\Graph$ have \kl{label}~$1$.
  Let us assume without loss of generality
  that $\Graph$~has at least two \kl{nodes}.
  (\kl{Single-node graphs} can easily be treated as a special case.)
  For each \kl{node}~$\Node[1] \in \NodeSet{\Graph}$,
  the new \kl{graph}~$\Graph'$
  has two copies $\Node[1]_0$ and~$\Node[1]_1$,
  and for each \kl{edge}
  $\Set{\Node[1], \Node[2]} \in \EdgeSet{\Graph}$,
  it contains the four \kl{edges}
  $\Set{\Node[1]_i, \Node[2]_j}_{i, j \in \Set{0, 1}}$.
  In addition,
  for each \kl{node}~$\Node[1] \in \NodeSet{\Graph}$
  whose \kl{label} is not~$1$,
  the new \kl{graph} also contains the \kl{edge}
  $\Set{\Node[1]_0, \Node[1]_1}$.
  An example is shown in Figure~\ref{fig:allselected-to-eulerian}.
  Notice that
  $\Graph \in \ALLSELECTED$
  if and only if
  $\Graph' \in \EULERIAN$,
  that $\Graph'$ is always \kl{connected}
  (as required by our definition of \kl{graphs}),
  and that the \kl{nodes} of~$\Graph$ can compute~$\Graph'$
  in \kl{constant round time} and \kl{polynomial step time}.
\end{proof}

\begin{figure}[tb]
  \centering
  \begin{tikzpicture}[
    semithick,>=stealth',on grid,
    vertex/.style={draw,fill=white,circle,minimum size=1.8ex,inner sep=0},
    named vertex/.style={draw,circle,inner sep=0,minimum size=4ex},
    box/.style={draw,dotted},
    hamilton/.style={ultra thick},
    label/.style={inner sep=0,align=left},
  ]
  \def\unitDist{3ex}
  \def\boxMargin{3ex}
  % u's gadget
  \node[vertex] (u_top) {};
  \node[vertex] (u_bot) at ($(u_top)+(-90:1.5*\unitDist)$) {};
  % v's gadget
  \node[vertex] (v_top) at ($(u_top)+(0:8*\unitDist)$) {};
  \node[vertex] (v_bot) at ($(v_top)+(-90:1.5*\unitDist)$) {};
  % w's gadget
  \node[vertex] (w_top) at ($(v_top)+(0:8*\unitDist)$) {};
  \node[vertex] (w_bot) at ($(w_top)+(-90:1.5*\unitDist)$) {};
  \path (w_top) edge (w_bot);
  % Connections between the gadgets
  \path
    (u_top) edge (v_top)
    (u_top) edge (v_bot)
    (u_bot) edge (v_top)
    (u_bot) edge (v_bot)
    (v_top) edge (w_top)
    (v_top) edge (w_bot)
    (v_bot) edge (w_top)
    (v_bot) edge (w_bot);
  % Cluster boxes
  \node[box] (cluster_u) [fit={($(u_top)+(-\boxMargin,\boxMargin)$)
                               ($(u_bot)+(\boxMargin,-\boxMargin)$)},rounded corners=3.5ex] {};
  \node[box] (cluster_v) [fit={($(v_top)+(-\boxMargin,\boxMargin)$)
                               ($(v_bot)+(\boxMargin,-\boxMargin)$)},rounded corners=3.5ex] {};
  \node[box] (cluster_w) [fit={($(w_top)+(-\boxMargin,\boxMargin)$)
                               ($(w_bot)+(\boxMargin,-\boxMargin)$)},rounded corners=3.5ex] {};
  \node[label,anchor=north] at ($(cluster_u.south)+(0,-0.3*\unitDist)$)
       {\kl{cluster} of~$\Node[1]$};
  \node[label,anchor=north] at ($(cluster_v.south)+(0,-0.3*\unitDist)$)
       {\kl{cluster} of~$\Node[2]$};
  \node[label,anchor=north] at ($(cluster_w.south)+(0,-0.3*\unitDist)$)
       {\kl{cluster} of~$\Node[3]$};
  \node (graph') [label,above left=0.1*\unitDist and 0.2*\unitDist of cluster_u.north west,anchor=north east]
                 {$\Graph'$:};
  % Original graph
  \node[named vertex] (u) [above=1.0*\unitDist of cluster_u.north] {$\Node[1]$};
  \node[named vertex] (v) [above=1.0*\unitDist of cluster_v.north] {$\Node[2]$};
  \node[named vertex] (w) [above=1.0*\unitDist of cluster_w.north] {$\Node[3]$};
  \node[label,above=0.3*\unitDist of u.north] {$1$};
  \node[label,above=0.3*\unitDist of v.north] {$1$};
  \node[label,above=0.3*\unitDist of w.north] {$0$};
  \node[label] at (graph' |- u) {$\Graph$:};
  \path
    (u) edge (v)
    (v) edge (w);
\end{tikzpicture}

%%% Local Variables:
%%% mode: latex
%%% TeX-master: "../lph-paper"
%%% End:
  \caption{
    Example illustrating
    the \kl{reduction} from $\ALLSELECTED$ to~$\EULERIAN$
    used in the proof of Proposition~\ref{prp:eulerian-lp-complete}.
    The \kl{graph}~$\Graph$ has all \kl{node labels} equal to~$1$
    if and only if
    the \kl{graph}~$\Graph'$ has an \kl{Eulerian cycle}.
    In this particular case,
    if \kl{node}~$\Node[3]$ of~$\Graph$ had \kl{label}~$1$,
    then its \kl{cluster} in~$\Graph'$ would lack the “vertical” \kl{edge},
    and thus all \kl{nodes} of~$\Graph'$ would have even \kl{degree},
    making~$\Graph'$ \kl{Eulerian}.
  }
  \label{fig:allselected-to-eulerian}
\end{figure}

While it is easy to determine if a given \kl{graph} is \kl{Eulerian},
it is much harder to determine if it is \kl{Hamiltonian}.
Indeed,
a characterization of $\HAMILTONIAN$
similar to Euler's characterization of $\EULERIAN$
is neither known nor expected to exist
(see, e.g., \cite[Ch.\,10]{DBLP:books/daglib/0030488}).
The next two propositions show that this is reflected in
the complexity of $\HAMILTONIAN$ in our model of computation:
the \kl(graph){property} is both $\LP$-\kl{hard} and $\coLP$-\kl{hard},
which implies that it lies neither in $\NLP$ nor in $\coNLP$
(see Corollary~\ref{cor:noneulerian-hamiltonian-nonhamiltonian-not-in-nlp}).

\begin{figure}[htb]
  \centering
  \begin{tikzpicture}[
    semithick,>=stealth',on grid,
    vertex/.style={draw,fill=white,circle,minimum size=1.8ex,inner sep=0},
    named vertex/.style={draw,circle,inner sep=0,minimum size=4ex},
    box/.style={draw,dotted,rounded corners=2ex},
    hamilton/.style={ultra thick},
    label/.style={inner sep=0,align=left},
  ]
  \def\unitDist{3ex}
  \def\horizDist{22ex}
  \def\vertDist{18ex}
  \def\boxMargin{7ex}
  \def\cycleRad{5.5ex}
  % u1's gadget
  \coordinate (u1_mid);
  \draw[hamilton] ($(u1_mid)+(60:\cycleRad)$) arc (60:210:\cycleRad);
  \draw[hamilton] ($(u1_mid)+(30:\cycleRad)$) arc (30:-30:\cycleRad);
  \draw[hamilton] ($(u1_mid)+(-120:\cycleRad)$) arc (-120:-60:\cycleRad);
  \draw ($(u1_mid)+(30:\cycleRad)$) arc (30:60:\cycleRad);
  \draw ($(u1_mid)+(-30:\cycleRad)$) arc (-30:-60:\cycleRad);
  \draw ($(u1_mid)+(-150:\cycleRad)$) arc (-150:-120:\cycleRad);
  \node[vertex] (u1_east1) at ($(u1_mid)+(60:\cycleRad)$) {};
  \node[vertex] (u1_east2) at ($(u1_mid)+(30:\cycleRad)$) {};
  \node[vertex] (u1_southeast1) at ($(u1_mid)+(-60:\cycleRad)$) {};
  \node[vertex] (u1_southeast2) at ($(u1_mid)+(-30:\cycleRad)$) {};
  \node[vertex] (u1_south1) at ($(u1_mid)+(-150:\cycleRad)$) {};
  \node[vertex] (u1_south2) at ($(u1_mid)+(-120:\cycleRad)$) {};
  % u2's gadget
  \node[vertex] (u2_mid) at ($(u1_mid)+(0:\horizDist)$) {};
  \draw[hamilton] ($(u2_mid)+(120:\cycleRad)$) arc (120:-210:\cycleRad);
  \draw ($(u2_mid)+(120:\cycleRad)$) arc (120:150:\cycleRad);
  \node[vertex] (u2_west1) at ($(u2_mid)+(120:\cycleRad)$) {};
  \node[vertex] (u2_west2) at ($(u2_mid)+(150:\cycleRad)$) {};
  \node[vertex] (u2_south1) at ($(u2_mid)+(-60:\cycleRad)$) {};
  \node[vertex] (u2_south2) at ($(u2_mid)+(-30:\cycleRad)$) {};
  \path (u2_west1) edge (u2_mid);
  % u3's gadget
  \coordinate (u3_mid) at ($(u1_mid)+(-90:\vertDist)$);
  \draw[hamilton] ($(u3_mid)+(120:\cycleRad)$) arc (120:-210:\cycleRad);
  \draw ($(u3_mid)+(120:\cycleRad)$) arc (120:150:\cycleRad);
  \node[vertex] (u3_east1) at ($(u3_mid)+(-30:\cycleRad)$) {};
  \node[vertex] (u3_east2) at ($(u3_mid)+(-60:\cycleRad)$) {};
  \node[vertex] (u3_north1) at ($(u3_mid)+(150:\cycleRad)$) {};
  \node[vertex] (u3_north2) at ($(u3_mid)+(120:\cycleRad)$) {};
  % u4's gadget
  \coordinate (u4_mid) at ($(u2_mid)+(-90:\vertDist)$);
  \draw[hamilton] ($(u4_mid)+(120:\cycleRad)$) arc (120:-210:\cycleRad);
  \draw ($(u4_mid)+(120:\cycleRad)$) arc (120:150:\cycleRad);
  \node[vertex] (u4_west1) at ($(u4_mid)+(-150:\cycleRad)$) {};
  \node[vertex] (u4_west2) at ($(u4_mid)+(-120:\cycleRad)$) {};
  \node[vertex] (u4_northwest1) at ($(u4_mid)+(150:\cycleRad)$) {};
  \node[vertex] (u4_northwest2) at ($(u4_mid)+(120:\cycleRad)$) {};
  \node[vertex] (u4_north1) at ($(u4_mid)+(60:\cycleRad)$) {};
  \node[vertex] (u4_north2) at ($(u4_mid)+(30:\cycleRad)$) {};
  % Connections between the gadgets
  \path[hamilton]
    (u1_east1) edge (u2_west1)
    (u1_east2) edge (u2_west2)
    (u1_southeast1) edge (u4_northwest1)
    (u1_southeast2) edge (u4_northwest2)
    (u1_south1) edge (u3_north1)
    (u1_south2) edge (u3_north2);
  \path
    (u2_south1) edge (u4_north1)
    (u2_south2) edge (u4_north2)
    (u3_east1) edge (u4_west1)
    (u3_east2) edge (u4_west2);
  % Cluster boxes
  \node[box] (cluster_u1) [fit={($(u1_mid)+(-\boxMargin,\boxMargin)$)
                                ($(u1_mid)+(\boxMargin,-\boxMargin)$)},rounded corners=6ex] {};
  \node[box] (cluster_u2) [fit={($(u2_mid)+(-\boxMargin,\boxMargin)$)
                                ($(u2_mid)+(\boxMargin,-\boxMargin)$)},rounded corners=6ex] {};
  \node[box] (cluster_u3) [fit={($(u3_mid)+(-\boxMargin,\boxMargin)$)
                                ($(u3_mid)+(\boxMargin,-\boxMargin)$)},rounded corners=6ex] {};
  \node[box] (cluster_u4) [fit={($(u4_mid)+(-\boxMargin,\boxMargin)$)
                                ($(u4_mid)+(\boxMargin,-\boxMargin)$)},rounded corners=6ex] {};
  \node[label,anchor=south] at ($(cluster_u1.north)+(0,0.3*\unitDist)$)
       {\kl{cluster} of~$\Node[1]_1$};
  \node[label,anchor=south] at ($(cluster_u2.north)+(0,0.3*\unitDist)$)
       {\kl{cluster} of~$\Node[1]_2$};
  \node[label,anchor=north] at ($(cluster_u3.south)+(0,-0.3*\unitDist)$)
       {\kl{cluster} of~$\Node[1]_3$};
  \node[label,anchor=north] at ($(cluster_u4.south)+(0,-0.3*\unitDist)$)
       {\kl{cluster} of~$\Node[1]_4$};
  \node (graph') [label,above left=0*\unitDist and 0*\unitDist of cluster_u1.north west,anchor=north east]
                 {$\Graph'$:};
  % Original graph
  \node[named vertex] (u1) [left=14*\unitDist of cluster_u1] {$\Node[1]_1$};
  \node[named vertex] (u2) at ($(u1)+(0:\horizDist)$)  {$\Node[1]_2$};
  \node[named vertex] (u3) at ($(u1)+(-90:\vertDist)$) {$\Node[1]_3$};
  \node[named vertex] (u4) at ($(u2)+(-90:\vertDist)$) {$\Node[1]_4$};
  \node[label,above=0.3*\unitDist of u1.north] {$1$};
  \node[label,above=0.3*\unitDist of u2.north] {$0$};
  \node[label,below=0.3*\unitDist of u3.south] {$1$};
  \node[label,below=0.3*\unitDist of u4.south] {$1$};
  \node[label] at ($(u1 |- graph')+(-1.4*\unitDist,0)$) {$\Graph$:};
  \path[hamilton]
    (u1) edge (u2)
         edge (u3)
         edge (u4);
  \path
    (u2) edge (u4)
    (u3) edge (u4);
\end{tikzpicture}

%%% Local Variables:
%%% mode: latex
%%% TeX-master: "../lph-paper"
%%% End:
  \caption{
    \emph{(repeated from Figure~\ref{fig:allselected-to-hamiltonian-overview})}
    Example illustrating
    the \kl{reduction} from $\ALLSELECTED$ to~$\HAMILTONIAN$
    used in the proof of Proposition~\ref{prp:hamiltonian-lp-hard}.
    The \kl{graph}~$\Graph$ has all \kl{node labels} equal to~$1$
    if and only if
    the \kl{graph}~$\Graph'$ has a \kl{Hamiltonian cycle}.
    The thick \kl{edges} in~$\Graph$ form a \kl{spanning tree},
    which is replicated by the thick \kl{edges} in~$\Graph'$.
    In this particular case,
    if \kl{node}~$\Node[1]_2$ of~$\Graph$ had \kl{label}~$1$,
    then its \kl{cluster} in~$\Graph'$ would lack the “central” \kl{node},
    and thus the thick \kl{edges} in~$\Graph'$
    would form a \kl{Hamiltonian cycle}.
  }
  \label{fig:allselected-to-hamiltonian}
\end{figure}

\begin{proposition}
  \label{prp:hamiltonian-lp-hard}
  $\HAMILTONIAN$ is $\LP$-\kl{hard}.
\end{proposition}

\begin{proof}
  By Remark~\ref{rem:allselected-lp-complete},
  it is sufficient to provide a \kl{locally polynomial reduction}
  from $\ALLSELECTED$ to $\HAMILTONIAN$.
  Given an arbitrary \kl{graph}~$\Graph$,
  we construct a \kl{graph}~$\Graph'$
  that has a \kl{Hamiltonian cycle}
  if and only if
  all \kl{nodes} of~$\Graph$ have \kl{label}~$1$.
  The main idea is that
  a \kl{Hamiltonian cycle} in~$\Graph'$
  will represent
  a depth-first traversal of a \kl{spanning tree} of~$\Graph$,
  using a method known as the Euler tour technique.
  For this purpose,
  each \kl{edge} of~$\Graph$ is represented by two \kl{edges} in~$\Graph'$,
  so that it can be traversed twice by a \kl{Hamiltonian cycle} in~$\Graph'$.
  More precisely,
  each \kl{node} $\Node[1] \inG \Graph$ of \kl{degree}~$\Degree$
  with \kl{neighbors} $\Node[2]_1, \dots, \Node[2]_{\Degree}$
  is represented in~$\Graph'$
  by a \kl{cycle} of length $\max \Set{3, 2\Degree}$
  of the form
  $\Node[1]_{\To\Node[2]_1},
   \Node[1]_{\From\Node[2]_1},
   \dots,
   \Node[1]_{\To\Node[2]_{\Degree}},
   \Node[1]_{\From\Node[2]_{\Degree}},
   \Node[1]_{\To\Node[2]_1}$.
  For each \kl{neighbor}~$\Node[2]_i$ of~$\Node[1]$,
  this \kl{cycle} contains two \kl{adjacent} \kl{nodes}
  $\Node[1]_{\To\Node[2]_i}$
  and
  $\Node[1]_{\From\Node[2]_i}$,
  which can be thought of as the “ports”
  that allow us to “go to” and “come from”~$\Node[2]_i$,
  respectively.
  (To ensure that there are enough \kl{nodes} to form a \kl{cycle},
  we add three dummy \kl{nodes} if $\Degree = 0$,
  and one dummy \kl{node} if $\Degree = 1$.)
  If $\Node[1]$'s \kl{label} differs from~$1$,
  then $\Graph'$ additionally contains a \kl{node}~$\Node[1]_{\Tag{bad}}$
  that is connected to
  exactly one \kl{node} of the \kl{cycle} representing~$\Node[1]$.
  Moreover,
  for each \kl{edge}
  $\Set{\Node[1], \Node[2]}$ of~$\Graph$,
  the \kl{graph}~$\Graph'$ contains the two \kl{edges}
  $\Set{\Node[1]_{\To\Node[2]}, \Node[2]_{\From\Node[1]}}$
  and
  $\Set{\Node[1]_{\From\Node[2]}, \Node[2]_{\To\Node[1]}}$.
  An example is shown in Figure~\ref{fig:allselected-to-hamiltonian}.

  Now,
  if all \kl{nodes} of~$\Graph$ are \kl{labeled} with~$1$,
  then any \kl{spanning tree} of~$\Graph$
  yields a \kl{Hamiltonian cycle} of~$\Graph'$.
  This \kl{cycle} includes
  all \kl{edges} of~$\Graph'$ of the form
  $\Set{\Node[1]_{\From\Node[2]}, \Node[1]_{\To\Node[3]}}$
  with $\Node[2] \neq \Node[3]$,
  and additionally,
  for each \kl{edge}
  $\Set{\Node[1], \Node[2]}$ of~$\Graph$,
  either the \kl{edges}
  $\Set{\Node[1]_{\To\Node[2]}, \Node[2]_{\From\Node[1]}}$
  and
  $\Set{\Node[1]_{\From\Node[2]}, \Node[2]_{\To\Node[1]}}$
  if $\Set{\Node[1], \Node[2]}$ belongs to the \kl{spanning tree},
  and otherwise the \kl{edges}
  $\Set{\Node[1]_{\To\Node[2]}, \Node[1]_{\From\Node[2]}}$
  and
  $\Set{\Node[2]_{\To\Node[1]}, \Node[2]_{\From\Node[1]}}$
  (see~Figure~\ref{fig:allselected-to-hamiltonian}).
  However,
  if at least one \kl{node}~$\Node[1]$ of~$\Graph$
  has a \kl{label} different from~$1$,
  then $\Graph'$ is not \kl{Hamiltonian},
  because the \kl{node}~$\Node[1]_{\Tag{bad}}$,
  which has \kl{degree}~$1$,
  cannot be part of any \kl{cycle}.
  Hence,
  $\Graph \in \ALLSELECTED$
  if and only if
  $\Graph' \in \HAMILTONIAN$.

  Notice that
  $\Graph'$ is guaranteed to be \kl{connected}
  (as required by our definition of \kl{graphs}),
  and that the \kl{nodes} of~$\Graph$ can compute~$\Graph'$
  in \kl{constant round time} and \kl{polynomial step time}.
  As an aside,
  note that we could reduce the number of \kl{nodes} of~$\Graph'$
  by a factor of two
  by contracting \kl{edges} of the form
  $\Set{\Node[1]_{\From\Node[2]}, \Node[1]_{\To\Node[3]}}$
  with $\Node[2] \neq \Node[3]$,
  but this would make the graphical representation somewhat less intuitive.
\end{proof}

\begin{proposition}
  \label{prp:hamiltonian-colp-hard}
  $\HAMILTONIAN$ is $\coLP$-\kl{hard}.
\end{proposition}

\begin{proof}
  By Remark~\ref{rem:allselected-lp-complete} and duality,
  $\NOTALLSELECTED$ is $\coLP$-\kl{complete},
  so it suffices to provide a \kl{locally polynomial reduction}
  from $\NOTALLSELECTED$ to $\HAMILTONIAN$.
  In essence,
  given any \kl{graph}~$\Graph$,
  we use the construction
  from the proof of Proposition~\ref{prp:hamiltonian-lp-hard}
  twice to create
  two \kl{subgraphs}~$\Graph'_{\Tag{top}}$ and~$\Graph'_{\Tag{bot}}$,
  and then connect them in such a way that
  the resulting \kl{graph}~$\Graph'$ has a \kl{Hamiltonian cycle}
  if and only if
  $\Graph$ has at least one unselected \kl{node}
  (i.e., a \kl{node} whose \kl{label} is not~$1$).
  By construction,
  $\Graph'_{\Tag{top}}$ and~$\Graph'_{\Tag{bot}}$
  each admit a \kl{Hamiltonian cycle},
  and the presence of an unselected \kl{node} in~$\Graph$
  will ensure the presence of two \kl{edges} in~$\Graph'$
  by which the two \kl{cycles} can be connected
  to form a \kl{Hamiltonian cycle} of~$\Graph'$.
  An example is provided in Figure~\ref{fig:notallselected-to-hamiltonian}.

  More formally,
  each \kl{node} $\Node[1] \inG \Graph$ of \kl{degree}~$\Degree$
  is represented in~$\Graph'$
  by two \kl{cycles} of length $(2\Degree + 3)$,
  which we will call the “top” and “bottom” \kl{cycles}.
  As in the proof of Proposition~\ref{prp:hamiltonian-lp-hard},
  for each \kl{neighbor}~$\Node[2]$ of~$\Node[1]$,
  each of the two \kl{cycles} contains two \kl{adjacent} \kl{nodes},
  which can be thought of as the “ports”
  that allow us to “go to” and “come from”
  the corresponding \kl{cycle} of~$\Node[2]$.
  That is,
  there are two adjacent \kl{nodes}
  $\Node[1]_{\To\Node[2]}^{\Tag{top}}$
  and
  $\Node[1]_{\From\Node[2]}^{\Tag{top}}$
  in the “top” \kl{cycle} of~$\Node[1]$,
  which are connected to the corresponding \kl{nodes}
  $\Node[2]_{\To\Node[1]}^{\Tag{top}}$
  and
  $\Node[2]_{\From\Node[1]}^{\Tag{top}}$
  of the “top” \kl{cycle} of~$\Node[2]$
  by means of the two \kl{edges}
  $\Set{\Node[1]_{\To\Node[2]}^{\Tag{top}}, \Node[2]_{\From\Node[1]}^{\Tag{top}}}$
  and
  $\Set{\Node[1]_{\From\Node[2]}^{\Tag{top}}, \Node[2]_{\To\Node[1]}^{\Tag{top}}}$.
  The “bottom” \kl{cycles} are connected by analogous \kl{nodes} and \kl{edges}.
  Moreover,
  to ensure that $\Graph'$ is \kl{connected}
  (as~required by our definition of \kl{graphs}),
  each \kl{cycle} contains a sequence of three additional \kl{nodes},
  named
  $\Node[1]_{\Down 1}^{\Tag{top}}$,
  $\Node[1]_{\Down 2}^{\Tag{top}}$,
  $\Node[1]_{\Down 3}^{\Tag{top}}$
  in the “top” \kl{cycle},
  and
  $\Node[1]_{\Up 1}^{\Tag{bot}}$,
  $\Node[1]_{\Up 2}^{\Tag{bot}}$,
  $\Node[1]_{\Up 3}^{\Tag{bot}}$
  in the “bottom” \kl{cycle}.
  For each \kl{node}~$\Node[1]$ of the original \kl{graph}~$\Graph$,
  the \kl{graph}~$\Graph'$ contains at least the \kl{edge}
  $\Set{\Node[1]_{\Down 2}^{\Tag{top}}, \Node[1]_{\Up 2}^{\Tag{bot}}}$.
  In addition,
  if $\Node[1]$~has a \kl{label} other than~$1$,
  then $\Graph'$ also contains the \kl{edge}
  $\Set{\Node[1]_{\Down 1}^{\Tag{top}}, \Node[1]_{\Up 1}^{\Tag{bot}}}$.

  Clearly,
  the new \kl{graph}~$\Graph'$ can be computed
  in \kl{constant round time} and \kl{polynomial step time}
  by the \kl{nodes} of~$\Graph$.
  To show the correctness of the construction,
  let $\Graph'_{\Tag{top}}$ and~$\Graph'_{\Tag{bot}}$
  be the \kl{induced subgraphs} of~$\Graph'$
  that contain all “top” \kl{cycles} and all “bottom” \kl{cycles},
  respectively.
  By the same argument as in
  the proof of Proposition~\ref{prp:hamiltonian-lp-hard},
  $\Graph'_{\Tag{top}}$ and~$\Graph'_{\Tag{bot}}$
  each have a \kl{Hamiltonian cycle},
  say $\HamCycle_{\Tag{top}}$ and~$\HamCycle_{\Tag{bot}}$.
  Now,
  if some \kl{node}~$\Node[1]$ of~$\Graph$ has a \kl{label} different from~$1$,
  then we can connect $\HamCycle_{\Tag{top}}$ and~$\HamCycle_{\Tag{bot}}$
  to form a \kl{Hamiltonian cycle} of~$\Graph'$.
  This can be achieved by adding the \kl{edges}
  $\Set{\Node[1]_{\Down 1}^{\Tag{top}}, \Node[1]_{\Up 1}^{\Tag{bot}}}$
  and
  $\Set{\Node[1]_{\Down 2}^{\Tag{top}}, \Node[1]_{\Up 2}^{\Tag{bot}}}$,
  and removing the \kl{edges}
  $\Set{\Node[1]_{\Down 1}^{\Tag{top}}, \Node[1]_{\Down 2}^{\Tag{top}}}$
  and
  $\Set{\Node[1]_{\Up 1}^{\Tag{bot}}, \Node[1]_{\Up 2}^{\Tag{bot}}}$
  (see~Figure~\ref{fig:notallselected-to-hamiltonian}).
  However,
  if all \kl{nodes} of~$\Graph$ are \kl{labeled} with~$1$,
  then~$\Graph'$ does not have a \kl{Hamiltonian cycle}.
  To see why,
  observe that in this case,
  all \kl{nodes} of the form
  $\Node[1]_{\Down 1}^{\Tag{top}}$,
  $\Node[1]_{\Down 3}^{\Tag{top}}$,
  $\Node[1]_{\Up 1}^{\Tag{bot}}$, or
  $\Node[1]_{\Up 3}^{\Tag{bot}}$
  have \kl{degree}~$2$.
  This implies that all \kl{edges} \kl{incident} to these \kl{nodes}
  must belong to any hypothetical \kl{Hamiltonian cycle}~$\HamCycle$ of~$\Graph'$,
  and hence that none of the \kl{edges} of the form
  $\Set{\Node[1]_{\Down 2}^{\Tag{top}}, \Node[1]_{\Up 2}^{\Tag{bot}}}$
  can be part of~$\HamCycle$.
  Since there are no other \kl{edges} connecting
  $\Graph'_{\Tag{top}}$ and~$\Graph'_{\Tag{bot}}$,
  the \kl{cycle}~$\HamCycle$ cannot exist.
\end{proof}

\begin{figure}[tb]
  \centering
  \begin{tikzpicture}[
    semithick,>=stealth',on grid,
    vertex/.style={draw,fill=white,circle,minimum size=1.8ex,inner sep=0},
    named vertex/.style={draw,circle,inner sep=0,minimum size=4ex},
    box/.style={draw,dotted,rounded corners=2ex},
    hamilton/.style={ultra thick},
    label/.style={inner sep=0,align=left},
  ]
  \def\unitDist{3ex}
  \def\boxMargin{7.1ex}
  \def\cycleRad{5.5ex}
  % u's gadget
  \coordinate (u_top);
  \draw[hamilton] ($(u_top)+(30:\cycleRad)$) arc (30:360:\cycleRad);
  \draw ($(u_top)+(30:\cycleRad)$) arc (30:0:\cycleRad);
  \node[vertex] (u_top_east1) at ($(u_top)+(30:\cycleRad)$) {};
  \node[vertex] (u_top_east2) at ($(u_top)+(0:\cycleRad)$) {};
  \node[vertex] (u_top_south1) at ($(u_top)+(-120:\cycleRad)$) {};
  \node[vertex] (u_top_south2) at ($(u_top)+(-90:\cycleRad)$) {};
  \node[vertex] (u_top_south3) at ($(u_top)+(-60:\cycleRad)$) {};
  \coordinate (u_bot) at ($(u_top)+(-90:5*\unitDist)$);
  \draw[hamilton] ($(u_bot)+(-30:\cycleRad)$) arc (-30:-360:\cycleRad);
  \draw ($(u_bot)+(-30:\cycleRad)$) arc (-30:0:\cycleRad);
  \node[vertex] (u_bot_east1) at ($(u_bot)+(0:\cycleRad)$) {};
  \node[vertex] (u_bot_east2) at ($(u_bot)+(-30:\cycleRad)$) {};
  \node[vertex] (u_bot_north1) at ($(u_bot)+(120:\cycleRad)$) {};
  \node[vertex] (u_bot_north2) at ($(u_bot)+(90:\cycleRad)$) {};
  \node[vertex] (u_bot_north3) at ($(u_bot)+(60:\cycleRad)$) {};
  \path (u_top_south2) edge (u_bot_north2);
  % v's gadget
  \coordinate (v_top) at ($(u_top)+(0:10*\unitDist)$);
  \draw[hamilton] ($(v_top)+(30:\cycleRad)$) arc (30:150:\cycleRad);
  \draw[hamilton] ($(v_top)+(0:\cycleRad)$) arc (0:-180:\cycleRad);
  \draw ($(v_top)+(30:\cycleRad)$) arc (30:0:\cycleRad);
  \draw ($(v_top)+(150:\cycleRad)$) arc (150:180:\cycleRad);
  \node[vertex] (v_top_west1) at ($(v_top)+(150:\cycleRad)$) {};
  \node[vertex] (v_top_west2) at ($(v_top)+(180:\cycleRad)$) {};
  \node[vertex] (v_top_east1) at ($(v_top)+(30:\cycleRad)$) {};
  \node[vertex] (v_top_east2) at ($(v_top)+(0:\cycleRad)$) {};
  \node[vertex] (v_top_south1) at ($(v_top)+(-120:\cycleRad)$) {};
  \node[vertex] (v_top_south2) at ($(v_top)+(-90:\cycleRad)$) {};
  \node[vertex] (v_top_south3) at ($(v_top)+(-60:\cycleRad)$) {};
  \coordinate (v_bot) at ($(v_top)+(-90:5*\unitDist)$);
  \draw[hamilton] ($(v_bot)+(0:\cycleRad)$) arc (0:180:\cycleRad);
  \draw[hamilton] ($(v_bot)+(-30:\cycleRad)$) arc (-30:-150:\cycleRad);
  \draw ($(v_bot)+(0:\cycleRad)$) arc (0:-30:\cycleRad);
  \draw ($(v_bot)+(-180:\cycleRad)$) arc (-180:-150:\cycleRad);
  \node[vertex] (v_bot_west1) at ($(v_bot)+(-180:\cycleRad)$) {};
  \node[vertex] (v_bot_west2) at ($(v_bot)+(-150:\cycleRad)$) {};
  \node[vertex] (v_bot_east1) at ($(v_bot)+(0:\cycleRad)$) {};
  \node[vertex] (v_bot_east2) at ($(v_bot)+(-30:\cycleRad)$) {};
  \node[vertex] (v_bot_north1) at ($(v_bot)+(120:\cycleRad)$) {};
  \node[vertex] (v_bot_north2) at ($(v_bot)+(90:\cycleRad)$) {};
  \node[vertex] (v_bot_north3) at ($(v_bot)+(60:\cycleRad)$) {};
  \path (v_top_south2) edge (v_bot_north2);
  % w's gadget
  \coordinate (w_top) at ($(v_top)+(0:10*\unitDist)$);
  \draw[hamilton] ($(w_top)+(150:\cycleRad)$) arc (150:-90:\cycleRad);
  \draw[hamilton] ($(w_top)+(-180:\cycleRad)$) arc (-180:-120:\cycleRad);
  \draw ($(w_top)+(150:\cycleRad)$) arc (150:180:\cycleRad);
  \draw ($(w_top)+(-120:\cycleRad)$) arc (-120:-90:\cycleRad);
  \node[vertex] (w_top_west1) at ($(w_top)+(150:\cycleRad)$) {};
  \node[vertex] (w_top_west2) at ($(w_top)+(180:\cycleRad)$) {};
  \node[vertex] (w_top_south1) at ($(w_top)+(-120:\cycleRad)$) {};
  \node[vertex] (w_top_south2) at ($(w_top)+(-90:\cycleRad)$) {};
  \node[vertex] (w_top_south3) at ($(w_top)+(-60:\cycleRad)$) {};
  \coordinate (w_bot) at ($(w_top)+(-90:5*\unitDist)$);
  \draw[hamilton] ($(w_bot)+(90:\cycleRad)$) arc (90:-150:\cycleRad);
  \draw[hamilton] ($(w_bot)+(120:\cycleRad)$) arc (120:180:\cycleRad);
  \draw ($(w_bot)+(-180:\cycleRad)$) arc (-180:-150:\cycleRad);
  \draw ($(w_bot)+(120:\cycleRad)$) arc (120:90:\cycleRad);
  \node[vertex] (w_bot_west1) at ($(w_bot)+(-180:\cycleRad)$) {};
  \node[vertex] (w_bot_west2) at ($(w_bot)+(-150:\cycleRad)$) {};
  \node[vertex] (w_bot_north1) at ($(w_bot)+(120:\cycleRad)$) {};
  \node[vertex] (w_bot_north2) at ($(w_bot)+(90:\cycleRad)$) {};
  \node[vertex] (w_bot_north3) at ($(w_bot)+(60:\cycleRad)$) {};
  \path[hamilton]
    (w_top_south1) edge (w_bot_north1)
    (w_top_south2) edge (w_bot_north2);
  % Connections between the gadgets
  \path[hamilton]
    (u_top_east1) edge (v_top_west1)
    (u_top_east2) edge (v_top_west2)
    (u_bot_east1) edge (v_bot_west1)
    (u_bot_east2) edge (v_bot_west2)
    (v_top_east1) edge (w_top_west1)
    (v_top_east2) edge (w_top_west2)
    (v_bot_east1) edge (w_bot_west1)
    (v_bot_east2) edge (w_bot_west2);
  % Cluster boxes
  \node[box] (cluster_u) [fit={($(u_top)+(-\boxMargin,\boxMargin)$)
                               ($(u_bot)+(\boxMargin,-\boxMargin)$)},rounded corners=6ex] {};
  \node[box] (cluster_v) [fit={($(v_top)+(-\boxMargin,\boxMargin)$)
                               ($(v_bot)+(\boxMargin,-\boxMargin)$)},rounded corners=6ex] {};
  \node[box] (cluster_w) [fit={($(w_top)+(-\boxMargin,\boxMargin)$)
                               ($(w_bot)+(\boxMargin,-\boxMargin)$)},rounded corners=6ex] {};
  \node[label,anchor=north] at ($(cluster_u.south)+(0,-0.3*\unitDist)$)
       {\kl{cluster} of~$\Node[1]$};
  \node[label,anchor=north] at ($(cluster_v.south)+(0,-0.3*\unitDist)$)
       {\kl{cluster} of~$\Node[2]$};
  \node[label,anchor=north] at ($(cluster_w.south)+(0,-0.3*\unitDist)$)
       {\kl{cluster} of~$\Node[3]$};
  \node (graph') [label,above left=0.1*\unitDist and -0.1*\unitDist of cluster_u.north west,anchor=north east]
                 {$\Graph'$:};
  % Original graph
  \node[named vertex] (u) [above=1.0*\unitDist of cluster_u.north] {$\Node[1]$};
  \node[named vertex] (v) [above=1.0*\unitDist of cluster_v.north] {$\Node[2]$};
  \node[named vertex] (w) [above=1.0*\unitDist of cluster_w.north] {$\Node[3]$};
  \node[label,above=0.3*\unitDist of u.north] {$1$};
  \node[label,above=0.3*\unitDist of v.north] {$1$};
  \node[label,above=0.3*\unitDist of w.north] {$0$};
  \node[label] at (graph' |- u) {$\Graph$:};
  \path
    (u) edge (v)
    (v) edge (w);
\end{tikzpicture}

%%% Local Variables:
%%% mode: latex
%%% TeX-master: "../lph-paper"
%%% End:
  \caption{
    Example illustrating
    the \kl{reduction} from $\NOTALLSELECTED$ to~$\HAMILTONIAN$
    used in the proof of Proposition~\ref{prp:hamiltonian-colp-hard}.
    The \kl{graph}~$\Graph$ has
    at least one \kl{node} with a \kl{label} different from~$1$
    if and only if
    the \kl{graph}~$\Graph'$ has a \kl{Hamiltonian cycle}.
    In this particular case,
    since \kl{node}~$\Node[3]$ of~$\Graph$ has \kl{label}~$0$,
    there is indeed a \kl{Hamiltonian cycle} in~$\Graph'$,
    represented by the thick \kl{edges}.
    This \kl{Hamiltonian cycle} can be thought of as
    the result of connecting the two “horizontally stretched” \kl{cycles}
    (the “top” one and the “bottom” one),
    using the two “vertical” \kl{edges} of $\Node[3]$'s \kl{cluster}.
    If $\Node[3]$~had \kl{label}~$1$,
    then one of these \kl{edges} would be missing,
    making it impossible to connect the two \kl{cycles}.
  }
  \label{fig:notallselected-to-hamiltonian}
\end{figure}

We now climb up one level in the \kl{locally polynomial hierarchy}
and investigate the notion of $\NLP$-\kl{completeness}.
Our treatment of centralized computing
as a special case of distributed computing
is particularly helpful here,
as it allows us to build directly on classical results from complexity theory.
We begin by recalling the Cook--Levin theorem,
which concerns the problem of determining
whether a given \kl{Boolean formula} is \kl(graph){satisfiable}.

\begin{theorem}[Cook and Levin~\cite{DBLP:conf/stoc/Cook71,Levin73}]
  \label{thm:cook-levin}
  $\SAT$ is $\NP$-\kl{complete}.
\end{theorem}

While this result was discovered a few years before Fagin's theorem,
it can also be obtained as a simple corollary of the latter
(see, e.g., \cite[Thm.~3.2.6]{DBLP:series/txtcs/GradelKLMSVVW07}).
In the following,
we will show that this observation extends to the distributed setting
by using Theorem~\ref{thm:local-fagin}
to obtain a generalized version of the Cook--Levin theorem.
But first,
we need to generalize the \kl{Boolean satisfiability} problem to \kl{graphs}.

\subparagraph*{Boolean graph satisfiability.}

\AP
\phantomintro{Boolean formula}%
\phantomintro{Boolean variable}%
\phantomintro{valuation}%
A \intro{Boolean graph} is a \kl{graph}~$\Graph$
whose \kl{nodes} are \kl{labeled} with (\kl{encodings} of) \kl{Boolean formulas}.
We call $\Graph$ \intro(graph){satisfiable}
if there exists a function~$\ValuationAssignment$
that assigns to each \kl{node} $\Node[1] \inG \Graph$
a \kl{valuation} of the \kl{Boolean variables} occurring
in $\Node[1]$'s \kl(Boolean){formula} $\Labeling{\Graph}(\Node[1])$
such that
\begin{itemize}
\item $\ValuationAssignment(\Node[1])$ satisfies
  $\Labeling{\Graph}(\Node[1])$, and
\item $\ValuationAssignment(\Node[1])$ is consistent
  with the \kl{valuations} of $\Node[1]$'s \kl{neighbors}, i.e.,
  $\ValuationAssignment(\Node[1])(\BoolVar) =
   \ValuationAssignment(\Node[2])(\BoolVar)$
  for every \kl{neighbor}~$\Node[2]$ of $\Node[1]$
  and every \kl{Boolean variable}~$\BoolVar$
  that occurs in both $\Labeling{\Graph}(\Node[1])$
  and~$\Labeling{\Graph}(\Node[2])$.
\end{itemize}

\AP
We denote the set of all \kl(graph){satisfiable} \kl{Boolean graphs}
by~$\intro*\SATGRAPH$.
The standard \intro{Boolean satisfiability} problem~$\SAT$
is simply the restriction of~$\SATGRAPH$ to \kl{single-node graphs},
i.e., $\intro*\SAT = \SATGRAPH \cap \NODE$.

Now we are ready to generalize the Cook--Levin theorem from~$\NP$ to $\NLP$.
The original result (Theorem~\ref{thm:cook-levin}) can be recovered
by restricting the following statement to \kl{single-node graphs},
as noted in Remark~\ref{rem:from-nlp-complete-to-np-complete}.

\begin{theorem}
  \label{thm:local-cook-levin}
  $\SATGRAPH$ is $\NLP$-\kl{complete}.
  This holds even if we impose that
  all \kl{locally polynomial reductions} must be \kl{topology-preserving}.
\end{theorem}

\begin{proof}
  Obviously,
  $\SATGRAPH$ lies in $\NLP$,
  since each \kl{node} can check
  in one \kl{communication round} and \kl{polynomial step time}
  whether a given \kl{valuation} is both
  locally satisfying and
  consistent with the \kl{valuations} of its \kl{neighbors}.

  It remains to show that
  $\SATGRAPH$ is $\NLP$-\kl{hard},
  i.e., that
  $\Property \ReducesTo{\GRAPH} \SATGRAPH$
  for every
  $\Property \in \NLP$,
  and to ensure that
  the \kl{locally polynomial reductions} involved are \kl{topology-preserving}.
  By Theorem~\ref{thm:local-fagin},
  we know that $\Property$ can be \kl{defined} by a $\SigmaLFOL{1}$-\kl{formula}
  $\ExistsRel{\SOVar_1} \dots \ExistsRel{\SOVar_n} \,
   \ForAll{\FOVar[1]} \, \Formula[1]\Of{\FOVar[1]}$.
  For each \kl{graph} $\Graph \in \GRAPH$,
  we now construct a \kl{Boolean graph}~$\Graph'$
  such that
  $\Graph \in \Property$
  if and only if
  $\Graph' \in \SATGRAPH$.
  The \kl{nodes} and \kl{edges} of~$\Graph'$
  are the same as those of~$\Graph$,
  and the \kl{labeling} function $\Labeling{\Graph'}$
  assigns a \kl{Boolean formula}~$\Formula[1]^{\Graph}_{\Node[1]}$
  to each \kl{node} $\Node[1] \inG \Graph$.
  This \kl(Boolean){formula}~$\Formula[1]^{\Graph}_{\Node[1]}$
  states that,
  for a given \kl{interpretation} of $\SOVar_1$, \dots, $\SOVar_n$,
  the $\BFL$-\kl{formula}~$\Formula[1]$ is satisfied
  at the \kl{element} of~$\StructRepr{\Graph}$ representing~$\Node[1]$
  and at all the \kl{elements} representing $\Node[1]$'s \kl{labeling bits}.
  To represent the relations $\SOVar_1$, \dots, $\SOVar_n$,
  we introduce \kl{Boolean variables} of the form
  $\BoolVar_{\SOVar(\Element_1, \dots, \Element_{\Arity})}$,
  with the intended meaning that the tuple of \kl{elements}
  $\Tuple{\Element_1, \dots, \Element_{\Arity}}
   \in (\Domain{\StructRepr{\Graph}})^{\Arity}$
  lies in the relation~$\SOVar$.

  Formally,
  we set
  \begin{equation*}
    \Formula[1]^{\Graph}_{\Node[1]} \,=\,
    \overbrace{
      \Translation_{\Set{\FOVar[1] \mapsto \Node[1]}}(\Formula[1])
      \vphantom{\BigAND}
    }^{\text{\kl{node}~$\Node[1]$}}
    \;\AND\;
    \overbrace{
      \BigAND_{
        \lalign{
          \Element \in (\Domain{\StructRepr{\Graph}} \setminus \NodeSet{\Graph}):\,
          \Element \,\NeighborRel{\StructRepr{\Graph}}\!\! \Node[1]
        }
      }
      \Translation_{\Set{\FOVar[1] \mapsto \Element}}(\Formula[1]),
    }^{\text{$\Node[1]$'s \kl{labeling bits}}}
  \end{equation*}
  where the translation function
  $\Translation_{\Assignment}(\Formula[2])$
  is defined inductively as follows
  for every $\BFL$-\kl{formula}~$\Formula[2]$
  and every \kl{variable assignment}
  $\Assignment \colon \FreeFO(\Formula[2]) \to \Domain{\StructRepr{\Graph}}$:
  \begin{itemize}
  \item \kl{Atomic} \kl{formulas} that do not involve a \kl{second-order variable}
    are replaced by
    their truth value in~$\Tuple{\StructRepr{\Graph}, \Assignment}$, i.e.,
    \begin{gather*}
      \Translation_{\Assignment}(\BitTrue{1}{\FOVar[2]}) = \,
      \begin{cases*}
        \True
        & if $\Assignment(\FOVar[2]) \in \BitSet{1}{\StructRepr{\Graph}}$, \\
        \False
        & otherwise,
      \end{cases*}
      \qquad
      \Translation_{\Assignment}(\FOVar[2] \Linked{i} \FOVar[3]) = \,
      \begin{cases*}
        \True
        & if $\Assignment(\FOVar[2]) \LinkRel{i}{\StructRepr{\Graph}}
              \Assignment(\FOVar[3])$, \\
        \False
        & otherwise,
      \end{cases*} \\
      \Translation_{\Assignment}(\FOVar[2] \Equal \FOVar[3]) = \,
      \begin{cases*}
        \True
        & if $\Assignment(\FOVar[2]) = \Assignment(\FOVar[3])$, \\
        \False
        & otherwise,
      \end{cases*}
    \end{gather*}
    for $i \in \Set{1, 2}$.
  \item \kl{Atomic} \kl{formulas} that involve a \kl{second-order variable}
    are replaced by the corresponding \kl{Boolean variable}, i.e.,\,
    $\Translation_{\Assignment}(
       \InRel{\SOVar}{\FOVar[2]_1,\dots,\FOVar[2]_{\Arity}}
     ) = \,
     \BoolVar_{
       \SOVar(
         \Assignment(\FOVar[2]_1), \dots,
         \Assignment(\FOVar[2]_{\Arity})
       )
     }$.
  \item \kl{Boolean connectives} are preserved:
    $\Translation_{\Assignment}(\NOT \Formula[2]) = \,
     \NOT \Translation_{\Assignment}(\Formula[2])$,
    and
    $\Translation_{\Assignment}(\Formula[2]_1 \OR \Formula[2]_2) = \,
     \Translation_{\Assignment}(\Formula[2]_1) \OR
     \Translation_{\Assignment}(\Formula[2]_2)$.
  \item \kl{First-order quantification}
    is expressed through a case distinction
    over all possible \kl{variable assignments}, i.e.,
    \begin{equation*}
      \Translation_{\Assignment}(
        \ExistsNb{\FOVar[3]}{\FOVar[2]} \; \Formula[2]
      ) = \,
      \BigOR_{
        \lalign{
          \Element \inS \StructRepr{\Graph}\!:\,
          \Element \,\NeighborRel{\StructRepr{\Graph}}\!\! \Assignment(\FOVar[2])
        }
      }
      \Translation_{\Version{\Assignment}{\FOVar[3]}{\Element}}(\Formula[2]).
    \end{equation*}
  \end{itemize}
  Clearly,
  every \kl{interpretation}
  of a $\Arity$-\kl{ary} \kl{relation variable}~$\SOVar$
  on~$\StructRepr{\Graph}$
  can be translated to a \kl{valuation} of the \kl{Boolean variables}
  $\Tuple{
     \BoolVar_{\SOVar(\Element_1, \dots, \Element_{\Arity})}
   }_{\Element_1, \dots, \Element_{\Arity} \inS \StructRepr{\Graph}}$,
  and vice versa.
  Hence,
  the \kl{structure}~$\StructRepr{\Graph}$ \kl{satisfies}
  $\ExistsRel{\SOVar_1} \dots \ExistsRel{\SOVar_n} \,
   \ForAll{\FOVar[1]} \, \Formula[1]\Of{\FOVar[1]}$
  if and only if
  there exists a (global) \kl{valuation}
  of all \kl{Boolean variables} occurring in~$\Graph'$
  that simultaneously satisfies
  all the \kl(Boolean){formulas} of~$\Graph'$.

  The only issue is that a \kl{distributed Turing machine}
  would require \kl{globally unique} \kl{identifiers}
  to compute~$\Graph'$ from~$\Graph$,
  since the \kl{Boolean variables} in~$\Graph'$
  allow to distinguish between
  all \kl{elements} of~$\StructRepr{\Graph}$.
  So instead of~$\Graph'$,
  we compute an \kl(graph){equisatisfiable} \kl{Boolean graph}~$\Graph''$
  for which \kl{locally unique} \kl{identifiers} suffice.
  Let~$\Radius$ be the maximum nesting depth
  of \kl(quantifier){bounded} \kl{first-order quantifiers} in~$\Formula[1]$
  (intuitively,
  the distance up to which $\Formula[1]$ can “see”).
  For a given $(\Radius + 1)$-\kl{locally unique}
  \kl{identifier assignment}~$\IdMap$ of~$\Graph$,
  we define~$\Graph''$ as the \kl{graph}
  that one obtains from~$\Graph'$ by rewriting all the \kl{Boolean formulas}
  in such a way that
  each \kl{node}~$\Node[1]$ and each of its \kl{labeling bits}
  are referred to using the \kl{identifier}~$\IdMap(\Node[1])$.
  Note that $\Graph''$ can be computed from~$\Graph$
  in \kl{round time} $(\Radius + 1)$ and \kl{polynomial step time}.
  (The size of~$\Formula[1]$ is constant with respect to~$\Graph$.)
  Furthermore,
  $\Graph''$ is \kl(graph){satisfiable}
  if and only if
  $\Graph'$ is so,
  because \kl{elements} that share the same \kl{identifier}
  are never referred to
  by the same \kl(Boolean){formula}
  or by two \kl(Boolean){formulas} belonging to \kl{adjacent} \kl{nodes}.
\end{proof}

While $\SATGRAPH$ itself may not be a particularly natural \kl{graph property},
we can use it as a basis for establishing
the $\NLP$-\kl{completeness} of more natural \kl(graph){properties},
in the same way that Karp~\cite{DBLP:conf/coco/Karp72}
used $\SAT$ to prove many other problems $\NP$-\kl{complete}.
Again,
this is made possible by the flexibility
of \kl{locally polynomial reductions},
which do not necessarily have to be \kl{topology-preserving}.
We now apply this approach to $3$-\kl{colorability}.

\begin{theorem}
  \label{thm:three-colorable}
  $\COLORABLE{3}$ is $\NLP$-\kl{complete}.
\end{theorem}

\begin{proof}[Proof sketch]
  By Example~\ref{ex:3-colorable} and Theorem~\ref{thm:local-fagin},
  the \kl(graph){property} $\COLORABLE{3}$ lies in~$\NLP$.
  \AP
  To show that it is also $\NLP$-\kl{hard},
  we first \kl{reduce} $\SATGRAPH$ to~$\intro*\tSATGRAPH$,
  the set of all \kl(graph){satisfiable} \kl{Boolean graphs}
  in which every \kl{node} is \kl{labeled}
  with a $\intro*\tCNF$ \kl(Boolean){formula},
  i.e., a \kl{Boolean formula} in conjunctive normal form
  consisting of clauses with at most three literals.
  \AP
  This is straightforward because the standard \kl{reduction}
  from $\SAT$ to $\intro*\tSAT$
  (i.e., to $\tSATGRAPH \cap \NODE$)
  can be trivially generalized:
  Given a \kl{Boolean graph}~$\Graph$
  and an \kl{identifier assignment}~$\IdMap$ of~$\Graph$,
  we create an \kl(graph){equisatisfiable} \kl{graph}~$\Graph'$
  by replacing the \kl(Boolean){formula}~$\Formula$
  at each \kl{node} $\Node[1] \inG \Graph$
  by an equisatisfiable $\tCNF$ \kl(Boolean){formula}~$\Formula'$
  whose size is proportional to the size of~$\Formula$.
  (This can be done using the Tseytin transformation.)
  The new \kl(Boolean){formula}~$\Formula'$
  may contain additional \kl{Boolean variables}
  that do not occur in~$\Formula$,
  but every satisfying \kl{valuation} of~$\Formula$
  can be extended to a satisfying \kl{valuation} of~$\Formula'$,
  and conversely,
  every satisfying \kl{valuation} of~$\Formula'$
  can be restricted to a satisfying \kl{valuation} of~$\Formula$.
  To ensure that the \kl{graphs}~$\Graph$ and~$\Graph'$
  are indeed \kl(graph){equisatisfiable},
  we make the new \kl(Boolean){variables}' names depend on
  the \kl{identifier}~$\IdMap(\Node[1])$.
  Thus,
  the \kl{valuations} of \kl{adjacent} \kl{nodes} in~$\Graph'$
  need to be consistent
  only for the original \kl(Boolean){variables},
  not for the newly introduced ones.

\begin{figure}[tb]
  \centering
  \begin{tikzpicture}[
    semithick,>=stealth',on grid,
    vertex/.style={draw,circle,minimum size=1.8ex,inner sep=0},
    named vertex/.style={draw,circle,inner sep=0,minimum size=4ex},
    box/.style={draw,dotted,rounded corners=2ex},
    label/.style={inner sep=0,align=left},
  ]
  \def\nodeDist{3ex}
  % Nodes of u's gadget
  \node[vertex] (false_u) {};
  \node[vertex] (ground_u) [below=2.0*\nodeDist of false_u] {};
  \node[vertex] (P2_u)    [below right=2.0*\nodeDist and 0.5*\nodeDist of ground_u] {};
  \node[vertex] (notP2_u) [below left =2.0*\nodeDist and 0.5*\nodeDist of ground_u] {};
  \node[vertex] (P1_u)    [left=1.5*\nodeDist of notP2_u] {};
  \node[vertex] (notP1_u) [left=\nodeDist of P1_u] {};
  \node[vertex] (notP3_u) [right=1.5*\nodeDist of P2_u] {};
  \node[vertex] (P3_u)    [right=\nodeDist of notP3_u] {};
  \node[vertex] (clause1_u) [below left=2.4*\nodeDist and 0.5*\nodeDist of notP2_u] {};
  \node[vertex] (clause2_u) [right=\nodeDist of clause1_u] {};
  \node[vertex] (clause3_u) at ($(clause1_u)+(-60:\nodeDist)$) {};
  \node[vertex] (clause4_u) [below=\nodeDist of clause3_u] {};
  \node[vertex] (clause5_u) [right=\nodeDist of clause4_u] {};
  \node[vertex] (clause6_u) at ($(clause4_u)+(-60:\nodeDist)$) {};
  \node[box] (clause_u) [fit=(clause1_u)(clause2_u)(clause3_u)(clause4_u)(clause5_u)(clause6_u)] {};
  \node[box] (formula_u) [fit={($(false_u)+(0,1.4*\nodeDist)$)
                               ($(notP1_u)-(2.4*\nodeDist,0)$)
                               ($(P3_u)+(0.8*\nodeDist,0)$)
                               ($(clause6_u)-(0,1.0*\nodeDist)$)},rounded corners=6ex] {};
  \node[label,above=0.1*\nodeDist of false_u.north] {$\mathit{false}$};
  \node[label,anchor=east,above left=0.15*\nodeDist and 0*\nodeDist of ground_u.west] {$\mathit{ground}$};
  \node[label,anchor=center,below=0.25*\nodeDist of P1_u.south] {$\BoolVar_1$};
  \node[label,anchor=center,below=0.1*\nodeDist of notP1_u.south] {$\Negated{\BoolVar}_1$};
  \node[label,anchor=center,below=0.25*\nodeDist of P2_u.south] {$\BoolVar_2$};
  \node[label,anchor=center,below=0.1*\nodeDist of notP2_u.south] {$\Negated{\BoolVar}_2$};
  \node[label,anchor=center,below=0.25*\nodeDist of P3_u.south] {$\BoolVar_3$};
  \node[label,anchor=center,below=0.1*\nodeDist of notP3_u.south] {$\Negated{\BoolVar}_3$};
  \node[label,rotate=90,anchor=north] at ($(clause_u.east)+(0.15*\nodeDist,0)$) {clause \\[-0.8ex] gadget};
  \node[label,rotate=90,anchor=north] at ($(formula_u.east |- clause_u)+(0.15*\nodeDist,0)$)
       {\kl(Boolean){formula} \\[-0.8ex] gadget};
  % Nodes of v's gadget
  \node[vertex] (false_v) [right=17*\nodeDist of false_u] {};
  \node[vertex] (ground_v) [below=2.0*\nodeDist of false_v] {};
  \node[vertex] (P4_v)    [below left =2.0*\nodeDist and 0.5*\nodeDist of ground_v] {};
  \node[vertex] (notP4_v) [below right=2.0*\nodeDist and 0.5*\nodeDist of ground_v] {};
  \node[vertex] (notP3_v) [left=1.5*\nodeDist of P4_v] {};
  \node[vertex] (P3_v)    [left=\nodeDist of notP3_v] {};
  \node[vertex] (P5_v)    [right=1.5*\nodeDist of notP4_v] {};
  \node[vertex] (notP5_v) [right=\nodeDist of P5_v] {};
  \node[vertex] (clause1_v) [below left=2.4*\nodeDist and 0.5*\nodeDist of P4_v] {};
  \node[vertex] (clause2_v) [right=\nodeDist of clause1_v] {};
  \node[vertex] (clause3_v) at ($(clause1_v)+(-60:\nodeDist)$) {};
  \node[vertex] (clause4_v) [below=\nodeDist of clause3_v] {};
  \node[vertex] (clause5_v) [right=\nodeDist of clause4_v] {};
  \node[vertex] (clause6_v) at ($(clause4_v)+(-60:\nodeDist)$) {};
  \node[box] (clause_v) [fit=(clause1_v)(clause2_v)(clause3_v)(clause4_v)(clause5_v)(clause6_v)] {};
  \node[box] (formula_v) [fit={($(false_v)+(0,1.4*\nodeDist)$)
                               ($(P3_v)-(0.8*\nodeDist,0)$)
                               ($(notP5_v)+(2.4*\nodeDist,0)$)
                               ($(clause6_v)-(0,1.0*\nodeDist)$)},rounded corners=6ex] {};
  \node[label,above=0.1*\nodeDist of false_v.north] {$\mathit{false}$};
  \node[label,anchor=west,above right=0.15*\nodeDist and 0.1*\nodeDist of ground_v.east] {$\mathit{ground}$};
  \node[label,anchor=center,below=0.25*\nodeDist of P3_v.south] {$\BoolVar_3$};
  \node[label,anchor=center,below=0.1*\nodeDist of notP3_v.south] {$\Negated{\BoolVar}_3$};
  \node[label,anchor=center,below=0.25*\nodeDist of P4_v.south] {$\BoolVar_4$};
  \node[label,anchor=center,below=0.1*\nodeDist of notP4_v.south] {$\Negated{\BoolVar}_4$};
  \node[label,anchor=center,below=0.25*\nodeDist of P5_v.south] {$\BoolVar_5$};
  \node[label,anchor=center,below=0.1*\nodeDist of notP5_v.south] {$\Negated{\BoolVar}_5$};
  \node[label,rotate=90,anchor=south] at ($(clause_v.west)-(0.15*\nodeDist,0)$) {clause \\[-0.8ex] gadget};
  \node[label,rotate=90,anchor=south] at ($(formula_v.west |- clause_v)-(0.15*\nodeDist,0)$)
       {\kl(Boolean){formula} \\[-0.8ex] gadget};
  % Nodes of connector gadgets
  \node[vertex] (false_uv) at ($(false_u)!0.4!(false_v)$) {};
  \node[vertex] (false_vu) at ($(false_u)!0.6!(false_v)$) {};
  \node[vertex] (ground_uv) at ($(ground_u)!0.4!(ground_v)$) {};
  \node[vertex] (ground_vu) at ($(ground_u)!0.6!(ground_v)$) {};
  \node[vertex] (P3_uv) [below=2.0*\nodeDist of ground_uv] {};
  \node[vertex] (P3_vu) [below=2.0*\nodeDist of ground_vu] {};
  % Cluster boxes
  \node[box] (cluster_u) [fit={($(false_uv)+(0.8*\nodeDist,2.0*\nodeDist)$)
                               ($(notP1_u)-(3.0*\nodeDist,0)$)
                               ($(clause6_u)-(0,1.6*\nodeDist)$)},rounded corners=6ex] {};
  \node[box] (cluster_v) [fit={($(false_vu)+(-0.8*\nodeDist,2.0*\nodeDist)$)
                               ($(notP5_v)+(3.0*\nodeDist,0)$)
                               ($(clause6_v)-(0,1.6*\nodeDist)$)},rounded corners=6ex] {};
  \node[label,anchor=north] at ($(cluster_u.south)+(0,-0.3*\nodeDist)$)
       {\kl{cluster} representing \kl{node}~$\Node[1]$};
  \node[label,anchor=north] at ($(cluster_v.south)+(0,-0.3*\nodeDist)$)
       {\kl{cluster} representing \kl{node}~$\Node[2]$};
  \node (graph') [label,above left=0.2*\nodeDist and -0.2*\nodeDist of cluster_u.north west,anchor=north east]
                 {$\Graph'$:};
  % Edges of u's gadget
  \path (false_u) edge (ground_u)
        (ground_u) edge (P1_u)
        (ground_u) edge (notP1_u)
        (ground_u) edge (P2_u)
        (ground_u) edge (notP2_u)
        (ground_u) edge (P3_u)
        (ground_u) edge (notP3_u)
        (P1_u) edge (notP1_u)
        (P2_u) edge (notP2_u)
        (P3_u) edge (notP3_u)
        (clause1_u) edge (clause2_u)
        (clause1_u) edge (clause3_u)
        (clause2_u) edge (clause3_u)
        (clause3_u) edge (clause4_u)
        (clause4_u) edge (clause5_u)
        (clause4_u) edge (clause6_u)
        (clause5_u) edge (clause6_u)
        (P1_u)    edge[out=-55,in=90]  (clause1_u)
        (notP2_u) edge[out=-55,in=90]  (clause2_u)
        (notP3_u) edge[out=-130,in=90] (clause5_u)
        (clause6_u) edge[out=170,in=180,looseness=1.8] (false_u)
        (clause6_u) edge[out=160,in=180,looseness=1.9] (ground_u)
        ;
  % Edges of v's gadget
  \path (false_v) edge (ground_v)
        (ground_v) edge (P3_v)
        (ground_v) edge (notP3_v)
        (ground_v) edge (P4_v)
        (ground_v) edge (notP4_v)
        (ground_v) edge (P5_v)
        (ground_v) edge (notP5_v)
        (P3_v) edge (notP3_v)
        (P4_v) edge (notP4_v)
        (P5_v) edge (notP5_v)
        (clause1_v) edge (clause2_v)
        (clause1_v) edge (clause3_v)
        (clause2_v) edge (clause3_v)
        (clause3_v) edge (clause4_v)
        (clause4_v) edge (clause5_v)
        (clause4_v) edge (clause6_v)
        (clause5_v) edge (clause6_v)
        (P3_v.-48)     edge[out=-75,in=175] (clause1_v)
        (P4_v)         edge[out=-55,in=90]  (clause2_v)
        (notP5_v.-130) edge[out=-112,in=16,looseness=1.1] (clause5_v)
        (clause6_v) edge[out=10,in=0,looseness=1.8] (false_v)
        (clause6_v) edge[out=20,in=0,looseness=1.9] (ground_v)
        ;
  % Edges of connector gadgets
  \path (false_u)  edge (false_uv)
        (false_u)  edge[bend left=13] (false_vu)
        (false_v)  edge[bend left=13] (false_uv)
        (false_v)  edge (false_vu)
        (false_uv) edge (false_vu)
        (ground_u)  edge (ground_uv)
        (ground_u)  edge[bend left=13] (ground_vu)
        (ground_v)  edge[bend left=13] (ground_uv)
        (ground_v)  edge (ground_vu)
        (ground_uv) edge (ground_vu)
        (P3_u)  edge (P3_uv)
        (P3_u)  edge[bend left=15] (P3_vu)
        (P3_v)  edge[bend left=15] (P3_uv)
        (P3_v)  edge (P3_vu)
        (P3_uv) edge (P3_vu)
        ;
  % Original graph
  \node[named vertex] (u) [above=1.0*\nodeDist of cluster_u.north] {$\Node[1]$};
  \node[named vertex] (v) [above=1.0*\nodeDist of cluster_v.north] {$\Node[2]$};
  \node[label,left=0.5*\nodeDist of u.west] {
    $\BoolVar_1 \OR \Negated{\BoolVar}_2 \OR \Negated{\BoolVar}_3$};
  \node[label,right=0.5*\nodeDist of v.east] {
    $\BoolVar_3 \OR \BoolVar_4 \OR \Negated{\BoolVar}_5$};
  \node[label] at (graph' |- u) {$\Graph$:};
  \path (u) edge (v);
\end{tikzpicture}

%%% Local Variables:
%%% mode: latex
%%% TeX-master: "../lph-paper"
%%% End:
  \caption{
    \emph{(repeated from Figure~\ref{fig:3satgraph-to-3colorable-overview})}
    Example illustrating
    the \kl{reduction} from $\tSATGRAPH$ to~$\COLORABLE{3}$
    used in the proof of Theorem~\ref{thm:three-colorable}.
    The \kl{Boolean graph}~$\Graph$ is \kl(graph){satisfiable}
    if and only if
    the \kl{graph}~$\Graph'$ is $3$-\kl{colorable}.
    The labels in the depiction of~$\Graph'$ serve explanatory purposes only
    and are not part of the \kl{graph}.
  }
  \label{fig:3satgraph-to-3colorable}
\end{figure}

  Next,
  we \kl{reduce} $\tSATGRAPH$ to~$\COLORABLE{3}$,
  again by generalizing to arbitrary \kl{graphs}
  the corresponding \kl{reduction} on \kl{single-node graphs}.
  The standard \kl{reduction}
  from $\tSAT$ to the \kl{string-encoded} version of $\COLORABLE{3}$
  converts any given $\tCNF$-\kl(Boolean){formula}~$\Formula$
  into a \kl{graph}~$\Graph[2]_{\Formula}$,
  which we will refer to as a \kl(Boolean){formula} gadget.
  This \kl{graph} contains
  two special \kl{nodes} called $\mathit{false}$ and $\mathit{ground}$,
  two \kl{nodes}~$\BoolVar$ and~$\Negated{\BoolVar}$
  for each \kl(Boolean){variable}~$\BoolVar$ occurring in~$\Formula$
  (to~represent the corresponding positive and negative literals),
  and a small gadget for each clause of~$\Formula$.
  The \kl{edges} are chosen in such a way
  that $\Graph[2]_{\Formula}$~is $3$-\kl{colorable}
  if and only if
  $\Formula$~is satisfiable.
  Moreover,
  if a $3$-\kl{coloring} exists,
  then we may assume without loss of generality
  that $\mathit{false}$ and $\mathit{ground}$ are colored $0$~and~$2$,
  respectively,
  and each \kl{node} representing a literal
  is colored by that literal's truth value
  ($0$~for false, and $1$ for true).
  For details, see, e.g.,
  \cite[Prp.~2.27]{DBLP:books/daglib/0019967}.

  We now generalize this construction to arbitrary \kl{graphs}.
  An example is provided in Figure~\ref{fig:3satgraph-to-3colorable}.
  Given a \kl{Boolean graph}~$\Graph$,
  we construct a \kl{graph}~$\Graph'$
  that is $3$-\kl{colorable} if and only if $\Graph$~is \kl(graph){satisfiable}.
  For each \kl{node} $\Node[1] \inG \Graph$
  \kl{labeled} with a \kl{Boolean formula}~$\Formula[1]$,
  the \kl{cluster} representing~$\Node[1]$ in~$\Graph'$
  contains a copy of
  the \kl(Boolean){formula} gadget~$\Graph[2]_{\Formula[1]}$.
  We denote this copy by~$\Graph[2]^{\Node[1]}$
  and mark its \kl{nodes} with a superscript~$\Node[1]$.
  To enforce that
  the \kl(Boolean){formula} gadgets of adjacent \kl{clusters}
  are colored consistently,
  we connect some of their \kl{nodes}
  by means of an additional gadget.
  More precisely,
  for $\Set{\Node[1], \Node[2]} \in \EdgeSet{\Graph}$,
  if we require that two \kl{nodes}
  $\Node[3]^{\Node[1]} \inG \Graph[2]^{\Node[1]}$
  and
  $\Node[3]^{\Node[2]} \inG \Graph[2]^{\Node[2]}$
  have the same color,
  then we connect them using the following connector gadget:
  
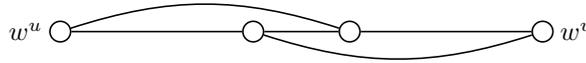
\begin{center}
    \begin{tikzpicture}[
    semithick,>=stealth',
    vertex/.style={draw,circle,minimum size=1.8ex,inner sep=0},
    label/.style={inner sep=0,align=left},
  ]
  \def\labelDist{0.5ex}
  \node[vertex] (w_u) {};
  \node[vertex] (w_v) [right=40ex of w_u] {};
  \node[vertex] (w_uv) at ($(w_u)!0.4!(w_v)$) {};
  \node[vertex] (w_vu) at ($(w_u)!0.6!(w_v)$) {};
  \node[label,left=\labelDist of w_u] {$\Node[3]^{\Node[1]}$};
  \node[label,right=\labelDist of w_v] {$\Node[3]^{\Node[2]}$};
  \path (w_u)  edge (w_uv)
        (w_u)  edge[bend left=18] (w_vu)
        (w_v)  edge[bend left=18] (w_uv)
        (w_v)  edge (w_vu)
        (w_uv) edge (w_vu)
        ;
\end{tikzpicture}

%%% Local Variables:
%%% mode: latex
%%% TeX-master: "../lph-paper"
%%% End:
  \end{center}
  Note that any valid $3$-\kl{coloring} of the connector gadget
  has to assign the same color
  to~$\Node[3]^{\Node[1]}$ and~$\Node[3]^{\Node[2]}$.
  Using this,
  we connect $\mathit{false}^{\Node[1]}$~to~$\mathit{false}^{\Node[2]}$,
  $\mathit{ground}^{\:\!\Node[1]}$~to~$\mathit{ground}^{\:\!\Node[2]}$, and
  $\BoolVar^{\:\!\Node[1]}$~to~$\BoolVar^{\:\!\Node[2]}$
  for any \kl{Boolean variable}~$\BoolVar$
  that occurs in the \kl(Boolean){formulas}
  of both~$\Node[1]$ and~$\Node[2]$
  (see Figure~\ref{fig:3satgraph-to-3colorable}).
  Thereby we ensure that
  each color has the same meaning in both \kl(Boolean){formula} gadgets
  and that \kl(Boolean){variables} shared by~$\Node[1]$ and~$\Node[2]$
  are assigned the same truth value.
  Besides,
  this also ensures that the \kl{graph}~$\Graph'$ is \kl{connected}
  (as~required by our definition of \kl{graphs}),
  even if some \kl{adjacent} \kl{nodes} of~$\Graph$
  do not share any \kl(Boolean){variables}.

  Clearly,
  $\Graph'$~can be computed from~$\Graph$
  by a \kl{distributed Turing machine}
  in two \kl{rounds} and \kl{polynomial step time}.
  (In the first \kl{round},
  the \kl{nodes} send their \kl{label} and \kl{identifier}
  to their \kl{neighbors};
  in the second \kl{round},
  they perform only local computations.)
  The \kl{cluster map} from~$\Graph'$ to~$\Graph$
  can be chosen such that
  each \kl{cluster} contains half of the \kl{nodes}
  of each connector gadget attached to it,
  as shown in Figure~\ref{fig:3satgraph-to-3colorable}.
\end{proof}

%%% Local Variables:
%%% mode: latex
%%% TeX-master: "../lph-paper"
%%% End:

\section{Infiniteness of the locally polynomial hierarchy}
\label{sec:hierarchy}

In this section,
we show that the \kl{locally polynomial hierarchy} does not collapse.
Intuitively,
this means that the more alternations we allow
between \kl(certificate){Eve} and \kl(certificate){Adam},
the more \kl{graph properties} we can express.
More precisely,
we prove that all inclusions
represented by solid lines in Figure~\ref{fig:hierarchy-full}
are strict,
and that classes represented on the same level are pairwise distinct.
For the remaining inclusions,
represented by dashed lines,
we expect the proof of strictness to be difficult,
since they are equalities if and only if $\PTIME = \coNP$.

\begin{figure}[tb]
  \centering
  \begin{tikzpicture}[
    semithick,on grid,node distance=9ex,
    every node/.style={draw,rounded rectangle,
                       minimum height=4ex,minimum width=10ex},
    extra/.style={draw=none,rectangle,minimum size=0,inner sep=0},
    strict/.style={},
    non strict/.style={thin,dashed},
    heavy/.style={ultra thick},
    pseudo/.style={thin,dotted},
  ]
  \def\nodeDist{9ex}
  % Sigma and Pi nodes
  \node[heavy] (s0) {$\LP$};
  \node[heavy] (s1) [above left=\nodeDist and 2/3*\nodeDist of s0] {$\SigmaLP{1}$};
  \node (p1) [above right=\nodeDist and 2/3*\nodeDist of s0] {$\PiLP{1}$};
  \node (s2) [above of=s1]        {$\SigmaLP{2}$};
  \node[heavy] (p2) [above of=p1] {$\PiLP{2}$};
  \node[heavy] (s3) [above of=s2] {$\SigmaLP{3}$};
  \node (p3) [above of=p2]        {$\PiLP{3}$};
  \node (s4) [above of=s3]        {$\SigmaLP{4}$};
  \node[heavy] (p4) [above of=p3] {$\PiLP{4}$};
  \node[heavy] (s5) [above of=s4] {$\SigmaLP{5}$};
  \node (p5) [above of=p4]        {$\PiLP{5}$};
  \node (s6) [above of=s5]        {$\SigmaLP{6}$};
  \node[heavy] (p6) [above of=p5] {$\PiLP{6}$};
  \node[extra] at ([yshift=2/3*\nodeDist]$(s6)!0.5!(p6)$) {$\vdots$};
  % co-Sigma and co-Pi nodes
  \node[heavy] (cs0) [right=4*\nodeDist of s0] {$\coLP$};
  \node[heavy] (cs1) [above left=\nodeDist and 2/3*\nodeDist of cs0] {$\coSigmaLP{1}$};
  \node (cp1) [above right=\nodeDist and 2/3*\nodeDist of cs0] {$\coPiLP{1}$};
  \node (cs2) [above of=cs1]        {$\coSigmaLP{2}$};
  \node[heavy] (cp2) [above of=cp1] {$\coPiLP{2}$};
  \node[heavy] (cs3) [above of=cs2] {$\coSigmaLP{3}$};
  \node (cp3) [above of=cp2]        {$\coPiLP{3}$};
  \node (cs4) [above of=cs3]        {$\coSigmaLP{4}$};
  \node[heavy] (cp4) [above of=cp3] {$\coPiLP{4}$};
  \node[heavy] (cs5) [above of=cs4] {$\coSigmaLP{5}$};
  \node (cp5) [above of=cp4]        {$\coPiLP{5}$};
  \node (cs6) [above of=cs5]        {$\coSigmaLP{6}$};
  \node[heavy] (cp6) [above of=cp5] {$\coPiLP{6}$};
  \node[extra] at ([yshift=2/3*\nodeDist]$(cs6)!0.5!(cp6)$) {$\vdots$};
  % Pseudo Sigma nodes
  \node (s0x) [draw=none,right=4*\nodeDist of cs0] {};
  \node (s1x) [pseudo,above left=\nodeDist and 2/3*\nodeDist of s0x] {$\SigmaLP{1}$};
  \node (s2x) [pseudo,above of=s1x] {$\SigmaLP{2}$};
  \node (s3x) [pseudo,above of=s2x] {$\SigmaLP{3}$};
  \node (s4x) [pseudo,above of=s3x] {$\SigmaLP{4}$};
  \node (s5x) [pseudo,above of=s4x] {$\SigmaLP{5}$};
  % Additional labels
  \node [extra,anchor=east,left=0ex of s0.west] {$\SigmaLP{0} \,{=}\,$};
  \node [extra,anchor=west,right=0ex of s0.east] {$\,{=}\, \PiLP{0}$};
  \node [extra,anchor=east,left=0ex of s1.west] {$\NLP \,{=}\,$};
  \node [extra,anchor=east,left=0ex of cs0.west] {$\coSigmaLP{0} \,{=}\,$};
  \node [extra,anchor=west,right=0ex of cs0.east] {$\,{=}\, \coPiLP{0}$};
  \node [extra,anchor=east,left=0ex of cs1.west] {$\coNLP \,{=}\,$};
  % Sigma and Pi edges
  \path[non strict]
    (s5) edge (s6)
    (p4) edge (p5)
    (s3) edge (s4)
    (p2) edge (p3)
    (s1) edge (s2)
    (s0) edge (p1);
  \path[strict]
    (s5) edge (p6)
    (p5) edge (s6)
    (p5) edge (p6)
    (s4) edge (s5)
    (s4) edge (p5)
    (p4) edge (s5)
    (s3) edge (p4)
    (p3) edge (s4)
    (p3) edge (p4)
    (s2) edge (s3)
    (s2) edge (p3)
    (p2) edge (s3)
    (s1) edge (p2)
    (p1) edge (s2)
    (p1) edge (p2)
    (s0) edge (s1);
  % co-Sigma and co-Pi edges
  \path[non strict]
    (cs5) edge (cs6)
    (cp4) edge (cp5)
    (cs3) edge (cs4)
    (cp2) edge (cp3)
    (cs1) edge (cs2)
    (cs0) edge (cp1);
  \path[strict]
    (cs5) edge (cp6)
    (cp5) edge (cs6)
    (cp5) edge (cp6)
    (cs4) edge (cs5)
    (cs4) edge (cp5)
    (cp4) edge (cs5)
    (cs3) edge (cp4)
    (cp3) edge (cs4)
    (cp3) edge (cp4)
    (cs2) edge (cs3)
    (cs2) edge (cp3)
    (cp2) edge (cs3)
    (cs1) edge (cp2)
    (cp1) edge (cs2)
    (cp1) edge (cp2)
    (cs0) edge (cs1);
  % Edges connecting non-co to co nodes
  \path[strict]
    (p1)  edge (cs3)
    (p3)  edge (cs5)
    (cs2) edge (p4)
    (cs4) edge (p6);
  \path[strict]
    (cp1) edge (s3x)
    (cp3) edge (s5x)
    (s2x) edge (cp4)
    (s4x) edge (cp6);
\end{tikzpicture}

%%% Local Variables:
%%% mode: latex
%%% TeX-master: "../lph-paper"
%%% End:
  \caption{
    \emph{(extended from Figure~\ref{fig:hierarchy-full-overview})}
    The \kl{locally polynomial hierarchy} (left)
    and its \kl{complement hierarchy} (right).
    The dotted classes on the far right
    are the same as those on the far left,
    repeated for the sake of readability.
    Only the lowest seven levels are shown,
    but the pattern extends infinitely.
    Each line (whether solid or dashed) indicates
    an inclusion of the lower class in the higher class.
    (This~holds by definition
    for classes in the same hierarchy,
    and by Proposition~\ref{prp:complementation} and duality
    for classes in separate hierarchies.)
    All inclusions represented by solid lines are proved to be strict,
    and classes located on the same level (regardless of which hierarchy)
    are proved to be pairwise distinct,
    even when restricted to \kl{graphs} of \kl{bounded structural degree}
    (by Proposition~\ref{prp:lp-vs-nlp},
    Theorem~\ref{thm:locally-polynomial-hierarchy},
    Corollaries~\ref{cor:locally-polynomial-hierarchy-missing-pieces},
    \ref{cor:complement-classes}
    and~\ref{cor:complementation-strict-inclusions},
    and duality).
    The inclusions represented by dashed lines are in fact equalities
    on \kl{graphs} of \kl{bounded structural degree}
    (by~Proposition~\ref{prp:bounded-graph-collapse}),
    but this statement is unlikely to generalize to arbitrary \kl{graphs},
    where it holds if and only if $\PTIME = \coNP$
    (by Remark~\ref{rem:unbounded-graph-collapse}).
    This means that
    from a distributed computing perspective,
    the classes depicted with thick borders are the most meaningful.
  }
  \label{fig:hierarchy-full}
\end{figure}

The picture becomes simpler when restricted to
\kl{graphs} of bounded maximum \kl{degree} and \kl{label} length.
On such \kl{graphs},
all of the previous separation results still hold,
but the inclusions represented by dashed lines in Figure~\ref{fig:hierarchy-full}
are actually equalities.
Thus,
in the bounded case,
the \kl{locally polynomial hierarchy}
boils down to the classes depicted with thick borders,
yielding a strict linear order of the following form:
\begin{equation*}
  \PiLP{0}\On{\bGRAPH{\MaxDegree}}    \, \subsetneqq \,
  \SigmaLP{1}\On{\bGRAPH{\MaxDegree}} \, \subsetneqq \,
  \PiLP{2}\On{\bGRAPH{\MaxDegree}}    \, \subsetneqq \,
  \SigmaLP{3}\On{\bGRAPH{\MaxDegree}} \, \subsetneqq \,
  \dots
\end{equation*}
\AP
Here,
$\intro*\bGRAPH{\MaxDegree}$ denotes
the set of \kl{graphs} of $\MaxDegree$-\kl{bounded structural degree},
for some $\MaxDegree \in \Naturals$.
Formally,
the \intro{structural degree} of a \kl{node}~$\Node[1]$ in a \kl{graph}~$\Graph$
is the number of \kl{elements}
that are connected to~$\Node[1]$
in the \kl{graph}'s \kl{structural representation}~$\StructRepr{\Graph}$, i.e.,
$\Card{
   \SetBuilder{
     \Element \inS \StructRepr{\Graph}
   }{
     \Element \NeighborRel{\StructRepr{\Graph}} \!\! \Node[1]
   }
 }$.
In other words,
$\Node[1]$'s \kl{structural degree} is the sum of
its \kl{degree} and its \kl{label} length.
We say that $\Graph$~is
of $\MaxDegree$-\intro{bounded structural degree}
if the \kl{structural degree} of every \kl{node} $\Node[1] \in \Graph$
is at most~$\MaxDegree$.

The remainder of this section is organized as follows.
In Section~\ref{ssec:warm-up},
we separate~$\LP$ from~$\NLP$ and~$\coLP$
by identifying simple \kl{graph properties}
that lie in one class but not the other.
Then,
in Section~\ref{ssec:climb-up},
we separate all higher levels of the \kl{locally polynomial hierarchy}
that end with an existential quantifier.
Our proof uses
an analogous result about \kl{monadic second-order logic} on~\kl{pictures},
as well as an automaton model
characterizing the \kl(monadic){existential fragment} of that logic.
Finally,
the picture is completed in Section~\ref{ssec:complete-picture},
where we establish
the remaining separations and inclusions shown in Figure~\ref{fig:hierarchy-full}
and identify \kl{graph properties}
that lie outside the \kl[locally polynomial hierarchy]{hierarchy}.

%–––––––––––––––––––––––––––––––––––––––––––––––––––––––––––––––––––––––––––––––
\subsection{Warming up at ground level}
\label{ssec:warm-up}

The connection to logic established in Section~\ref{sec:fagin}
will be useful for separating the higher levels of
the \kl{locally polynomial hierarchy} and its \kl{complement hierarchy}.
But since this connection does not hold for the lowest level,
we must use a different approach to separate~$\LP$ from~$\NLP$ and~$\coLP$.
In this subsection,
we do so using elementary arguments
based on symmetry breaking and the pigeonhole principle.
This also gives us the opportunity to gain a better intuition
for the lower levels of the \kl{locally polynomial hierarchy}
by revisiting the \kl{graph properties}
presented in Section~\ref{ssec:example-formulas}
and taking advantage of the \kl{hardness} and \kl{completeness} results
established in Section~\ref{sec:reductions}.

\begin{proposition}
  \label{prp:lp-vs-nlp}
  Some \kl{graph properties} can be \kl{verified} but not \kl{decided}
  by a \kl{locally polynomial machine},
  even when restricted to \kl{graphs} of \kl{bounded structural degree}.
  More precisely,
  $\LP\On{\bGRAPH{\MaxDegree}} \subsetneqq \NLP\On{\bGRAPH{\MaxDegree}}$
  for all $\MaxDegree \geq 2$,
  and a fortiori
  $\LP \subsetneqq \NLP$.
\end{proposition}

\begin{proof}
  The \kl{graph property} $\COLORABLE{2}$ clearly lies in~$\NLP$
  (simply adapt the \kl{formula} from Example~\ref{ex:3-colorable}
  and apply Theorem~\ref{thm:local-fagin}).
  In the following,
  we show that it does not lie in~$\LP$,
  even when restricted to \kl{graphs} of $2$-\kl{bounded structural degree}.

  Suppose,
  for the sake of contradiction,
  that
  $\COLORABLE{2} \cap \bGRAPH{2} \in \LP\On{\bGRAPH{2}}$.
  By Lemma~\ref{lem:restrictive-arbiters},
  this means that there is
  a \kl{restrictive} $\LP$\nobreakdash-\kl{decider}~$\Machine$
  for~$\COLORABLE{2}$ on~$\bGRAPH{2}$.
  Assume that $\Machine$~operates
  under $\IdentRadius$-\kl{locally unique} \kl{identifiers},
  and consider an unlabeled \kl{cycle graph}~$\Graph$
  with \kl{node} set
  $\Set{\Node[1]_1, \dots, \Node[1]_n}$
  such that
  $n$~is odd and greater than~$2\IdentRadius$.
  Since $\Graph$~is of odd length,
  it is not $2$-\kl{colorable}.
  We now construct a new \kl{cycle graph}~$\Graph'$
  with \kl{node} set
  $\Set{\Node[1]_1, \dots, \Node[1]_n,\Node[1]_1', \dots, \Node[1]_n'}$
  by “gluing together” two copies of~$\Graph$ as follows:
  
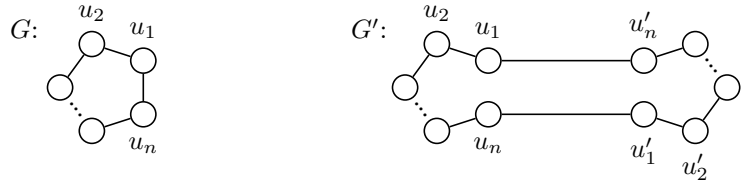
\begin{center}
    \begin{tikzpicture}[
    semithick,
    vertex/.style={draw,circle,minimum size=2ex},
    label/.style={},
  ]
  \def\nodeDist{4ex}
  \def\graphDist{30ex}
  \begin{scope}
    \node[vertex] (u1) at (36:\nodeDist) {};
    \node[vertex] (u2) at (108:\nodeDist) {};
    \node[vertex] (u3) at (180:\nodeDist) {};
    \node[vertex] (u4) at (252:\nodeDist) {};
    \node[vertex] (u5) at (324:\nodeDist) {};
    \node[label] at (144:2.2*\nodeDist) {$\Graph$:};
    \node[label,above=0 of u1] {$\Node[1]_1$};
    \node[label,above=0 of u2] {$\Node[1]_2$};
    \node[label,below=0 of u5] {$\Node[1]_n$};
    \path (u1) edge (u2)
          (u2) edge (u3)
          (u3) edge[draw=none]
               node[sloped] {$\scriptstyle\boldsymbol{{\dots}}$}
               (u4)
          (u4) edge (u5)
          (u5) edge (u1);
  \end{scope}
  \begin{scope}[xshift=\graphDist]
    \node[vertex] (u1Left) at (36:\nodeDist) {};
    \node[vertex] (u2Left) at (108:\nodeDist) {};
    \node[vertex] (u3Left) at (180:\nodeDist) {};
    \node[vertex] (u4Left) at (252:\nodeDist) {};
    \node[vertex] (u5Left) at (324:\nodeDist) {};
    \node[label] at (144:2.2*\nodeDist) {$\Graph'$:};
    \node[label,above=0 of u1Left] {$\Node[1]_1$};
    \node[label,above=0 of u2Left] {$\Node[1]_2$};
    \node[label,below=0 of u5Left] {$\Node[1]_n$};
    \path (u1Left) edge (u2Left)
          (u2Left) edge (u3Left)
          (u3Left) edge[draw=none]
                   node[sloped] {$\scriptstyle\boldsymbol{{\dots}}$}
                   (u4Left)
          (u4Left) edge (u5Left);
  \end{scope}
  \begin{scope}[xshift=\graphDist+5*\nodeDist]
    \node[vertex] (u3Right) at (360:\nodeDist) {};
    \node[vertex] (u4Right) at (72:\nodeDist) {};
    \node[vertex] (u5Right) at (144:\nodeDist) {};
    \node[vertex] (u1Right) at (216:\nodeDist) {};
    \node[vertex] (u2Right) at (288:\nodeDist) {};
    \node[label,below=0 of u1Right] {$\Node[1]_1'$};
    \node[label,below=0 of u2Right] {$\Node[1]_2'$};
    \node[label,above=0 of u5Right] {$\Node[1]_n'$};
    \path (u1Right) edge (u2Right)
          (u2Right) edge (u3Right)
          (u3Right) edge[draw=none]
                    node[sloped] {$\scriptstyle\boldsymbol{{\dots}}$}
                    (u4Right)
          (u4Right) edge (u5Right)
          (u5Right) edge (u1Left)
          (u5Left)  edge (u1Right);
  \end{scope}
\end{tikzpicture}

%%% Local Variables:
%%% mode: latex
%%% TeX-master: "../lph-paper"
%%% End:
  \end{center}
  Since $\Graph'$~is of even length,
  it is $2$-\kl{colorable}.
  Given any $\IdentRadius$-\kl{locally unique}
  \kl{identifier assignment}~$\IdMap$ of~$\Graph$,
  let $\IdMap' \colon \NodeSet{\Graph'} \to \Set{0, 1}^\KleeneStar$
  be such that
  $\IdMap'(\Node[1]_i) =
   \IdMap'(\Node[1]_i') =
   \IdMap(\Node[1]_i)$
  for $i \in \Range[1]{n}$.
  As $n \geq 2\IdentRadius$,
  the function~$\IdMap'$ is
  an $\IdentRadius$-\kl{locally unique} \kl{identifier assignment} of~$\Graph'$.
  Moreover,
  the \kl{verdict} of $\Node[1]_i$ and~$\Node[1]_i'$
  in $\Result{\Machine}{\Graph', \IdMap'}$
  is the same as
  the \kl{verdict} of~$\Node[1]_i$
  in $\Result{\Machine}{\Graph, \IdMap}$,
  because the \kl{tape contents} of these \kl{nodes}
  are the same in every \kl{communication round} and \kl{computation step}.
  Hence,
  $\Machine$~\kl{accepts}~$\Graph$
  if and only if
  it \kl{accepts}~$\Graph'$.
  This contradicts our assumption that $\Machine$~is
  a \kl{restrictive} $\LP$\nobreakdash-\kl{decider}
  for~$\COLORABLE{2}$ on~$\bGRAPH{2}$,
  since both $\Graph$ and~$\Graph'$
  are of $2$-\kl{bounded structural degree},
  but only $\Graph'$ is $2$-\kl{colorable}.
\end{proof}

\begin{corollary}
  \label{cor:three-colorable-not-in-lp}
  $\COLORABLE{3}$ does not lie in $\LP$.
\end{corollary}

\begin{proof}
  By Theorem~\ref{thm:three-colorable},
  we know that
  $\COLORABLE{3}$ is $\NLP$-\kl{hard},
  and by Proposition~\ref{prp:lp-vs-nlp},
  we know that $\NLP \nsubseteq \LP$.
  This implies that $\COLORABLE{3}$ cannot lie in~$\LP$,
  since otherwise we could show that $\NLP \subseteq \LP$.
  To do so,
  it would suffice to transform a hypothetical
  $\LP$\nobreakdash-\kl{decider}~$\Machine$ for~$\COLORABLE{3}$
  into an $\LP$\nobreakdash-\kl{decider}~$\Machine'$
  for an arbitrary \kl(graph){property} ${\Property \in \NLP}$.
  The \kl{nodes} running~$\Machine'$ would first apply
  a \kl{locally polynomial reduction} from~$\Property$ to~$\COLORABLE{3}$.
  Then,
  each \kl{node} would simulate~$\Machine$ on its \kl{cluster}
  and \kl{accept} precisely if all \kl{nodes} of the \kl{cluster} do so.
  (See, e.g., Figure~\ref{fig:3satgraph-to-3colorable}
  on page~\pageref{fig:3satgraph-to-3colorable}
  for an illustration of two \kl{clusters}.)
\end{proof}

\begin{proposition}
  \label{prp:colp-vs-nlp}
  The classes $\coLP$ and~$\NLP$ are incomparable,
  even when restricted to \kl{graphs} of \kl{bounded structural degree}.
  More precisely,
  $\coLP\On{\bGRAPH{\MaxDegree}} \Incomparable \NLP\On{\bGRAPH{\MaxDegree}}$
  for all $\MaxDegree \geq 3$,
  and a fortiori
  $\coLP \Incomparable \NLP$.
\end{proposition}

\begin{proof}
  Since $\NOTALLSELECTED$ lies in~$\coLP$
  (by Remark~\ref{rem:allselected-lp-complete} it is even $\coLP$-\kl{complete}),
  it suffices to show that
  $\NOTALLSELECTED \cap \bGRAPH{3} \notin \NLP\On{\bGRAPH{3}}$.
  Indeed,
  this statement implies that
  $\coLP\On{\bGRAPH{3}} \nsubseteq \NLP\On{\bGRAPH{3}}$,
  and by duality that
  $\LP\On{\bGRAPH{3}} \nsubseteq \coNLP\On{\bGRAPH{3}}$,
  and hence also that
  $\NLP\On{\bGRAPH{3}} \nsubseteq \coLP\On{\bGRAPH{3}}$.

  Suppose then, for the sake of contradiction, that
  $\NOTALLSELECTED \cap \bGRAPH{3} \in \NLP\On{\bGRAPH{3}}$.
  By Lemma~\ref{lem:restrictive-arbiters},
  this means that there exists
  a \kl{restrictive} $\NLP$\nobreakdash-\kl{verifier}~$\Machine$
  for~$\NOTALLSELECTED$ on~$\bGRAPH{3}$.
  Assume that $\Machine$~operates
  under $\IdentRadius$-\kl{locally unique} \kl{identifiers}
  and $\Tuple{\CertifRadius_1, \CertifPolynomial}$\nobreakdash-%
  \kl{bounded certificates}
  and runs in \kl{round time}~$\Radius_2$,
  and let
  $\Radius = \max\Set{2\IdentRadius, \CertifRadius_1, \Radius_2}$.
  In the following,
  we restrict our attention to \kl{cycle graphs}
  whose \kl{nodes} are all \kl{labeled} with a single bit
  and whose length is a multiple of~$(\Radius + 1)$.
  Notice that on such \kl{graphs},
  we can construct an $\IdentRadius$-\kl{locally unique} \kl{identifier assignment}
  by cyclically assigning each \kl{node} an \kl{identifier}
  corresponding to a number in $\Range[0]{\Radius}$,
  \kl{encoded} as a binary string.
  If we do so,
  then the length of an
  $\Tuple{\CertifRadius_1, \CertifPolynomial}$\nobreakdash-\kl{bounded certificate}
  is at most
  $m = \CertifPolynomial
  \bigl(
    (2\Radius + 1) \cdot (\Ceiling{\log_2\Radius} + 3)
  \bigr)$.
  Thus,
  for our choice of \kl{graphs} and \kl{identifier assignments},
  there are no more than
  $n = (\Radius + 1) \cdot (2^{m + 2} - 2)^{2\Radius + 1}$
  possible ways to assign
  \kl{labels}, \kl{identifiers}, and \kl{certificates}
  to the $\Radius_2$-\kl{neighborhood} of any \kl{node}.

  First,
  let $\Graph$ be a \kl{cycle graph}
  whose length is greater than~$n$ and a multiple of~$(\Radius + 1)$
  such that
  exactly one \kl{node}~$\Node[1]$ has \kl{label}~$0$
  (the unselected \kl{node})
  and all others have \kl{label}~$1$.
  Let $\IdMap$ be a cyclic
  $\IdentRadius$\nobreakdash-\kl{locally unique} \kl{identifier assignment}
  of~$\Graph$ as described above.
  Since $\Graph \in \NOTALLSELECTED \cap \bGRAPH{3}$,
  there exists an
  $\Tuple{\CertifRadius_1, \CertifPolynomial}$\nobreakdash-%
  \kl{bounded certificate assignment}~$\CertifMap$ of $\Tuple{\Graph, \IdMap}$
  such that
  $\Result{\Machine}{\Graph, \IdMap, \CertifMap} \equiv \Accept$.
  By the pigeonhole principle,
  there must be two distinct \kl{nodes} $\Node[2]$ and~$\Node[2]'$ in~$\Graph$
  whose $\Radius_2$\nobreakdash-\kl{neighborhoods} are indistinguishable
  because
  the \kl{labels}, \kl{identifiers}, and \kl{certificates} therein
  are all the same.
  Now consider the \kl{cycle graph}~$\Graph'$
  obtained from~$\Graph$ by taking the \kl{path}
  between $\Node[2]$ and~$\Node[2]'$
  that does not contain the $0$-\kl{labeled} \kl{node}~$\Node[1]$,
  and identifying $\Node[2]$ with~$\Node[2]'$.
  Let $\IdMap'$ and~$\CertifMap'$ be
  the restrictions of $\IdMap$ and~$\CertifMap$ to~$\Graph'$.
  Notice that
  $\IdMap'$~is $\IdentRadius$-\kl{locally unique},
  $\CertifMap'$~is
  $\Tuple{\CertifRadius_1, \CertifPolynomial}$\nobreakdash-\kl(certificate){bounded},
  and
  $\Result{\Machine}{\Graph', \IdMap', \CertifMap'} \equiv \Accept$
  because every \kl{node} of~$\Graph'$ reaches the same \kl{verdict}
  in~$\Result{\Machine}{\Graph', \IdMap', \CertifMap'}$
  as in~$\Result{\Machine}{\Graph, \IdMap, \CertifMap}$.
  We conclude that $\Graph' \in \NOTALLSELECTED \cap \bGRAPH{3}$,
  which contradicts the fact that
  all \kl{nodes} of~$\Graph'$ are \kl{labeled} with~$1$.
\end{proof}

\begin{corollary}
  \label{cor:lp-complementation}
  The class of \kl{graph properties}
  \kl{decidable} by a \kl{locally polynomial machine}
  is not closed under \kl{complementation},
  even when restricted to \kl{graphs} of \kl{bounded structural degree}.
  More precisely,
  $\LP\On{\bGRAPH{\MaxDegree}} \neq \coLP\On{\bGRAPH{\MaxDegree}}$
  for all $\MaxDegree \geq 3$,
  and a fortiori
  $\LP \neq \coLP$.
\end{corollary}

\begin{proof}
  This follows immediately from Proposition~\ref{prp:colp-vs-nlp}.
\end{proof}

\begin{corollary}
  \label{cor:non-three-colorable-not-in-nlp}
  $\NONCOLORABLE{3}$ does not lie in $\NLP$.
\end{corollary}

\begin{proof}
  By Theorem~\ref{thm:three-colorable},
  we know that
  $\COLORABLE{3}$ is $\NLP$-\kl{hard},
  which by duality means that
  $\NONCOLORABLE{3}$ is $\coNLP$-\kl{hard}.
  This implies that $\NONCOLORABLE{3}$ cannot lie in~$\NLP$,
  since otherwise we could show that $\coNLP \subseteq \NLP$,
  contradicting Proposition~\ref{prp:colp-vs-nlp}.
  The argument is analogous to the proof of
  Corollary~\ref{cor:three-colorable-not-in-lp},
  with the additional observation that
  a \kl{node} can simulate an $\NLP$\nobreakdash-\kl{verifier}
  on its \kl{cluster}
  by interpreting its own \kl{certificate} as an \kl{encoding} of
  the \kl{certificates} of all \kl{nodes} of the \kl{cluster}.
\end{proof}

\begin{corollary}
  \label{cor:noneulerian-hamiltonian-nonhamiltonian-not-in-nlp}
  None of the following \kl{graph properties} lies in $\NLP$:
  $\NONEULERIAN$, $\HAMILTONIAN$, $\NONHAMILTONIAN$.
\end{corollary}

\begin{proof}
  By the duals of
  Propositions~\ref{prp:eulerian-lp-complete} and~\ref{prp:hamiltonian-lp-hard},
  and by Proposition~\ref{prp:hamiltonian-colp-hard},
  all these \kl{graph properties} are $\coLP$-\kl{hard}.
  Using the same argument as in the proofs of
  Corollaries~\ref{cor:three-colorable-not-in-lp}
  and~\ref{cor:non-three-colorable-not-in-nlp},
  we conclude from Proposition~\ref{prp:colp-vs-nlp}
  that none of them lies in~$\NLP$.
\end{proof}

%–––––––––––––––––––––––––––––––––––––––––––––––––––––––––––––––––––––––––––––––
\subsection{Climbing up the hierarchy}
\label{ssec:climb-up}

We now turn to the more technical part of our separation proof,
which uses the connection to logic
provided by Theorem~\ref{thm:local-hierarchy-equivalence}
to leverage two results about
\kl{monadic second-order logic} on \kl{pictures}
(stated in Theorems~\ref{thm:mso-hierarchy} and~\ref{thm:equivalence-ts-emso}).
In Section~\ref{sssec:digression-on-pictures},
we use these results
to show that the \kl{local second-order hierarchy} is infinite
when restricted to \kl{pictures}.
Then,
in Section~\ref{sssec:from-pictures-to-graphs},
we transfer this infiniteness result
from \kl{pictures} to \kl{graphs},
where
the \kl{local second-order hierarchy}
coincides with
the \kl{locally polynomial hierarchy}.
As this two-part proof
may seem a little convoluted,
we also briefly explain
in Section~\ref{sssec:direct-infiniteness-proof}
why a more straightforward approach would yield a weaker result.

\AP
\intro{Monadic second-order logic} is the fragment of \kl{second-order logic}
that can only \kl{quantify} over sets instead of arbitrary relations.
That is,
for \kl{second-order quantifications} of the form
$\ExistsRel{\SOVar} \, \Formula[1]$,
the \kl{second-order variable}~$\SOVar$
must necessarily be of \kl{arity}~$1$.
We analogously define
\intro{local monadic second-order logic} as
the corresponding fragment of \kl{local second-order logic}.

\AP
Accordingly,
the \intro[monadic second-order hierarchy]{monadic}
\phantomintro{local monadic second-order hierarchy}%
versions of
the \kl[second-order hierarchy]{second-order}
and \kl{local second-order hierarchies}
are obtained by restricting
\kl{second-order quantification} in each class
to \kl{unary} relations.
To denote the monadic classes,
we add the letter “m” as a prefix to the corresponding non-monadic classes.
This gives us the classes of \kl{formulas}
$\intro*\mSigmaFOL{\Level}$,
$\intro*\mPiFOL{\Level}$,
$\intro*\mSigmaLFOL{\Level}$,
$\intro*\mPiLFOL{\Level}$,
for~$\Level \in \Naturals$,
and the classes of \kl{structure properties}
$\intro*\mSigmaFO{\Level}$,
$\intro*\mPiFO{\Level}$,
$\intro*\mSigmaLFO{\Level}$,
$\intro*\mPiLFO{\Level}$
that can be \kl{defined} by \kl{formulas} of the respective classes.
We will refer to $\mSigmaFOL{1}$ and~$\mSigmaLFOL{1}$
as the the \intro(monadic){existential fragments}
of \kl{monadic second-order logic} and \kl{local monadic second-order logic}.

%...............................................................................
\subsubsection{A digression on pictures}
\label{sssec:digression-on-pictures}

\AP
In the context of this paper,
\kl{pictures} are matrices of fixed-length binary strings.
More precisely,
for any $\BitLen \in \Naturals$ and $\PicHeight,\PicWidth \in \Positives$,
a $\BitLen$-bit \intro{picture}~$\Picture$
of \intro(picture){size} $\Tuple{\PicHeight, \PicWidth}$
is an $(\PicHeight \times \PicWidth)$-matrix
whose entries are bit strings of length~$\BitLen$.
We refer to the positions
$\Tuple{i, j} \in \Range*{\PicHeight} \times \Range*{\PicWidth}$
as $\Picture$'s~\intro{pixels}.
Alternatively,
a \kl{picture} can also be viewed as a function
$\Picture \colon \Range*{\PicHeight} \times \Range*{\PicWidth}
          \to \Set{0, 1}^{\BitLen}$.
All the usual terminology of matrices applies;
for instance,
the \kl{pixel} $\Tuple{1, 1}$ is referred to as the top-left corner.
The class of all $\BitLen$-bit \kl{pictures}
(of arbitrary \kl(picture){size})
is denoted by $\intro*\PIC{\BitLen}$,
and any subset $\Property \subseteq \PIC{\BitLen}$
is called a \intro{picture property}.

\AP
As with \kl{graphs},
we can evaluate \kl{logical formulas} on \kl{pictures}
by identifying each $\BitLen$-bit \kl{picture}~$\Picture$ with a \kl{structure}.
Formally,
the \intro(picture){structural representation} of~$\Picture$
is the \kl{structure}
$\intro*\StructReprP{\Picture} =
 \Tuple{
   \Domain{\StructReprP{\Picture}},
   \BitSet{1}{\StructReprP{\Picture}},
   \dots,
   \BitSet{\BitLen}{\StructReprP{\Picture}},
   {\LinkRel{1}{\StructReprP{\Picture}}},
   {\LinkRel{2}{\StructReprP{\Picture}}}
 }$
of \kl{signature}~$\Tuple{\BitLen, 2}$
with
\kl{domain}
$\Domain{\StructReprP{\Picture}} =
  \Range*{\PicHeight} \times \Range*{\PicWidth}$,
\kl{unary} relations
$\BitSet{1}{\StructReprP{\Picture}},
 \dots,
 \BitSet{\BitLen}{\StructReprP{\Picture}}
 \subseteq
 \Domain{\StructReprP{\Picture}}$\!
such that
$\Tuple{i, j} \in \BitSet{k}{\StructReprP{\Picture}}$
precisely when the $k$-th bit of $\Picture(i, j)$ is~$1$,
and “vertical” and “horizontal” successor relations
${\LinkRel{1}{\StructReprP{\Picture}}},
 {\LinkRel{2}{\StructReprP{\Picture}}}
 \subseteq
 (\Domain{\StructReprP{\Picture}} \!\times \Domain{\StructReprP{\Picture}})$
such that
$\Tuple{i, j} \LinkRel{1}{\StructReprP{\Picture}} \Tuple{i + 1, j}$
and
$\Tuple{i, j} \LinkRel{2}{\StructReprP{\Picture}} \Tuple{i, j + 1}$
for all suitable $i, j$.
An example is provided in Figure~\ref{fig:picture}.

\begin{figure}[htb]
  \centering
  % https://tex.stackexchange.com/a/55604
\makeatletter
\tikzset{circle split part fill/.style  args={#1,#2}{%
    alias=tmp@name, % Jake's idea !!
    postaction={%
      insert path={%
        \pgfextra{%
          \pgfpointdiff{\pgfpointanchor{\pgf@node@name}{center}}%
          {\pgfpointanchor{\pgf@node@name}{east}}%
          \pgfmathsetmacro\insiderad{\pgf@x}
          \fill[#1] (\pgf@node@name.base) ([xshift=-\pgflinewidth]\pgf@node@name.east) arc
          (0:180:\insiderad-\pgflinewidth)--cycle;
          \fill[#2] (\pgf@node@name.base) ([xshift=\pgflinewidth]\pgf@node@name.west)  arc
          (180:360:\insiderad-\pgflinewidth)--cycle;
        }}}}}
\makeatother

\begin{tikzpicture}[
    semithick,>=stealth',on grid,
    inner sep=0,
    pixel/.style={minimum size=7ex},
    vertex/.style={circle split,rotate=90,minimum size=2ex},
    00/.style={draw=white,circle split part fill={white,white}},
    01/.style={draw=black,circle split part fill={white,black}},
    10/.style={draw=black,circle split part fill={black,white}},
    11/.style={draw=black,circle split part fill={black,black}},
  ]
  \def\pixelDist{7ex}
  \begin{scope}
    \matrix (p) [matrix of math nodes,left delimiter={[},right delimiter={]},
                 inner xsep=0ex,inner ysep=0.5ex,
                 row sep={\pixelDist,between origins},
                 column sep={\pixelDist,between origins},
                 matrix anchor=p-1-1.center] {
      00 & 01 & 00 & 01 \\
      10 & 11 & 10 & 11 \\
      00 & 01 & 00 & 01 \\
    };
    \node [left=6ex of p-1-1] {$\Picture$:};
  \end{scope}
  \begin{scope}[xshift=40ex]
    \node[vertex,00] (p11) {};
    \node[vertex,10] (p21) [below=\pixelDist of p11] {};
    \node[vertex,00] (p31) [below=\pixelDist of p21] {};
    \node[vertex,01] (p12) [right=\pixelDist of p11] {};
    \node[vertex,11] (p22) [below=\pixelDist of p12] {};
    \node[vertex,01] (p32) [below=\pixelDist of p22] {};
    \node[vertex,00] (p13) [right=\pixelDist of p12] {};
    \node[vertex,10] (p23) [below=\pixelDist of p13] {};
    \node[vertex,00] (p33) [below=\pixelDist of p23] {};
    \node[vertex,01] (p14) [right=\pixelDist of p13] {};
    \node[vertex,11] (p24) [below=\pixelDist of p14] {};
    \node[vertex,01] (p34) [below=\pixelDist of p24] {};
    \node [left=5ex of p11] {$\StructReprP{\Picture}$:};
    \foreach \n in {p11,p31,p13,p33}
      \node[draw,circle,minimum size=2ex] at (\n) {};
    \path[->] (p11) edge (p21)
              (p21) edge (p31)
              (p12) edge (p22)
              (p22) edge (p32)
              (p13) edge (p23)
              (p23) edge (p33)
              (p14) edge (p24)
              (p24) edge (p34);
    \path[->,densely dotted]
              (p11) edge (p12)
              (p12) edge (p13)
              (p13) edge (p14)
              (p21) edge (p22)
              (p22) edge (p23)
              (p23) edge (p24)
              (p31) edge (p32)
              (p32) edge (p33)
              (p33) edge (p34);
  \end{scope}
\end{tikzpicture}

%%% Local Variables:
%%% mode: latex
%%% TeX-master: "../lph-paper"
%%% End:
  \caption{
    \emph{(repeated from Figure~\ref{fig:picture-overview})}
    A $2$-bit \kl{picture}~$\Picture$ of \kl(picture){size} $\Tuple{3, 4}$
    and its \kl(picture){structural representation}~$\StructReprP{\Picture}$.
    The sets $\BitSet{1}{\StructReprP{\Picture}}$
    and~$\BitSet{2}{\StructReprP{\Picture}}$
    are represented by \kl{elements} whose left and right halves,
    respectively,
    are colored black,
    and the relations
    $\LinkRel{1}{\StructReprP{\Picture}}$ and~$\LinkRel{2}{\StructReprP{\Picture}}$
    are represented by solid and dotted arrows,
    respectively.
  }
  \label{fig:picture}
\end{figure}

\kl{Monadic second-order logic} on \kl{pictures}
has been fairly well-understood since the early 2000s.
In particular,
Matz, Schweikardt, and Thomas~\cite[Thm.~1]{DBLP:journals/iandc/MatzST02}
have shown that the \kl{monadic second-order hierarchy} is infinite
on several kinds of \kl{structures},
including \kl{pictures}.
Here,
we state only the part of their result we need,
in a stronger form obtained by
Matz~\cite[Thm.~2.26]{DBLP:journals/tcs/Matz02}.

\begin{theorem}
  [Matz, Schweikardt, Thomas
  \cite{DBLP:journals/tcs/Matz02,DBLP:journals/iandc/MatzST02}]
  \label{thm:mso-hierarchy}
  The \kl{monadic second-order hierarchy} on \kl{pictures} is infinite.
  More precisely,
  $\mSigmaFO{\Level}\On{\PIC{\BitLen}}$
  and
  $\mPiFO{\Level}\On{\PIC{\BitLen}}$
  are incomparable,
  which implies that
  $\mSigmaFO{\Level}\On{\PIC{\BitLen}} \subsetneqq
   \mPiFO{\Level + 1}\On{\PIC{\BitLen}}$
  and
  $\mPiFO{\Level}\On{\PIC{\BitLen}} \subsetneqq
   \mSigmaFO{\Level + 1}\On{\PIC{\BitLen}}$,
  for all $\Level \in \Positives$ and $\BitLen \in \Naturals$.
\end{theorem}

The ultimate goal of Section~\ref{ssec:climb-up}
is to transfer part of Theorem~\ref{thm:mso-hierarchy}
from \kl{monadic second-order logic} on \kl{pictures}
to \kl{local second-order logic} on \kl{graphs}
(and thus to the \kl{locally polynomial hierarchy}).
This will culminate in Theorem~\ref{thm:locally-polynomial-hierarchy}
on page~\pageref{thm:locally-polynomial-hierarchy}.
As a first milestone towards this goal,
we establish a partial levelwise equivalence between the two logics
when restricted to \kl{pictures}.
We do this in two steps,
showing roughly that the expressive power of
\kl{local second-order logic} on \kl{pictures} remains unaffected
if we first weaken \kl{second-order quantification}
and then strengthen \kl{first-order quantification}.
The outcome is presented in Theorem~\ref{thm:equivalence-lso-mso}.

We begin by reducing unrestricted \kl{second-order quantification}
to \kl{quantification} over sets,
exploiting the fact that \kl{local second-order logic} on \kl{pictures}
allows us to represent arbitrary relations as collections of sets.
Since the analogous observation
also holds for \kl{graphs} of \kl{bounded structural degree},
which will be useful in Section~\ref{ssec:complete-picture},
we prove a more general statement
for all \kl{structures} of \kl(structure){bounded degree}.

\AP
Formally,
a \kl{structure}~$\Struct$ is of $\MaxDegree$-\intro(structure){bounded degree},
for $\MaxDegree \in \Naturals$,
if every \kl{element}
$\Element[1] \inS \Struct$
is connected to at most~$\MaxDegree$ other \kl{elements},
i.e.,
$\Card{
   \SetBuilder{
     \Element[2] \inS \Struct
   }{
     \Element[2] \NeighborRel{\Struct} \! \Element[1]
   }
 }
 \leq \MaxDegree$.
We denote by $\intro*\bSTRUCT{\MaxDegree}$
the set of all \kl{structures} of $\MaxDegree$-\kl(structure){bounded degree}.
Notice that
$\PIC{\BitLen} \subseteq \bSTRUCT{4}$
and
$\bGRAPH{\MaxDegree} \subseteq \bSTRUCT{\MaxDegree}$
for all $\BitLen, \MaxDegree \in \Naturals$.

\begin{proposition}
  \label{prp:equivalence-lso-lmso}
  When restricted to \kl{structures} of \kl(structure){bounded degree},
  each level of the \kl{local second-order hierarchy}
  is equivalent to the corresponding level of
  the \kl{local monadic second-order hierarchy}.
  More precisely,
  $\SigmaLFO{\Level}\On{\bSTRUCT{\MaxDegree}} =
   \mSigmaLFO{\Level}\On{\bSTRUCT{\MaxDegree}}$
  and
  $\PiLFO{\Level}\On{\bSTRUCT{\MaxDegree}} =
   \mPiLFO{\Level}\On{\bSTRUCT{\MaxDegree}}$
  for all $\Level, \MaxDegree \in \Naturals$.
\end{proposition}

\begin{proof}
  The statement holds trivially for $\Level = 0$,
  so in the following we assume that $\Level > 0$.
  We need only show inclusions from left to right,
  since
  $\mSigmaLFOL{\Level}$
  is a syntactic fragment of
  $\SigmaLFOL{\Level}$
  and
  $\mPiLFOL{\Level}$
  is a syntactic fragment of
  $\PiLFOL{\Level}$.
  Consider any \kl{formula}
  $\Formula[1] \in
  \SigmaLFOL{\Level} \cup \PiLFOL{\Level}$,
  and let~$\Radius \in \Naturals$ be
  the maximum nesting depth
  of \kl(quantifier){bounded} \kl{first-order quantifiers}
  in~$\Formula[1]$.
  Intuitively,
  $\Radius$~is the distance
  up to which each \kl{element} can “see”
  when evaluating~$\Formula[1]$
  on a \kl{structure}.

  We now show how to use multiple sets of \kl{elements}
  to encode a single relation of arbitrary \kl{arity},
  taking advantage of the fact that
  we have a bound on the number of \kl{elements}
  that lie within distance~$2\Radius$ of any fixed \kl{element}.
  More precisely,
  for any \kl{structure}~$\Struct$
  of $\MaxDegree$-\kl(structure){bounded degree},
  the number of \kl{elements}
  reachable in at most $2\Radius$~steps
  from a given \kl{element}
  cannot be greater than
  $\MaxName = \sum_{i \in \Range{2\Radius}} \MaxDegree^i$.
  Hence,
  we can assign each \kl{element}
  $\Element \inS \Struct$
  a number $\Name \in \Range{\MaxName}*$,
  called its \emph{name},
  that uniquely identifies it within distance~$2\Radius$.
  (If $\Struct$ is the \kl{structural representation} of a \kl{graph},
  this is like an $\Radius$-\kl{locally unique} \kl{identifier assignment},
  but extended to all \kl{elements}.)
  These names will be represented
  by a collection of \kl{unary} \kl{variables}
  $\Tuple{\MSOVar[1]_{\Name}}_{\Name \in \Range{\MaxName}*}$,
  each to be \kl{interpreted} as
  the set of \kl{elements} having the corresponding name~$\Name$.
  Based on this,
  if~$\Formula[1]$ contains a \kl{second-order variable}~$\SOVar$
  of \kl{arity}~$\Arity \geq 2$,
  we proceed as follows to represent~$\SOVar$
  by a collection of \kl{unary} \kl{variables}
  $\Tuple{
    \MSOVar[2]_{\SOVar(*, \Name_2, \dots, \Name_{\Arity})}
  }_{\Name_2, \dots, \Name_{\Arity} \in \Range{\MaxName}*}$:
  for any \kl{elements}
  $\Element_1, \dots, \Element_{\Arity} \inS \Struct$
  such that
  $\Element_2, \dots, \Element_{\Arity}$
  are located within distance~$2\Radius$ of~$\Element_{1}$
  and are assigned the names
  $\Name_2, \dots, \Name_{\Arity} \in \Range{\MaxName}*$,
  we stipulate that $\Element_1$ lies in the set
  $\MSOVar[2]_{\SOVar(*, \Name_2, \dots, \Name_{\Arity})}$
  if and only if the tuple
  $\Tuple{\Element_1, \dots, \Element_{\Arity}}$
  lies in the relation~$\SOVar$.
  Notice that we do not encode the entire relation~$\SOVar$,
  but only its restriction to tuples whose \kl{elements} lie
  within a distance of at most~$2\Radius$
  from the first \kl{element}~$\Element_1$.
  This is sufficient for our purposes,
  since $\Formula[1]$ cannot refer to any other tuples.
  (Intuitively,
  when evaluating~$\Formula[1]$,
  each \kl{element}~$\Element$ can only “see” tuples
  whose \kl{elements} are all reachable within $\Radius$~steps,
  so those \kl{elements} cannot be further than~$2\Radius$ apart from each other,
  and in particular not further than~$2\Radius$ from the first \kl{element}.)

  Formally,
  given~$\Radius$,
  we define a function~$\Translation_{\Radius}$
  that translates~$\Formula[1]$
  into a \kl{formula} of \kl{local monadic second-order logic}
  with \kl{free} \kl{variables}
  $\Tuple{\MSOVar[1]_{\Name}}_{\Name \in \Range{\MaxName}*}$.
  The translation works under the assumption
  that these \kl{variables}
  will be \kl{interpreted} as $\Radius$-locally unique names,
  as described above.
  We define~$\Translation_{\Radius}$ by induction
  on the structure of \kl{local second-order} \kl{formulas}
  (bearing in mind
  that the unique \kl(quantifier){unbounded} \kl{first-order quantifier}
  is always \kl{universal},
  and that both \kl{existential} and \kl{universal} \kl{second-order quantifiers}
  must be handled explicitly,
  since they are syntactically restricted to the \kl{quantifier} prefix
  and thus cannot occur under negation):
  \begin{itemize}
  \item \kl{Atomic} \kl{formulas}
    that do not involve a \kl{relation variable} of higher \kl{arity}
    are kept unchanged, i.e.,
    $\Translation_{\Radius}(\Formula[3]) = \Formula[3]$
    if $\Formula[3]$ is of the form
    $\BitTrue{i}{\FOVar[1]}$,\:
    $\FOVar[1] \Linked{i} \FOVar[2]$,\:
    $\FOVar[1] \Equal \FOVar[2]$, or
    $\InRel{\SOVar}{\FOVar[1]}$,
    where $\SOVar$ is of \kl{arity}~$1$.
  \item \kl{Atomic} \kl{formulas}
    involving a \kl{relation variable}~$\SOVar$ of \kl{arity} $\Arity \geq 2$
    are rewritten in terms of the corresponding \kl{unary} \kl{variables}:
    \begin{equation*}
      \Translation_{\Radius}
      \bigr(
        \InRel{\SOVar}{\FOVar_1, \dots, \FOVar_{\Arity}}
      \bigl)
      \: = \,
      \BigOR_{\lalign{\Name_2, \dots, \Name_{\Arity} \in \Range{\MaxName}*}}
      \Bigl(
        \InRel{\MSOVar[2]_{\SOVar(*, \Name_2, \dots, \Name_{\Arity})}}{\FOVar_1}
        \, \AND \,
        \InRel{\MSOVar[1]_{\Name_2}}{\FOVar_2}
        \, \AND \dots \AND \,
        \InRel{\MSOVar[1]_{\Name_{\Arity}}}{\FOVar_{\Arity}}
      \Bigr).
    \end{equation*}
  \item \kl{Boolean connectives},
    \kl{first-order quantifiers}, and
    \kl{second-order quantifiers} over \kl{unary} relations
    are preserved:
    $\Translation_{\Radius}(\NOT \Formula[3]) = \,
     \NOT \Translation_{\Radius}(\Formula[3])$,\:
    $\Translation_{\Radius}(\Formula[3]_1 \OR \Formula[3]_2) = \,
     \Translation_{\Radius}(\Formula[3]_1) \OR
     \Translation_{\Radius}(\Formula[3]_2)$,\:
    $\Translation_{\Radius}(
       \ExistsNb{\FOVar[1]}{\FOVar[2]} \; \Formula[3]
     ) = \,
     \ExistsNb{\FOVar[1]}{\FOVar[2]} \; \Translation_{\Radius}(\Formula[3])$,\:
    $\Translation_{\Radius}(
       \ForAll{\FOVar[1]} \, \Formula[3]
     ) = \,
     \ForAll{\FOVar[1]} \, \Translation_{\Radius}(\Formula[3])$,
    and
    $\Translation_{\Radius}(
       \QuantifierRel{\SOVar} \, \Formula[3]
     ) = \,
     \QuantifierRel{\SOVar} \, \Translation_{\Radius}(\Formula[3])$
    if $\SOVar$ is of \kl{arity}~$1$,
    where $\QuantifierRel{}$~can be $\ExistsRel{}$ or~$\ForAllRel{}$.
  \item Each \kl{second-order quantifier}
    binding a \kl{relation variable}~$\SOVar$ of \kl{arity} $\Arity \geq 2$
    is replaced by a collection of \kl{second-order quantifiers}
    binding the corresponding \kl{unary} \kl{variables}:
    $\Translation_{\Radius}(
       \QuantifierRel{\SOVar} \, \Formula[3]
     ) = \,
     \QuantifierRel{
       \bigTuple{
         \MSOVar[2]_{\SOVar(*, \Name_2, \dots, \Name_{\Arity})}
       }_{\Name_2, \dots, \Name_{\Arity} \in \Range{\MaxName}*}
     } \,
     \Translation_{\Radius}(\Formula[3])$,\,
    where $\QuantifierRel{}$~can be $\ExistsRel{}$ or~$\ForAllRel{}$.
  \end{itemize}

  Applying this translation to the initial \kl{formula}~$\Formula[1]$,
  we obtain a \kl{formula} $\Translation_{\Radius}(\Formula[1])$
  that expresses the same \kl(structure){property} as~$\Formula[1]$
  on \kl{structures} of $\MaxDegree$-\kl(structure){bounded degree},
  assuming that the \kl{variables}
  $\Tuple{\MSOVar[1]_{\Name}}_{\Name \in \Range{\MaxName}*}$
  are \kl{interpreted} as intended.
  Notice that $\Translation_{\Radius}$ preserves
  the alternation level of \kl{second-order quantifiers}.
  Thus,
  if
  $\Formula[1] \in \SigmaLFOL{\Level}$,
  then
  $\Translation_{\Radius}(\Formula[1]) \in \mSigmaLFOL{\Level}$,
  and if
  $\Formula[1] \in \PiLFOL{\Level}$,
  then
  $\Translation_{\Radius}(\Formula[1]) \in \mPiLFOL{\Level}$.

  It remains to bind
  $\Vector{\MSOVar[1]} =
   \Tuple{\MSOVar[1]_{\Name}}_{\Name \in \Range{\MaxName}*}$
  with a \kl{second-order quantifier}
  that is relativized to admissible \kl{interpretations}.
  To maintain the level of alternation,
  we choose the same type of \kl{quantifier}
  as the first \kl{second-order quantifier} of~$\Formula[1]$.
  Let us first assume that
  $\Formula[1] \in \SigmaLFOL{\Level}$,
  and let the result of the translation be
  $\Translation_{\Radius}(\Formula[1]) =
   \ExistsRel{\Vector{\MSOVar[2]}_1}
   \ForAllRel{\Vector{\MSOVar[2]}_2} \dots
   \QuantifierRel{\Vector{\MSOVar[2]}_{\Level}} \,
   \ForAll{\FOVar[1]} \, \Formula[2]\Of{\FOVar[1]}$,
  where each $\Vector{\MSOVar[2]}_i$
  is a tuple of \kl{unary} \kl{variables},
  and $\QuantifierRel{}$ is~$\ForAllRel{}$ if $\Level$~is even
  and $\ExistsRel{}$ otherwise.
  From this we derive the $\mSigmaLFOL{\Level}$-\kl{formula}
  \begin{equation*}
    \Formula[1]' = \,
    \ExistsRel{\Vector{\MSOVar[1]}} \,
    \ExistsRel{\Vector{\MSOVar[2]}_1} \,
    \ForAllRel{\Vector{\MSOVar[2]}_2} \dots
    \QuantifierRel{\Vector{\MSOVar[2]}_{\Level}} \:
    \ForAll{\FOVar[1]} \,
    \bigl(
      \Fml{UniqueName}\Of{\FOVar[1]}
      \, \AND \,
      \Formula[2]\Of{\FOVar[1]}
    \bigr),
  \end{equation*}
  which is equivalent to~$\Formula[1]$
  on \kl{structures} of $\MaxDegree$-\kl(structure){bounded degree}.
  Here,
  $\Fml{UniqueName}\Of{\FOVar[1]}$
  is the $\BFL$-\kl{formula}
  \begin{equation*}
    \Bigl( \,
      \BigOR_{\lalign{\Name \in \Range{\MaxName}*}}
      \InRel{\MSOVar[1]_{\Name}}{\FOVar[1]}
    \, \Bigr)
    \, \AND \,
    \Bigl( \,
      \BigAND_{\lalign{\Name_1,\Name_2 \in \Range{\MaxName}*:\, \Name_1 \neq \Name_2}}
      \NOT
      \bigl(
        \InRel{\MSOVar[1]_{\Name_1}}{\FOVar[1]}
        \AND
        \InRel{\MSOVar[1]_{\Name_2}}{\FOVar[1]}
      \bigr)
    \, \Bigr)
    \, \AND \:
    \ForAllLoc{\FOVar[2]}{2\Radius}{\FOVar[1]}
    \Bigl( \,
      \BigAND_{\lalign{\Name \in \Range{\MaxName}*}}
      \NOT
      \bigl(
        \InRel{\MSOVar[1]_{\Name}}{\FOVar[1]}
        \AND
        \InRel{\MSOVar[1]_{\Name}}{\FOVar[2]}
      \bigr)
    \, \Bigr),
  \end{equation*}
  which states that
  $\FOVar[1]$~is assigned a name
  that is unique within distance~$2\Radius$ of~$\FOVar[1]$.
  (This is very similar to the \kl{formula} $\WellColored\Of{\FOVar[1]}$
  from Example~\ref{ex:3-colorable} on page~\pageref{ex:3-colorable}.)

  The case where
  $\Formula[1] \in \PiLFOL{\Level}$
  is quite similar,
  but this time we bind~$\Vector{\MSOVar[1]}$
  with a \kl{universal} \kl{quantifier}.
  To relativize this \kl{quantifier} appropriately,
  we ensure that
  the $\BFL$\nobreakdash-\kl{formula} evaluated at each \kl{element}~$\FOVar[1]$
  holds true
  whenever
  an \kl{element} within distance~$\Radius$ of~$\FOVar[1]$
  violates the condition of $\Radius$-locally unique names.
  (This is a simple form of the quantifier relativization
  performed in the proofs of
  Lemma~\ref{lem:restrictive-arbiters} and
  Theorem~\ref{thm:local-hierarchy-equivalence};
  it is sound because
  the names of any conflicting \kl{elements}
  can always be replaced with $\Radius$-locally unique names
  without affecting the names of any other \kl{elements}.)
  Thus,
  if the translation yields
  $\Translation_{\Radius}(\Formula[1]) =
   \ForAllRel{\Vector{\MSOVar[2]}_1}
   \ExistsRel{\Vector{\MSOVar[2]}_2} \dots
   \QuantifierRel{\Vector{\MSOVar[2]}_{\Level}} \,
   \ForAll{\FOVar[1]} \, \Formula[2]\Of{\FOVar[1]}$,
  then we construct the $\mPiLFOL{\Level}$-\kl{formula}
  \begin{equation*}
    \Formula[1]' = \,
    \ForAllRel{\Vector{\MSOVar[1]}} \,
    \ForAllRel{\Vector{\MSOVar[2]}_1} \,
    \ExistsRel{\Vector{\MSOVar[2]}_2} \dots
    \QuantifierRel{\Vector{\MSOVar[2]}_{\Level}} \:
    \ForAll{\FOVar[1]} \,
    \bigl( \,
      \ForAllLoc{\FOVar[2]}{\Radius}{\FOVar[1]} \,
      (\Fml{UniqueName}\Of{\FOVar[2]})
      \, \IMP \,
      \Formula[2]\Of{\FOVar[1]}
    \bigr),
  \end{equation*}
  which is equivalent to~$\Formula[1]$
  on \kl{structures} of $\MaxDegree$-\kl(structure){bounded degree}.
\end{proof}

\AP
Next,
we want to strengthen \kl{first-order quantification}.
To do this,
we take advantage of an automaton model for \kl{pictures}
introduced by Giammarresi and Restivo~\cite{DBLP:journals/ijprai/GiammarresiR92},
which is closely related to \kl{monadic second-order logic}.
A “machine” in this model,
called a $\BitLen$\nobreakdash-bit \intro[tiling system]{tiling~system},
is defined as a tuple
$\TilingSystem = \Tuple{\TStateSet, \TTileSet}$,
where
$\BitLen$~is a nonnegative integer,
$\TStateSet$ is a finite set of \intro(tiling system){states},
and
$\TTileSet \subseteq
 \bigl(
   (\Set{0, 1}^{\BitLen} \times \TStateSet) \cup
   \Set{\TBoundarySymbol}
 \bigr)^4$
is a set of $(2 \times 2)$-\intro{tiles}.
Each \kl{tile} in~$\TTileSet$
consists of entries that are
either a $\BitLen$-bit string accompanied by
a \kl(tiling system){state} in~$\TStateSet$,
or the special boundary symbol $\intro*\TBoundarySymbol$
(assumed not to be contained in $\Set{0, 1}^{\BitLen} \times \TStateSet$).

\AP
A $\BitLen$-bit \kl{tiling system}
$\TilingSystem = \Tuple{\TStateSet, \TTileSet}$
operates similarly to
a nondeterministic finite automaton generalized to two dimensions:
given a \kl{picture}~$\Picture$,
it first nondeterministically assigns
a \kl(tiling system){state} of~$\TStateSet$ to each \kl{pixel} of~$\Picture$,
and then checks that this assignment of \kl(tiling system){states}
respects the “transitions” that are allowed by~$\TTileSet$.
More precisely,
a $\BitLen$-bit \kl{picture}~$\Picture$
of \kl(picture){size} $\Tuple{\PicHeight, \PicWidth}$
is \intro(tiling system){accepted} by~$\TilingSystem$
if there exists an assignment
$\Range*{\PicHeight} \times \Range*{\PicWidth} \to \TStateSet$
such that
each $(2 \times 2)$-subblock of~$\Picture$
matches some \kl{tile} of~$\TTileSet$,
assuming that the entire \kl{picture} is surrounded by a frame
consisting of $\TBoundarySymbol$'s
(to detect the borders),
and that a \kl{pixel} matches
$\Tuple{\String, \TState} \in \Set{0, 1}^{\BitLen} \times \TStateSet$
precisely if its value is~$\String$
and its assigned \kl(tiling system){state} is~$\TState$.
The \kl{picture property} \intro{recognized} by~$\TilingSystem$
consists of those $\BitLen$-bit \kl{pictures}
that are \kl(tiling system){accepted} by~$\TilingSystem$.
We write $\intro*\TS$
for the class of \kl{picture properties}
that are \kl{recognized} by some \kl{tiling system}.

Exploiting a locality property of \kl{first-order logic},
Giammarresi, Restivo, Seibert, and Thomas~%
\cite[Thm.~3.1]{DBLP:journals/iandc/GiammarresiRST96}
have shown that \kl{tiling systems} capture precisely
the nondeterministic level of the \kl{monadic second-order hierarchy}
on \kl{pictures}:

\begin{theorem}
  [Giammarresi, Restivo, Seibert, Thomas
  {\cite{DBLP:journals/iandc/GiammarresiRST96}}]
  \label{thm:equivalence-ts-emso}
  \kl{Tiling systems} are equivalent to
  the \kl(monadic){existential fragment} of \kl{monadic second-order logic}
  on \kl{pictures}.
  That is,
  $\TS\On{\PIC{\BitLen}} =
   \mSigmaFO{1}\On{\PIC{\BitLen}}$
  for all $\BitLen \in \Naturals$.
\end{theorem}

This result gives us the key to move
from \kl(quantifier){bounded} to arbitrary \kl{first-order quantification}.
The following corollary is based on the observation
that \kl{tiling systems} can be easily described
in \kl{local monadic second-order logic}.

\begin{corollary}
  \label{cor:equivalence-lemso-emso}
  When restricted to \kl{pictures},
  the \kl(monadic){existential fragment} of \kl{local monadic second-order logic}
  is equivalent to
  the \kl(monadic){existential fragment} of \kl{monadic second-order logic}.
  That is,
  $\mSigmaLFO{1}\On{\PIC{\BitLen}} =
   \mSigmaFO{1}\On{\PIC{\BitLen}}$
  for all $\BitLen \in \Naturals$.
\end{corollary}

\begin{proof}
  Since $\mSigmaLFOL{1}$ can be seen as
  a syntactic fragment of $\mSigmaFOL{1}$,
  it suffices to show that
  $\TS\On{\PIC{\BitLen}} \subseteq
   \mSigmaLFO{1}\On{\PIC{\BitLen}}$
  and then apply Theorem~\ref{thm:equivalence-ts-emso}.
  We thus require a translation~$\Translation$
  from \kl{tiling systems}
  to $\mSigmaLFOL{1}$-\kl{sentences} on \kl{pictures}
  such that
  a \kl{picture}~$\Picture$ is \kl(tiling system){accepted}
  by a \kl{tiling system}~$\TilingSystem$
  if and only if
  its \kl(picture){structural representation}~$\StructReprP{\Picture}$
  \kl{satisfies}~$\Translation(\TilingSystem)$.

  By inspecting the proof of Theorem~\ref{thm:equivalence-ts-emso}
  in~{\cite[Thm.~3.1]{DBLP:journals/iandc/GiammarresiRST96}},
  it is easy to see that
  the $\mSigmaFOL{1}$\nobreakdash-\kl{sentence} provided there
  can be rewritten as an equivalent $\mSigmaLFOL{1}$-\kl{sentence},
  essentially by replacing \kl(quantifier){unbounded} \kl{first-order quantifiers}
  with their \kl(quantifier){bounded} counterparts.
  We therefore only give a high-level description of the construction.
  For $\TilingSystem = \Tuple{\TStateSet, \TTileSet}$,
  the \kl{formula}~$\Translation(\TilingSystem)$
  is of the form
  \begin{equation*}
    \ExistsRel{
      \Tuple{
        \MSOVar_{\TState}
      }_{
        \TState \in \TStateSet
      }
    } \,
    \ForAll{\FOVar[1]}
    \bigl(
      \Fml{OneState}\Of{\FOVar[1]}
      \AND
      \Fml{LegalTiling}\Of{\FOVar[1]}
    \bigr),
  \end{equation*}
  where
  each $\MSOVar_{\TState}$ is a \kl{unary} \kl{relation variable}
  intended to represent the set of \kl{pixels}
  in \kl(tiling system){state}~$\TState$,
  $\Fml{OneState}\Of{\FOVar[1]}$
  is a $\BFL$-\kl{formula} stating that
  exactly one \kl(tiling system){state}
  has been assigned to \kl{pixel}~$\FOVar[1]$,
  and $\Fml{LegalTiling}\Of{\FOVar[1]}$
  is another $\BFL$-\kl{formula} stating that
  each of the $(2 \times 2)$-subblocks containing \kl{pixel}~$\FOVar[1]$
  corresponds to some \kl{tile} of~$\TTileSet$.
  Since the boundary markers~$\TBoundarySymbol$ surrounding the \kl{picture}
  are not represented by any \kl{elements}
  in the \kl{structure}~$\StructReprP{\Picture}$,
  the \kl{formula} $\Fml{LegalTiling}\Of{\FOVar[1]}$
  performs a case distinction on whether $\FOVar[1]$ lies
  in one of the four corners,
  along one of the four borders,
  or somewhere else inside the \kl{picture}.
  This can be written as a conjunction of the form
  \begin{align*}
    \Fml{LegalTiling}\Of{\FOVar[1]} = \;
    & \Fml{TopLeft}\Of{\FOVar[1]} \, \AND \,
      \Fml{TopRight}\Of{\FOVar[1]} \, \AND \,
      \Fml{BottomLeft}\Of{\FOVar[1]} \, \AND \,
      \Fml{BottomRight}\Of{\FOVar[1]} \, \AND \, {} \\
    & \Fml{Top}\Of{\FOVar[1]} \, \AND \,
      \Fml{Bottom}\Of{\FOVar[1]} \, \AND \,
      \Fml{Left}\Of{\FOVar[1]} \, \AND \,
      \Fml{Right}\Of{\FOVar[1]} \, \AND \,
      \Fml{Inside}\Of{\FOVar[1]},
  \end{align*}
  where,
  for example,
  the conjunct $\Fml{TopLeft}\Of{\FOVar[1]}$ states that
  if $\FOVar[1]$ lies in the top-left corner
  (i.e., if it has neither a “vertical” nor a “horizontal” predecessor),
  then there must be some \kl{tile}
  $\bigl(
     \begin{smallmatrix}
       \TBoundarySymbol & \TBoundarySymbol \\
       \TBoundarySymbol & \String, \TState
     \end{smallmatrix}
   \bigr)
   \in \TTileSet$
  such that
  $\FOVar[1]$ has value~$\String$
  and lies in \kl(tiling system){state}~$\TState$.
  The other conjuncts are similar.
\end{proof}

By combining Proposition~\ref{prp:equivalence-lso-lmso}
(which weakens \kl{second-order quantification})
and Corollary~\ref{cor:equivalence-lemso-emso}
(which strengthens \kl{first-order quantification}),
we can now derive a partial levelwise equivalence between
\kl{local second-order logic} and \kl{monadic second-order logic}
on \kl{pictures}.

\begin{theorem}
  \label{thm:equivalence-lso-mso}
  When restricted to \kl{pictures},
  every level of the \kl{local second-order hierarchy}
  that ends with a block of existential quantifiers
  is equivalent to the corresponding level of
  the \kl{monadic second-order hierarchy}.
  That is,
  $\SigmaLFO{\Level}\On{\PIC{\BitLen}} =
   \mSigmaFO{\Level}\On{\PIC{\BitLen}}$
  if $\Level$ is odd,
  and
  $\PiLFO{\Level}\On{\PIC{\BitLen}} =
   \mPiFO{\Level}\On{\PIC{\BitLen}}$
  if $\Level$ is even,
  for all $\Level \in \Positives$ and $\BitLen \in \Naturals$.
\end{theorem}

\begin{proof}
  We proceed by induction on~$\Level$.
  For $\Level = 1$,
  it suffices to apply Proposition~\ref{prp:equivalence-lso-lmso}
  and then Corollary~\ref{cor:equivalence-lemso-emso}, i.e.,
  \begin{equation*}
    \SigmaLFO{1}\On{\PIC{\BitLen}}
    \overset{\text{Prp.\ref{prp:equivalence-lso-lmso}}}{=}
    \mSigmaLFO{1}\On{\PIC{\BitLen}}
    \overset{\text{Cor.\ref{cor:equivalence-lemso-emso}}}{=}
    \mSigmaFO{1}\On{\PIC{\BitLen}}.
  \end{equation*}
  For $\Level \geq 2$,
  let us assume that $\Level$ is even,
  the other case being completely analogous.
  We have
  \begin{equation*}
    \PiLFO{\Level}\On{\PIC{\BitLen}}
    \overset{\text{Prp.\ref{prp:equivalence-lso-lmso}}}{=}
    \mPiLFO{\Level}\On{\PIC{\BitLen}}
    \; = \;
    \mPiFO{\Level}\On{\PIC{\BitLen}},
  \end{equation*}
  by first applying Proposition~\ref{prp:equivalence-lso-lmso}
  and then using the fact that
  $\mPiLFOL{\Level}$ and~$\mPiFOL{\Level}$
  are defined in terms of
  $\mSigmaLFOL{\Level - 1}$ and~$\mSigmaFOL{\Level - 1}$,
  for which the induction hypothesis already provides an equivalence.
  More precisely,
  if \,$\ForAllRel{\MSOVar_1} \dots \ForAllRel{\MSOVar_n}(\Formula)$
  is an $\mPiLFOL{\Level}$-\kl{sentence}
  evaluated on $\BitLen$-bit \kl{pictures},
  where $\Formula$ starts with a block of
  \kl{existential quantifiers} over sets,
  then $\Formula$ can be evaluated as
  an $\mSigmaLFOL{\Level - 1}$-\kl{sentence}
  on $(\BitLen + n)$-bit \kl{pictures}.
  The idea is simply to interpret each \kl{atomic} \kl{formula}
  $\InRel{\MSOVar_i}{\FOVar[1]}$ as
  $\BitTrue{\BitLen + i}{\FOVar[1]}$,
  for $i \in \Range[1]{n}$.
  The analogous observation holds for $\mPiFOL{\Level}$-\kl{sentences}
  on $\BitLen$-bit \kl{pictures},
  whose \kl{subformulas} can be interpreted as
  $\mSigmaFOL{\Level - 1}$-\kl{sentences}
  on $(\BitLen + n)$-bit \kl{pictures}.
  Since we already know that
  $\mSigmaLFO{\Level - 1}\On{\PIC{\BitLen + n}} =
   \mSigmaFO{\Level - 1}\On{\PIC{\BitLen + n}}$,
  this implies that
  $\mPiLFO{\Level}\On{\PIC{\BitLen}} =
   \mPiFO{\Level}\On{\PIC{\BitLen}}$.
\end{proof}

%...............................................................................
\subsubsection{From pictures to graphs}
\label{sssec:from-pictures-to-graphs}

With the partial levelwise equivalence
obtained in Theorem~\ref{thm:equivalence-lso-mso},
we can already transfer
part of the separation result of Matz, Schweikardt, and Thomas
(Theorem~\ref{thm:mso-hierarchy})
from \kl{monadic second-order logic} to \kl{local second-order logic},
while remaining in the realm of \kl{pictures}.
To further transfer the result from \kl{pictures} to \kl{graphs},
we now show how to \kl(graph){encode} $0$-bit \kl{pictures}
as \kl{graphs} of $4$-\kl{bounded structural degree},
and how to translate \kl{formulas} from one type of \kl{structure} to the~other.

\AP
The \intro{graph encoding} of the $0$-bit \kl{picture}~$\Picture$
of \kl(picture){size} $\Tuple{\PicHeight, \PicWidth}$
is a \kl{graph}~$\GraphEnc{\Picture}$
that represents each \kl{pixel} of~$\Picture$
by five \kl{nodes}:
one main \kl{node} ($\Tag{pxl}$),
and four auxiliary \kl{nodes}
($\Tag{in}_1$, $\Tag{in}_2$, $\Tag{out}_1$, $\Tag{out}_2$),
which can be thought of as the incoming and outgoing “ports” of the \kl{pixel}.
Each main \kl{node} is connected to its four ports,
and the ports of any adjacent \kl{pixels} are connected
in such a way as to represent the relations
$\LinkRel{1}{\StructReprP{\Picture}}$ and~%
$\LinkRel{2}{\StructReprP{\Picture}}$.
Formally,
$\GraphEnc{\Picture}$ is defined by
the set of \kl{nodes}
\begin{equation*}
  \NodeSet{\GraphEnc{\Picture}} =
  \Range*{\PicHeight} \times \Range*{\PicWidth} \,\times\,
  \Set{
    \Tag{pxl},
    \Tag{in}_1, \Tag{in}_2,
    \Tag{out}_1, \Tag{out}_2
  },
\end{equation*}
the set of \kl{edges}
\begin{align*}
  \EdgeSet{\GraphEnc{\Picture}} = \;
  &\bigSetBuilder{ \,
    \Set{\Tuple{i, j, \Tag{pxl}}, \Tuple{i, j, \Symbol}} \,
  }{ \,
    i \in \Range*{\PicHeight}, \:
    j \in \Range*{\PicWidth}, \:
    \Symbol \in
    \Set{
      \Tag{in}_1, \Tag{in}_2,
      \Tag{out}_1, \Tag{out}_2
    } \,
  }
  \\
  {} \cup \;
  &\bigSetBuilder{ \,
    \Set{\Tuple{i, j, \Tag{out}_1}, \Tuple{i + 1, j, \Tag{in}_1}} \,
  }{ \,
    i \in \Range*{\PicHeight}*, \:
    j \in \Range*{\PicWidth} \,
  }
  \\
  {} \cup \;
  &\bigSetBuilder{ \,
    \Set{\Tuple{i, j, \Tag{out}_2}, \Tuple{i, j + 1, \Tag{in}_2}} \,
  }{ \,
    i \in \Range*{\PicHeight}, \:
    j \in \Range*{\PicWidth}* \,
  },
\end{align*}
and the \kl{labeling} function
\begin{alignat*}{2}
  \Labeling{\GraphEnc{\Picture}} \colon \quad
  \Tuple{i, j, \Tag{pxl}}    \mapsto \EmptyString, \qquad
  \Tuple{i, j, \Tag{in}_1}  &\mapsto 00, &\qquad
  \Tuple{i, j, \Tag{out}_1} &\mapsto 10, \\
  \Tuple{i, j, \Tag{in}_2}  &\mapsto 01, &
  \Tuple{i, j, \Tag{out}_2} &\mapsto 11,
\end{alignat*}
\AP
where
$i \in \Range*{\PicHeight}$,\,
$j \in \Range*{\PicWidth}$,
and $\intro*\EmptyString$ denotes the empty string.
An example is provided in Figure~\ref{fig:picture-graph}.
Notice that $\GraphEnc{\Picture}$ is always of $4$-\kl{bounded structural degree}.

\begin{figure}[tb]
  \centering
  \input{fig/picture-graph.tex}
  \caption{
    The $0$-bit \kl{picture}~$\Picture$ of \kl(picture){size} $\Tuple{2, 2}$,
    its \kl{graph encoding}~$\GraphEnc{\Picture}$,
    and their respective \kl{structural representations}
    $\StructReprP{\Picture}$ and~$\StructRepr{\GraphEnc{\Picture}}$.
    We follow the same graphical conventions as
    in Figures~\ref{fig:graph} and~\ref{fig:picture}.
  }
  \label{fig:picture-graph}
\end{figure}

\begin{figure}[htb]
  \centering
  \begin{tikzpicture}[
    semithick,>=stealth',on grid,
    label distance=0,
    vertex/.style={draw,circle,minimum size=2ex,inner sep=0,font=\scriptsize},
    bit/.style={vertex},
    set/.style={fill=black,text=white},
  ]
  \def\nodeDist{13ex}
  \def\bitNodeDist{6.5ex}
  \node[vertex] (pxl)
       [label=below left:$\Tag{pxl}$] {};
  \node[vertex] (in1)  at ([shift={(+90:\nodeDist)}]pxl)
       [label=right:$\Tag{in}_1$] {};
  \node[vertex] (in2)  at ([shift={(180:\nodeDist)}]pxl)
       [label=below:$\Tag{in}_2$] {};
  \node[vertex] (out1) at ([shift={(-90:\nodeDist)}]pxl)
       [label=left:$\Tag{out}_1$] {};
  \node[vertex] (out2) at ([shift={(0:\nodeDist)}]pxl)
       [label=above:\;$\Tag{out}_2$] {};
  \node[bit]     (in1-1) at ([shift={(+155:\bitNodeDist)}]in1) {1};
  \node[bit]     (in1-2) at ([shift={(-155:\bitNodeDist)}]in1) {2};
  \node[bit]     (in2-1) at ([shift={(+115:\bitNodeDist)}]in2) {1};
  \node[bit,set] (in2-2) at ([shift={(+65:\bitNodeDist)}]in2) {2};
  \node[bit,set] (out1-1) at ([shift={(+25:\bitNodeDist)}]out1) {1};
  \node[bit]     (out1-2) at ([shift={(-25:\bitNodeDist)}]out1) {2};
  \node[bit,set] (out2-1) at ([shift={(-115:\bitNodeDist)}]out2) {1};
  \node[bit,set] (out2-2) at ([shift={(-65:\bitNodeDist)}]out2) {2};
  \path[<->] (pxl) edge (in1)
                   edge (in2)
                   edge (out1)
                   edge (out2);
  \path[->] (in1-1) edge (in1-2)
            (in2-1) edge (in2-2)
            (out1-1) edge (out1-2)
            (out2-1) edge (out2-2);
  \path[->,densely dotted]
            (in1) edge (in1-1)
                  edge (in1-2)
            (in2) edge (in2-1)
                  edge (in2-2)
            (out1) edge (out1-1)
                   edge (out1-2)
            (out2) edge (out2-1)
                   edge (out2-2);
\end{tikzpicture}

%%% Local Variables:
%%% mode: latex
%%% TeX-master: "../lph-paper"
%%% End:
  \caption{
    The \kl{gadget}~$\Gadget$ representing
    a single \kl{pixel} whose value is the empty string.
    $\Gadget$ occurs four times
    in the \kl{structure}~$\StructRepr{\GraphEnc{\Picture}}$
    shown in Figure~\ref{fig:picture-graph}.
    We follow the same graphical conventions as in Figure~\ref{fig:graph}.
    The names of the \kl{elements} serve explanatory purposes only
    and are not part of the~\kl{gadget}.
  }
  \label{fig:gadget}
\end{figure}

\AP
If we look at
the \kl{graph encoding}'s
\kl{structural representation}~$\StructRepr{\GraphEnc{\Picture}}$
(also illustrated in Figure~\ref{fig:picture-graph}),
we see that each \kl{pixel} is represented by
the \kl{gadget} shown in Figure~\ref{fig:gadget}.
Formally,
the \intro{gadget}~$\intro*\Gadget$
representing any \kl{pixel} $\Tuple{i, j}$
of a $0$-bit \kl{picture}~$\Picture$
corresponds to the \kl{structural representation}
of the \kl{subgraph} of~$\GraphEnc{\Picture}$
that is \kl{induced} by
$\Set{i} \times \Set{j} \times
 \Set{
   \Tag{pxl},
   \Tag{in}_1, \Tag{in}_2,
   \Tag{out}_1, \Tag{out}_2
 }$.
It is convenient to identify the \kl{domain} of~$\Gadget$
with the set
\begin{equation*}
  \Domain{\Gadget} = \,
  \Set{
    \Tag{pxl},
    \Tag{in}_1, \Tag{in}_2,
    \Tag{out}_1, \Tag{out}_2
  }
  \: \cup \,
  \bigl(
    \Set{
      \Tag{in}_1, \Tag{in}_2,
      \Tag{out}_1, \Tag{out}_2
    } \times \Range*{2}
  \bigr),
\end{equation*}
where
the \kl{elements} in
$\Set{
   \Tag{pxl},
   \Tag{in}_1, \Tag{in}_2,
   \Tag{out}_1, \Tag{out}_2
 }$
represent \kl{nodes},
and the remaining \kl{elements} represent \kl{labeling bits}.
Considering Cartesian products to be associative,
this allows us to identify the \kl{domain}
of the entire \kl{structure}~$\StructRepr{\GraphEnc{\Picture}}$
with the set of \kl{elements}
$\Domain{\StructRepr{\GraphEnc{\Picture}}} =
 \Domain{\StructReprP{\Picture}} \times \Domain{\Gadget}$,
where the first component specifies the \kl{pixel},
and the second component specifies the \kl{gadget} \kl{element}.

The following lemma states
that the expressive power of the \kl{local second-order hierarchy}
remains essentially the same
whether we consider \kl{pictures} or \kl{graph encodings} of \kl{pictures}.
This is because we can translate \kl{formulas}
from one type of \kl{structure} to the other
without changing the alternation level of \kl{second-order quantifiers}.
Consequently,
any separation result established for \kl{pictures}
carries over to \kl{graphs}.

\begin{lemma}
  \label{lem:translation-picture-graph}
  Let $\Level \in \Naturals$.
  \begin{enumerate}
  \item \label{itm:forward-translation}
    For every
    $\SigmaLFOL{\Level}$-\kl{sentence}~$\Formula$
    evaluated on $0$-bit \kl{pictures},
    there is a
    $\SigmaLFOL{\Level}$-\kl{sentence}~$\Formula'$
    evaluated on \kl{graphs}
    such that
    $\StructReprP{\Picture} \Satisfies \Formula$
    if and only if
    $\StructRepr{\GraphEnc{\Picture}} \Satisfies \Formula'$
    for all $\Picture \in \PIC{0}$.
  \item \label{itm:backward-translation}
    Conversely,
    for every
    $\SigmaLFOL{\Level}$-\kl{sentence}~$\Formula$
    evaluated on \kl{graphs},
    there is a
    $\SigmaLFOL{\Level}$-\kl{sentence}~$\Formula'$
    evaluated on $0$-bit \kl{pictures}
    such that
    $\StructRepr{\GraphEnc{\Picture}} \Satisfies \Formula$
    if and only if
    $\StructReprP{\Picture} \Satisfies \Formula'$
    for all $\Picture \in \PIC{0}$.
  \end{enumerate}
  The analogous statements hold for $\PiLFOL{\Level}$-\kl{sentences}.
\end{lemma}

\begin{proof}
  The forward direction
  (Statement~\ref{itm:forward-translation})
  is straightforward.
  \kl{First-order quantification} on
  $\StructRepr{\GraphEnc{\Picture}}$
  is relativized to \kl{elements} of
  $\Domain{\StructReprP{\Picture}} \times \Set{\Tag{pxl}}$,
  which correspond to the central \kl{element} of the \kl{gadget},
  and \kl{atomic} \kl{formulas} that refer to the relations
  $\LinkRel{1}{\StructReprP{\Picture}}$
  and
  $\LinkRel{2}{\StructReprP{\Picture}}$
  of~$\StructReprP{\Picture}$
  are rewritten in terms of the representation
  of these relations in~$\StructRepr{\GraphEnc{\Picture}}$.
  There is no need to explicitly relativize \kl{second-order quantification},
  since relations can only be evaluated for \kl{elements}
  represented by \kl{first-order variables} anyway.

  We start by defining some helper \kl{formulas}.
  The \kl{formula}
  \begin{equation*}
    \Fml{IsPixel}\Of{\FOVar[1]} =
    \NOT \ExistsNb{\FOVar[2]}{\FOVar[1]} \,
    \bigl(
      \FOVar[1] \Linked{2} \FOVar[2]
      \, \OR \,
      \FOVar[2] \Linked{2} \FOVar[1]
    \bigr)
  \end{equation*}
  states that $\FOVar[1]$ is a “\kl{pixel} center”,
  corresponding to \kl{element} $\Tag{pxl}$ of the \kl{gadget}.
  For $i \in \Set{1, 2}$,
  the \kl{formulas}
  \begin{align*}
    \Fml{HasBitFalse}_i\Of{\FOVar[1]} &= \,
    \ExistsNb{\FOVar[2]}{\FOVar[1]} \,
    \bigl(
      \Fml{AtBitPos}_i\Of{\FOVar[1], \FOVar[2]}
      \, \AND \,
      \NOT \BitTrue{1}{\FOVar[2]}
    \bigr) \quad \text{and}\\
    \Fml{HasBitTrue}_i\Of{\FOVar[1]} &= \,
    \ExistsNb{\FOVar[2]}{\FOVar[1]} \,
    \bigl(
      \Fml{AtBitPos}_i\Of{\FOVar[1], \FOVar[2]}
      \, \AND \,
      \BitTrue{1}{\FOVar[2]}
    \bigr)
  \end{align*}
  state that
  the $i$-th \kl{labeling bit} of \kl{node}~$\FOVar[1]$ has value~$0$ and~$1$,
  respectively,
  where the \kl{subformulas}
  \begin{equation*}
    \Fml{AtBitPos}_i\Of{\FOVar[1], \FOVar[2]} =
    \begin{cases*}
      \FOVar[1] \Linked{2} \FOVar[2]
      \: \AND \,
      \ExistsNb{\FOVar[3]}{\FOVar[1]} \, (\FOVar[2] \Linked{1} \FOVar[3])
      & if $i = 1$,
      \\
      \FOVar[1] \Linked{2} \FOVar[2]
      \: \AND \,
      \ExistsNb{\FOVar[3]}{\FOVar[1]} \, (\FOVar[3] \Linked{1} \FOVar[2])
      & if $i = 2$
    \end{cases*}
  \end{equation*}
  identify~$\FOVar[2]$ as the $i$-th \kl{labeling bit} of~$\FOVar[1]$.
  Building on that,
  \begin{equation*}
    \Fml{IsIn}_1\Of{\FOVar[1]} =
    \Fml{HasBitFalse}_{1}\Of{\FOVar[1]}
    \AND
    \Fml{HasBitFalse}_{2}\Of{\FOVar[1]}
  \end{equation*}
  states that $\FOVar[1]$ is a “vertical input port”,
  corresponding to \kl{element} $\Tag{in}_1$ of the \kl{gadget}.
  Similarly,
  we define
  $\Fml{IsIn}_2\Of{\FOVar[1]}$,
  $\Fml{IsOut}_1\Of{\FOVar[1]}$, and
  $\Fml{IsOut}_2\Of{\FOVar[1]}$,
  which identify the
  “horizontal input port”,
  “vertical output port”, and
  “horizontal output port”,
  respectively.

  We now show by structural induction
  that there is a translation~$\Translation$
  from \kl{formulas} on $0$-bit \kl{pictures}
  to \kl{formulas} on \kl{graphs}
  such that
  for every \kl{formula}~$\Formula$,
  every \kl{picture} $\Picture \in \PIC{0}$,
  and every \kl{variable assignment}~$\Assignment$
  of $\Free(\Formula[1])$ on~$\StructReprP{\Picture}$,
  we have
  $\StructReprP{\Picture}, \Assignment \Satisfies \Formula[1]$
  if and only if
  $\StructRepr{\GraphEnc{\Picture}}, \Assignment' \Satisfies \Translation(\Formula[1])$.
  Here,
  $\Assignment'$ is the \kl{variable assignment}
  of $\Free(\Formula[1])$ on~$\StructRepr{\GraphEnc{\Picture}}$
  such that
  $\Assignment'(\FOVar[1]) =
   \bigTuple{\Assignment(\FOVar[1]), \Tag{pxl}}$
  for all $\FOVar[1] \in \FreeFO(\Formula[1])$,
  and
  \begin{equation*}
    \Assignment'(\SOVar) =
    \bigSetBuilder{
      \bigTuple{
        \Tuple{\Element_1, \Tag{pxl}},
        \dots,
        \Tuple{\Element_{\Arity}, \Tag{pxl}}
      }
    }{
      \Tuple{\Element_1, \dots, \Element_{\Arity}} \in \Assignment(\SOVar)
    }
  \end{equation*}
  for all $\SOVar \in \FreeSO(\Formula[1])$
  of \kl{arity} $\Arity \in \Positives$.

  \begin{itemize}
  \item To express that one \kl{pixel} of~$\Picture$
    is the “vertical” or “horizontal” successor of another,
    we state that
    the appropriate ports of the corresponding \kl{nodes} of~$\GraphEnc{\Picture}$
    are connected, i.e.,
    \begin{equation*}
      \Translation(\FOVar[1] \Linked{i} \FOVar[2]) \, = \,
      \ExistsLoc{\FOVar[3]_1, \FOVar[3]_2}{2}{\FOVar[1]}
      \bigl(
        \Fml{IsOut}_i\Of{\FOVar[3]_1}
        \AND
        \Fml{IsIn}_i\Of{\FOVar[3]_2}
        \AND
        \FOVar[1] \Linked{1} \FOVar[3]_1
                  \Linked{1} \FOVar[3]_2
                  \Linked{1} \FOVar[2]
      \bigr).
    \end{equation*}
  \item \kl{Atomic} \kl{formulas} for equality and other relations
    are kept unchanged:
    $\Translation(\FOVar[1] \Equal \FOVar[2]) =
     (\FOVar[1] \Equal \FOVar[2])$,
    and
    $\Translation
     \bigl(
       \InRel{\SOVar}{\FOVar[1]_1,\dots,\FOVar[1]_{\Arity}}
     \bigr) =
     \InRel{\SOVar}{\FOVar[1]_1,\dots,\FOVar[1]_{\Arity}}$
    for any \kl{second-order variable}~$\SOVar$
    of \kl{arity} $\Arity \in \Positives$.
  \item \kl{Boolean connectives} are also preserved:
    $\Translation(\NOT \Formula[1]) =
     \NOT \Translation(\Formula[1])$,
    and
    $\Translation(\Formula[1] \OR \Formula[2]) =
     \Translation(\Formula[1]) \OR \Translation(\Formula[2])$.
  \item \kl{First-order quantifiers} are relativized to “\kl{pixel} centers”.
    More precisely,
    the \kl(quantifier){unbounded} \kl{universal quantifier}
    at the outermost scope of an $\LFOL$-\kl{formula}
    is translated by
    $\Translation(\ForAll{\FOVar[1]} \, \Formula[1]) =
     \ForAll{\FOVar[1]}
     \bigl(
       \Fml{IsPixel}\Of{\FOVar[1]} \IMP \Translation(\Formula[1])
     \bigr)$.
    For \kl(quantifier){bounded} \kl{first-order quantifiers},
    the idea is the same,
    but we have to take into account
    that two adjoining “\kl{pixel} centers”
    lie at a distance of~$3$ from each other,
    thus
    $\Translation(\ExistsNb{\FOVar[2]}{\FOVar[1]} \; \Formula[1]) \, = \,
     \ExistsLoc{\FOVar[2]}{3}{\FOVar[1]}
     \bigl(
       \FOVar[2] \NotEqual \FOVar[1] \AND
       \Fml{IsPixel}\Of{\FOVar[2]} \AND
       \Translation(\Formula[1])
     \bigr)$.
  \item \kl{Second-order quantifiers} are not affected by the translation, i.e.,
    $\Translation(\QuantifierRel{\SOVar} \, \Formula[1]) =
     \QuantifierRel{\SOVar} \: \Translation(\Formula[1])$
    for any \kl{second-order variable}~$\SOVar$
    of \kl{arity} $\Arity \in \Positives$,
    where $\QuantifierRel{}$~is either $\ExistsRel{}$ or~$\ForAllRel{}$.
    We explain why this works in the case of \kl{existential quantifiers};
    the argument for \kl{universal quantifiers} is completely analogous.
    Specifically,
    $\StructReprP{\Picture}\!, \Assignment
     \Satisfies \ExistsRel{\SOVar} \, \Formula[1]$
    is equivalent to the existence of
    $\Relation \subseteq (\Domain{\StructReprP{\Picture}})^{\Arity}$
    such that
    $\StructReprP{\Picture}\!, \Version{\Assignment}{\SOVar}{\Relation}
     \Satisfies \Formula[1]$,
    which by induction is equivalent to the existence of
    $\Relation' \, \subseteq \,
     (\Domain{\StructReprP{\Picture}} \! \times \Set{\Tag{pxl}})^{\Arity}$
    such that
    $\StructRepr{\GraphEnc{\Picture}},
     \Version{\Assignment'}{\SOVar}{\Relation'}
     \Satisfies \Translation(\Formula[1])$,
    where
    $\Assignment'$~is defined as above.
    This in turn is equivalent to the existence of
    \begin{equation*}
      \Relation' \, \subseteq \,
      \bigl(
        \Domain{\StructReprP{\Picture}} \! \times \Set{\Tag{pxl}}
      \bigr)^{\Arity}
      \qquad \text{and} \qquad
      \Relation'' \subseteq \,
      \bigl(
        \Domain{\StructRepr{\GraphEnc{\Picture}}}
      \bigr)^{\Arity}
      \setminus
      \bigl(
        \Domain{\StructReprP{\Picture}} \! \times \Set{\Tag{pxl}}
      \bigr)^{\Arity}
    \end{equation*}
    such that
    $\StructRepr{\GraphEnc{\Picture}},
     \Version{\Assignment'}{\SOVar}{\Relation' \cup \Relation''}
     \Satisfies \Translation(\Formula[1])$,
    where $\Relation''$ is irrelevant
    because we ensure that all \kl{first-order variables}
    refer to “\kl{pixel} centers”.
    Finally,
    since $\Relation' \cup \Relation''$ can be any $\Arity$-\kl{ary} relation on
    $\Domain{\StructRepr{\GraphEnc{\Picture}}}$\!,
    the last condition is equivalent to
    $\StructRepr{\GraphEnc{\Picture}}, \Assignment'
     \Satisfies \ExistsRel{\SOVar} \: \Translation(\Formula[1])$.
  \end{itemize}

  Note that our translation preserves
  the alternation level of \kl{second-order quantifiers}
  and the fact that there is exactly one
  \kl(quantifier){unbounded} \kl{universal quantifier}
  nested directly below the \kl{second-order quantifiers}.
  That is,
  if $\Formula[1]$ is a
  $\SigmaLFOL{\Level}$- or $\PiLFOL{\Level}$-\kl{formula},
  then so is~$\Translation(\Formula[1])$.
  Hence,
  Statement~\ref{itm:forward-translation} corresponds to
  the special case of the induction hypothesis where
  $\Formula[1]$ does not have any \kl{free} \kl{variables}.

  \bigskip

  The backward direction
  (Statement~\ref{itm:backward-translation})
  is a bit more tedious because
  $\StructRepr{\GraphEnc{\Picture}}$
  has $\CardS{\Gadget}$ times
  as many \kl{elements} as~$\StructReprP{\Picture}$.
  To simulate the additional \kl{elements}
  when translating a \kl{formula}~$\Formula[1]$
  from \kl{graphs} to \kl{pictures},
  we introduce a “virtual \kl{variable assignment}”
  $\VirtualAssignment \colon \FreeFO(\Formula[1]) \to \Domain{\Gadget}$
  that tells us for each \kl{free} \kl{first-order variable} of~$\Formula[1]$
  to which \kl{element} of the \kl{gadget} it corresponds.
  In combination with the actual \kl{variable assignment}~$\Assignment'$
  on~$\StructReprP{\Picture}$,
  which tells us the corresponding \kl{pixel},
  this allows us to reference every \kl{element} of
  $\StructRepr{\GraphEnc{\Picture}}$.
  Our translation is thus parameterized by~$\VirtualAssignment$.
  We handle \kl{first-order quantification}
  by combining actual \kl{quantification} with
  a case distinction over all possible values of~$\VirtualAssignment$,
  and \kl{second-order quantification}
  by representing
  each $\Arity$-\kl{ary} \kl{relation variable}~$\SOVar$ of~$\Formula[1]$
  by a collection of \kl{variables}
  $\SOVar_{\Tuple{\Element[1]_1, \dots, \Element[1]_{\Arity}}}$,
  one for each $\Arity$-tuple
  $\Tuple{\Element[1]_1, \dots, \Element[1]_{\Arity}}$
  of \kl{gadget} \kl{elements}.

  Formally,
  we show by structural induction
  that there is a parameterized translation~$\Translation_{\VirtualAssignment}$
  from \kl{formulas} on \kl{graphs}
  to \kl{formulas} on $0$-bit \kl{pictures}
  such that
  for every \kl{formula}~$\Formula$,
  every \kl{picture} $\Picture \in \PIC{0}$,
  and every \kl{variable assignment}~$\Assignment$
  of $\Free(\Formula[1])$ on~$\StructRepr{\GraphEnc{\Picture}}$,
  we have
  $\StructRepr{\GraphEnc{\Picture}}, \Assignment \Satisfies \Formula[1]$
  if and only if
  $\StructReprP{\Picture}, \Assignment'
   \Satisfies \Translation_{\VirtualAssignment}(\Formula[1])$.
  Here,
  $\Assignment'$ is the \kl{variable assignment}
  of $\Free(\Formula[1])$ on~$\StructReprP{\Picture}$
  and
  $\VirtualAssignment$ is the “virtual \kl{variable assignment}”
  $\FreeFO(\Formula[1]) \to \Domain{\Gadget}$
  such that
  $\Assignment(\FOVar[1]) =
   \bigTuple{\Assignment'(\FOVar[1]), \VirtualAssignment(\FOVar[1])}$
  for all $\FOVar[1] \in \FreeFO(\Formula[1])$,
  and
  \begin{equation*}
    \Assignment(\SOVar) \, = \,
    \smashoperator{
      \bigcup_{
        \hspace{8.5ex}
        \Element[1]_1, \dots, \Element[1]_{\Arity}
        \inS \Gadget
      }
    } \;
    \bigSetBuilder{
      \bigTuple{
        \Tuple{\Element[2]_1, \Element[1]_1},
        \dots,
        \Tuple{\Element[2]_{\Arity}, \Element[1]_{\Arity}}
      }
    }{
      \Tuple{\Element[2]_1, \dots, \Element[2]_{\Arity}}
      \in \Assignment'(\SOVar_{\Tuple{\Element[1]_1, \dots, \Element[1]_{\Arity}}})
    }
  \end{equation*}
  for all $\SOVar \in \FreeSO(\Formula[1])$
  of \kl{arity} $\Arity \in \Positives$.

  \begin{itemize}
  \item To express that an \kl{element} of~$\StructRepr{\GraphEnc{\Picture}}$
    lies in the set~$\BitSet{1}{\StructRepr{\GraphEnc{\Picture}}}$,
    we state that
    the corresponding \kl{pixel} of~$\StructReprP{\Picture}$
    is mapped by~$\VirtualAssignment$
    to a \kl{gadget} \kl{element}
    that lies in~$\BitSet{1}{\Gadget}$.
    Hence,
    \begin{equation*}
      \Translation_{\VirtualAssignment}
      \bigl(
        \BitTrue{1}{\FOVar[1]}
      \bigr) =
      \begin{cases*}
        \True
        & if $\VirtualAssignment(\FOVar[1]) \in \BitSet{1}{\Gadget}$,
        \\
        \False
        & otherwise.
      \end{cases*}
    \end{equation*}
  \item To express that two \kl{elements} of~$\StructRepr{\GraphEnc{\Picture}}$
    are connected by
    the relation~$\LinkRel{i}{\StructRepr{\GraphEnc{\Picture}}}$,
    we need to distinguish
    the case where they belong to the same \kl{gadget}
    from the case where they are in two adjacent \kl{gadgets}.
    In the first case,
    the corresponding \kl{pixels} of~$\StructReprP{\Picture}$
    must coincide and be mapped by~$\VirtualAssignment$
    to two \kl{gadget} \kl{elements}
    that are connected by~$\LinkRel{i}{\Gadget}$.
    In the second case,
    the connection must necessarily be a “$\Linked{1}$”-link
    from an “input port” to an “output port”,
    or vice versa,
    and the corresponding \kl{pixels} of~$\StructReprP{\Picture}$
    must be connected accordingly
    by the “vertical” or “horizontal” successor relation:
    \begin{equation*}
      \Translation_{\VirtualAssignment}
      \bigl(
        \FOVar[1] \Linked{i} \FOVar[2]
      \bigr) =
      \begin{cases*}
        \FOVar[1] \Equal \FOVar[2]
        & if $\VirtualAssignment(\FOVar[1])
              \LinkRel{i}{\Gadget} \!
              \VirtualAssignment(\FOVar[2])$,
        \\
        \FOVar[1] \Linked{j} \FOVar[2]
        & if $i = 1$,\,
          $\VirtualAssignment(\FOVar[1]) = \Tag{out}_j$,\, and\,
          $\VirtualAssignment(\FOVar[2]) = \Tag{in}_j$,\,
          where $j \in \Set{1, 2}$,
        \\
        \FOVar[2] \Linked{j} \FOVar[1]
        & if $i = 1$,\,
          $\VirtualAssignment(\FOVar[1]) = \Tag{in}_j$,\, and\,
          $\VirtualAssignment(\FOVar[2]) = \Tag{out}_j$,\,
          where $j \in \Set{1, 2}$,
        \\
        \False
        & otherwise.
      \end{cases*}
    \end{equation*}
  \item In order for
    two \kl{elements} of~$\StructRepr{\GraphEnc{\Picture}}$ to be equal,
    they must correspond to the same \kl{pixel} of~$\StructReprP{\Picture}$
    and the same \kl{element} of the \kl{gadget}, i.e.,
    \begin{equation*}
      \Translation_{\VirtualAssignment}
      \bigl(
        \FOVar[1] \Equal \FOVar[2]
      \bigr) =
      \begin{cases*}
        \FOVar[1] \Equal \FOVar[2]
        & if $\VirtualAssignment(\FOVar[1]) = \VirtualAssignment(\FOVar[2])$,
        \\
        \False
        & otherwise.
      \end{cases*}
    \end{equation*}
  \item To express that
    $\Arity$ \kl{elements} of~$\StructRepr{\GraphEnc{\Picture}}$
    are $\SOVar$-related,
    for some \kl{second-order variable}~$\SOVar$
    of \kl{arity}~$\Arity$,
    we state that the corresponding \kl{pixels} of~$\StructReprP{\Picture}$
    are related by the appropriate copy of~$\SOVar$,
    which is determined by the \kl{gadget} \kl{elements}
    that $\VirtualAssignment$ assigns to each \kl{pixel}.
    That is,
    $\Translation_{\VirtualAssignment}
     \bigl(
       \InRel{\SOVar}{\FOVar[1]_1, \dots, \FOVar[1]_{\Arity}}
     \bigr) =
     \InRel{
       \SOVar_{
         \Tuple{\VirtualAssignment(\FOVar[1]_1), \dots,
                \VirtualAssignment(\FOVar[1]_{\Arity})}
       }
     }{
       \FOVar[1]_1, \dots, \FOVar[1]_{\Arity}
     }$.
  \item \kl{Boolean connectives} are preserved:
    $\Translation_{\VirtualAssignment}
     \bigl(
       \NOT \Formula[2]
     \bigr) =
     \NOT \Translation_{\VirtualAssignment}(\Formula[2])$,
    and
    $\Translation_{\VirtualAssignment}
     \bigl(
       \Formula[2]_1 \OR \Formula[2]_2
     \bigr) =
     \Translation_{\VirtualAssignment}(\Formula[2]_1) \OR
     \Translation_{\VirtualAssignment}(\Formula[2]_2)$.
  \item \kl{First-order quantification} over~$\StructRepr{\GraphEnc{\Picture}}$
    is expressed through a combination of
    \kl{first-order quantification} over~$\StructReprP{\Picture}$
    and a case distinction over the \kl{gadget} \kl{element}
    to which $\VirtualAssignment$ maps the \kl{quantified} \kl{variable}.
    For the (unique) \kl(quantifier){unbounded}
    \kl{first-order quantifier} of $\LFOL$,
    this simply means
    \begin{equation*}
      \Translation_{\VirtualAssignment}
      \bigl(
        \ForAll{\FOVar[1]} \, \Formula[1]
      \bigr) = \,
      \ForAll{\FOVar[1]}
      \BigAND_{
        \lalign{\Element \inS \Gadget}
      }
      \Translation_{\Version{\VirtualAssignment}{\FOVar[1]}{\Element}}(\Formula[1]).
    \end{equation*}
    For \kl(quantifier){bounded} \kl{first-order quantifiers},
    the case distinction is a bit more involved
    because we must take into account
    the topology of~$\StructRepr{\GraphEnc{\Picture}}$:
    each \kl{element} of~$\StructRepr{\GraphEnc{\Picture}}$
    is connected to its \kl{neighbors} within the same \kl{gadget},
    and additionally,
    an “input” or “output port” is also connected to its counterpart
    in the appropriate adjacent \kl{gadget}.
    This yields
    \begin{equation*}
      \Translation_{\VirtualAssignment}
      \bigl(
        \ExistsNb{\FOVar[2]}{\FOVar[1]} \; \Formula[1]
      \bigr) = \,
      \ExistsLoc{\FOVar[2]}{1}{\FOVar[1]}
      \bigl(
        \Fml{IntraGadget}_{\Formula[1]}\Of{\FOVar[1], \FOVar[2]}
        \OR
        \Fml{InterGadget}_{\Formula[1]}\Of{\FOVar[1], \FOVar[2]}
      \bigr),
    \end{equation*}
    where
    \begin{equation*}
      \Fml{IntraGadget}_{\Formula[1]}\Of{\FOVar[1], \FOVar[2]} = \,
      (\FOVar[2] \Equal \FOVar[1])
      \, \AND \,
      \BigOR_{
        \lalign{
          \Element \;
          \NeighborRel{\Gadget}
          \VirtualAssignment(\FOVar[1])
        }
      }
      \Translation_{\Version{\VirtualAssignment}{\FOVar[2]}{\Element}}(\Formula[1]),
    \end{equation*}
    and
    \begin{equation*}
      \Fml{InterGadget}_{\Formula[1]}\Of{\FOVar[1], \FOVar[2]} =
      \begin{cases*}
        \FOVar[2] \Linked{i} \FOVar[1]
        \, \AND \,
        \Translation_{\Version{\VirtualAssignment}{\FOVar[2]}{\Tag{out}_i}}(\Formula[1])
        & if $\VirtualAssignment(\FOVar[1]) = \Tag{in}_i$,
          for $i \in \Set{1, 2}$,
        \\
        \FOVar[1] \Linked{i} \FOVar[2]
        \, \AND \,
        \Translation_{\Version{\VirtualAssignment}{\FOVar[2]}{\Tag{in}_i}}(\Formula[1])
        & if $\VirtualAssignment(\FOVar[1]) = \Tag{out}_i$,
          for $i \in \Set{1, 2}$,
        \\
        \False & otherwise.
      \end{cases*}
    \end{equation*}
  \item Each \kl{second-order quantification}
    over~$\StructRepr{\GraphEnc{\Picture}}$
    is expressed through
    multiple \kl{second-order quantifications}
    over~$\StructReprP{\Picture}$.
    More precisely,
    each $\Arity$-\kl{ary} relation~$\Relation$
    on~$\StructRepr{\GraphEnc{\Picture}}$
    is represented as the union of
    $\CardS{\Gadget}^{\Arity}$
    pairwise disjoint $\Arity$-\kl{ary} relations.
    Each such relation
    $\Relation_{\Tuple{\Element_1, \dots, \Element_{\Arity}}}$
    contains precisely those $\Arity$-tuples of~$\Relation$
    whose components correspond to the \kl{gadget} \kl{elements}
    $\Element_1, \dots, \Element_{\Arity}$
    (in that order).
    Hence,
    \begin{equation*}
      \Translation_{\VirtualAssignment}
      \bigl(
        \QuantifierRel{\SOVar} \, \Formula[1]
      \bigr) = \,
      \QuantifierRel{
        \Tuple{
          \SOVar_{\Tuple{\Element_1, \dots, \Element_{\Arity}}}
        }_{
          \Element_1, \dots, \Element_{\Arity} \inS \Gadget
        }
      } \,
      \bigl(
        \Translation_{\VirtualAssignment}(\Formula[1])
      \bigr),
    \end{equation*}
    where $\QuantifierRel{}$~can be either $\ExistsRel{}$ or~$\ForAllRel{}$.
  \end{itemize}

  Again,
  our translation preserves
  the alternation level of \kl{second-order quantifiers}
  and the fact that there is exactly one
  \kl(quantifier){unbounded} \kl{universal quantifier}
  nested directly below the \kl{second-order quantifiers}.
  Hence,
  Statement~\ref{itm:backward-translation} corresponds to
  the special case of the induction hypothesis where
  $\Formula[1]$ does not have any \kl{free} \kl{variables}
  (which means, in particular, that
  the “virtual \kl{variable assignment}”~$\VirtualAssignment$ is empty).
\end{proof}

We now have everything at hand to transfer
part of Theorem~\ref{thm:mso-hierarchy}
from \kl{monadic second-order logic} on \kl{pictures}
to \kl{local second-order logic} on \kl{graphs},
and thus to the \kl{locally polynomial hierarchy}.
The result is stated in the following theorem
and illustrated in Figure~\ref{fig:hierarchy-partial}.

\begin{figure}[htb]
  \centering
  \begin{minipage}[c]{0.28\textwidth}
    \begin{tikzpicture}[
    semithick,on grid,node distance=9ex,
    every node/.style={draw,rounded rectangle,
                       minimum height=4ex,minimum width=10ex},
    strict/.style={},
    non strict/.style={thin,dashed},
    heavy/.style={ultra thick},
  ]
  \def\nodeDist{9ex}
  \node (s1) [heavy] {$\SigmaLP{\Level - 1}$};
  \node (p1) [right=4/3*\nodeDist of s1] {$\PiLP{\Level - 1}$};
  \node (s2) [above of=s1]       {$\SigmaLP{\Level}$};
  \node (p2) [heavy,above of=p1] {$\PiLP{\Level}$};
  \node (s3) [heavy,above of=s2] {$\SigmaLP{\Level + 1}$};
  \node (p3) [above of=p2]       {$\PiLP{\Level + 1}$};
  \path[non strict]
    (s2) edge (s3)
    (s2) edge (p3)
    (p2) edge (p3)
    (s1) edge (s2)
    (p1) edge (s2)
    (p1) edge (p2);
  \path[strict]
    (p2) edge (s3)
    (s1) edge (p2);
\end{tikzpicture}

%%% Local Variables:
%%% mode: latex
%%% TeX-master: "../lph-paper"
%%% End:
  \end{minipage}
  \begin{minipage}[c]{0.45\textwidth}
    \vspace{2ex}
    \caption{
      A partial separation result for the \kl{locally polynomial hierarchy}
      obtained in Theorem~\ref{thm:locally-polynomial-hierarchy}
      for every even integer $\Level \geq 2$.
      Each line indicates an inclusion of the lower class in the higher class.
      The inclusions represented by solid lines are proved to be strict,
      even when restricted to \kl{graphs} of \kl{bounded structural degree}.
      This forms the basis for the fuller separation result
      shown in Figure~\ref{fig:hierarchy-full}.
    }
    \label{fig:hierarchy-partial}
  \end{minipage}
\end{figure}

\begin{theorem}
  \label{thm:locally-polynomial-hierarchy}
  The \kl{locally polynomial hierarchy} is infinite,
  even when restricted to \kl{graphs} of \kl{bounded structural degree}.
  More precisely,
  $\SigmaLP{\Level - 1}\On{\bGRAPH{\MaxDegree}} \subsetneqq
   \PiLP{\Level}\On{\bGRAPH{\MaxDegree}} \subsetneqq
   \SigmaLP{\Level + 1}\On{\bGRAPH{\MaxDegree}}$,
  and a fortiori
  $\SigmaLP{\Level - 1} \subsetneqq
   \PiLP{\Level} \subsetneqq
   \SigmaLP{\Level + 1}$,
  for every even integer $\Level \geq 2$
  and every integer $\MaxDegree \geq 4$.
\end{theorem}

\begin{proof}
  By Theorem~\ref{thm:local-hierarchy-equivalence},
  the statement can be equivalently formulated
  in terms of the \kl{local second-order hierarchy} on \kl{graphs}:
  $\SigmaLFO{\Level - 1}\On{\bGRAPH{4}} \subsetneqq
   \PiLFO{\Level}\On{\bGRAPH{4}} \subsetneqq
   \SigmaLFO{\Level + 1}\On{\bGRAPH{4}}$
  for every even integer $\Level \geq 2$.
  To prove it,
  we start with the analogous separation result
  for the \kl{monadic second-order hierarchy} on $0$-bit \kl{pictures},
  which holds by Theorem~\ref{thm:mso-hierarchy}:
  \begin{equation*}
    \mSigmaFO{\Level - 1}\On{\PIC{0}}
    \, \subsetneqq \,
    \mPiFO{\Level}\On{\PIC{0}}
    \, \subsetneqq \,
    \mSigmaFO{\Level + 1}\On{\PIC{0}}
  \end{equation*}
  By Theorem~\ref{thm:equivalence-lso-mso},
  this can be rewritten
  in terms of \kl{local second-order logic}:
  \begin{equation*}
    \SigmaLFO{\Level - 1}\On{\PIC{0}}
    \, \subsetneqq \,
    \PiLFO{\Level}\On{\PIC{0}}
    \, \subsetneqq \,
    \SigmaLFO{\Level + 1}\On{\PIC{0}}
    \label{eq:picture-inequalities}
    \tag{$\ast$}
  \end{equation*}
  We now transfer this result
  from $0$-bit \kl{pictures}
  to \kl{graphs} of $4$-\kl{bounded structural degree}.
  The first inequality of~\eqref{eq:picture-inequalities}
  tells us that there exists a \kl{picture property}
  $\Property \in \PiLFO{\Level}\On{\PIC{0}}$
  that does not lie in
  $\SigmaLFO{\Level - 1}\On{\PIC{0}}$.
  Applying the forward translation provided by
  Lemma~\ref{lem:translation-picture-graph}.\ref{itm:forward-translation}
  and the fact that
  \kl{graph encodings} of $0$-bit \kl{pictures}
  are of $4$-\kl{bounded structural degree},
  we infer the existence of a \kl{graph property}
  $\Property' \in
   \PiLFO{\Level}\On{\bGRAPH{4}}$
  such that
  $\Picture \in \Property$
  if and only if
  $\GraphEnc{\Picture} \in \Property'$,
  for every $0$-bit \kl{picture}~$\Picture$.
  Similarly,
  the backward translation provided by
  Lemma~\ref{lem:translation-picture-graph}.\ref{itm:backward-translation}
  lets us deduce that
  $\Property' \notin
   \SigmaLFO{\Level - 1}\On{\bGRAPH{4}}$,
  because otherwise we would have
  $\Property \in
   \SigmaLFO{\Level - 1}\On{\PIC{0}}$.
  Hence,
  $\SigmaLFO{\Level - 1}\On{\bGRAPH{4}}
   \subsetneqq
   \PiLFO{\Level}\On{\bGRAPH{4}}$.
  Analogously,
  we can conclude from the second inequality of~\eqref{eq:picture-inequalities}
  that
  $\PiLFO{\Level}\On{\bGRAPH{4}}
   \subsetneqq
   \SigmaLFO{\Level + 1}\On{\bGRAPH{4}}$.
\end{proof}

%...............................................................................
\subsubsection{Why not a simpler infiniteness proof?}
\label{sssec:direct-infiniteness-proof}

The infiniteness proof presented in the previous sections
may seem unnecessarily complicated.
If our only objective were to show that
the \kl{locally polynomial hierarchy} is infinite,
the detour through \kl{pictures}
in Section~\ref{sssec:digression-on-pictures}
could be avoided entirely.
A more direct argument can be derived from the following two facts:
\begin{itemize}
  \item The infiniteness of the \kl{monadic second-order hierarchy}
    established by Matz, Schweikardt, and Thomas
    (Theorem~\ref{thm:mso-hierarchy})
    holds not only on \kl{pictures},
    but also on \kl{graphs} of \kl{bounded structural degree}.
  \item When restricted to \kl{graphs} of \kl{bounded structural degree},
    the \kl{locally polynomial hierarchy}
    is equivalent,
    though not levelwise equivalent,
    to \kl{monadic second-order logic}.
    As we will see in Section~\ref{ssec:complete-picture},
    one direction of this equivalence
    follows as a corollary from our previous results
    (Corollary~\ref{cor:mso-subsumes-hierarchy-on-bounded-graphs}
    on page~\pageref{cor:mso-subsumes-hierarchy-on-bounded-graphs}),
    while the other could be easily shown
    using the closure properties of the \kl{locally polynomial hierarchy}.
\end{itemize}
Together,
these two statements imply that
the \kl{locally polynomial hierarchy} does not collapse.
Moreover,
similarly to the classical \kl{polynomial hierarchy}
(see, e.g.,~\cite[Thm.~5.4]{DBLP:books/daglib/0023084}),
one can show that
the \kl{locally polynomial hierarchy} would collapse
to~$\SigmaLP{\Level}$
if that class were equal to~$\PiLP{\Level + 1}$,
and likewise to~$\PiLP{\Level}$
if that class were equal to~$\SigmaLP{\Level + 1}$.
The proof is fully analogous to that for the logarithmic-size hierarchy
of Feuilloley, Fraigniaud, and Hirvonen
(see \cite[Thm.~5]{DBLP:journals/tcs/FeuilloleyFH21}).
Therefore,
the infiniteness of the \kl{locally polynomial hierarchy}
immediately implies its strictness,
i.e., that
any two consecutive levels starting with different quantifiers
are~distinct.

Why, then, did we take a detour through \kl{pictures}
to separate the individual levels of the \kl[locally polynomial hierarchy]{hierarchy}?
The reason is twofold.
First,
the proof via \kl{pictures} remains fully constructive,
providing concrete (albeit artificial) \kl{graph properties}
that separate the different levels.
These are precisely the \kl{graph-encoded} versions
of the \kl(picture){properties} identified by Matz, Schweikardt, and Thomas
for the \kl{monadic second-order hierarchy} on $0$-bit \kl{pictures}
(see \cite[\S\,2.4]{DBLP:journals/iandc/MatzST02}).
Second,
the non-constructive proof sketched above
becomes more complicated
when restricted to \kl{graphs} of \kl{bounded structural degree}.
For instance,
to show that
$\SigmaLP{\Level} = \PiLP{\Level + 1}$
implies
$\SigmaLP{\Level} = \SigmaLP{\Level + 2}$,
the \kl{certificates} chosen by \kl(certificate){Eve} in her initial move
must somehow be encoded
within the constant-radius \kl{neighborhoods} of the \kl{nodes}.
The most natural choice would be
to encode the \kl{certificates} in the \kl{node labels},
but this is not possible on \kl{graphs} of \kl{bounded structural degree},
where
the maximum \kl{label} size is fixed
independently of the \kl{graph property}.
Hence,
the \kl{certificates} would need to be encoded
into the \kl{graph} topology using additional \kl{nodes}.
This would have to be done in a way
that neither increases the maximum \kl{degree} of the \kl{graphs}
nor modifies the given quantifier alternation level
when
converting between \kl{arbiters} (or \kl{formulas})
on original and encoded \kl{graphs}.
While possible,
this would not be any easier
than our approach via \kl{graph encodings} of~\kl{pictures}.

%–––––––––––––––––––––––––––––––––––––––––––––––––––––––––––––––––––––––––––––––
\subsection{Completing the picture}
\label{ssec:complete-picture}

In this subsection,
we establish all the remaining separations and inclusions
shown in Figure~\ref{fig:hierarchy-full} on page~\pageref{fig:hierarchy-full},
and then conclude by identifying \kl{graph properties}
that lie outside the \kl{locally polynomial hierarchy}.

Our first goal is to prove that
the inclusions represented by dashed lines in Figure~\ref{fig:hierarchy-full}
are equalities when restricted to \kl{graphs} of \kl{bounded structural degree}.
To do this,
we first show that on such \kl{graphs},
we can refine the notion of \kl{restrictive arbiters}
introduced in Section~\ref{sec:restrictive-arbiters}
to require that \kl{identifier assignments} are necessarily \kl{small}.

\AP
Let $\Level$ be a nonnegative integer,
$\IdentRadius$ and~$\CertifRadius$ be positive integers,
$\CertifPolynomial$~be a polynomial function,
$\BaseProperty$~be an $\LP$-\kl(graph){property},
and
$\Machine_1, \dots, \Machine_{\Level}$
be \kl{certificate restrictors}
for $\Tuple{\CertifRadius, \CertifPolynomial}$\nobreakdash-\kl{bounded certificates}
under $\IdentRadius$\nobreakdash-\kl{locally unique} \kl{identifiers}.
A \intro{small-restrictive} \,$\SigmaLP{\Level}$\nobreakdash-\reintro{arbiter}
for a \kl{graph property}~$\Property$
on~$\BaseProperty$
under
$\IdentRadius$-\kl{locally unique} \kl{identifiers}
and
$\Tuple{\CertifRadius, \CertifPolynomial}$\nobreakdash-\kl{bounded certificates}
restricted by
$\Machine_1, \dots, \Machine_{\Level}$
is a \kl{locally polynomial machine}~$\Machine$
that satisfies the same equivalence
as a \kl{restrictive} $\SigmaLP{\Level}$\nobreakdash-\kl{arbiter}
(see on page~\pageref{def:restrictive-arbiter})
for every \kl{graph}~$\Graph \in \BaseProperty$
and every
\emph{\kl{small}} $\IdentRadius$-\kl{locally unique}
\kl{identifier assignment}~$\IdMap$ of~$\Graph$.
That is,
the equivalence does not have to hold for \kl{arbitrary-sized} \kl{identifiers}.
We analogously define
\reintro{small-restrictive} \,$\PiLP{\Level}$\nobreakdash-\reintro{arbiters}.

The following lemma is a refinement of Lemma~\ref{lem:restrictive-arbiters}
for \kl{small-restrictive arbiters}
on \kl{graphs} of \kl{bounded structural degree}.

\begin{lemma}
  \label{lem:small-restrictive-arbiters}
  Let $\Level, \MaxDegree \in \Naturals$
  and $\Property \subseteq \GRAPH$.
  The \kl{graph property} $\Property \cap \bGRAPH{\MaxDegree}$
  belongs to~$\SigmaLP{\Level}\On{\bGRAPH{\MaxDegree}}$
  if and only if
  $\Property$~has
  a \kl{small-restrictive} $\SigmaLP{\Level}$-\kl{arbiter}
  on~$\bGRAPH{\MaxDegree}$.
  The analogous statement holds for~$\PiLP{\Level}\On{\bGRAPH{\MaxDegree}}$.
\end{lemma}

\begin{proof}
  We prove only the first statement,
  since the proof for~$\PiLP{\Level}\On{\bGRAPH{\MaxDegree}}$ is completely analogous.
  By definition,
  if $\Property \cap \bGRAPH{\MaxDegree}$ belongs
  to~$\SigmaLP{\Level}\On{\bGRAPH{\MaxDegree}}$,
  then there exists
  a \kl{permissive} $\SigmaLP{\Level}$\nobreakdash-\kl{arbiter}~$\Machine$
  for a \kl{graph property}~$\Property'$
  such that
  $\Property' \cap \bGRAPH{\MaxDegree} = \Property \cap \bGRAPH{\MaxDegree}$,
  and thus
  $\Machine$~is also a \kl{small-restrictive} $\SigmaLP{\Level}$-\kl{arbiter}
  for~$\Property$ on~$\bGRAPH{\MaxDegree}$
  under \kl{unrestricted} \kl{certificates}.

  For the converse,
  since $\bGRAPH{\MaxDegree}$ is an $\LP$-\kl(graph){property},
  it suffices by Lemma~\ref{lem:restrictive-arbiters}
  to convert a \kl{small-restrictive arbiter}
  into a \kl{restrictive arbiter}
  that operates under \kl{arbitrary-sized} \kl{identifiers}.
  Let $\Machine$ be
  a \kl{small-restrictive} $\SigmaLP{\Level}$\nobreakdash-\kl{arbiter}
  for~$\Property$ on~$\bGRAPH{\MaxDegree}$
  under $\IdentRadius$-\kl{locally unique} \kl{identifiers}
  and $\Tuple{\CertifRadius, \CertifPolynomial}$\nobreakdash-\kl{bounded certificates}
  restricted by
  $\Machine_1, \dots, \Machine_{\Level}$.
  We need to distinguish two cases:
  \begin{enumerate}
  \item If $\Level = 0$,
    then $\Machine$~is in fact a \kl{small-restrictive} $\LP$-\kl{decider}
    for~$\Property$ on~$\bGRAPH{\MaxDegree}$.
    This means that
    every \kl{node} reaches its \kl{verdict}
    simply by examining a portion of the input \kl{graph},
    i.e., without having to consider any \kl{certificates}
    that could potentially depend on a particular choice of \kl{identifiers}.
    Suppose that $\Machine$ runs in \kl{round time}~$\RunRadius$.
    We construct~$\Machine'$ such that when it is \kl{executed}
    on a \kl{graph}~$\Graph$,
    the \kl{nodes} first communicate
    for $(\RunRadius + 2 \IdentRadius)$~\kl{rounds}
    to reconstruct their $(\RunRadius + 2 \IdentRadius)$-\kl{neighborhoods}.
    Then,
    each \kl{node}~$\Node[1]$ simulates~$\Machine$ locally
    on its $\RunRadius$-\kl{neighborhood}
    under every possible
    \kl{small} $\IdentRadius$-\kl{locally unique} \kl{identifier assignment}
    of~$\Graph$ restricted to~$\Neighborhood{\Graph}{\RunRadius}{\Node[1]}$,
    and finally \kl{accepts} if it did so in every simulation.
    The condition of being
    \kl{small} $\IdentRadius$-\kl{locally unique} can be respected
    because $\Node[1]$~knows the $2 \IdentRadius$-\kl{neighborhood}
    of every \kl{node} in~$\Neighborhood{\Graph}{\RunRadius}{\Node[1]}$.
    Note that
    by running the simulation under every possible \kl{identifier assignment},
    we avoid the problem
    of the \kl{nodes} having to agree on a particular one.
    This can be done in constant \kl{step time}
    because
    the restriction to \kl{graphs} of $\MaxDegree$-\kl{bounded structural degree}
    entails a constant upper bound on
    the number of (partial) \kl{identifier assignments}
    each \kl{node} has to consider.
    The \kl{machine}~$\Machine'$ obtained this way
    is a \kl{restrictive} $\LP$-\kl{decider}
    for~$\Property$ on~$\bGRAPH{\MaxDegree}$
    under \kl{arbitrary-sized}
    $\IdentRadius'$-\kl{locally unique} \kl{identifiers},
    where $\IdentRadius' = \RunRadius + 2 \IdentRadius$.
  \item If $\Level > 0$,
    then $\Machine$~can be simulated by a \kl{machine}~$\Machine'$
    that uses the first \kl{certificate assignment}
    to \kl{encode} \kl{small} \kl{identifiers}.
    More precisely,
    let the given \kl{small-restrictive arbiter}~$\Machine$ be such that
    for every \kl{graph} $\Graph \in \bGRAPH{\MaxDegree}$
    and every \kl{small} $\IdentRadius$-\kl{locally unique}
    \kl{identifier assignment}~$\IdMap$ of~$\Graph$,
    \begin{equation*}
      \Graph \in \Property
      \; \iff \;
      \exists \CertifMap_1 \,
      \forall \CertifMap_2
      \dots
      \Quantifier \CertifMap_{\Level}:
      \Result{\Machine}{
        \Graph, \,
        \IdMap, \,
        \CertifMap_1 \CertifConcat
        \CertifMap_2 \CertifConcat \dots \CertifConcat
        \CertifMap_{\Level}
      } \equiv
      \Accept,
    \end{equation*}
    where all quantifiers range over
    $\Tuple{\CertifRadius, \CertifPolynomial}$-\kl{bounded certificate assignments}
    of $\Tuple{\Graph, \IdMap}$
    with the additional restriction that
    $\Result{\Machine_i}{
       \Graph, \,
       \IdMap, \,
       \CertifMap_1 \CertifConcat \dots \CertifConcat
       \CertifMap_i
     } \equiv
    \Accept$
    for all $i \in \Range[1]{\Level}$.
    We now construct~$\Machine'$
    with appropriately chosen
    constants $\IdentRadius', \CertifRadius' \in \Positives$,
    polynomial~$\CertifPolynomial'$,
    and \kl{certificate restrictors}
    $\Machine_1', \dots, \Machine_{\Level}'$
    such that
    for every \kl{graph} $\Graph \in \bGRAPH{\MaxDegree}$
    and every \kl{arbitrary-sized} $\IdentRadius'$-\kl{locally unique}
    \kl{identifier assignment}~$\IdMap'$ of~$\Graph$,
    \begin{equation*}
      \Graph \in \Property
      \; \iff \;
      \exists \CertifMap_1' \,
      \forall \CertifMap_2
      \dots
      \Quantifier \CertifMap_{\Level}:
      \Result{\Machine'}{
        \Graph, \,
        \IdMap', \,
        \CertifMap_1' \CertifConcat
        \CertifMap_2 \CertifConcat \dots \CertifConcat
        \CertifMap_{\Level}
      } \equiv
      \Accept,
    \end{equation*}
    where all quantifiers range over
    $\Tuple{\CertifRadius', \CertifPolynomial'}$-\kl{bounded certificate assignments}
    of $\Tuple{\Graph, \IdMap'}$
    with the additional restriction that
    $\Result{\Machine_i'}{
       \Graph, \,
       \IdMap', \,
       \CertifMap_1' \CertifConcat
       \CertifMap_2 \CertifConcat \dots \CertifConcat
       \CertifMap_i
     } \equiv
    \Accept$
    for all $i \in \Range[1]{\Level}$.
    The \kl{certificate restrictors} are chosen such that
    $\CertifMap_1'$ \kl{encodes} both
    a \kl{small} $\IdentRadius$-\kl{locally unique}
    \kl{identifier assignment}~$\IdMap$ of~$\Graph$
    and an
    $\Tuple{\CertifRadius, \CertifPolynomial}$-%
    \kl{bounded certificate assignment}~$\CertifMap_1$
    of $\Tuple{\Graph, \IdMap}$
    satisfying the restrictions imposed by~$\Machine_1$,
    and the remaining \kl{certificate assignments}
    $\CertifMap_2, \dots, \CertifMap_{\Level}$
    are
    $\Tuple{\CertifRadius, \CertifPolynomial}$-\kl(certificate){bounded}
    with respect to $\Tuple{\Graph, \IdMap}$
    and satisfy the restrictions imposed by
    $\Machine_2, \dots, \Machine_{\Level}$.
    The \kl{machine}~$\Machine'$ itself
    then simply simulates~$\Machine$ on~$\Graph$ under~$\IdMap$ and
    $\CertifMap_1 \CertifConcat
     \CertifMap_2 \CertifConcat \dots \CertifConcat
     \CertifMap_{\Level}$.

    Note that it is easy to construct
    $\Machine_1', \dots, \Machine_{\Level}'$
    such that they satisfy \kl{local repairability}
    (and thus the definition of a \kl{certificate restrictor}).
    In particular,
    the restriction of $\IdMap$~being
    \kl{small} $\IdentRadius$-\kl{locally unique}
    is compatible with \kl{local repairability}
    because
    if a \kl{node} has an invalid \kl{identifier}
    (i.e., too large or not $\IdentRadius$-\kl{locally unique}),
    then by the proof of Remark~\ref{rem:small-sized-identifiers},
    we can assign it a valid \kl{identifier}
    without affecting the validity of the other \kl{nodes}' \kl{identifiers}.
    Also note that we need
    the restriction to \kl{graphs} of $\MaxDegree$-\kl{bounded structural degree}
    only for the case $\Level = 0$,
    not for the case $\Level > 0$.
    \qedhere
  \end{enumerate}
\end{proof}

Using the notion of \kl{small-restrictive arbiters},
we now show that on \kl{graphs} of \kl{bounded structural degree},
it is useless to let \kl(certificate){Adam} choose
the last \kl{certificate assignment}.
Intuitively,
this is because on such \kl{graphs},
the \kl{step running time} of the \kl{nodes}
can be arbitrarily large with respect to
their local input and the messages they receive,
so they can use brute force
to perform universal quantification over the last \kl{certificate}.
This is analogous to the observation made by
Feuilloley, Fraigniaud, and Hirvonen~\cite{DBLP:journals/tcs/FeuilloleyFH21}
for their alternation hierarchy
(which does not impose any restrictions on
the processing power of the \kl{nodes}).

\begin{proposition}
  \label{prp:bounded-graph-collapse}
  When restricted to \kl{graphs} of \kl{bounded structural degree},
  every level of the \kl{locally polynomial hierarchy}
  that ends with a universal quantifier
  is equivalent to the level directly below it
  that lacks that final quantifier.
  Formally,
  $\SigmaLP{\Level}\On{\bGRAPH{\MaxDegree}} =
   \SigmaLP{\Level + 1}\On{\bGRAPH{\MaxDegree}}$
  if~$\Level$~is odd,
  and
  $\PiLP{\Level}\On{\bGRAPH{\MaxDegree}} =
   \PiLP{\Level + 1}\On{\bGRAPH{\MaxDegree}}$
  if~$\Level$~is even,
  for all $\Level, \MaxDegree \in \Naturals$.
\end{proposition}

\begin{proof}
  We prove only the second equality,
  the proof of the first being completely analogous.
  Let $\Property$ be a \kl{graph property}
  in $\PiLP{\Level + 1}\On{\bGRAPH{\MaxDegree}}$,
  and let $\Machine$ be
  a \kl{restrictive} $\PiLP{\Level + 1}$-\kl{arbiter}
  for~$\Property$ on~$\bGRAPH{\MaxDegree}$
  under $\IdentRadius$-\kl{locally unique} \kl{identifiers}
  and \kl{unrestricted}
  $\Tuple{\CertifRadius, \CertifPolynomial}$\nobreakdash-\kl{bounded certificates}.
  By definition,
  for every \kl{graph} $\Graph \in \bGRAPH{\MaxDegree}$
  and every $\IdentRadius$\nobreakdash-\kl{locally unique}
  \kl{identifier assignment}~$\IdMap$ of~$\Graph$,
  we have
  \begin{equation*}
    \Graph \in \Property
    \; \iff \;
    \forall \CertifMap_1 \,
    \exists \CertifMap_2
    \dots
    \forall \CertifMap_{\Level + 1}:
    \Result{\Machine}{
      \Graph, \,
      \IdMap, \,
      \CertifMap_1 \CertifConcat
      \CertifMap_2 \CertifConcat \dots \CertifConcat
      \CertifMap_{\Level + 1}
    } \equiv
    \Accept,
  \end{equation*}
  where all quantifiers range over
  $\Tuple{\CertifRadius, \CertifPolynomial}$-\kl{bounded certificate assignments}
  of $\Tuple{\Graph, \IdMap}$.
  Now,
  to show that
  $\Property \in \PiLP{\Level}\On{\bGRAPH{\MaxDegree}}$,
  it suffices by Lemma~\ref{lem:small-restrictive-arbiters}
  to provide
  a \kl{small-restrictive} $\PiLP{\Level}$\nobreakdash-\kl{arbiter}~$\Machine'$
  for~$\Property$ on~$\bGRAPH{\MaxDegree}$.
  More precisely,
  we construct~$\Machine'$
  with an appropriately chosen constant~$\IdentRadius'$
  such that
  for every \kl{graph} $\Graph \in \bGRAPH{\MaxDegree}$
  and every \kl{small} $\IdentRadius'$-\kl{locally unique}
  \kl{identifier assignment}~$\IdMap$ of~$\Graph$,
  we have
  \begin{equation*}
    \Graph \in \Property
    \; \iff \;
    \forall \CertifMap_1 \,
    \exists \CertifMap_2
    \dots
    \exists \CertifMap_{\Level}:
    \Result{\Machine'}{
      \Graph, \,
      \IdMap, \,
      \CertifMap_1 \CertifConcat
      \CertifMap_2 \CertifConcat \dots \CertifConcat
      \CertifMap_{\Level}
    } \equiv
    \Accept,
  \end{equation*}
  where all quantifiers range over
  $\Tuple{\CertifRadius, \CertifPolynomial}$-\kl{bounded certificate assignments}
  of $\Tuple{\Graph, \IdMap}$.

  Given the \kl{certificates} assigned to them by
  $\CertifMap_1, \dots, \CertifMap_{\Level}$,
  the \kl{nodes} running~$\Machine'$ simulate~$\Machine$
  for every possible choice of~$\CertifMap_{\Level + 1}$.
  To do so,
  they first communicate for
  $\CertifRadius' = \Radius_{\Machine} + \CertifRadius$
  \kl{rounds},
  where $\Radius_{\Machine}$ is a (constant) bound
  on the \kl{round running time} of~$\Machine$.
  Setting $\IdentRadius' = \CertifRadius'$,
  this allows each \kl{node} $\Node[1]$ of~$\Graph$
  to reconstruct its $\CertifRadius'$-\kl{neighborhood}
  $\Neighborhood{\Graph}{\CertifRadius'}{\Node[1]}$
  and the \kl{identifiers} and \kl{certificates}
  of all \kl{nodes} therein.
  Then,
  $\Node[1]$~simulates~$\Machine$ locally
  for every
  $\Tuple{\CertifRadius, \CertifPolynomial}$-\kl(certificate){bounded}
  \kl{certificate assignment}~$\CertifMap_{\Level + 1}$
  of~$\Tuple{\Graph, \IdMap}$
  restricted to~$\Neighborhood{\Graph}{\Radius_{\Machine}}{\Node[1]}$.
  Note that $\Node[1]$~can respect
  the condition of
  $\Tuple{\CertifRadius, \CertifPolynomial}$-\kl(certificate){boundedness}
  because it knows the $\CertifRadius$-\kl{neighborhood}
  of every \kl{node} in~$\Neighborhood{\Graph}{\Radius_{\Machine}}{\Node[1]}$.
  Finally,
  $\Node[1]$ \kl{accepts} precisely if
  it has \kl{accepted} in every simulation.

  Intuitively,
  the reason why this approach works is that
  there is a universal quantification
  on both~$\CertifMap_{\Level + 1}$
  and on the \kl{nodes}
  in the \kl{acceptance} criterion of \kl{distributed Turing machines}.
  Therefore,
  we can reverse the order of quantification
  (by letting the \kl{nodes} perform the quantification over~$\CertifMap_{\Level + 1}$)
  without changing the semantics.

  The \kl{step running time} of~$\Machine'$ at~$\Node[1]$
  is clearly exponential in
  \begin{equation*}
    \sum_{
      \lalign{
        \Node[2] \inG \Neighborhood{\Graph}{\Radius_{\Machine}\!}{\Node[1]}
      }
    }
    \hspace{4ex}
    \CertifPolynomial
    \Bigl( \,
      \sum_{
        \lalign{
          \Node[3] \inG \Neighborhood{\Graph}{\CertifRadius}{\Node[2]}
        }
      }
      1 +
      \Length{\Labeling{\Graph}(\Node[3])} +
      \Length{\IdMap(\Node[3])} \,
    \Bigr).
  \end{equation*}
  However,
  since we require
  $\Graph$~to be of $\MaxDegree$-\kl{bounded structural degree}
  and $\IdMap$~to be \kl{small} $\CertifRadius'$-\kl{locally unique},
  this value is bounded by a constant
  that depends only on $\MaxDegree$, $\CertifRadius'$, and~$\CertifPolynomial$.
  Hence,
  $\Machine'$~runs in constant and thus \kl{polynomial step time}.
\end{proof}

The equalities established in Proposition~\ref{prp:bounded-graph-collapse}
not only simplify the \kl{locally polynomial hierarchy}
on \kl{graphs} of \kl{bounded structural degree},
but also imply,
in combination with the previous separation results,
that the remaining inclusions in the \kl[locally polynomial hierarchy]{hierarchy}
are strict.
This is particularly relevant on arbitrary \kl{graphs}.

\begin{corollary}
  \label{cor:locally-polynomial-hierarchy-missing-pieces}
  The separation results for the \kl{locally polynomial hierarchy}
  stated in Proposition~\ref{prp:lp-vs-nlp}
  and Theorem~\ref{thm:locally-polynomial-hierarchy}
  can be extended as follows:
  $\PiLP{\Level - 1} \subsetneqq
   \PiLP{\Level}$,
  and\,
  $\SigmaLP{\Level} \subsetneqq
   \SigmaLP{\Level + 1}$,
  and\,
  $\PiLP{\Level - 1} \subsetneqq
   \SigmaLP{\Level}  \subsetneqq
   \PiLP{\Level + 1}$,
  for every even integer $\Level \geq 2$.
\end{corollary}

\begin{proof}
  By Proposition~\ref{prp:lp-vs-nlp}
  and Theorem~\ref{thm:locally-polynomial-hierarchy}
  (for the inequalities),
  and by Proposition~\ref{prp:bounded-graph-collapse}
  (for the equalities),
  we have
  \begin{equation*}
    \begin{array}[t]{@{}ccccccc@{}}
        \PiLP{\Level - 2}\On{\bGRAPH{4}}    & \subsetneqq
      & \SigmaLP{\Level - 1}\On{\bGRAPH{4}} & \subsetneqq
      & \PiLP{\Level}\On{\bGRAPH{4}}        & \subsetneqq
      & \SigmaLP{\Level + 1}\On{\bGRAPH{4}}
      \\[1ex]
      \RotateRelation{=} & & \RotateRelation{=} & & \RotateRelation{=}
      \\[1ex]
        \PiLP{\Level - 1}\On{\bGRAPH{4}} &
      & \SigmaLP{\Level}\On{\bGRAPH{4}}  &
      & \PiLP{\Level + 1}\On{\bGRAPH{4}}\:\!,
    \end{array}
  \end{equation*}
  which by transitivity yields
  $\PiLP{\Level - 1}\On{\bGRAPH{4}} \subsetneqq
   \PiLP{\Level}\On{\bGRAPH{4}}$,
  and\,
  $\SigmaLP{\Level}\On{\bGRAPH{4}} \subsetneqq
   \SigmaLP{\Level + 1}\On{\bGRAPH{4}}$,
  and\,
  $\PiLP{\Level - 1}\On{\bGRAPH{4}} \subsetneqq
   \SigmaLP{\Level}\On{\bGRAPH{4}}  \subsetneqq
   \PiLP{\Level + 1}\On{\bGRAPH{4}}$.
  This implies the desired inequalities on arbitrary \kl{graphs}.
\end{proof}

Since the proof of Proposition~\ref{prp:bounded-graph-collapse}
relies on exhaustive search,
it is unlikely to generalize to arbitrary \kl{graphs}.
This can be restated as follows:

\begin{remark}
  \label{rem:unbounded-graph-collapse}
  The statement of Proposition~\ref{prp:bounded-graph-collapse}
  generalizes to arbitrary \kl{graphs}
  if and only if\,
  $\PTIME = \coNP$.
\end{remark}

\begin{claimproof}
  If $\PTIME = \coNP$,
  then the proof of Proposition~\ref{prp:bounded-graph-collapse}
  does not require the restriction to
  \kl{graphs} of \kl{bounded structural degree}.
  Indeed,
  the local simulations described there
  can then be performed by a polynomial-time algorithm
  that is equivalent to
  testing all possible \kl{certificate assignments} in parallel.

  Conversely,
  if the statement of Proposition~\ref{prp:bounded-graph-collapse}
  holds for arbitrary \kl{graphs},
  then in particular we have
  $\PiLP{0} = \PiLP{1}$,
  which entails
  $\PiLP{0}\On{\NODE} = \PiLP{1}\On{\NODE}$,
  and thus $\PTIME = \coNP$.
\end{claimproof}

Next,
we focus on the relationship between
the \kl{locally polynomial hierarchy} and its \kl{complement hierarchy}.
Using our previous results,
it is easy to show that the two hierarchies are completely distinct.

\begin{corollary}
  \label{cor:complement-classes}
  Classes on the same level of the \kl{locally polynomial hierarchy}
  are neither \kl{complement classes} of each other,
  nor are they closed under \kl{complementation},
  even when restricted to \kl{graphs} of \kl{bounded structural degree}.
  More precisely,
  $\SigmaLP{\Level}\On{\bGRAPH{\MaxDegree}} \nsubseteq
   \coPiLP{\Level}\On{\bGRAPH{\MaxDegree}}$
  if~$\Level$~is odd,
  and
  $\PiLP{\Level}\On{\bGRAPH{\MaxDegree}} \nsubseteq
   \coSigmaLP{\Level}\On{\bGRAPH{\MaxDegree}}$
  if~$\Level$~is even,
  for all $\Level \in \Naturals$ and $\MaxDegree \geq 4$.
  Moreover,
  $\SigmaLP{\Level}\On{\bGRAPH{\MaxDegree}} \neq
   \coSigmaLP{\Level}\On{\bGRAPH{\MaxDegree}}$
  and
  $\PiLP{\Level}\On{\bGRAPH{\MaxDegree}} \neq
   \coPiLP{\Level}\On{\bGRAPH{\MaxDegree}}$,
  for all $\Level \in \Naturals$ and $\MaxDegree \geq 4$.
\end{corollary}

\begin{proof}
  The statement for $\Level = 0$ reduces to
  $\LP\On{\bGRAPH{4}} \neq \coLP\On{\bGRAPH{4}}$,
  which holds by Corollary~\ref{cor:lp-complementation}.

  For $\Level \geq 1$,
  we first show that none of the classes is closed under \kl{complementation},
  building on the analogous result
  for the \kl{monadic second-order hierarchy} on $0$-bit \kl{pictures}.
  By Theorem~\ref{thm:mso-hierarchy},
  we know that
  $\mSigmaFO{\Level}\On{\PIC{0}}$
  and
  $\mPiFO{\Level}\On{\PIC{0}}$
  are incomparable for all $\Level \in \Positives$.
  Since these two classes are \kl{complement classes} of each other
  ($\FOL$~being closed under negation),
  this means that neither class is closed under \kl{complementation}.
  By Theorem~\ref{thm:equivalence-lso-mso},
  this implies that
  $\SigmaLFO{\Level}\On{\PIC{0}}$ is not closed under \kl{complementation}
  if $\Level$ is odd,
  and that
  $\PiLFO{\Level}\On{\PIC{0}}$ is not closed under \kl{complementation}
  if $\Level$ is even.
  We can transfer this result from $0$-bit \kl{pictures}
  to \kl{graphs} of $4$-\kl{bounded structural degree}
  by using Lemma~\ref{lem:translation-picture-graph}
  in the same way as in the proof of Theorem~\ref{thm:locally-polynomial-hierarchy}.
  This in turn allows us to conclude by Theorem~\ref{thm:local-hierarchy-equivalence}
  that
  $\SigmaLP{\Level}\On{\bGRAPH{4}}$
  is not closed under \kl{complementation} if $\Level$~is odd,
  and that
  $\PiLP{\Level}\On{\bGRAPH{4}}$
  is not closed under \kl{complementation} if $\Level$~is even.
  The analogous statement for the remaining cases
  follows by Proposition~\ref{prp:bounded-graph-collapse},
  which tells us that
  $\SigmaLP{\Level}\On{\bGRAPH{4}} =
   \SigmaLP{\Level - 1}\On{\bGRAPH{4}}$
  if~$\Level$~is even,
  and
  $\PiLP{\Level}\On{\bGRAPH{4}} =
   \PiLP{\Level - 1}\On{\bGRAPH{4}}$
  if~$\Level$~is odd.
  Hence,
  we have
  $\SigmaLP{\Level}\On{\bGRAPH{4}} \neq
   \coSigmaLP{\Level}\On{\bGRAPH{4}}$
  and
  $\PiLP{\Level}\On{\bGRAPH{4}} \neq
   \coPiLP{\Level}\On{\bGRAPH{4}}$
  for all $\Level \in \Naturals$.

  It remains to show for $\Level \geq 1$
  that the classes on level~$\Level$
  are not \kl{complement classes} of each other.
  If $\Level$~is odd,
  suppose for the sake of contradiction that
  $\SigmaLP{\Level}\On{\bGRAPH{4}} \subseteq
   \coPiLP{\Level}\On{\bGRAPH{4}}$.
  In combination with Proposition~\ref{prp:bounded-graph-collapse},
  this allows us to write the chain of inclusions
  \begin{equation*}
    \SigmaLP{\Level}\On{\bGRAPH{4}} \; \subseteq \;
    \coPiLP{\Level}\On{\bGRAPH{4}}
    \overset{\text{Prp.\ref{prp:bounded-graph-collapse}}}{=}
    \coPiLP{\Level - 1}\On{\bGRAPH{4}} \; \subseteq \;
    \coSigmaLP{\Level}\On{\bGRAPH{4}},
  \end{equation*}
  which contradicts the already established fact that
  $\SigmaLP{\Level}\On{\bGRAPH{4}}$
  is not closed under \kl{complementation}.
  Analogously,
  we can show that
  $\PiLP{\Level}\On{\bGRAPH{4}} \nsubseteq
   \coSigmaLP{\Level}\On{\bGRAPH{4}}$
  if~$\Level$~is even.
\end{proof}

Although the \kl{locally polynomial hierarchy}
is distinct from its \kl{complement hierarchy},
there are inclusions between the two.
This can be shown by generalizing
Examples~\ref{ex:not-all-selected} and~\ref{ex:not-3-colorable}
from Section~\ref{ssec:example-formulas}.
The strictness of these inclusions
is immediate by Corollary~\ref{cor:complement-classes}.

\begin{proposition}
  \label{prp:complementation}
  In the \kl{locally polynomial hierarchy},
  \kl{complementation} can be achieved
  at the cost of two or three additional quantifier alternations.
  More precisely,
  $\coSigmaLP{\Level} \subseteq \PiLP{\Level + 2}$
  if $\Level$~is even,
  and
  $\coPiLP{\Level} \subseteq \SigmaLP{\Level + 2}$
  if $\Level$~is odd,
  for all $\Level \in \Positives$.
\end{proposition}

\begin{proof}
  Let $\Level$ be even
  and $\Property$ be a \kl{graph property} in~$\SigmaLP{\Level}$.
  By Theorem~\ref{thm:local-hierarchy-equivalence},
  $\Property$ can be \kl{defined} by a $\SigmaLFOL{\Level}$-\kl{formula}
  of the form
  $\ExistsRel{\Vector{\SOVar}_1}
   \ForAllRel{\Vector{\SOVar}_2} \dots
   \ForAllRel{\Vector{\SOVar}_{\Level}} \,
   \ForAll{\FOVar[1]} \, \Formula[2]\Of{\FOVar[1]}$.
  Based on that,
  the \kl{complement}~$\Complement{\Property}$
  can be defined by the \kl{formula}
  $\ForAllRel{\Vector{\SOVar}_1}
   \ExistsRel{\Vector{\SOVar}_2} \dots
   \ExistsRel{\Vector{\SOVar}_{\Level}} \,
   \Fml{ExistsBadNode}$,
  where
  $\Fml{ExistsBadNode}$ is a $\SigmaLFOL{3}$-\kl{formula}
  with \kl{free} \kl{variables} in
  $\Vector{\SOVar}_1, \dots, \Vector{\SOVar}_{\Level}$
  that is equivalent to
  $\Exists{\FOVar[1]} \, \NOT \Formula[2]\Of{\FOVar[1]}$.
  The definition of $\Fml{ExistsBadNode}$
  is the same as in Example~\ref{ex:not-3-colorable},
  except that now we use $\NOT \Formula[2]\Of{\FOVar[1]}$
  instead of $\NOT \WellColored\Of{\FOVar[1]}$
  to instantiate the \kl{formula} schema
  $\PointsTo{\Formula[3]}\Of{\FOVar[1]}$.
  Hence,
  again by Theorem~\ref{thm:local-hierarchy-equivalence},
  $\Complement{\Property} \in \PiLP{\Level + 2}$.

  The proof for $\Level$ odd and $\Property \in \PiLP{\Level}$
  is completely analogous.
\end{proof}

\begin{corollary}
  \label{cor:complementation-strict-inclusions}
  The inclusions stated in Proposition~\ref{prp:complementation} are strict,
  even when restricted to \kl{graphs} of \kl{bounded structural degree}.
\end{corollary}

\begin{proof}
  By Corollary~\ref{cor:complement-classes},
  we know that
  $\PiLP{\Level}\On{\bGRAPH{\MaxDegree}} \nsubseteq
   \coSigmaLP{\Level}\On{\bGRAPH{\MaxDegree}}$
  for all even $\Level \in \Naturals$ and $\MaxDegree \geq 4$.
  This implies that
  $\PiLP{\Level + 2}\On{\bGRAPH{\MaxDegree}} \nsubseteq
   \coSigmaLP{\Level}\On{\bGRAPH{\MaxDegree}}$.
  Analogously,
  we obtain
  $\SigmaLP{\Level + 2}\On{\bGRAPH{\MaxDegree}} \nsubseteq
   \coPiLP{\Level}\On{\bGRAPH{\MaxDegree}}$
  for all odd $\Level \in \Naturals$ and $\MaxDegree \geq 4$.
\end{proof}

\bigskip

We end this section by identifying several natural \kl{graph properties}
that cannot be expressed
at any level of the \kl{locally polynomial hierarchy}.
Again,
the connection to logic proves valuable:
as a corollary of our previous results,
we obtain that on \kl{graphs} of \kl{bounded structural degree},
the \kl{locally polynomial hierarchy}
is (semantically) included in \kl{monadic second-order logic}.
Thus,
if a \kl{graph property} is not \kl{definable}
in \kl{monadic second-order logic}
on \kl{graphs} of \kl{bounded structural degree},
then it lies outside the \kl{locally polynomial hierarchy}.
In fact,
using the closure properties of the \kl{locally polynomial hierarchy},
in particular Proposition~\ref{prp:complementation},
it is easy to show that the reverse inclusion also holds,
i.e.,
the \kl{locally polynomial hierarchy}
is equivalent,
though not levelwise equivalent,
to \kl{monadic second-order logic}
on \kl{graphs} of \kl{bounded structural degree}.
But since we do not need this equivalence here,
we show only the first inclusion
(which holds even levelwise).

\begin{corollary}
  \label{cor:mso-subsumes-hierarchy-on-bounded-graphs}
  When restricted to \kl{graphs} of \kl{bounded structural degree},
  each level of the \kl{locally polynomial hierarchy}
  is subsumed by the corresponding level of
  the (standard) \kl{monadic second-order hierarchy}.
  More precisely,
  $\SigmaLP{\Level}\On{\bGRAPH{\MaxDegree}} \subseteq
   \mSigmaFO{\Level}\On{\bGRAPH{\MaxDegree}}$
  and
  $\PiLP{\Level}\On{\bGRAPH{\MaxDegree}} \subseteq
   \mPiFO{\Level}\On{\bGRAPH{\MaxDegree}}$
  for all $\Level, \MaxDegree \in \Naturals$.
\end{corollary}

\begin{proof}
  By Theorem~\ref{thm:local-hierarchy-equivalence}
  and Proposition~\ref{prp:equivalence-lso-lmso},
  when restricted to
  \kl{graphs} of $\MaxDegree$-\kl{bounded structural degree},
  each level of the \kl{locally polynomial hierarchy}
  is equivalent to the corresponding level of
  the \kl{local monadic second-order hierarchy},
  starting from level~$1$.
  Moreover,
  it is easy to verify that this also holds for level~$0$
  (given that we can exhaustively treat all constant-radius \kl{neighborhoods}).
  This immediately implies the claim,
  since \kl{local monadic second-order logic}
  can be seen as a syntactic fragment of
  standard \kl[monadic second-order logic]{monadic second-order~logic}.
\end{proof}

We now look at some examples of \kl{graph properties}
for which the previous corollary implies
that they lie outside the \kl{locally polynomial hierarchy}.
These include the following \kl(graph){properties},
which intuitively require counting the number of \kl{nodes}:
\begin{itemize}
  \item\AP $\intro*\PRIME$:
    the set of \kl{graphs}
    whose \kl(graph){cardinality} is a prime number.
  \item\AP $\intro*\SQUARE$:
    the set of \kl{graphs}
    whose \kl(graph){cardinality} is a perfect square.
  \item\AP $\intro*\HALFSELECTED$:
    the set of \kl{graphs}
    in which exactly half of the \kl{nodes} have the \kl{label}~$1$.
\end{itemize}
We also consider a \kl(graph){property}
that requires
mapping potentially distant parts of a \kl{graph} to each other:
\begin{itemize}
  \item\AP $\intro*\AUTOMORPHIC$:
    the set of \kl{graphs} that have a \intro{nontrivial automorphism},
    i.e., an isomorphism to themselves
    that is not the identity function.
\end{itemize}
As we will see below,
the inexpressibility of these \kl(graph){properties}
follows from classical results in automata theory.

\begin{corollary}
  \label{cor:properties-outside-hierarchy}
  There are \kl{graph properties},
  such as $\PRIME$, $\SQUARE$, $\HALFSELECTED$, and $\AUTOMORPHIC$,
  that lie outside the \kl{locally polynomial hierarchy},
  even when restricted to \kl{graphs} of \kl{bounded structural degree}.
  More precisely,
  $\Property \cap \bGRAPH{\MaxDegree} \notin
   \SigmaLP{\Level}\On{\bGRAPH{\MaxDegree}}$
  for all $\Level \in \Naturals$ and $\MaxDegree \geq 3$,
  where $\Property$~can be $\PRIME$, $\SQUARE$, $\HALFSELECTED$, or $\AUTOMORPHIC$.
\end{corollary}

\begin{proof}[Proof sketch]
  The claim follows directly from
  Corollary~\ref{cor:mso-subsumes-hierarchy-on-bounded-graphs}
  and the fact that
  $\PRIME$ and $\SQUARE$
  are not \kl{monadic second-order} \kl{definable}
  on~$\bGRAPH{2}$,
  while
  $\HALFSELECTED$ and $\AUTOMORPHIC$
  are not \kl{monadic second-order} \kl{definable}
  on~$\bGRAPH{3}$.
  These inexpressibility results
  are well known for arbitrary \kl{graphs},
  so we only sketch their proofs
  to verify that they remain valid
  when restricted to
  (\kl{connected}) \kl{graphs} of \kl{bounded structural degree},
  assuming the same \kl{structural representation} used throughout this paper.

  For each \kl(graph){property},
  its inexpressibility can be shown by
  a contradiction argument of the following form:
  Suppose the \kl(graph){property} is \kl{definable}
  by a \kl{monadic second-order} \kl{formula}~$\Formula[1]$
  on \kl{graphs} of \kl{bounded structural degree}.
  Then $\Formula[1]$
  can be converted into
  a \kl{monadic second-order} \kl{formula}~$\Formula[1]'$
  that \kl{defines} a nonregular language
  on (structural representations of) finite words.
  This is a contradiction,
  because the Büchi-Elgot-Trakhtenbrot theorem
  (see, e.g., \cite[Thm.~3.1]{DBLP:reference/hfl/Thomas97})
  states that
  a language is
  \kl{monadic second-order} \kl{definable} on finite words
  if and only if
  it is regular
  (i.e., recognizable by a finite-state automaton).
  We now give some more details for each \kl(graph){property}:
  \begin{itemize}
    \item $\PRIME$:
      If there were
      a \kl{monadic second-order} \kl{formula}~$\Formula[1]$ on~$\bGRAPH{2}$
      \kl{defining} $\PRIME \cap \bGRAPH{2}$,
      then we could convert~$\Formula[1]$ into
      a \kl{monadic second-order} \kl{formula}~$\Formula[1]'$ on finite words
      \kl{defining} the language
      $\SetBuilder{a^n}{\text{$n$ is a prime}}$.
      In $\Formula[1]'$
      we would simply interpret
      the $a$'s as (unlabeled) \kl{nodes} and
      the successor relation as an undirected \kl{edge} relation.
      However,
      by the pumping lemma for regular languages,
      this language is not regular
      (see, e.g., \mbox{\cite[\S\,4.1]{DBLP:books/daglib/0016921}}).
    \item $\SQUARE$:
      The proof for $\SQUARE$ is the same as for $\PRIME$,
      since
      $\SetBuilder{a^n}{\text{$n$ is a perfect square}}$
      is also a nonregular language.
    \item $\HALFSELECTED$:
      Similarly,
      we could convert
      a hypothetical \kl{monadic second-order} \kl{formula}
      for the \kl{graph property}
      $\HALFSELECTED \cap \bGRAPH{3}$
      into a \kl{formula} for the nonregular language
      $\SetBuilder{w \in \Set{a,b}^\KleeneStar}{\text{$w$ contains as many $a$'s as $b$'s}}$.
      But this case is a bit more tedious than the previous ones.
      The $a$'s
      (or equivalently the $b$'s)
      would have to be interpreted as
      \kl{nodes} \kl{labeled} with the string~$1$,
      so we would have to
      simulate additional \kl{elements} for the \kl{labeling bits}.
      (This is very similar to what we did in the proof of
      Lemma~\ref{lem:translation-picture-graph}.\ref{itm:backward-translation}.)
      The bound of~$3$ on the \kl{structural degree}
      accounts for these additional \kl{elements}
      (see Figure~\ref{fig:graph} and our definition of \kl{structural representations}).
    \item $\AUTOMORPHIC$:
      For the last \kl(graph){property},
      we simply observe that
      the construction sketched in \cite[Prp.~5.13]{DBLP:books/daglib/0030804}
      would allow us to convert
      a hypothetical \kl{formula} for
      $\AUTOMORPHIC \cap \bGRAPH{3}$
      into a \kl{formula} for the nonregular language
      $\SetBuilder{a^nbcd^n}{n \geq 2}$.
      \qedhere
  \end{itemize}
\end{proof}

%%% Local Variables:
%%% mode: latex
%%% TeX-master: "../lph-paper"
%%% End:

\section{Discussion}
\label{sec:discussion}

We have extended the \kl{polynomial hierarchy}
to the \textsc{local} model of distributed computing.
Some major results of complexity theory generalize well to this setting,
including Fagin's theorem and the Cook--Levin theorem.
Moreover,
we could go beyond what is known in the centralized setting
by showing that the \kl{locally polynomial hierarchy} is infinite.
Descriptive complexity was very helpful in this regard,
as it allowed us to build directly
on sophisticated results from logic and automata theory,
in particular the infiniteness of
the \kl{monadic second-order hierarchy} on \kl{pictures}.

It seems highly unlikely that this paper
will provide any new insights into
major open problems in complexity theory,
such as $\PTIME$~versus~$\NP$.
This is because our separation results rely on
the distributed nature of the \textsc{local} model.
They do not hold in cases where distributedness is irrelevant,
such as on \kl{single-node graphs},
or when \kl{certificates} can be replaced by local computation
(see Proposition~\ref{prp:bounded-graph-collapse}).
However,
our findings may provide a new perspective on
the concept of locality in distributed computing.

\subparagraph*{Measuring locality.}

Within the \textsc{local} model,
\kl{round-time} complexity is certainly the most natural and widely studied
measure of locality.
It tells us
the radius up to which each \kl{node} must see
in order to solve a given problem.
But,
as pointed out by Feuilloley~\cite[\S\,4.4]{DBLP:journals/dmtcs/Feuilloley21},
if we require the radius to be constant,
and compensate for this by introducing nondeterminism,
then \kl{certificate} size becomes a natural measure of locality.
Intuitively,
\kl{certificate} size tells us
how much global information about the \kl{graph}
each \kl{node} must receive from the prover~(\kl(certificate){Eve})
in order to \kl{verify} a given \kl(graph){property}.
Now,
if we go one step further and additionally require \kl{certificate} size
to depend only on a \kl{node}'s constant-radius \kl{neighborhood},
and compensate for this in turn by introducing quantifier alternation,
then the level of alternation arguably becomes our new measure of locality.
Its meaning is more abstract,
as it represents the number of moves in a two-player game,
but the longer the game,
the more global information the two players can prove or disprove.

Since the \kl{locally polynomial hierarchy} is infinite,
it provides, at least in principle,
a fine-grained measure of locality based on alternation.
What remains to be seen is how meaningful its different levels are,
given that the \kl(graph){properties} used to separate them
involve \kl{graph encodings} of \kl{pictures},
which make little sense from a distributed computing perspective.
In this paper,
we have seen some initial clues.
At the bottom of the \kl[locally polynomial hierarchy]{hierarchy},
the class $\PiLP{0} = \LP$
represents, by definition, purely local \kl(graph){properties}.
A canonical example of such a \kl(graph){property}
is \kl{Eulerianness},
which is $\LP$-\kl{complete}
(by Proposition~\ref{prp:eulerian-lp-complete}).
One level higher,
$\SigmaLP{1} = \NLP$ represents
\kl(graph){properties} that are almost, but not quite, local.
A canonical example of such a \kl(graph){property}
is $3$\nobreakdash-\kl{colorability},
which is $\NLP$-\kl{complete}
(by~Theorem~\ref{thm:three-colorable}).
In contrast,
the \kl{complements} of \kl{Eulerianness} and $3$-\kl{colorability}
are more global,
as neither of them lies in~$\SigmaLP{1}$
(by Corollaries~\ref{cor:non-three-colorable-not-in-nlp}
and~\ref{cor:noneulerian-hamiltonian-nonhamiltonian-not-in-nlp}).
We could only place
non-\kl{Eulerianness} in~$\SigmaLP{3}$
and non-$3$-\kl{colorability} in~$\PiLP{4}$
(by Proposition~\ref{prp:complementation}),
leaving open whether there are matching lower bounds.
The third level of the \kl[locally polynomial hierarchy]{hierarchy}
seems particularly noteworthy in that
$\SigmaLP{3}$~contains a number of natural \kl(graph){properties}
that can be expressed using spanning trees,
such as
\kl{Hamiltonicity},
\kl{acyclicity},
non-$2$-\kl{colorability},
and having an odd number of \kl{nodes}
(see Example~\ref{ex:hamiltonian} and the discussion that follows).
Again,
we leave open the exact complexity of these \kl(graph){properties},
but we have seen that
\kl{Hamiltonicity} does not lie in~$\SigmaLP{1}$
(by Corollary~\ref{cor:noneulerian-hamiltonian-nonhamiltonian-not-in-nlp}).
In addition,
we have identified some \kl{graph properties}
that lie outside the \kl[locally polynomial hierarchy]{hierarchy},
a fact that indicates their inherently global nature.
These include certain
\kl(graph){cardinality}-dependent \kl(graph){properties},
such as the number of \kl{nodes} being a prime or a perfect square,
as well as
the \kl(graph){property} of having a \kl{nontrivial automorphism}
(see~Corollary~\ref{cor:properties-outside-hierarchy}).

To gain a better intuition for
the higher levels of the \kl[locally polynomial hierarchy]{hierarchy},
the notions of \kl{hardness} and \kl{completeness}
under \kl{locally polynomial reductions}
could be helpful.
In the centralized setting,
Meyer and Stockmeyer~\cite{DBLP:conf/focs/MeyerS72}
generalized the Cook--Levin theorem
to classes of quantified Boolean formulas,
thus providing \kl{complete} problems
for all levels of the \kl{polynomial hierarchy}
(see, e.g.,~\cite[\S\,5.2.2]{DBLP:books/daglib/0023084}).
Although these problems are rather artificial in themselves,
they have been used to prove the \kl{completeness}
of more natural problems,
especially on the second and third levels of the \kl{polynomial hierarchy}
(see~\cite{Schaefer01}).
A similar strategy could be pursued in the distributed setting.
It should be straightforward
to further generalize our distributed version of the Cook--Levin theorem
(Theorem~\ref{thm:local-cook-levin})
to cover the entire \kl{locally polynomial hierarchy},
and based on that,
we may find more natural \kl{complete} \kl(graph){properties}
for higher levels of the \kl[locally polynomial hierarchy]{hierarchy}.
Given the $\NLP$-\kl{completeness} of $3$\nobreakdash-\kl{colorability},
a promising candidate would be
the generalization of $3$-\kl{colorability} to a family of multi-round games,
as defined by
Ajtai, Fagin, and Stockmeyer~\cite[\S\,11]{DBLP:journals/jcss/AjtaiFS00};
we saw one member of this family in Example~\ref{ex:3-round-3-colorable}.

While it is to be expected that many \kl{graph properties} of interest
are not \kl{complete} for any level of the \kl{locally polynomial hierarchy},
we may still be able to derive lower bounds for them
by proving their \kl{hardness}
for certain levels of the \kl[locally polynomial hierarchy]{hierarchy}.
For example,
although \kl{Hamiltonicity} is probably
not \kl{complete} for any level of the \kl{locally polynomial hierarchy},%
\footnote{%
  This is because
  from the work of
  Ajtai, Fagin, and Stockmeyer~\cite[\S\,11]{DBLP:journals/jcss/AjtaiFS00}
  we can conclude that
  each level of the \kl{locally polynomial hierarchy}
  contains a \kl{graph property}
  whose \kl{string-encoded} version is \kl{complete}
  for the corresponding level of the classical \kl{polynomial hierarchy}.
  For example,
  it is easy to see that
  $\PiLP{2}$ contains the \kl(graph){property} $\ROUNDCOLORABLE{2}{3}$,
  which holds for a given \kl{graph}~$\Graph$
  if every $3$-color assignment to the \kl{leaves} of~$\Graph$
  can be extended to a valid $3$-\kl{coloring} of~$\Graph$.
  Thus,
  if $\HAMILTONIAN$ were $\PiLP{2}$-\kl{hard},
  then by simulating a \kl{distributed Turing machine} with a centralized one,
  we could get a polynomial-time reduction
  from $\Encoding(\ROUNDCOLORABLE{2}{3})$ to $\Encoding(\HAMILTONIAN)$,
  where $\Encoding \colon \GRAPH \to \NODE$
  is some \kl{encoding} of \kl{graphs} as strings.
  But this would imply
  the collapse of the \kl{polynomial hierarchy} to~$\NP$,
  since
  $\Encoding(\ROUNDCOLORABLE{2}{3})$
  is $\PiP{2}$-\kl{complete},
  and $\Encoding(\HAMILTONIAN)$ lies in~$\NP$.
}
we were still able to show that it does not lie in~$\SigmaLP{1}$
by establishing its $\coLP$-\kl{hardness}.
More generally,
just as distributedness made it easier to separate the different levels
of the \kl[locally polynomial hierarchy]{hierarchy},
it can also make it easier
to prove unconditional lower bounds for individual \kl{graph properties}.

\subparagraph*{Beyond polynomial bounds.}

As the primary goal of this paper was to explore
the connections between
standard complexity theory and local distributed decision,
a natural starting point was to impose polynomial bounds
on the processing time and \kl{certificate} sizes of the \kl{nodes}.
This allowed us to build on classical results
using descriptive complexity
and to reinterpret
major open questions in standard complexity theory
as particularly challenging special cases of network computing.
However,
it could be argued that polynomial bounds are not the most canonical choice
if one wishes to use quantifier alternation purely as a measure of locality.
In that case,
the main concern is not to limit the individual processing power of the \kl{nodes},
but rather to restrict the \kl{certificates} in a way
that preserves the local nature of the \kl{arbitrating} algorithm.
As explained in Section~\ref{ssec:related-work},
the three alternation hierarchies based on~$\LD$
do not meet this requirement,
since they allow \kl{certificate} sizes
to depend on the entire input \kl{graph}.

\AP
It turns out that
we can generalize the \kl{locally polynomial hierarchy}
without compromising its potential as a measure of locality,
simply by replacing polynomial bounds with arbitrary bounds.
This leads us to define the class~$\intro*\LB$
(for~\emph{locally bounded time}),
which consists of all \kl{graph properties}
that can be decided by a \kl{distributed Turing machine}
operating under \kl{locally unique} \kl{identifiers}
and running in \kl{constant round time}
and arbitrary \kl{step time}
(i.e., \kl{step time} bounded by some arbitrary computable function).
Based on this,
we obtain the \intro{locally bounded hierarchy}
$\Set{\intro*\SigmaLB{\Level}\!, \intro*\PiLB{\Level}}_{\Level \in \Naturals}$,
which is defined analogously to the \kl{locally polynomial hierarchy},
except that the \kl{certificate assignments} are
$\Tuple{\CertifRadius, f}$-\kl(certificate){bounded}
for some arbitrary computable function
$f \colon \Naturals \to \Naturals$.

Most of our results carry over directly to this generalized setting.
This includes all separation results and,
consequently, the infiniteness of the hierarchy.
The reason is that our separations already hold
on \kl{graphs} of \kl{bounded structural degree},
where
the \kl[locally bounded hierarchy]{locally bounded} and
\kl{locally polynomial hierarchies}
are equivalent.
On arbitrary \kl{graphs},
the \kl{locally bounded hierarchy} even exhibits a “cleaner” structure
in that it forms a strict linear order,
while the \kl{locally polynomial hierarchy}
presumably does so only on \kl{graphs} of \kl{bounded structural degree}
(see Proposition~\ref{prp:bounded-graph-collapse}
and Remark~\ref{rem:unbounded-graph-collapse}).
Moreover,
if we generalize \kl{locally polynomial reductions}
to reductions computable in
\kl{constant round time} and arbitrary \kl{step time},
then all our \kl{hardness} and \kl{completeness} results
can be extended
to the corresponding classes of the \kl{locally bounded hierarchy}.
This even holds for our distributed version of the Cook--Levin theorem
(Theorem~\ref{thm:local-cook-levin}),
although it would have to be proved directly
instead of relying on descriptive complexity.%
\footnote{%
  A similar observation can be made in the centralized setting:
  the standard proof of the Cook--Levin theorem
  already shows how to construct a \kl{Boolean formula}
  that encodes the possible space-time diagrams
  of any given nondeterministic Turing machine
  whose running time is bounded by some known computable function.
  While not particularly useful for classical complexity theory,
  this implies, for instance,
  that \kl{Boolean satisfiability} is
  $\NEXPTIME$\nobreakdash-complete under exponential-time reductions.
  (Note, however,
  that exponential-time reductions are not closed under composition,
  and that they allow us to reduce
  any problem in $\twoEXPTIME$ to a problem in $\EXPTIME$.)
}

Indeed,
descriptive complexity is the only aspect of this paper
for which there does not seem to be a direct generalization
to the \kl{locally bounded hierarchy}.
This is quite striking,
considering how central Fagin's theorem (in its generalized form)
has been to our approach.
As a conceptual guide,
it motivated the introduction
of \kl{locally unique} \kl{identifiers}
and \kl(certificate){locally bounded} \kl{certificates}.
As a technical tool,
it first demonstrated
the robustness of the \kl{locally polynomial hierarchy},
then provided a shortcut
to the first \kl{completeness} result for $\NLP$,
and finally served as a bridge to the realm of logic and automata theory,
where we proved most of our separation results.
In a way,
we lose Fagin's theorem
when generalizing the \kl[locally polynomial hierarchy]{hierarchy},
but the insights gained from it remain fully applicable.

\subparagraph*{Toward the essence of locality.}

Having proposed quantifier alternation as a measure of locality,
the obvious next question is
how this new metric compares with existing ones.
The most direct comparison can be made with
\kl{certificate} size in the nondeterministic setting,
as this also measures the locality of \kl{graph properties}
(rather than construction problems).
Specifically,
we focus on the locally-checkable-proofs hierarchy
introduced by Göös and Suomela~\cite{DBLP:journals/toc/GoosS16},
which is the most general nondeterministic model
in the distributed-decision literature.
This hierarchy is based on a single-round game
in which \kl(certificate){Eve} plays alone but is granted more power:
she can choose arbitrarily large \kl{certificates}
that may refer to \kl{globally unique} \kl{identifiers}.
Instead of restricting \kl(certificate){Eve}'s \kl{certificates},
one measures
how large they must be
to \kl{verify} a given \kl(graph){property}.
For any function
$f \colon \Naturals \to \Naturals$,
the class~$\LCP{f}$ contains those \kl{graph properties}
that \kl(certificate){Eve} can prove to the \kl{nodes}
using \kl{certificates} of size at most~$f(n)$,
where $n$~is the total number of \kl{nodes} in the \kl{graph}.

\begin{figure}[tb]
  \centering
  \begin{tikzpicture}[%
    semithick,on grid,>=stealth',node distance=6.5ex,
    property/.style={rectangle,draw=none,minimum size=0},
    strict/.style={},
    class/.style={draw,rounded rectangle,minimum height=4ex},
    lb/.style={class,minimum width=10ex},
    lcp/.style={class,minimum width=20ex},
    ruler/.style={minimum size=0,inner sep=0,font=\footnotesize,label distance=1.5ex},
    measure/.style={},
    ]
  \def\nodeDist{6.5ex}
  \def\propertySep{2.5ex}
  \def\hierarchyDist{53ex}
  % LB hierarchy
  \node[lb] (s0) {$\LB\vphantom{\SigmaLB{0}}$};
  \node[lb] (s1) [above of=s0] {$\SigmaLB{1}$};
  \node[lb] (p2) [above of=s1] {$\PiLB{2}$};
  \node[lb] (s3) [above of=p2] {$\SigmaLB{3}$};
  \node[lb] (p4) [above of=s3] {$\PiLB{4}$};
  \node (dots) at ([yshift=3/4*\nodeDist]p4) {$\vdots$};
  \coordinate (outside) at ([yshift=2*\nodeDist]p4);
  \path[strict]
    (s3) edge (p4)
    (p2) edge (s3)
    (s1) edge (p2)
    (s0) edge (s1)
    ;
  % LCP hierarchy
  \node[lcp] (lcp0)     [right=\hierarchyDist of s0] {$\LCP{0}$};
  \node[lcp] (lcpConst) [right=\hierarchyDist of s1] {$\LCP{O(1)}$};
  \node[lcp] (lcpLog)   [right=\hierarchyDist of s3] {$\LCP{O(\log n)}$};
  \node[lcp] (lcpPoly)  [right=\hierarchyDist of outside,anchor=center] {$\LCP{\poly(n)}$};
  \path[strict]
    (lcpPoly)  edge (lcpLog)
    (lcpLog)   edge (lcpConst)
    (lcpConst) edge (lcp0)
    ;
  % Graph properties
  \node (eulerian) [property]
    at ($(s0.east)!0.5!(lcp0.west)$) {$\EULERIAN$};
  \node (3-colorable) [property]
    at ($(s1.east)!0.5!(lcpConst.west)$) {$\COLORABLE{3}$};
  \node (odd) [property]
    at ($(s3.east)!0.5!(lcpLog.west)$) {$\ODD$};
  \node (acyclic) [property]
    at ([yshift=-\propertySep]odd) {$\ACYCLIC$};
  \node (hamiltonian) [property]
    at ([yshift=-\propertySep]acyclic) {$\HAMILTONIAN$};
  \node (non-2-colorable) [property]
    at ([yshift=-\propertySep]hamiltonian) {$\NONCOLORABLE{2}$};
  \node (non-3-colorable) [property]
    at ($(p4.east)!0.5!(lcpLog.west|-p4)$) {$\NONCOLORABLE{3}$};
  \node (automorphic) [property]
    at ($(p4.east|-lcpPoly)!0.5!(lcpPoly.west)$) {$\AUTOMORPHIC$};
  \node (prime) [property]
    at ([yshift=-\propertySep]automorphic) {$\PRIME$};
  \path[->,shorten >=1ex]
    (eulerian) edge (s0)
               edge (lcp0)
    (3-colorable) edge (s1)
                  edge (lcpConst)
    (odd.east) edge (lcpLog.west)
    (acyclic.east) edge (lcpLog.184)
    (hamiltonian.east) edge (lcpLog.187)
    (non-2-colorable.east) edge (lcpLog.190)
    (non-3-colorable.east) edge[out=0,in=180] (lcpPoly.186)
    (automorphic) edge (lcpPoly)
    (prime.east) edge[out=0,in=180] (lcpLog.175)
    ;
  \path[->,shorten >=1ex,densely dashed]
    (odd.west) edge (s3.east)
    (acyclic.west) edge (s3.-9)
    (hamiltonian.west) edge (s3.-17)
    (non-2-colorable.west) edge (s3.-24)
    (non-3-colorable) edge (p4)
    ;
  \path[black!30,line width=0.6ex]
    ([xshift=-6ex]$(prime-|dots)!0.55!(dots)$) edge ([xshift=-5ex]$(prime)!0.55!(dots-|prime)$);
  % LB measure of locality
  \node (r4) [ruler,left=3ex of p4.west,label={[ruler]left:4}] {\large$-$};
  \node (r3) [ruler,label={[ruler]left:3}] at (s3-|r4) {\large$-$};
  \node (r2) [ruler,label={[ruler]left:2}] at (p2-|r4) {\large$-$};
  \node (r1) [ruler,label={[ruler]left:1}] at (s1-|r4) {\large$-$};
  \node (r0) [ruler,label={[ruler]left:0}] at (s0-|r4) {\large$-$};
  \path[ruler] ([yshift=-1.5ex]r0.center) edge ([yshift=3ex]r4.center);
  % LCP measure of locality
  \node (global) [measure,anchor=west] at ([xshift=2ex]lcpPoly.east) {global};
  \node (local)  [measure,anchor=west] at ([xshift=2ex]lcp0.east) {local};
  \path[measure,->] (local) edge (global.south-|local);
\end{tikzpicture}

%%% Local Variables:
%%% mode: latex
%%% TeX-master: "../lph-paper"
%%% End:
  \caption{
    \emph{(repeated from Figure~\ref{fig:lb-vs-lcp-overview})}
    An example-based comparison between
    the \kl{locally bounded hierarchy}
    $\Set{\SigmaLB{\Level}\!, \PiLB{\Level}}_{\Level \in \Naturals}$
    introduced in this paper
    and the locally-checkable-proofs hierarchy
    $\Set{\LCP{f}}_{f:\, \Naturals \to \Naturals}$
    introduced by Göös and Suomela~\cite{DBLP:journals/toc/GoosS16}.
    (The interpretation of the latter as a measure of locality
    is due to Feuilloley~\cite[\S\,4.4]{DBLP:journals/dmtcs/Feuilloley21}.)
    A vertical line between two complexity classes
    indicates that
    the lower class is strictly included in the higher one.
    An arrow from a \kl{graph property} to a complexity class
    indicates that
    the \kl(graph){property} is contained in the class.
    If the arrow is solid,
    there is also a matching lower bound,
    i.e., the \kl(graph){property} is not contained in any lower class.
    \kl(graph){Properties} above the thick horizontal line
    lie outside the \kl{locally bounded hierarchy}.
  }
  \label{fig:lb-vs-lcp}
\end{figure}

In Figure~\ref{fig:lb-vs-lcp},
we informally compare the \kl{locally bounded hierarchy}
with the locally-checkable-proofs hierarchy
by mapping various \kl{graph properties} onto both scales.
The two hierarchies roughly align
on four levels of locality:
\begin{itemize}
  \item \emph{Purely local}:
    \kl(graph){properties} in $\LB$ and $\LCP{0}$.
    The correspondence between these two classes is not surprising,
    given their very similar definitions.
  \item \emph{Almost local}:
    \kl(graph){properties} in $\SigmaLB{1}$ and $\LCP{O(1)}$.
    This alignment is also no coincidence,
    as both classes impose similar restrictions
    on \kl(certificate){Eve}'s \kl{certificates}
    (locally \kl(certificate){bounded} for the former,
    and constant size for the latter);
    on \kl{graphs} of \kl{bounded structural degree},
    these constraints are equivalent.
  \item \emph{Intermediate}:
    \kl(graph){properties} in $\SigmaLB{3}$ and $\LCP{O(\log n)}$.
    It is more surprising that
    many natural \kl{graph properties} fall into
    these two seemingly unrelated classes.
    The key reason is that
    both allow \kl(certificate){Eve}
    to provide a verifiable spanning tree,
    which can be leveraged in a variety of ways.
  \item \emph{Inherently global}:
    \kl(graph){properties}
    outside $\bigcup_{\Level \in \Naturals} \SigmaLB{\Level}$
    but within $\LCP{\poly(n)}$.
    As mentioned in Section~\ref{ssec:background} 
    (see~“\nameref{par:nondeterminism}”),
    quadratic-size \kl{certificates}
    with \kl{globally unique} \kl{identifiers}
    suffice to \kl{verify} all \kl{graph properties},
    including the inherently global ones.
\end{itemize}

However,
Figure~\ref{fig:lb-vs-lcp} also reveals
that the two hierarchies
are not perfectly aligned.
For instance,
having a prime number of \kl{nodes}
is an inherently global \kl(graph){property} in the \kl{locally bounded hierarchy},
but only an intermediate \kl(graph){property}
in the locally-checkable-proofs hierarchy.
This is because
any \kl(graph){cardinality}-dependent \kl(graph){property}
can be verified using a spanning tree
when logarithmic-size \kl{certificates} are allowed.
Conversely,
non-$3$-\kl{colorability} ranks relatively low
in the \kl{locally bounded hierarchy},
where \kl{complementation} costs at most three alternations,
but this \kl(graph){property} is nearly maximally global
in the locally-checkable-proofs hierarchy.
(Göös and Suomela established a near-tight lower bound of
$\Omega(n^2 / \log n)$;
see~\cite[Thm.~6.4]{DBLP:journals/toc/GoosS16}).

These divergences highlight a fundamental issue:
there is no precise, universally accepted definition of locality%
---the very concept we seek to measure. 
As suggested in Section~\ref{ssec:background},
locality relates to the amount of information
that a single \kl{node} must obtain about the rest of the \kl{graph}
to determine a given \kl(graph){property}.
But this quantity inevitably depends on
how the information is represented and interpreted,
and thus on the underlying formalism.
Each formalism may have idiosyncrasies
that make certain tasks difficult,
such as counting in the \kl{locally bounded hierarchy},
or \kl{complementation} in the locally-checkable-proofs hierarchy.
Therefore,
perhaps the most productive path forward is
to develop even more measures of locality
and compare them systematically.
The aspects on which they all agree
are most likely to capture the true essence of locality.

%%% Local Variables:
%%% mode: latex
%%% TeX-master: "../lph-paper"
%%% End:

\bibliography{references}

\end{document}